\documentclass{article}
\usepackage[english]{babel}

\usepackage[letterpaper,top=2cm,bottom=2cm,left=3cm,right=3cm,marginparwidth=1.75cm]{geometry}

\usepackage{amsmath}
\usepackage{amssymb}
\usepackage{amsthm}
\usepackage{authblk}
\usepackage{bbm}
\usepackage{cancel}
\usepackage[sorting=none,citestyle=numeric-comp]{biblatex} 
\usepackage{booktabs}
\usepackage{tikz}
\usetikzlibrary{quantikz2}
\usetikzlibrary{shapes.geometric, shapes.symbols, shapes.misc}
\usepackage[compat=0.6]{yquant}
\usepackage{multirow}
\usepackage{pdflscape} 
\usepackage{afterpage}
\usepackage{titletoc}
\usepackage{tocbibind}
\usepackage[page,header]{appendix}

\usepackage{rotating} 

\usepackage{todonotes}  

\usetikzlibrary{patterns}
\usepackage[normalem]{ulem}
\usepackage[shortlabels]{enumitem}
\usepackage{float}
\usepackage{graphicx}
\usepackage[colorlinks=true, allcolors=blue]{hyperref}
\usepackage{orcidlink}
\usepackage{subcaption}
\useyquantlanguage{groups}
\usepackage{siunitx}
\usepackage[ruled,vlined]{algorithm2e}

\SetKwInput{KwRegisters}{Registers}
\SetKw{Apply}{apply}
\SetKw{Compute}{compute}
\SetKw{Uncompute}{uncompute}

\DeclareMathOperator*{\argmax}{arg\,max}

\makeatletter
\newcommand*{\bdiv}{%
  \nonscript\mskip-\medmuskip\mkern5mu%
  \mathbin{\operator@font div}\penalty900\mkern5mu%
  \nonscript\mskip-\medmuskip
}
\makeatother

\newtheorem{theorem}{Theorem}[section] 
\newtheorem{lemma}[theorem]{Lemma} 
\newtheorem{corollary}[theorem]{Corollary} 
\newtheorem{definition}[theorem]{Definition} 

\newcommand{\lowQuantumVolumeEst}{10^{21}}
\newcommand{\highQuantumVolumeEst}{10^{39}}
\newcommand{\lowTgateCountEst}{10^{15}}
\newcommand{\highTgateCountEst}{10^{18}}

\title{Detailed assessment of calculating drag force with quantum computers: Explicit time-evolution precludes exponential advantage for nonlinear differential equations}

\begin{document}
\author[1]{John Penuel} 
\author[2,3]{Amara Katabarwa} 
\author[3]{Peter D. Johnson} 
\author[4]{Parker Kuklinski} 
\author[4]{Benjamin Rempfer} 
\author[3]{Collin Farquhar} 
\author[3]{Yudong Cao\thanks{ycao@zapataquantum.com}} 
\author[1]{Michael C. Garrett\thanks{mike.garrett@l3harris.com}} 

\affil[1]{L3Harris Technologies, Inc., Palm Bay, FL, USA}
\affil[2]{Georgia Tech Research Institute, Atlanta, GA, USA}
\affil[3]{Zapata Quantum, Boston, MA, USA}
\affil[4]{MIT Lincoln Laboratory, Lexington, MA, USA}

\maketitle

\begin{abstract}
\noindent This study examines the potential for fault-tolerant quantum computers to provide utility in fluid dynamics simulations, with a focus on drag force calculations for ship hull design. We assess whether quantum algorithms can surpass classical computational limits by generating detailed quantum resource estimates (QREs) in terms of logical qubits and $T$-gate counts. Our analysis is based on a quantum algorithm leveraging Carleman linearization of the lattice Boltzmann method (LBM), which has been suggested to offer exponential speedup. We develop efficient block encodings for LBM matrices and a method for amplitude-encoding drag force. We apply the method to the simple case of fluid flow past a sphere across a range of Reynolds numbers ($\mathrm{Re}$). We estimate the required (logical qubits)$\times$($T$-gates), finding them to be prohibitively large, ranging from $\lowQuantumVolumeEst$ to $\highQuantumVolumeEst$. While classical simulations scale as $O(\mathrm{Re}^3)$, our QREs exhibit a modest polynomial scaling of $O(\mathrm{Re}^{2.69})$, indicating no exponential quantum advantage. We attribute this limitation to an intrinsic power-law relationship between spatial grid resolution and time-stepping requirements that is a fundamental characteristic of explicit methods for evolving nonlinear differential equations. Thus, quantum computers are unlikely to provide utility in applications that require time-evolving fluids and other systems of nonlinear differential equations.
\end{abstract}

\pagebreak
\tableofcontents 

\pagebreak
\section{Introduction} \label{sec:intro}

Quantum computing has made significant progress in the past few years.  Exciting experimental developments in \cite{Hu2019, bluvstein2023, Silva2024, acharya2023, anderson2021} have implemented quantum error correction (QEC), enabling the use of logical qubits.  Theoretical progress has also kept pace.  For example, the resource cost of magic state distillation (an important subroutine in QEC) has been drastically reduced \cite{litinski2019,gidney2024a}. The discovery of quantum low-density parity-check codes is also significant due to their improved encoding rate over the dominant surface codes.  For quantum algorithms there has also been a grand unification with the discovery of the quantum singular value transformation \cite{Gilyen2019long}, whereby varied algorithms like quantum search, matrix inversion and phase estimation have been brought under the same umbrella \cite{Martyn2021}. All of these developments imply the era of fault-tolerant quantum computing is on the horizon.  

Despite these advances, the future utility of full-scale fault-tolerant quantum computers for real-world applications is not clear. As different parts of the quantum stack improve and fault-tolerant quantum computing is matured, a thorough assessment of the potential utility and the size and performance of the quantum platform required to unlock utility becomes an important guide for research and development. This drives the need for ``computational benchmark problems'' that evaluate the problem-solving capabilities of full-stack computing systems, enabling comparisons of utility between disparate architectures, quantum and classical alike \cite{proctor2024a}.

Given
\begin{enumerate}
    \item estimates for return on investment for a quantum processing unit (QPU) that is able to solve some commercially interesting application, \label{item:roi}
    \item quantum \textit{resource estimates} (e.g., estimated number of logical/physical qubits, non-Clifford gates, and other performance characteristics) that the QPU would need to solve the application, and \label{item:resource-est}
    \item estimated cost and schedule to develop said QPU, \label{item:cost}
\end{enumerate}
an investor would have all the information necessary to determine if funding the development of such a QPU is worthwhile.  The goal of this work is to follow this thread for isothermal applications in incompressible computational fluid dynamics (CFD) and provide quantified estimates for \ref{item:roi} and \ref{item:resource-est},  thus empowering the investor when a QPU developer proposes \ref{item:cost}.  

This paper is one in a series of studies from the US Defense Advanced Research Projects Agency (DARPA) Quantum Benchmarking (QB) program \cite{darpa_qb_website}. Rather than beginning with a specific quantum algorithm and contriving amenable problems, the QB program began with the bulk collection of classically hard, full-scale problems of significant utility to industry, defense, and/or academia. Other quantum-amenable, high-utility applications are being considered in separate QB program publications \cite{QBothers}.
Each includes estimates of the utility to be gained by a faster and/or more accurate solver, and the quantum \textit{resource estimates} (QREs) that a QPU would need.

CFD is among the earliest applications of classical computing, and yet persists as one of its most challenging despite roughly a century of significant advancements. In particular, the direct numerical simulation (DNS) of turbulent flow (high Reynolds numbers, $\mathrm{Re}\sim 10^4$ and above) typically demands computational resources that exceed current supercomputing capabilities by many orders of magnitude. This is due to the need to resolve spatial features down to the Kolmogorov scale, at which the smallest eddies dissipate, and so the required number of spatial grid points grows as $\sim \mathrm{Re}^{9/4}$ \cite{Rogallo1984}. CFD practitioners must instead rely on coarse-grained methods such as Reynolds-averaged Navier-Stokes (RANS), which reduce the computational burden by relying on approximate turbulence models, albeit at the cost of reduced accuracy.

CFD applications can be categorized into compressible or incompressible regimes. This work focuses on the incompressible or weakly compressible regime, wherein solid obstructions are traveling at low speeds relative to the speed of sound in the medium (i.e., Mach $0.3$ or lower) and the density of the medium remains roughly constant throughout the fluid domain. Incompressible CFD can be used to study industrial mixing, blood flow, weather patterns, marine engineering tasks, and a wide variety of other problems.  As an example of an industrial-scale problem, we consider the simulation-driven design of ship hulls \cite{siemens_sim_design}.  We consider the problem of calculating the drag force due to steady-state flow on three specific hull designs in section \ref{sec:problem_instances}. This entails time-evolving a uniform initial state to a steady state, and subsequently calculating the net drag force over the surface of the hull.  As a presently accessible waypoint, we address the more geometrically simple problem of drag force on a sphere.  While the flow-past-a-sphere problem is not commercially interesting, it provides a parameterized verification instance to ensure a compute platform can adequately solve geometrically simple instances before proceeding to industrial-scale applications.

Potential quantum advantage in incompressible CFD stems from the ability to solve a linear system of equations with runtime scaling logarithmically in the number of variables \cite{Harrow2009}. This was used by \cite{Berry2014} to realize an exponential speedup in solving linear differential equations.  Later advancements have improved the dependence on parameters such as precision and condition number and extended this approach to weakly nonlinear differential equations \cite{Berry2017, Liu2021, Krovi2023}.  The requirement of sufficiently weak nonlinearity motivates our focus on the incompressible regime. There are presently three basic approaches to solving linear differential equations. The first -- and focus of this work -- is through solving a large system of linear equations based on the underlying differential equation \cite{jennings2023a, jennings2023b}.
The second method uses a time-marching approach \cite{Fang2023}.  The third method leverages methods of quantum simulation and using a linear combination of Hamiltonian simulations \cite{An2023, dong2023optimal}.

Our goal in this paper is to quantify the utility of being able to complete CFD simulations faster and more accurately on a fault-tolerant QPU, and to estimate the quantum hardware resources required. Inspired by \cite{Li2025, Sanavio_2024, bakker2024quantum, Kumar2024, wawrzyniak2024a}, we built upon the work of \cite{jennings2023b} investigating the use of quantum computers for solving the nonlinear differential equation within the lattice Boltzmann method (LBM) formulation for incompressible CFD.  This is the first attempt at fully costing out an end-to-end workflow including the costs of classical data input and readout.  We pay special attention to the following three caveats associated with this approach (e.g., see \cite{dalzell2023a} and references therein):
\begin{enumerate}
    \item Data loading: The cost of loading classical data into a QPU threatens any potential speedup. The problems considered in this work have uniform initial conditions and geometrically simple parametric boundary conditions.  Loading more sophisticated geometries and initial conditions will \textit{increase} quantum resource estimates.
    \item Carleman linearization: This procedure requires that the nonlinearity must be sufficiently weak \cite{Liu2021, Krovi2023}. In the LBM, this scales with Mach number. The problems considered in this work are therefore restricted to the incompressible regime below Mach 0.3.
    \item Data readout: The inefficiency of data readout from a QPU (e.g., quantum tomography) threatens any potential speedup. The problems considered in this work focus on extracting only scalar functions of the solution state, with drag force chosen as the example. 
\end{enumerate}

As a waypoint toward assessing the feasibility of using quantum computers to improve ship hull design, we estimate the quantum resources required for the simpler case of drag force on a sphere. We estimate the product of (logical qubits)$\times$($T$-gates) to range from $\lowQuantumVolumeEst$ to $\highQuantumVolumeEst$, for Reynolds numbers ranging from $10^1$ to $10^8$. These estimates already include a $\sim 10^5\times$ reduction attained from bespoke block-encodings of the LBM collision matrices. These lofty QREs suggest a pessimistic outlook on the potential of future quantum computers to provide utility in CFD applications. We note, however, that quantum algorithm development and resource estimation for differential equation applications is relatively nascent, as compared with more mature applications such as ground state energy estimation (GSEE) in quantum chemical systems:  Comparing early GSEE resource estimates from \cite{Wecker2014} to later improved resource estimates in \cite{Lee2021} and more recently \cite{caesura2025a}, we see the estimated number of $T$-gates reduced by as much as $\sim 10^9\times$.  Important resource estimation contributions for differential equation applications were recently made by \cite{jennings2023b}, which gives the number of calls to the block encoding of the matrix that defines the linear differential equation. This drew upon their work that does detailed resource estimation for the matrix inversion method \cite{jennings2023a}.  

Of even greater concern is how our lofty QREs scale with Reynolds number (the impatient reader should skip ahead to Sec.~\ref{subsec:qres}). We observed a $\sim O(\mathrm{Re}^{2.69})$ power-law scaling. Hence, we found no exponential advantage with respect to Reynolds number, but rather a weak (sub-quadratic) polynomial advantage relative to the $O(\mathrm{Re}^{3})$ scaling for classical DNS. We attribute the lack of exponential advantage to the intrinsic relationship between temporal and spatial resolution imposed by explicit methods for time-evolving systems of differential equations. The Courant-Freidrichs-Lewy (CFL) condition for numerical convergence requires that physical effects do not propagate across more than one spatial grid spacing within the duration of a single time step. Consequently, the number of time steps has an implicit polynomial dependence on the number of grid points. As pointed out in \cite{Li2025} (based upon \cite{Berry2017}), the theoretically predicted exponential advantage (complexity scaling polylogarithmically in the number of spatial grid points) is negated if the simulation time (effectively the number of time steps) has implicit polynomial dependence on the number of spatial grid points. Our observed power-law scaling was recently confirmed by an independent but similar analysis \cite{jennings2025a}.

It is our hope that this work will serve as an example of how the prospective utility of quantum computers in applications of differential equations can be thoroughly assessed as a step toward application-driven benchmarking. Our detailed results show that current quantum algorithms will not provide significant utility in explicitly time-evolving fluids or other systems of nonlinear differential equations. Recent improvements beyond the quantum approach we considered (e.g., \cite{jennings2025a,pocrnic2025constantfactorimprovementsquantumalgorithms}) remain limited to weak polynomial advantages. Either fundamental algorithmic improvements are required to circumvent the linear dependence on the number of time steps, or the candidate applications must be restricted to those wherein the number of time steps does not depend polynomially (or worse) on the number of spatial grid points. For example, implicit methods commonly allow for a relaxation of the CFL condition. For calculations of steady-state properties, it may be worth exploring alternative methods to directly solve for the steady state rather than time-evolving from an initial condition.  Utility may also be found in applications of differential equations that permit superquadratic, albeit subexponential, speedups, such as simulating coupled classical oscillators \cite{Babbush2023a}.

\subsection{Definitions and organization}

We define the following terms, which are common to the set of DARPA QB program publications.

\begin{definition}[Application Instance]\label{def:application_instance}
A general type of computational task to be performed.  For example: solve a system of linear equations.
\end{definition}

\begin{definition}[Problem Instance]\label{def:problem_instance}
A specific realization of an application instance.  For example: solve a system of linear equations $Ax=b$ where the data for $A$, $b$ and all other parameters and constraints are provided.
\end{definition}

\begin{definition}[Performance Metric]\label{def:performance_metric}
A measure that quantifies the ability of a compute device and algorithm to perform a computation.  For example: wall-clock run time in units of seconds.
\end{definition}

\begin{definition}[Utility Threshold]\label{def:utility_threshold}
Specific values of a performance metric that correspond to the capability required to solve a problem instance.  The units of a utility threshold are the same units as the associated performance metric.  For example: solve a system of linear equations $Ax=b$ (where $A$, $b$ and all other data is provided) within a wall-clock run time of 600 seconds.
\end{definition}

\begin{definition}[Utility Estimate]
The quantitative benefit to a user if a utility threshold is surpassed.  The default unit of utility is US dollars.  For example: if a compute device and algorithm can solve the system of linear equations $Ax=b$ (where all data is provided) within a wall-clock run time of 600 seconds, then the user may expect estimated benefit of \$Z US dollars.
\end{definition}

This paper is organized as follows. In section \ref{sec:classical_approach}, we summarize current classical approaches to our \textit{application instance} of simulating steady-state incompressible fluid flow to calculate drag force on ship hulls. We define the relevant \textit{performance metrics} and identify \textit{utility thresholds} tied to current performance limitations. In section \ref{sec:utility}, we attempt to quantify the economic benefit of improved computational performance, with a focus on \textit{utility estimates} associated with ship hull design. In section \ref{sec:problem_instances}, we provide detailed formulations of \textit{problem instances}: flow past a sphere at various Reynolds numbers and three utility-scale instances of flow past model-scale and full-scale ship hulls. In section \ref{sec:quantum_workflow}, we describe our quantum approach, based on Carleman linearization of the LBM, and present a method for extracting drag force from the solution state. In section \ref{sec:quantum-resource-estimates}, we describe our approach to quantum resource estimation and present our results, including arguments about the use of QPUs and quantum algorithms for solving CFD and other general systems of differential equations. Concluding remarks are in section \ref{sec:conclusion}.  

A significant amount of the LBM formulation and quantum algorithm implementation details are relegated to appendices.  Of particular note is appendix \ref{sec:quantum_compilation}, which describes the quantum circuit, the development of bespoke LBM-specific block encodings, and comparison of bespoke block encodings with unstructured block encodings.

\section{Classical approach}
\label{sec:classical_approach}

The goal for the CFD simulations is to ``solve'' the nonlinear Navier-Stokes equations that govern fluid flow to time-evolve a uniform initial state to the steady-state velocity vector field and scalar pressure field, and then compute some measurable against these fields.  The typical classical computing approach is to first discretize the spatial domain into a set of non-overlapping polytope cells known as a mesh. The velocity and pressure fields are then calculated at the cell centers from the resulting set of difference equations over the meshed space.  The finite difference method uses uniformly-sized polytopes, while the finite volume method (FVM) may use polytopes of variable size and shape that are tailored to the domain geometry.

The parameter that drives up computational cost is the Reynolds number defined in equation \eqref{eq:Re}.  The Reynolds number (Re) is the ratio of inertial to viscous forces in a fluid.  At low Re (below $\sim 10^3$) the flow is laminar, while at higher Re ($\sim 10^4$ and above)  the flow is turbulent with chaotic changes in flow velocity and pressure.  To properly resolve features of turbulence at higher Re, the polytopes of the mesh need to be exceedingly small, which results in a larger set of difference equations to solve.  The length scale necessary to resolve the smallest turbulent features is known as the Kolmogorov length scale. Direct numerical simulation (DNS) down to the Kolmogorov scale typically incurs prohibitive computational costs, which scale as $O(\mathrm{Re}^3)$ and limit current supercomputers to $\mathrm{Re}\sim 10^4$ \cite{hoyas2022a}. The lattice Boltzmann method (LBM) (section \ref{sec:lbm-formulation}) provides an alternative DNS approach with various advantages and disadvantages compared to Navier-Stokes solvers \cite{Krueger2017TheLB}, though classical LBM computational costs are generally no less exorbitant. 

A common concession is to utilize Reynolds-averaged Navier-Stokes (RANS) solvers, which use a mesh cell size much larger than the Kolmogorov scale, but employ various averaging models to account for turbulent velocity and pressure features.  This greatly reduces computational cost, albeit at reduced accuracy. To provide greater accuracy than RANS without resorting to full DNS, Large Eddy Simulation (LES) lessens the reliance on approximate turbulence models by using mesh sizes that resolve some of the larger-scale turbulent flow, though the commensurate increase in computational cost is still often prohibitive.

\subsection{Ship design workflow}

Traditional ship design workflows rely heavily on experimental testing. This primarily consists of constructing scale models and measuring key performance indicators such as hull resistance, trim, and sinkage, in a towing tank facility. Due to the reduced scale and inability to reproduce open-ocean conditions, towing tank data is limited in accuracy. A single design workflow may require the testing of multiple design variants, incurring significant time and resource costs (see section \ref{sec:utility}) and therefore restricting the number of candidate designs \cite{wheeler2021a}.

Decades of advances in classical CFD have led to the advent of simulation-driven ship design \cite{siemens_sim_design} as a means to both lessen the burden of experimental testing and improve the performance of resulting designs. This engineering process emphasizes extensive simulation of the vessel early in the design cycle.  This has the potential to (1) improve/optimize designs by thorough exploration of the parameter space, (2) give a high confidence that the design will meet requirements and regulations as built, and (3) perform requirements sensitivity analysis to identify requirements that are driving a design into a particular (possibly suboptimal) configuration.  The cost for extensive simulation is computing resources, which translates to wall-clock time, and may start to drive completion cost or schedule.  Given enough time, current software and hardware can complete the required simulations, albeit with accuracy concessions.

For this study we consider drag force due to obstruction geometry as the exemplar key performance indicator to be output by the solver. This is one of several key performance indicators relevant to ship hull design.  We propose additional outputs in section \ref{sec:conclusion} for future work.  After the solver has calculated the pressure field $p$ and velocity flow field $\mathbf{u}$, the magnitude of the drag force $\mathcal{F}$, assumed to be in the direction of the bulk free-stream velocity unit vector $\hat{\mathbf{u}}_\mathrm{bulk}$, is calculated by numerically integrating the expression:
\begin{align}
    \mathcal{F}
    &=
    \int_{\mathcal{S}} d{\mathcal{S}} (\hat{\mathbf{n}} \cdot \hat{\mathbf{u}}_\mathrm{bulk}) p
    +
    \int_{\mathcal{S}} d{\mathcal{S}} (\hat{\boldsymbol{\tau}} \cdot \hat{\mathbf{u}}_\mathrm{bulk}) \tau_\mathrm{w}.
\end{align}
The first term is the pressure/form drag and the second term is skin/friction drag. The integral/summation runs over the surface of the solid object ${\mathcal{S}}$. $\hat{\mathbf{n}}$ is the unit vector normal to surface $d{\mathcal{S}}$ and $\hat{\boldsymbol{\tau}}$ is the unit vector in the direction of shear stress parallel to the surface $d{\mathcal{S}}$.  The wall shear stress $\tau_\mathrm{w}$ may be calculated using finite-difference methods to approximate the expression
\begin{align}
    \tau_\mathrm{w} = \rho\nu \left. \frac{\partial u}{\partial h}\right|_{h=0},
\end{align}
where $h$ is distance from the surface along $\hat{\mathbf{n}}$, $\rho$ is the density of the fluid, $\nu$ is the kinematic viscosity of the fluid, and
\begin{align}
    \hat{\boldsymbol{\tau}} = \lim_{h\rightarrow 0} \frac{\mathbf{u}}{u},
\end{align}
noting that formally $u=0$ on the surface of the obstruction.

\subsection{Utility thresholds} \label{sec:thresholds}

In order to unlock some portion of the utility estimates in section \ref{sec:utility} below, a future CFD platform must exceed the utility thresholds associated with a particular problem instance. Utility thresholds are tied to the performance of current CFD platforms, which makes them moving targets.  As future algorithmic and hardware improvements result in increased performance, they likewise improve the utility thresholds. Therefore, to unlock a significant portion of the utility estimates, a future CFD platform must exceed utility thresholds by one or more orders of magnitude. Incremental performance gains yield only incremental utility.

We quantify performance according to two main performance metrics: runtime and accuracy. Runtime is highly problem-dependent, and is specified for each problem instance in section \ref{sec:problem_instances} for a RANS-based solver. Accuracy varies according to method, being highly dependent on the level of spatial resolution and reliance on turbulence models. Though accuracy is also problem-dependent, the typical uncertainty range for an output of interest calculated via RANS, LES, and DNS are \cite{xiao2019a,verma2023a}:

\begin{itemize}
\item RANS: Uncertainty between 5\% and 20\%, with 10\% being an optimistic estimate. Usually RANS models need to be tuned mercilessly for the specific problem being considered, and then a 10\% error would be considered a good outcome \cite{verma2023a}. There are many sources of uncertainty in RANS, and elaborate methods for quantifying them \cite{xiao2019a}.
\item LES: Uncertainty between 1\% and 5\%, depending on the problem at hand, and with extremely high computational cost. Uncertainty is non-trivial to estimate \cite{xiao2019a} and also highly dependent on the subgrid model being used. The most commonly used models (e.g., Smagorinsky) are not universally applicable to all problems, and may need to be modified appropriately, which is still a topic of basic research \cite{verma2023a}. 
\item DNS: Uncertainty less than 1\%; we use 0.1\% for the purpose of defining an accuracy threshold. DNS accuracy is limited by machine precision and error associated with numerical integration techniques (e.g. convergence threshold, time step size, etc.). DNS remains classically intractable for most utility-scale problems. 
\end{itemize}

Though the quantum approach we describe in section \ref{sec:quantum_workflow} is based on the LBM (a form of DNS), we do not compare against classical LBM simulations as these are rarely used in ship hull design. Indeed, DNS of any form is classically intractable or at least considered prohibitively expensive. We therefore base our utility thresholds on the more commonly used RANS solvers, still widely regarded as the standard workhorse for simulation-driven ship design and other utility-scale CFD applications. As LBM-based simulations will significantly exceed RANS accuracy, we define the following two utility thresholds on runtime:

\begin{enumerate}
    \item Incremental Utility Threshold: DNS-level accuracy at runtimes shorter than current classical DNS.
    \item Transformative Utility Threshold: DNS-level accuracy at runtimes comparable to current classical RANS.
\end{enumerate}

An LBM-based quantum solver running 10x faster than current classical DNS would only unlock incremental utility since this would still amount to impractically long runtimes for most applications. We do not provide runtime estimates for current classical DNS of our considered problem instances. In utility-scale problems such as the instances involving ship hulls, the Reynolds numbers of $10^8$--$10^9$ equate to a runtime increase of at least $\sim10^{12}\times$ over current world-record DNS of turbulent flows (e.g., $\mathrm{Re}\approx 10^4$ reported in \cite{hoyas2022a}). In this regard, completion of an LBM-based simulation with any measurable runtime (e.g. one year or less) for a problem instance with $\mathrm{Re}\ge 10^5$ would satisfy this utility threshold.

An LBM-based quantum solver running at speeds comparable to current classical RANS would unlock transformative utility, as described and estimated in section \ref{sec:utility}. We provide runtime estimates for classical RANS-based simulations of the flow-past-a-sphere problem instances in section \ref{sec:problem_instances}.  We implemented the problem instances on ANSYS\textregistered\; Fluent\textregistered\; CFD software on an engineering-grade laptop for a conservative estimate.  This establishes the modest utility threshold that a quantum platform would have to surpass in order to provide transformative economic benefit or return on investment. The true spirit of the DARPA QB program is to advance scientific computing via any algorithm or platform: quantum, classical, hybrid, or other.  Thus the units for the performance metrics of interest and utility thresholds are software and hardware agnostic.  

By quantifying the utility of future quantum solvers according to their performance relative to utility thresholds based on the limitations of classical solvers, we are attempting to compare classical solvers to quantum solvers. However, each paradigm is using a completely different computational workflow. There are alternative choices to be made for each paradigm and workflow.  We felt it was important to only compare the most suitable known methods for each compute platform.
\section{Utility estimation}
\label{sec:utility}

The overarching goal of the DARPA QB program is to clarify the future utility of full-scale fault-tolerant quantum computers. While incompressible CFD, and differential equation solving more generally, is relatively nascent in terms of quantum algorithm development, we chose to focus on this application instance because the associated utility is unquestionably large. CFD was among the earliest applications in classical computing, and is ubiquitous in modern science and engineering. The incompressible regime encompasses all liquid-phase and low-velocity (below Mach 0.3) gas-phase scenarios, and thus accounts for a significant fraction of CFD applications. In this section we attempt to estimate the utility, in US dollars of economic benefit, associated with improved computational performance in incompressible CFD and for ship hull design in particular.

\subsection{Utility per user group}

The \textit{utility} or economic benefit of an improved compute platform will be realized in different ways to different user groups.  We categorize the utility estimate into the the following three user groups: compute platform developers, engineering services users, and derived product end users.  A compute platform developer produces the actual compute hardware or software that is used and stands to gain economic utility from the sale or lease of an improved platform.  The engineering services user purchases or subscribes to the compute platform and uses the platform to design derived products (e.g., ship hulls).  The engineering services user stands to gain economic utility from the sale of an improved product design or potentially a more economical design workflow.  The final end user purchases the improved product from the engineering services user and stands to gain economic utility because (1) the product may have an improved design, and/or (2) the product may have been produced at a reduced cost due to some efficiencies realized by the improved compute platform.

In Appendix \ref{apx:utility} we provide estimates (or bounds on estimates) for all the facets of the utility associated with meeting or exceeding our \emph{transformative} utility threshold: DNS-level accuracy at runtimes comparable to current classical RANS. We arrive at the following order-of-magnitude estimates for each user group:
\begin{itemize}
    \item Software/compute market utility of \$10M--\$100M for the application of simulation-driven ship design, and $\sim$\$1B as a weak upper bound for applications of incompressible CFD more broadly.
    \item Engineering services market utility of \$100M--\$1B for the application of simulation-driven ship design (upper bound assumes towing tank tests are rendered obsolete), and $\sim$\$100B as a weak upper bound for applications of incompressible CFD more broadly.
    \item Derived product end-user utility of $\sim$\$10B as a weak lower bound associated with savings on fuel and other operational costs associated with worldwide cargo shipping.
\end{itemize}

\subsection{Incremental utility}
\label{sec:incremental_utility}

The utility gained by exceeding the \emph{incremental} utility threshold (DNS-level accuracy at runtimes shorter than current classical DNS) stems from improving our fundamental scientific understanding of turbulence. In this regard, even simple verification problem instances such as flow past a sphere can offer some utility at Reynolds numbers beyond the reach of classical DNS.  This provides downstream utility by informing better turbulence models in approximate CFD methods such as RANS and LES. Improved accuracy in RANS and LES translates to utility gain for all user groups described above: CFD software developers, users of CFD software, and users of CFD-derived products.

Whereas meeting or exceeding the \emph{transformational} utility threshold (replace RANS with DNS) results in reducing uncertainty from $\sim$10\% to $\sim$0.1\%, let us suppose that meeting or exceeding the \emph{incremental} utility threshold ultimately leads to a reduction in RANS uncertainty from $\sim$10\% to $\sim$9\% via turbulence model improvement. We therefore suggest that the utility associated with the incremental threshold is $10\times$ to $100\times$ smaller than the utility associated with the transformational threshold. This is consistent with numerical estimates of $\sim$\$1B (total for all user groups) we calculated using a parametric utility estimation tool \cite{JobTBP} initially developed for estimating the utility of simulating the Fermi-Hubbard model with quantum computers \cite{agrawal2024a}.
\section{Problem instances}\label{sec:problem_instances}

In this work, we consider problems instances with various Reynolds numbers from laminar to turbulent regimes.  As we show later, problem instances with higher Reynolds numbers drive up the spatial and temporal discretization fidelity for both the classical workflow and the quantum workflow.

A classical CFD solver typically outputs the entire time-evolved fluid velocity vector field and other scalar fields. The practitioner then has full access to the fields and may call a variety of post-processing analytics to calculate key design performance indicators, produce interesting charts, images, and animations. A similar calculation on a QPU would prepare a quantum state proportional to the entire velocity vector field.  However, a full readout of the field would destroy any computational advantage due to the prohibitive cost of quantum tomography. Thus, it is not sufficient for a quantum platform to merely prepare a state representation of the field; the platform must then calculate some scalar measurable of interest against the quantum state and read out the measurable.  This is a significant downside of the quantum platform (though it has been suggested that quantum algorithms for data processing and machine learning may enable more efficient output analysis \cite{kiani2022a}).  Instead of calculating the velocity field once and running a variety of post processing analytics, a quantum platform must re-solve the velocity field for each noncommuting measurable scalar quantity of interest. In the problem instances formulated below, we choose the magnitude of the drag force exerted on the object as the exemplar scalar output.

Loading a large amount of input data into a QPU also threatens to destroy computational advantage.  Thus, a QPU cannot accommodate CFD problems with elaborate initial conditions and complicated flow geometry.  The QPU would prefer to handle CFD problems where geometry, boundary, and initial conditions are uniform or otherwise have a compact representation that is independent of problem size. We therefore choose a spherical geometry for simple verification instances (section \ref{sec:sphere_instances}), and note that the geometries in utility-scale ship hull instances (section \ref{sec:shiphulls}) may be represented by a discrete set of polytopes. We choose a uniform fluid velocity vector field as initial conditions.

Table \ref{tab:problem_instances_overview} provides a list of the problem instances analyzed.
\begin{table}[ht]
\centering
\small
\begin{tabular}{cll}
\toprule
\textbf{No.} & \textbf{Study/Design Description} & \textbf{Details} \\
\midrule
1 & Flow past a sphere - $\mathrm{Re}=10^1$ & Laminar flow (section \ref{sec:sphere_instances}) \\
2 & Flow past a sphere - $\mathrm{Re}=10^2$ & Laminar flow (section \ref{sec:sphere_instances}) \\
3 & Flow past a sphere - $\mathrm{Re}=10^3$ & Transition to turbulent (section \ref{sec:sphere_instances}) \\
4 & Flow past a sphere - $\mathrm{Re}=10^4$ & Transition to turbulent (section \ref{sec:sphere_instances}) \\
5 & Flow past a sphere - $\mathrm{Re}=10^5$ & Turbulent flow (section \ref{sec:sphere_instances}) \\
6 & Flow past a sphere - $\mathrm{Re}=10^6$ & Turbulent flow (section \ref{sec:sphere_instances}) \\
7 & Flow past a sphere - $\mathrm{Re}=10^7$ & Turbulent flow (section \ref{sec:sphere_instances}) \\
8 & Flow past a sphere - $\mathrm{Re}=10^8$ & Turbulent flow (section \ref{sec:sphere_instances}) \\
\midrule
9 & Model-scale ship hull design:  & Turbulent flow (section \ref{sec:jbc}) \\
  & \quad \quad Japan Bulk Carrier (JBC) - $\mathrm{Re}\sim 10^7$ & \\ 
\midrule
10 & Model-scale ship hull design: & Turbulent flow (section \ref{sec:kcs}) \\
   & \quad \quad KRISO container ship (KCS) - $\mathrm{Re}\sim 10^7$ & \\ 
\midrule
11 & Full-scale ship hull design: & Turbulent flow (section \ref{sec:mv-regal}) \\
  & \quad \quad MV Regal - $\mathrm{Re}\sim 10^9$ & \\
\bottomrule
\end{tabular}
\caption{Problem instances analyzed, where Re is the Reynolds number.}
\label{tab:problem_instances_overview}
\end{table}

\subsection{Flow past a sphere} \label{sec:sphere_instances}

We consider flow past a sphere with eight parameter regimes summarized in table \ref{tab:sphere_physical_parameters}. For our implementation, we set the computational domain (in meters) to 
\begin{equation}
    \{  (x,y,z)\in [-5, 5]\times[-4 , 4]\times[-4, 4] \}.
\end{equation}
The spherical obstruction is centered at $(0,0,0)$ and has a diameter equal to 1 meter. The fluid flow travels in the positive $x$-direction with ambient velocity $\mathbf{u} = (u_x, 0 , 0)$.  The Reynolds number (Re) is defined as
\begin{equation}\label{eq:Re} 
    \text{Re} = \frac{u_x L}{\nu},
\end{equation}
where $L$ is the characteristic length and $\nu$ is the kinematic viscosity.  In the flow-past-a-sphere cases, $L$ is the sphere diameter.  In the ship hull simulations, $L$ is the hull length.  The relationship can be rewritten to calculate the ambient velocity $u_x$ as a function of a desired Reynolds number as
\begin{equation}\label{eq:velocity_as_function_of_Re}
u_x  = \frac{\nu \text{Re}}{L}.
\end{equation}
For the desired set of Reynolds numbers $\{10^1, 10^2, ...,10^8\}$, the maximum velocity $u_x$ is about 100 \si{m/s} (corresponding to the $\mathrm{Re}=10^8$ case).   This velocity remains an order of magnitude below the speed of sound in water, ensuring all parameter sets are well within the incompressible regime (cavitation effects notwithstanding).

\begin{table}[ht]
\centering
\label{tab:fluid-properties}
\begin{tabular}{cl}
\toprule
Reynolds Number Re & Velocity $u_x$ (\si{m/s}) \\
\midrule
\(10^1\) & \(1.003 \times 10^{-5}\) \\
\(10^2\) & \(1.003 \times 10^{-4}\) \\
\(10^3\) & \(1.003 \times 10^{-3}\) \\
\(10^4\) & \(1.003 \times 10^{-2}\) \\
\(10^5\) & \(1.003 \times 10^{-1}\) \\
\(10^6\) & \(1.003\) \\
\(10^7\) & \(10.03\) \\
\(10^8\) & \(100.3\) \\
\bottomrule
\end{tabular}
\caption{Parameters for flow past a sphere instances, with kinematic viscosity $\nu = 1.003 \times 10^{-6}\;\mathrm{m}^2/\mathrm{s}$ and density $\rho = 998.21\;\mathrm{kg}/\mathrm{m}^3$.}
    \label{tab:sphere_physical_parameters}
\end{table}

The geometric description above is sufficient for any practitioner to implement the simulation; no additional files are necessary. The benchmark practitioner is asked to calculate the measurable:  drag force on the sphere in the $x$-direction. Since the benchmark practitioner is only required to report the measurable of interest, the benchmark practitioner may choose an alternative appropriate computational domain. In other words, the computational domain specified above and in all instances is not a required parameter of the problem instance. 

As a conservative comparison to modern classical CFD solvers, we implemented the simulations in ANSYS\textregistered\; Fluent\textregistered\; 2023R. The simulations were run on an engineering-grade laptop with an Intel\textregistered\; Xeon \textregistered\; W-11855M CPU @ 3.20GHz and 96GB.  ANSYS solves for steady state using Reynolds-averaged Navier Stokes (RANS) methods rather than time evolving, though this is not appropriate for high Reynolds number instances with transient vortex shedding.  The results in table \ref{table:sphere_runtimes} show that even for the simple flow-past-a-sphere problem, a significant amount of computational power is required and increases rapidly with the Reynolds number.
\begin{table}[h]
\centering
\begin{tabular}{ccc}
\toprule
\textbf{Reynolds Number} & \textbf{Mesh Generation Time} & \textbf{Run Time} \\
\midrule
$10^3$ & 6 minutes & 1 hour \\  
$10^4$ & 8 minutes & 1 hour \\
$10^5$ & 20 minutes & 3 hours \\
$10^6$ & 4 hours & 3 hours \\
$10^7$ & 5 hours & 5 hours \\
$10^8$ & Not enough memory. & \\
\bottomrule
\end{tabular}
\caption{Simulation run times for different Reynolds numbers.}
\label{table:sphere_runtimes}
\end{table}

The runtime results shown in table \ref{table:sphere_runtimes} define the utility thresholds on runtime that any other hardware/software platform would have to surpass in order to be considered a viable alternative.  

\subsection{Ship hull design} \label{sec:shiphulls}

In this section we present the parameters for three utility-scale problem instances using publicly available ship hull CAD files.

\subsubsection{Model-scale ship hull design:  Japan Bulk Carrier (JBC)}\label{sec:jbc}

The Tokyo 2015 Workshop on CFD in Ship Hydrodynamics \cite{tokyo2015}, provided three ``model scale'' ship CAD files, test case parameters, and empirical results on the models.  We have chosen two of the three model ships to analyze:  JBC (7 meters long) and KCS (approximately 7 meters long).  

We have also chosen to deviate from some of the test case parameters and measurables to be consistent across all of our test cases and consistent with our target measurable:  drag force.  The workshop requires calculation of total resistance,  as well as conditions of self-propulsion, and other cases that we are not considering.  We feel it is still worthwhile to include the problem instances with modified parameters because the CAD designs are well-studied by practitioners and existing simulations may be easily modified to act as competitive benchmark performers.

The problem instance parameters are:

\begin{tikzpicture}
    \node [draw, rounded corners, fill=blue!10, inner sep=10pt, outer sep=0, text width=15cm] (box) {
        \begin{minipage}{.95\textwidth}
            \begin{itemize}
                \item Fluid is water at 20 \(^\circ\)C with all parameters consistent with table \ref{tab:sphere_physical_parameters}.
                \item JBC CAD drawing and measurements available at \cite{tokyo2015} (specifically \url{https://www.t2015.nmri.go.jp/jbc_gc.html}).  The CAD drawing does not include rudder or propeller.
                \item Velocity \(u_x = 1.179 \, \si{m/s}\) per Case 1.1a. ($\text{Re}=8.23\times 10^6$)
                \item Draft of the (model) ship hull is 0.4125 \si{m}.
                \item Projected frontal area of the submerged ship hull is \(0.46319 \, \text{m}^2\).
                \item Computational domain (in meters) in our implementation is $\{(x, y, z) \in [-7, 7] \times [-2, 2] \times [-1, 0.4125]\}$
                with the center of the ship hull at $(0, 0, 0)$. The benchmark practitioner may choose an alternative
appropriate domain.
            \end{itemize}
        \end{minipage}
    };
\end{tikzpicture}

\subsubsection{Model-scale ship hull design:  KRISO Container Ship (KCS)}\label{sec:kcs}

The KRISO Container Ship (KCS) is the second model-scale ship test case published in \cite{tokyo2015}.  The problem instance parameters we choose are:

\begin{tikzpicture}
    \node [draw, rounded corners, fill=blue!10, inner sep=10pt, outer sep=0, text width=15cm] (box) {
        \begin{minipage}{.95\textwidth}
            \begin{itemize}
    \item Fluid is water at $20^{\circ}$ C with all parameters consistent with table \ref{tab:sphere_physical_parameters}.
    \item KCS CAD drawing and measurements available at \cite{tokyo2015} (specifically \url{https://www.t2015.nmri.go.jp/kcs_gc.html}).  The drawing should include the rudder, but not the propeller.
    \item Velocity $u_x =  0.915~\si{m/s}$ per \hyperlink{hhttps://www.t2015.nmri.go.jp/Instructions_KCS/Case_2.1/Case_2-1.html}{Case 2.1-1}. ($\text{Re}=6.39\times 10^6$) 
    \item Draft of the (model) ship hull is 0.34178 \si{m} (scaled down from the full-scale ship draft at 10.8 \si{m}).
    \item Projected frontal area of the submerged ship hull is 0.34301 $\si{m^2}$ (calculated by projecting drawing onto a plane and scaling to model size).
    \item Computational domain (in meters) in our implementation is $\{  (x,y,z)\in [-10, 4]\times[-2 , 2]\times[-1, 0.34178] \}$ with the base of the rudder of the ship at $(0,0,0)$.  The benchmark practitioner may choose an alternative appropriate domain.
            \end{itemize}
        \end{minipage}
    };
\end{tikzpicture}

\subsubsection{Full-scale ship hull design:  MV Regal}\label{sec:mv-regal}

For the final ship hull design problem instance, we use the test case studied in the 2016 Workshop on Ship Scale Hydrodynamic Computer Simulation published by Lloyd's Register \cite{lr-8428}.  The work describes the sea trials of the MV Regal ship, experimental measurements during sea trials and compares those measurements to various calculated measurements from CFD solvers.  They used a 3D laser scan geometry of the as-built vessel while in dry dock.  The CAD file of the geometry is available at \url{www.jores.net} \cite{www.jores.net}.  This is a full-scale problem instance with a ship hull length of 138 meters.  The problem instance parameters are:

\begin{tikzpicture}
    \node [draw, rounded corners, fill=blue!10, inner sep=10pt, outer sep=0, text width=15cm] (box) {
        \begin{minipage}{.95\textwidth}
            \begin{itemize}
    \item Fluid is water at $20^{\circ}$ C with all parameters consistent with table \ref{tab:sphere_physical_parameters}.
    \item MV Regal CAD drawing available at \cite{www.jores.net}.  The drawing should include the rudder, but not the propeller.
    \item Velocity $u_x =  7.2022 ~\si{m/s}$ (14 knots) per Case 1.4. ($\text{Re}=9.91\times 10^8$)
    \item Draft of the ship hull during sea trails varied forward to aft, but we use a consistent draft of 5.25 \si{m}.
    \item Projected frontal area of the submerged ship hull is 120.055 $\si{m^2}$ (calculated from design drawings).  
    \item Computational domain (in meters) in our implementation is $\{  (x,y,z)\in [-100, 200]\times[-35 , 35]\times[-10, 5.25] \}$ with the base of the rudder of the ship at $(0,0,0)$.  The benchmark practitioner may choose an alternative appropriate domain.
            \end{itemize}
        \end{minipage}
    };
\end{tikzpicture}

\section{Quantum approach}
\label{sec:quantum_workflow}

\subsection{Overview}

\begin{figure}[H]
    \centering
    \tikzstyle{startstop} = [rectangle, rounded corners, minimum width=3cm, minimum height=1cm,text centered, draw=black, fill=white!30]
    \tikzstyle{process} = [rectangle, rounded corners, minimum width=3cm, minimum height=1cm, text centered, draw=black, fill=white!30]
    \tikzstyle{decision} = [diamond, minimum width=3cm, minimum height=1cm, text centered, draw=black, fill=green!30]
    \tikzstyle{arrow} = [thick,->,>=stealth]

\begin{tikzpicture}[node distance=1cm]

\node (start) [startstop] {Problem Instance};
\node (proc1) [process, below left=of start] {Finite Volume Mesh};
\node (proc2) [process, below=of proc1] {Pressure-based Solver};
\node (proc3) [process, below=of proc2] {Calculate drag force};
\node (proc4) [process, below right=of start] {LBM Grid};
\node (proc5) [process, below=of proc4] {CL of LBM Dynamics};
\node (proc6) [process, below=of proc5] {Quantum Solver};
\node (proc7) [process, below=of proc6] {Drag force measurable};

\draw [arrow] (start) -- (proc1);
\draw [arrow] (proc1) -- (proc2);
\draw [arrow] (proc2) -- (proc3);
\draw [arrow] (start) -- (proc4);
\draw [arrow] (proc4) -- (proc5);
\draw [arrow] (proc5) -- (proc6);
\draw [arrow] (proc6) -- (proc7);
\end{tikzpicture}

    \caption{Top-level classical workflow (left path) vs. quantum workflow (right path).  The workflows are well-suited to each platform.}
    \label{fig:tikz_flow_chart_arrows}
\end{figure}

The QPU workflow we consider is contrasted with the classical workflow in figure \ref{fig:tikz_flow_chart_arrows}, and is based on Li \textit{et al.} \cite{Li2025} with enhancements including data input costs and measuring a scalar quantity of interest (drag force). Quantum mechanics is a linear theory, and leading quantum approaches to solving nonlinear differential equations rely on linearization procedures such as Carleman linearization (CL) to yield a larger set of linear differential equations \cite{Liu2021, Krovi2023}.  The need to reproduce the nonlinear dynamics with a linear unitary theory, coupled with the inefficiency of loading classical data (e.g., mesh geometry), creates a significantly different workflow for QPUs compared to classical computers.

Carleman linearization of nonlinear differential equations requires that the combined strength of any nonlinear terms be sufficiently small relative to the linear dissipative terms \cite{Liu2021, Krovi2023}, in order to ensure an acceptable truncation error.  This motivates the use of LBM over Navier-Stokes: The strength of nonlinearity in the Navier-Stokes equations scales with the Reynolds number and therefore precludes the simulation of turbulent flows (that have large Reynolds numbers).  In the LBM formulation, the nonlinearity instead scales with the Mach number.  Our focus is on the incompressible or weakly compressible CFD regime, which has a small Mach number, but may have a large Reynolds number (turbulent).  

Quantum algorithms for solving linear differential equations can be combined with the Carleman-linearized LBM to prepare a quantum state encoding the evolution of the phase space density of the fluid over time. The drag force on the obstruction can then be calculated as a linear combination of the encoded phase space density at some evolved time corresponding to steady-state flow \cite{Ladd1993NumericalSO, Krueger2017TheLB}. Given we have the drag force represented as an inner product of two vectors, if we are able to prepare quantum states with amplitudes proportional to these two vectors (with proportionality constants known), 
then a quantum algorithm for amplitude estimation such as iterative quantum amplitude estimation \cite{grinko2021iterative} may be used to estimate the drag force.

In the following subsections we will review the essential components of this quantum workflow for drag estimation, illustrated in figure \ref{fig:tikz_flow_chart_arrows}:
\begin{enumerate}
    \item Section \ref{sec:lbm-formulation}: Convert the CFD problem instance into a lattice Boltzmann method (LBM) formulation and approximate the LBM collision operator with a cubic polynomial per \cite{Li2025}.
    \item Section \ref{sec:cl-lbm}: Perform Carleman linearization of the continuous dynamics of the approximated LBM formulation per \cite{Li2025, bakker2024quantum, itani2023quantum}.
    \item Section \ref{sec:quantum-lin-DE-solver}: Use the quantum algorithms from \cite{Berry2017, jennings2023b} to solve the linear system of differential equations (i.e., evolve a quantum state that is proportional to the velocity vector field \emph{with history}).\label{step:quant-diff-eq-alg}
    \item Section \ref{sec:quantum_approach_to_estimating_drag_force}: Calculate drag force of the obstruction via amplitude estimation of an inner product.\label{step:drag-measureable}
\end{enumerate}

\subsection{The lattice Boltzmann method (LBM)}\label{sec:lbm-formulation}

\begin{table}
\centering
\begin{tabular}{|c|l|}
\hline
\bf{Term} & \bf{Description} \\ 
\hline
$\mathbb{R}_+$ & the set of positive real numbers.\\
\hline
$\mathbb{Z}_+$ & the set of positive integers.\\
\hline
$\mathbb{Z}_p$ & $\mathbb{Z}_p = \{0,1,2,...,p-1\}$ the set of nonnegative integers modulo $p$. \\
\hline
$D \in \mathbb{Z}_+$ & number of spatial dimension (we focus on $D=3$).\\
\hline
$Q \in \mathbb{Z}_+$ & number of lattice velocity vectors (we focus on $Q=27$). \\
\hline
D1Q3 & lattice topology in dimension $D=1$, with $Q=3$ lattice velocity vectors. \\
\hline
D2Q9 & lattice topology in dimension $D=2$, with $Q=9$ lattice velocity vectors. \\
\hline
D3Q27 & lattice topology in dimension $D=3$, with $Q=27$ lattice velocity vectors.  \\
      & See table \ref{tab:D3Q27-constants} for a complete list of constants.\\
\hline
$\Delta t \in \mathbb{R}_+$ & time step size. [units: s] \\ 
\hline
$\Delta x \in \mathbb{R}_+$ & spatial grid node size. [units: \si{m}]\\
\hline
$t^\star \in \mathbb{R}_+$ & physical time. [units: \si{s}]\\
\hline
$t \in \mathbb{Z}_+$ & LBM time in steps. $t=t^\star/\Delta t$\\
\hline
$n_x, n_y, n_z \in \mathbb{Z}_+$ & number of spatial grid nodes in the $x, y,$ and $z$ dimensions, respectively.\\
\hline 
$n \in \mathbb{Z}_+$ & total number of spatial grid nodes. $n = n_xn_yn_z$ \\
\hline
$ \mathbf{x} \in \mathbb{Z}_{n_x} \times \mathbb{Z}_{n_y} \times \mathbb{Z}_{n_z}$ & a spatial grid node position in 3 dimensional space. \\
& We use Greek letter subscripts for different nodes. E.g., $\mathbf{x}_\alpha, \mathbf{x}_\beta$, etc.\\
\hline
$\mathbf{X} \subset \mathbb{Z}^D_+$ & the set of grid nodes that form the domain under study. 
 $|\mathbf{X}|=n$.\\
\hline
$\mathbf{C} = \{ +1, -1, 0 \}^3$ & the set of D3Q27 lattice velocity vectors (see table \ref{tab:D3Q27-constants}). $|\mathbf{C}|=Q$ \\
\hline
$i \in \mathbb{Z}_Q$ & typically used as the index of the lattice velocity vector.\\
\hline
$\mathbf{c}_i \in \mathbf{C}$ & the $i^{\text{th}}$ lattice velocity vector. \\
\hline
$f_i(\mathbf{x},t) \in \mathbb{R}_+ \cup \{ 0 \}$ & distribution function of the number of fictive particles within volume $\Delta x^3$ \\ & at spatial grid node $\mathbf{x}$ and time step $t$, traveling with velocity vector $\mathbf{c}_i$. \\ 
& Our objective is to accurately and efficiently calculate $f_i(\mathbf{x},t)$, and \\ 
& relevant functions (or \emph{measurables}) thereof.\\
\hline
$f^{\text{eq}}_i(\mathbf{x},t) \in \mathbb{R}$ & the equilibrium distribution at position $\mathbf{x}$ with velocity vector $\mathbf{c}_i$ at time $t$.  \\
& The distribution functions $f_i(\mathbf{x},t)$ tend towards $f^{\text{eq}}_i(\mathbf{x},t)$ in time (see \eqref{eq:lbm-feq}). \\
\hline
$\tilde{f}^{\text{eq}}_i(\mathbf{x},t)  \in \mathbb{R}$ & a cubic polynomial approximation to $f^{\text{eq}}_i(\mathbf{x},t)$ derived in \eqref{eq:lbm-feq-approx}.  \\
\hline
$w_i \in \mathbb{R}_+$ & weight corresponding to velocity vector $\mathbf{c}_i$ in the equilibrium distribution. \\
& We note that $w_i \ge 0$ and $\sum_i w_i = 1$ (see table \ref{tab:D3Q27-constants}).\\
\hline
$\tau \in \mathbb{R}_+$ & a lattice topology parameter that sets the relaxation time towards the \\ 
& equilibrium distribution $f^{\text{eq}}_i(\mathbf{x},t)$. \\
& Equation \eqref{eq:lbm_param_consistency_7.14} establishes a consistency relationship with $\tau$, physical \\
& kinematic viscosity $\nu$, and other parameters.\\
\hline
$\nu \in \mathbb{R}_+$ & physical property of kinematic viscosity related to $\tau$ in \eqref{eq:lbm_param_consistency_7.14}. [units: \si{m^2/s}]\\
\hline
$\rho(\mathbf{x},t) \in \mathbb{R}_+$ & fluid density at position $\mathbf{x}$ at time $t$.  See \eqref{eq:lbm-rho}.\\
\hline
$\rho  \in \mathbb{R}_+$ & a shorthand notation for $\rho(\mathbf{x},t)$.\\
\hline
$\mathbf{u}(\mathbf{x},t) \in \mathbb{R}^D$ & the mean velocity of fictive particles at position $\mathbf{x}$ at time $t$. See \eqref{eq:lbm-u-1}.\\
\hline
$\mathbf{u} \in \mathbb{R}^D$ & a shorthand notation for $\mathbf{u}(\mathbf{x},t)$.\\
\hline
$c_s  \in \mathbb{R}_+$ & the lattice speed of sound (for the D3Q27 lattice, $c_s = 1/ \sqrt{3}$). \\
\hline
$\Omega ( f_i (\mathbf{x},t)  )$ & the BGK collision operator.  See \eqref{eq:lbm-bgk}.\\
\hline
\end{tabular}
\caption{Terms used in the LBM. Except for those quantities with units explicitly specified (e.g., ``physical'' quantities), all other quantities are in lattice units (i.e., combinations of $\Delta x$ and $\Delta t$) or otherwise dimensionless.}
\label{tab:lbm-params}
\end{table}

The lattice Boltzmann method (LBM) discretizes the computational domain into a regular grid of spatial nodes $\mathbf{x} \in \mathbf{X}$, where $|\mathbf{X}|=n$, and seeks to compute the distribution functions $f_i(\mathbf{x},t)$ of fluid particles within the volume $\Delta x^3$ centered at spatial grid node $\mathbf{x}$ at time step $t$, traveling on lattice velocity vector $\mathbf{c}_i \in \mathbf{C}$. The set of velocity vectors $\mathbf{C}$ is a small discrete set and $|\mathbf{C}|=Q$.  When using the LBM, lattice grid steps and time steps are scaled to 1.  Table \ref{tab:lbm-params} provides a quick reference for all LBM notation.

Fluid density at lattice node $\mathbf{x}$ and time step $t$ is $\rho(\mathbf{x},t)$ and defined as
\begin{equation}
\rho(\mathbf{x},t) = \sum_i f_i(\mathbf{x},t).
\label{eq:lbm-rho}
\end{equation}

The macroscopic fluid velocity $\mathbf{u}(\mathbf{x},t)$ is
\begin{equation}
\mathbf{u}(\mathbf{x},t) = \frac{1}{\rho}\sum_i f_i(\mathbf{x},t) \mathbf{c}_i,
\label{eq:lbm-u-1}
\end{equation}
or equivalently
\begin{equation}
\rho \mathbf{u}(\mathbf{x},t)=\sum_i f_i(\mathbf{x},t) \mathbf{c}_i,
\label{eq:lbm-u-2}
\end{equation}
where $\rho$ is shortened notation for $\rho(\mathbf{x},t)$.

When implementing the LBM, we scale physical fluid density $\rho$ to $1$ throughout the computational domain.  In the incompressible or weakly compressible case, $\rho$ minimally deviates from $\rho \approx 1$.  To handle the $\frac{1}{\rho}$ term in \eqref{eq:lbm-u-1} and later in \eqref{eq:lbm-feq-rearrange}, we will use a  first-order Taylor series expansion about the expansion point $\bar{\rho}=1$ \cite{Li2025}.  This approximation is not used for the classical execution of the LBM, but it will help us later for the Carleman Linearization procedure in section \ref{sec:cl-lbm}.
\begin{equation}
\frac{1}{\rho} \approx 2 - \rho, \mbox{ when } | 1- \rho | < \varepsilon_\rho \mbox{ for } 0 < \varepsilon_\rho \ll 1.
\label{eq:lbm-approx-inv-rho}
\end{equation}

The LBM repeats two operations: (1) collision and (2) streaming until a stopping condition is met. The collision operation describes the rate at which the distribution functions $f_i(\mathbf{x},t)$ locally converge to an equilibrium distribution $f^{\text{eq}}_i(\mathbf{x},t)$.  For the collision operation, we use the common Bhatnagar-Gross-Krook (BGK) \cite{bgk_1954} operator:
\begin{equation}
\Omega \left( f_i(\mathbf{x},t) \right) = -\frac{1}{\tau}\left(f_i(\mathbf{x},t) - f^{\text{eq}}_i(\mathbf{x},t)  \right),
\label{eq:lbm-bgk}
\end{equation}
where $\tau$ is a relaxation parameter that sets the rate that distribution functions $f_i(\mathbf{x},t)$ tend towards the equilibrium distribution  $f^{\text{eq}}_i(\mathbf{x},t)$.  The $\tau$ parameter is related to the physical fluid kinematic viscosity $\nu$ (see table \ref{tab:lbm-params} for a list of all parameters).  The relationship between $\tau$ and viscosity and other parameters is discussed in section \ref{sec:lbm-params-for-instances} and \eqref{eq:lbm_param_consistency_7.14}.  The equilibrium distribution $f^{\text{eq}}_i(\mathbf{x},t)$ is defined as
\begin{subequations}\label{eq:lbm-feq}
\begin{align}
f^{\text{eq}}_i(\mathbf{x},t) &= 
w_i \rho \left(1 + \frac{\mathbf{u} \cdot \mathbf{c}_i}{c_s^2} + \frac{\left(\mathbf{u} \cdot \mathbf{c}_i\right)^2}{2c_s^4} + \frac{\mathbf{u}\cdot\mathbf{u}}{2c_s^2} \right) 
\label{eq:lbm-feq-c_s}\\
&= 
w_i \rho \left( (1) + (3) \mathbf{u} \cdot \mathbf{c}_i + \left(\frac{9}{2}\right)\left(\mathbf{u} \cdot \mathbf{c}_i\right)^2 + \left( \frac{-3}{2} \right) \mathbf{u}\cdot\mathbf{u} \right)
\label{eq:lbm-feq-coeffs}\\
&= 
w_i \left( (1) \rho + (3)\rho \mathbf{u} \cdot \mathbf{c}_i  + \left(\frac{9}{2}\right)\left( \frac{1}{\rho} \right)\left(\rho \mathbf{u} \cdot \mathbf{c}_i\right)^2 + \left( \frac{-3}{2} \right) \left( \frac{1}{\rho}\right) \rho \mathbf{u}\cdot \rho \mathbf{u} \right) \label{eq:lbm-feq-rearrange}\\
&\approx
\underbrace{w_i \left( (1) \rho + (3)\rho \mathbf{u} \cdot \mathbf{c}_i  + \left(\frac{9}{2}\right)(2-\rho)\left(\rho \mathbf{u} \cdot \mathbf{c}_i\right)^2 + \left( \frac{-3}{2} \right)  (2 - \rho) \rho \mathbf{u}\cdot \rho \mathbf{u}\right)}_{\tilde{f}^{\text{eq}}_i(\mathbf{x},t)},
\label{eq:lbm-feq-approx}
\end{align}
\end{subequations}
where $\rho$ and $\mathbf{u}$ are shorthand notations for $\rho(\mathbf{x},t)$ and $\mathbf{u}(\mathbf{x},t)$, respectively.  The equilibrium distribution $f^{\text{eq}}_i(\mathbf{x},t)$ in \eqref{eq:lbm-feq-c_s} is the Taylor expansion of the Maxwell distribution to quadratic degree with appropriate weights $w_i$ chosen for specific lattice topologies \cite{Li2025}.  The expansion enables the LBM to reproduce the Navier-Stokes equations up to an error of $O(\mbox{Ma}^2)$ where Ma is the mach number \cite{Li2025}.  When the lattice speed of sound $c_s$ is $1/\sqrt{3}$, we obtain \eqref{eq:lbm-feq-coeffs}.  Equation \eqref{eq:lbm-feq-rearrange} isolates the $\frac{1}{\rho}$ terms and equation \eqref{eq:lbm-feq-approx} substitutes the approximation \eqref{eq:lbm-approx-inv-rho} for the $\frac{1}{\rho}$ terms.  When \eqref{eq:lbm-rho} is substituted for $\rho$ and \eqref{eq:lbm-u-2} is substituted for $\rho \mathbf{u}$, then \eqref{eq:lbm-feq-approx} is a cubic polynomial in terms of the distribution functions $f_i(\mathbf{x},t)$.  Finally, $\tilde{f}^{\text{eq}}_i(\mathbf{x},t)$ is a cubic polynomial approximation to $f^{\text{eq}}_i(\mathbf{x},t)$.  In section \ref{sec:cl-lbm}, we incorporate the polynomial approximation \eqref{eq:lbm-feq-approx} into $\frac{\partial}{\partial t} f_i(\mathbf{x},t)$ and describe the Carleman linearization process.

Circling back: the LBM repeats two steps, (1) collision and (2) streaming. The collision step is:
\begin{subequations}
\begin{align}
f^{\text{collision}}_i(\mathbf{x},t) &:= f_i(\mathbf{x}, t) - \Omega\left( f_i(\mathbf{x},t) \right)\\
&:= f_i(\mathbf{x}, t) -\frac{1}{\tau}\left(f_i(\mathbf{x},t) - f^{\text{eq}}_i(\mathbf{x},t) \right).
\end{align}
\label{eq:lbm-collision}
\end{subequations}
During the streaming operation, the post-collision population $f^{\text{collision}}_i(\mathbf{x},t)$ at node $\mathbf{x}$ is streamed to its neighboring node at $\mathbf{x}+\mathbf{c}_i$, where $\mathbf{c}_i$ is the $i^{\text{th}}$ lattice velocity vector (figure \ref{fig:streaming-picture-1}). The streaming step is
\begin{align}
f_i(\mathbf{x}+\mathbf{c}_i,t+1) &:= f^{\text{collision}}_i(\mathbf{x}, t).
\label{eq:lbm-simple-streaming}
\end{align}

The streaming step described in \eqref{eq:lbm-simple-streaming} is more complicated when boundary conditions are imposed and the particle populations are interacting with inlets, outlets, and hard surfaces.  For this work, we implement a ``halfway bounce-back'' method described in \cite{Krueger2017TheLB}.  In the spatial domain $\mathbf{X}$, some grid point $\mathbf{x}$ may be a solid node or a fluid node.  All nodes are centered inside of a cubic voxel with side length $\Delta x$.  Since the voxels are centered on the grid nodes, the walls of the voxels are exactly halfway between grid nodes while the boundary surface between fluid nodes and solid nodes is in general staircase-like.  Consider some population $f_j(\mathbf{x}+\mathbf{c}_i,t)$ at a fluid node $\mathbf{x}+\mathbf{c}_i$ that is streaming in the direction $\mathbf{c}_j = -\mathbf{c}_i$ towards its solid neighbor node at $\mathbf{x}$.  The population $f_j(\mathbf{x}+\mathbf{c}_i,t)$ streaming in the direction $\mathbf{c}_j$ is bounced back in the opposite direction $\mathbf{c}_i = -\mathbf{c}_j$ and the population at $f_i(\mathbf{x}+\mathbf{c}_i,t)$ is replaced by $f_j(\mathbf{x}+\mathbf{c}_i,t)$ in the following time step.  See the bounce back scenario in figure \ref{fig:bb-picture-1}.

\begin{figure}[h]
    \centering
    \begin{tikzpicture}
        
        \draw (6,0) 
        node[
            circle, 
            draw, 
            fill=white, 
            inner sep=3pt, 
            minimum size=5mm, 
            label=below:{$\mathbf{x}$},
            label=above:{Fluid}
        ] 
        (F1) 
        {};
        
        \draw (12,0)
        node[
            circle, 
            draw, 
            fill=white, 
            inner sep=3pt, 
            minimum size=5mm, 
            label=below:{$\mathbf{x} + \mathbf{c}_{i}$},
            label=above:{Fluid}
        ]
        (F2)
        {};
        
        \draw[->] (F1) -- node[above] {$\mathbf{c}_{i}$} (9,0);
        \draw[->] (F2) -- node[above] {$\mathbf{c}_{i}$} (15,0);
    \end{tikzpicture}   
    \caption{Free streaming. The population $f_i(\mathbf{x},t)$ (left node) becomes the population $f_i(\mathbf{x}+\mathbf{c}_i,t+1)$ (right node).}
    \label{fig:streaming-picture-1}
\end{figure}

\begin{figure}[h]
\centering
\begin{tikzpicture}

\draw (6,0) 
node[
    circle, 
    draw, 
    fill=black, 
    inner sep=3pt, 
    minimum size=5mm, 
    label=below:{$\mathbf{x}$},
    label=above:{Solid}
] 
(S) 
{};

\draw (12,0)
node[
    circle, 
    draw, 
    fill=white, 
    inner sep=3pt, 
    minimum size=5mm, 
    label=below:{$\mathbf{x} + \mathbf{c}_{i}$},
    label=above:{Fluid}
]
(F)
{};

\draw[->] (F) -- node[above] {$\mathbf{c}_j = -\mathbf{c}_i$} (9,0);
\draw[->] (F) -- node[above] {$\mathbf{c}_i$} (15,0);

\end{tikzpicture}
\caption{Bounce back. The population $f_j(\mathbf{x} + \mathbf{c}_i,t)$ (flowing leftward from right node) becomes the population $f_i(\mathbf{x} +\mathbf{c}_i,t+1)$ (flowing rightward from right node). Note the change in velocity vector index $f_j \rightarrow f_i$, where $\mathbf{c}_i = -\mathbf{c}_j$.}
\label{fig:bb-picture-1}
\end{figure}

The populations are updated as
\begin{align}
f_i(\mathbf{x}+\mathbf{c}_i,t + 1) 
&:= 
\begin{cases}
f^{\text{collision}}_i(\mathbf{x}, t)
& \text{ if } \mathbf{x} \text{ is a fluid node}\\
& \text{ and } \mathbf{x}+\mathbf{c}_i \text{ is a fluid node.}\\
f^{\text{collision}}_j(\mathbf{x} + \mathbf{c}_i, t)
& \text{ if } \mathbf{x} \text{ is a solid node}\\
& \text{ and } \mathbf{x}+\mathbf{c}_i \text{ is a fluid node.}\\
& \text{ where } \mathbf{c}_j = -\mathbf{c}_i.\\
0 
& \text{ if } \mathbf{x}+\mathbf{c}_i \text{ is a solid node.}\\
& \text{(the populations are always zero at solid nodes.)}.\\
\end{cases}
\label{eq:lbm-bb-1}
\end{align}

We do not include inlet/outlet boundary conditions in our LBM formulation.  We start the system in a homogeneous state of local equilibrium (Eq.~\eqref{eq:lbm-feq}) corresponding to some uniform initial velocity,
\begin{align}
    \mathbf{u}(\mathbf{x},t=0) & = (u_x^\mathrm{init},0,0) \\
    f_i(\mathbf{x},t=0) & = f_i^\mathrm{eq}(\mathbf{x},t=0) \quad \mathrm{with} \quad \rho(\mathbf{x},t=0)=1,
\end{align}
and simulate the dynamics up to some final evolution time $T$ using periodic, wrap-around boundaries in the $x, y,$ and $z$ dimensions.  This choice was made for three reasons:
\begin{enumerate}
    \item naively imposing a fixed velocity for the bounced-back populations at the inlet/outlet can induce sound pressure wave artifacts in the domain,
    \item more sophisticated inlet/outlet conditions are still an active topic of research, and 
     \item the entering/exiting velocity populations need to be compatible with the bulk physics of the problem, which is typically unknown before simulation.
\end{enumerate}
The choice of periodic boundaries brings up two concerns.  The first is that our quantum resource estimate counts are lower since we are not encoding inlet/outlet conditions as an external forcing/driving term.  The second is that without an inlet driving the velocity field, the simulation will tend towards a zero velocity field.  The concern here is that the velocity field at the final evolution time $T$ will be lower than an inlet-driven velocity field and thus the drag force measurable will also be lower.  Future work may incorporate advanced non-reflecting inlet/outlet boundary conditions. A recent extension of quantum LBM to incorporate inlets, outlets, and external forcing is described in \cite{Jennings2026a}.

\subsection{Carleman linearization of the lattice Boltzmann method}\label{sec:cl-lbm}

The goal of this section is to reformulate the discrete space/time step dynamics of the LBM into a continuous-time linear ODE of the form $\frac{\partial \phi}{\partial t} = A\phi + b$, following the approach of Li \textit{et al.} \cite{Li2025}, compatible with established quantum linear ODE solvers and their resource estimates \cite{Berry2014,Berry2017,jennings2023b}. A recent alternative approach wherein Carleman linearization is directly applied to the discrete-time LBM is presented in \cite{jennings2025a}.

Consider the lattice Boltzmann equation as a differential equation where we can describe $\frac{ \partial }{ \partial t} f_i(\mathbf{x},t)$ as a function of the $f$-variables. There will be a contribution to $\frac{ \partial }{ \partial t} f_i(\mathbf{x},t)$ from the streaming operation and from the collision operation.  We now consider lattice time $t$ as continuous rather than discrete steps. Recall that we are working in lattice time rather than physical time $t^\star$.

The streaming operator for one LBM time step is described in \eqref{eq:lbm-bb-1}.  As one discrete LBM time step elapses, the current population $f_i(\mathbf{x})$ is reduced to zero and is completely replaced by some neighboring population $f_j(\mathbf{y})$ where the neighboring population depends on bounce back conditions described in \eqref{eq:lbm-bb-1}.  So the time derivative of the streaming operation is
\begin{align}
    \frac{ \partial }{ \partial t} f_i(\mathbf{x},t) 
    &=
    -f_i(\mathbf{x},t)  \underbrace{+f_j(\mathbf{y},t).}_{\text{defined by \eqref{eq:lbm-bb-1} for }f_i(\mathbf{x},t)}
        \label{eq:partial-derivative-LBM-streaming}
\end{align}
Formally, this continuous-time form of the streaming operation should include a second-order memory term so as to retain the LBM's second-order accuracy with respect to time step size $\Delta t$ \cite{Ubertini2010a}. While our omission of this memory term greatly simplfies the resource estimation process,  it was noted in \cite{jennings2025a} that mere first-order accuracy requires prohibitively small time steps $\Delta t \ll \tau$. Our neglect of this requirement is yet another manner in which our QREs should be regarded as underestimates of utility-scale requirements (see Sec.~\ref{subsec:methodolgy}).

The collision operator for one LBM time step is \eqref{eq:lbm-collision}.  The time derivative of the collision operator is:
\begin{subequations}
\begin{align}
\frac{ \partial }{ \partial t} f_i(\mathbf{x},t) 
&= \lim_{\epsilon \rightarrow 0} \frac{f_i(\mathbf{x},t + \epsilon)  - f_i(\mathbf{x},t)}{\epsilon}\\
&= \lim_{\epsilon \rightarrow 0} \frac{ \cancel{f_i(\mathbf{x},t)} - \frac{ \cancel{\epsilon}}{\tau} \left( f_i(\mathbf{x},t) - f^{\text{eq}}_i(\mathbf{x},t) \right) - \cancel{f_i(\mathbf{x},t)} }{ \cancel{\epsilon}}\\
&= \frac{ -1}{\tau} \left( f_i(\mathbf{x},t) - f^{\text{eq}}_i(\mathbf{x},t) \right)\label{eq:cl-partial}\\
&\approx \frac{ -1}{\tau} \left( f_i(\mathbf{x},t) - \tilde{f}^{\text{eq}}_i(\mathbf{x},t) \right)\label{eq:partial-derivative-LBM-collision-approx}
\end{align}
\end{subequations}
where $\tilde{f}^{\text{eq}}_i(\mathbf{x},t)$ is the cubic approximation for $f^{\text{eq}}_i(\mathbf{x},t)$ from \eqref{eq:lbm-feq-approx}.

Since we are now considering a continuous evolution of the populations, in our notation we will leave out the time step $t$ from our notation and we refer to $f_i(\mathbf{x},t)$ as simply $f_i(\mathbf{x})$.  We may also refer to the entire vector of $f_i(\mathbf{x})$-variables as simply $f \in \mathbb{R}^{nQ}$.  Various schemes may be developed to determine the index that a specific $f_i(\mathbf{x})$ appears in the vector $f$.  The scheme we use is described in appendix \ref{sec:notation}.

The dynamics for streaming in \eqref{eq:partial-derivative-LBM-streaming} are linear with respect to $f$-variables.  The dynamics for collision in \eqref{eq:partial-derivative-LBM-collision-approx} are cubic in $f$-variables because we have substituted the cubic approximation $\tilde{f}^{\text{eq}}_i(\mathbf{x},t)$ for $f^{\text{eq}}_i(\mathbf{x},t)$.  Therefore we can prepare a system of equations of the form:
\begin{subequations}
\begin{align}
    \frac{\partial }{\partial t} f &\approx Sf + F_1 f + F_2 f^{\otimes2} + F_3 f^{\otimes3}\\
    \, &\approx (S+F_1)f + F_2 f^{\otimes2} + F_3 f^{\otimes3},
\end{align}
\label{eq:lbm-multilinear-algebra}
\end{subequations}
where $S$, $F_1$, $F_2$, $F_3$ are sparse matrices, and $f^{\otimes 2} = f \otimes f$ and $f^{\otimes 3} = f \otimes f \otimes f$ represent second- and third-order monomials composed of $f$-variable factors, respectively.  The $S$-matrix accounts for contributions from streaming operations and the $F$-matrices account for contributions from the collision operation.  The matrices are large, so we don't want to store them in memory. Instead we have derived functions $\tilde{S}(r,c)$, $\tilde{F}_1(r,c)$, $\tilde{F}_2(r,c)$, $\tilde{F}_3(r,c)$ to compute the elements of each matrix in row $r$, column $c$ on the fly.  The adventurously detail-oriented reader is encouraged to review appendix \ref{sec:derivation-of-matrix-coeffs}.

The system in \eqref{eq:lbm-multilinear-algebra} is still nonlinear.  We will work with the truncated Carleman linearized approximation to \eqref{eq:lbm-multilinear-algebra}:
\begin{subequations}
\begin{align}
    \frac{\partial}{\partial t} 
    \underbrace{\begin{bmatrix} f \\ f^{\otimes 2} \\ f^{\otimes 3} \end{bmatrix}}_{\phi}
    &\approx
    \underbrace{
    \begin{bmatrix} 
    (S + F_1) & F_2 & F_3 \\ 
    0 & (S+F_1)^{[2]} & F_2^{[2]} \\
    0 & 0 & (S+F_1)^{[3]} 
    \end{bmatrix}
    }_{A}
    \underbrace{\begin{bmatrix} f \\ f^{\otimes 2} \\ f^{\otimes 3} \end{bmatrix}}_{\phi}
    + \underbrace{\mathbf{0}}_{b}\\
    \frac{\partial}{\partial t}\phi 
    &\approx
    A\phi + b\label{eq:3rd_order_carleman_linearized_system_renamed}
\end{align}
\label{eq:3rd_order_carleman_linearized_system}
\end{subequations}
Carleman Linearization creates an infinite dimensional system.  We have truncated the system in \eqref{eq:3rd_order_carleman_linearized_system} at third-order variables similar to \cite{Li2025, bakker2024quantum, Sanavio_2024}.  As convenient, we can rename variables and constants in the system as in \eqref{eq:3rd_order_carleman_linearized_system_renamed} so we can analyze the system as an ODE with $\phi$-variables, dynamics matrix $A$, and external forcing term $b = \mathbf{0}$.

$A$ is not block-upper-triangular. The blocks of $A$ are not the same size.  Table \ref{tab:F_matrix_dimensions} provides a summary of the dimensions.  
The blocks of $A$ are defined as follows:
\begin{align}
    (S+F_1)^{[2]} 
    &=
    \mathbb{I} \otimes (S+F_1) + (S+F_1) \otimes \mathbb{I} 
    \label{eq:S_plus_F_1_2_submatrix}
    \\
    F_2^{[2]}
    &=
    \mathbb{I} \otimes F_2 + F_2 \otimes \mathbb{I} 
    \label{eq:F_2_2_submatrix}
    \\
    (S+F_1)^{[3]}
    &=
    \mathbb{I} \otimes \mathbb{I} \otimes (S+F_1) + \mathbb{I} \otimes (S+F_1) \otimes \mathbb{I} + (S+F_1) \otimes \mathbb{I} \otimes \mathbb{I} 
    \label{eq:S_plus_F_1_3_submatrix}
\end{align}
The identify matrix $\mathbb{I}$ that appears in \eqref{eq:S_plus_F_1_2_submatrix}, \eqref{eq:F_2_2_submatrix}, and \eqref{eq:S_plus_F_1_3_submatrix} is $nQ \times nQ$ to be consistent with our original $f$-variables.  A derivation of \eqref{eq:S_plus_F_1_2_submatrix}, \eqref{eq:F_2_2_submatrix}, and \eqref{eq:S_plus_F_1_3_submatrix} is provided in appendix \ref{sec:appendix_derivation_of_submatrices}.

\begin{table}[h]
\centering
\begin{tabular}{|c|l|l|l|l|l|}
\hline
Item & Number of rows & Number of columns \\
\hline
$(S+F_1)$ & $nQ$ &  $nQ$\\
$F_2$ & $nQ$ &  $(nQ)^2$\\
$F_3$ & $nQ$ &  $(nQ)^3$\\
$(S+F_1)^{[2]}$ & $(nQ)^2$ &  $(nQ)^2$\\
$F_2^{[2]}$ & $(nQ)^2$ &  $(nQ)^3$\\
$(S+F_1)^{[3]}$ & $(nQ)^3$ &  $(nQ)^3$\\
$\phi_1 = f$ & $nQ$ & 1\\
$\phi_2 = f^{\otimes2}$ & $(nQ)^2$ & 1\\
$\phi_3 = f^{\otimes3}$ & $(nQ)^3$ & 1\\
$\phi$ & $nQ + (nQ)^2 + (nQ)^3$ & 1 \\
$A$ & $nQ + (nQ)^2 + (nQ)^3$ & $nQ + (nQ)^2 + (nQ)^3$ \\
\hline
\multicolumn{3}{|l|}{Recall that we have $n$ grid points and $Q$ lattice velocity vectors.}\\
\hline
\end{tabular}
\caption{Problem dimensions.}
\label{tab:F_matrix_dimensions}
\end{table}

\subsection{Time domain convergence of the Carleman linearized system}\label{sec:time_domain_convergence_main_body_text}

Carleman linearization lifts a nonlinear differential equation into an infinite system of linearized differential equations.  The infinite system is truncated at some order, which induces error.  We need to verify our truncated, linearized system will be within some acceptable error tolerance for the evolution time domain $t \in [0,T]$.  In appendix \ref{sec:appendix_time_domain_convergence_of_CL}, we build on \cite{Forets2018} and introduce theorem \ref{thm:max_convergence_time} for our specific LBM formulation, that allows us to calculate bounds on the evolution time $T_c$ where there is a guarantee that the solution to the truncated, Carleman-linearized system will be within some tolerance of the true solution.  

The numerical results/bounds from applying theorem \ref{thm:max_convergence_time} are shown in table \ref{tab:time_domain_T_convergence_bounds_for_sphere} in section \ref{sec:lbm-params-for-instances} alongside other problem instance parameters.  

\subsection{Quantum linearized differential equation solver}\label{sec:quantum-lin-DE-solver}

After linearizing the nonlinear differential equation, the dominant strategy to solve the linear form has been to transform this into a linear system of equations and leverage the exponential savings in memory from the solution procedures in \cite{Harrow2009, jennings2023a, subasi2019}.  
We use \cite{Berry2017} to transform our linear differential equation into a system of linear equations.  Then we use the linear systems solver from \cite{jennings2023a} to encode the solution. We capture quantum resource estimates using the results from \cite{jennings2023a} and incorporating detailed block-encoding costs.

As noted in \cite{Berry2017}, the solution of our linear differential equation is
\begin{equation}
    \label{eq:solution_vector}
    \ket{\phi(t)} = e^{At}\ket{\phi(0)} + (e^{At} -I )A^{-1}\ket{b}.
\end{equation}
Using the truncated Taylor series approximations $e^{At} \approx \sum_{j=0}^k \frac{(At)^j}{j!}$ and $(e^{At} -I )A^{-1} \approx \sum_{j=1}^k \frac{(At)^{j-1}}{j!}$ with truncation size $k$, \cite{Berry2017} constructs a system of linear equations \eqref{eq:berry2017_system_of_equations} based on the $L_{(m,k,p)}$ matrix, where $m$ is the number of time steps, $k$ is the truncation size, and $p$ is the number of trailing steps to hold the solution constant and boost the probability of obtaining the solution at the final time.

We can set up our system of equations with $\ket{\phi_{tj}}$ variables where $t$ will label the iteration in time and $j$ will label a summand in the Taylor approximations. We have the following relation to encode different summands between $\ket{\phi_{tj}} = \frac{Ah}{j} \ket{\phi_{t,j-1}}$ for time step size $h$, number of discrete time steps $m$, $0 \le t < m$, and $2 \le j < k$. The first variable for the first summand at an iteration step $\ket{\phi_{t,0}}$ is related to the result of the last iteration step by  $\ket{\phi_{t,0}} = \sum_{j=0}^k \ket{\phi_{t-1}, j}$ and lastly we have the relationship between the first summand and the second summand at an iteration step being $\ket{\phi_{t,1}} = A \ket{\phi_{t,0}} + h \ket{b}$.  From \cite{Berry2017}, this results in a system of linear equations of the form 
\begin{align}
    L_{(m,k,p)} 
    \left(
        \begin{array}{c}
        \ket{\phi_{t=0}}\\
        m \begin{dcases}
            \ket{\phi_{t=1}}\\
            \ket{\phi_{t=2}}\\
            \vdots \\
            \ket{\phi_{t=m-1}}\\
        \end{dcases}\\
        p \begin{dcases}
            \ket{\phi_{t=m-1}}\\
            \ket{\phi_{t=m-1}}\\
            \vdots\\
        \end{dcases}
        \end{array}
    \right)
    &= 
    \left( 
        \begin{array}{c}
            \ket{\phi_{\text{init}}}  \\
            m \begin{dcases}
                k+1 \begin{dcases}
                h\ket{b} \\
                0\\
                \vdots\\
                \end{dcases}
                \\
                k+1 \begin{dcases}
                h\ket{b}\\
                0\\
                \vdots\\
                \end{dcases}
                \\
                \vdots\\
            \end{dcases}
            \\
            p \begin{dcases}
                0\\
                \vdots\\
            \end{dcases}
            \\
        \end{array}
    \right).    \label{eq:berry2017_system_of_equations}
\end{align}
For illustration, a \textit{small example} $L_{(m=2,k=3,p=2)}$ matrix is
\begin{equation}
\label{eq:inverted_matrix}
L_{(2,3,2)}= 
\left( 
\begin{array}{ccccccccccc}
I &  &  &  &  &  &  &  &  &  &  \\
-Ah & I &  &  &  &  &  &  &  &  &  \\
 & -Ah/2 & I &  &  &  &  &  &  &  &  \\
 &  & -Ah/3 & I &  &  &  &  &  &  &  \\
 -I& -I &  -I& -I & I &  &  &  &  &  &  \\
 &  &  &  & -Ah & I &  &  &  &  &  \\
 &  &  &  &  & -Ah/2 & I &  &  &  &  \\
 &  &  &  &  &  & -Ah/3 & I &  &  &  \\
 &  &  &  &  -I&  -I& -I & -I & I &  &  \\
 &  &  &  &  &  &  &  & -I & I &  \\
 &  &  &  &  &  &  &  &  & -I & I \\
\end{array}
\right).
\end{equation}

Once we have solved the linear equations, we have a \textit{history state} that contains the solution for all time steps $0 \le t < m = \lceil T/h \rceil$, where $T$ is the final evolution time. Time step $h$ is bounded by $h \le 1/||A||_2$. We derive a bound for $||A||_2$ in appendix \ref{sec:bounding_spectral_norm_of_A}. We shall use this state and quantum amplitude estimation to extract an estimate of the drag force. The very general workflow is shown in figure \ref{fig:lucid_chart_figure}.

Quantum circuits are needed to block encode the classical data needed to specify the Carleman $A$ matrix and equations.  A significant amount of resources are used for these operations.  Appendix \ref{sec:quantum_compilation} discusses the details of block encoding.

\begin{figure}[H]
    \centering
    \begin{tikzpicture}[
  node distance=0.25cm and 1cm,
  every node/.style={align=center},
  box/.style={draw, rectangle, minimum width=1cm, minimum height=1cm},
  block/.style={draw, rectangle, minimum width=2cm, minimum height=1cm},
  arrow/.style={draw, -{Latex[length=3mm]}}
]

\node[box] (input1) {$S + F_1$};
\node[box, right=of input1, xshift=1cm] (input2) {$F_2$};
\node[box, right=of input2, xshift=1cm] (input3) {$F_3$};

\node[block, below=of input1, yshift=-1cm] (blockencoding1) {$\begin{pmatrix} S + F_1 & * \\ * & * \end{pmatrix}$};
\node[block, below=of input2, right=of blockencoding1] (blockencoding2) {$\begin{pmatrix} F_2 & * \\ * & * \end{pmatrix}$};
\node[block, right=of blockencoding2, xshift=0cm] (blockencoding3) {$\begin{pmatrix} F_3 & * \\ * & * \end{pmatrix}$};

\node[block, below=of blockencoding2, yshift=-1cm] (UA) {$U_A = \begin{pmatrix} A & * \\ * & * \end{pmatrix}$};
\node[block, below=of UA, yshift=-1cm] (UL) {$U_L = \begin{pmatrix} L & * \\ * & * \end{pmatrix}$};

\node[block, below=of UL, yshift=-1cm] (linear) {\it Linear system solver by adiabatic trajectories};
\node[block, below=of linear, yshift=-1cm] (amplification) {\it Quantum amplitude amplification};
\node[block, below=of amplification, yshift=-1cm] (estimation) {\it Quantum amplitude estimation};

\draw[arrow] (input1.south) -- (blockencoding1.north);
\draw[arrow] (input2.south) -- (blockencoding2.north);
\draw[arrow] (input3.south) -- (blockencoding3.north);

\draw[arrow] (blockencoding1.south) -- (UA.north);
\draw[arrow] (blockencoding2.south) -- (UA.north);
\draw[arrow] (blockencoding3.south) -- (UA.north);

\draw[arrow] (UA.south) -- (UL.north);
\draw[arrow] (UL.south) -- (linear.north);
\draw[arrow] (linear.south) -- (amplification.north);
\draw[arrow] (amplification.south) -- (estimation.north);

\node[align=center, left=of input1, xshift=-1cm] {\it Input from classical workflow};
\node[align=center, left=of blockencoding1, yshift=0cm, xshift=-0.5cm] {\it Block encoding for \\ \it input classical matrices};
\node[align=center, left=of UA, xshift=-3cm] {\it Block encoding the\\ \it Carleman-linearized system};
\node[align=center, left=of UL, xshift=-4cm] {\it Block encoding the\\ \it $L$ matrix};
\node[align=center, left=of linear, xshift=-1.5cm] {\it Linear inversion on\\ \it quantum computer};
\node[align=center, left=of amplification, xshift=-2.25cm] {\it Increase probability \\ \it of solution};
\node[align=center, left=of estimation, xshift=-2.5cm] {\it Extract the drag force};

\end{tikzpicture}
    
    \caption{Overview of the quantum algorithmic pipeline workflow for extracting the drag coefficient. See figure \ref{fig:drag_est_call_graph} for details on how to break apart this diagram from a resource estimation perspective.}
    \label{fig:lucid_chart_figure}
\end{figure}

\subsection{Quantum approach to estimating the drag force}\label{sec:quantum_approach_to_estimating_drag_force}

Here we describe the quantum method for estimating the drag force scalar $\mathcal{F}$ in the flow-past-a-sphere scenario.
 The costs of running this on the quantum computer will be assessed in Section \ref{sec:quantum-resource-estimates}.  We consider drag force on the sphere in the $x$-direction (the direction of bulk flow).  We now describe how amplitude estimation can be used to yield an estimate of this quantity by encoding $\mathcal{F}$ into a quantum amplitude. 
As described in \cite{Li2025}, and reviewed in the previous two sections, the quantum LBM method provides a quantum circuit that generates a state that is proportional to the $K^{\text{th}}$-order Carleman linearization encoding
\begin{align}
\label{eq:carleman_encoding}
\phi
&=
f \oplus f^{\otimes 2}\oplus \ldots \oplus f^{\otimes K} 
\\
&=
f\oplus \bigoplus^{K}_{k=2} f^{\otimes k},
\end{align}
\noindent where we have separated out the first-order block $f$ from the higher-order blocks $f^{\otimes k}$ (with $k>1$), since only the first-order block encodes the full nonlinear dynamics. As explained in \cite{jennings2025a}, the higher-order blocks omit successively higher-order nonlinear terms from their evolution, culminating in the highest-order block $f^{\otimes K}$ evolving only according to the linear dynamics.
This referenced work discusses the resulting challenge this poses for efficiently estimating observables from a quantum encoding of this vector.  Since $|f|=O(\sqrt{n})$, the norm of $|f^{\otimes k}|$ increases as a function of $k$, with $|f^{\otimes K}|$ being the largest.  Therefore, to exclude $\bigoplus^{K}_{k=2} f^{\otimes k}$ from the estimates of observables via amplitude estimation is costly.
A prior work \cite{costa2025further} explores some approaches to improving this estimation efficiency, although \cite{jennings2025a} discusses several associated limitations. 
The approach of \cite{jennings2025a} resolves this issue by using a ``shifted'' LBM wherein the norm of the state vector is $O(1)$ instead of $O(\sqrt{n})$.

In light of these challenges, our aim is to explore the \emph{limitations} of using quantum computers for simulating fluid dynamics.
Accordingly, we will proceed by assuming that the $K$th-order term in the expansion is equal to $f^{\otimes K}$, such that,
\begin{subequations}
\begin{align}
\label{eq:carleman_encoding}
\phi
&=
f \oplus f^{\otimes 2}\oplus \ldots \oplus f^{\otimes K} 
\\
&=
\bigoplus^K_{k=1} f^{\otimes k}.
\end{align}
\end{subequations}
\noindent With this assumption, our resource estimates will underestimate the true costs.
Also with this assumption, if the algorithm were run as described below, there would be a significant bias in the drag force estimate on account of the dominant amplitude weighting of $f^{\otimes K}$, which only encodes the linear dynamics.
Therefore, accurately estimating drag force in practice would require adapting the algorithm presented below, for example along the lines of the approaches explored in \cite{costa2025further, jennings2025a}, at the cost of increased runtime.

Encoding $\phi$ into a quantum state requires re-normalizing this vector by the quantity
\begin{subequations}
\begin{align}
\label{eq:solution_vector_classical}
||\phi||_2^2
&=\sum^K_{k=1} ||f||_2^{2k}\\
&=||f||_2\sum^{K-1}_{k=0} ||f||_2^{2k}\\
&=||f||_2\left(\frac{1-||f||_2^{K}}{1-||f||_2}\right).
\end{align}
\end{subequations}
We will choose to encode $\phi$ into a quantum state as follows,
\begin{align}
    \ket{\phi}\equiv U_{\phi}\ket{0\ldots 0} &= \sum_{k=1}^K\frac{||f||_2^k}{||\phi||^2_2} \ket{k}\ket{f}^{\otimes k}\ket{0^M}^{\otimes K-k},
\end{align}
where the register with $\ket{k}$ is used to label the terms in the direct sum of \eqref{eq:carleman_encoding}
and the state $\ket{f}$ is defined such that $f_i(\mathbf{x})=f_i(x,y,z)=||f||_2\langle x,y,z,i\ket{f}$ and is encoded into $M$ qubits.
Note that, while the size of each sector in \eqref{eq:carleman_encoding} increases with $k$, our encoding embeds each of these subspaces in a larger Hibert space all of the same size corresponding to that of the $K^{\text{th}}$ sector. This will incur a marginal extra qubit cost at the benefit of a fairly efficient Carleman matrix block encoding (see appendix \ref{subsec:carleman_matrix_block_encoding}).
The amplitude estimation efficiency is derived from the fact that the second register is unentangled from the others and its state is $\ket{f}$.
Accordingly, we can consider $U_{\phi}$ to yield a block encoding of the distribution function $f$ at time $t$ with subnormalization $||f||_2^2$, meaning that there is some norm-1 vector $r$ such that
\begin{align}
    f\otimes r = ||f||_2 U_{\phi}\ket{0... 0} = ||f||_2 \ket{f}\otimes\ket{\textup{remainder}},
\end{align}
where $\ket{\textup{remainder}} := \sum_{k=1}^K\frac{||f||_2^k}{||\phi||_2} \ket{k}\ket{f}^{\otimes k-1}\ket{0^M}^{\otimes K-k}$ is the state of the system apart from the second register.
Having established this encoding, we are then able to use $U_{\phi}$ to estimate properties of  $f$, including the drag force on the sphere.

To develop a quantum method for estimating the drag force $\mathcal{F}$, in the following we show that $\mathcal{F}$ can be written as a linear function of $f$. This means that there exists some vector $v$ such that
\begin{align}
    \mathcal{F} = v \cdot f.
    \label{eq:f_x_def}
\end{align}
Once this is established, then an approach to estimating $\mathcal{F}$ on a quantum computer is to encode $\mathcal{F}$ in an amplitude of $\ket{f}$ and then use quantum amplitude estimation to determine it.
The linearity of $\mathcal{F}$ is established as follows.
The drag force in the $x$-direction can be derived from the total momentum exchange during one streaming step in the lattice Boltzmann method \cite{Ladd1993NumericalSO, Krueger2017TheLB}.  First, we define the set of boundary fluid nodes $\mathcal{B}$ and the set of solid wall nodes $\mathcal{W}$:
\begin{align}
    \mathcal{B} = \{ \mathbf{x} \in \mathbf{X} : \mathbf{x} \text{ is a fluid node and } \exists \mathbf{c}_i \in \mathbf{C} \text{ such that } \mathbf{x} + \mathbf{c}_i \text{ is a solid node.} \},
    \label{eq:set_of_boundary_nodes}
\end{align}
and 
\begin{align}
    \mathcal{W} = \{ \mathbf{x} \in \mathbf{X} : \mathbf{x} \text{ is a solid node and } \exists \mathbf{c}_i \in \mathbf{C} \text{ such that } \mathbf{x} + \mathbf{c}_i \text{ is a fluid node.} \}.
    \label{eq:set_of_wall_nodes}
\end{align}
We define the set of boundary links $\mathcal{B}_{\text{links}}$ as
\begin{align}
    \mathcal{B}_{\text{links}} = \{ (\mathbf{x}, i) \in \mathbf{X} \times \mathbb{Z}_Q : \mathbf{x} \text{ is a fluid node and } \mathbf{x} + \mathbf{c}_i \text{ is a solid node, where } \mathbf{c}_i \in \mathbf{C} \}.
    \label{eq:set_of_boundary_links}
\end{align}
Recall that $\mathbf{C}$ is the set of all lattice velocity vectors defined by table \ref{tab:D3Q27-constants}.  Note that one boundary node $\mathbf{x} \in \mathcal{B}$ may participate in one or more boundary links.  The total momentum change $\Delta \mathbf{p}$ is calculated as the sum over all boundary links that have one node outside of the sphere (fluid node) and one node inside of the sphere (solid node) and given by the equation
\begin{align}
\Delta \mathbf{p} = \Delta x^3 \sum_{(\mathbf{x},i)\in \mathcal{B}_{\text{links}}} (f_i(\mathbf{x}) + f_{\bar{i}}(\mathbf{x})) \mathbf{c}_i.
\end{align}
The two terms represent the distribution functions for particles entering and bouncing back from the solid wall nodes, and the notation $f_{\bar{i}}(\mathbf{x})$ implies $f_{\bar{i}}(\mathbf{x}) = f_j(\mathbf{x})$, where $\mathbf{c}_j = -\mathbf{c}_i$.  The $\Delta x^3$ ensures the result is a momentum rather than a momentum density, and we have omitted the prefactor $\rho$ since we have scaled to $\rho\approx 1$.
The force in the $x$-direction is then given by the following linear combination,
\begin{align}
    \mathcal{F} = \frac{\Delta \mathbf{p}\cdot \hat{x}}{\Delta t} = \frac{\Delta x^3}{\Delta t}\sum_{(\mathbf{x},i)\in\mathcal{B}_{\text{links}}}(f_i(\mathbf{x})+f_{\bar{i}}(\mathbf{x}))c_{ix} = v \cdot f, \label{eq:drag_force}
\end{align}
where $c_{ix}=\mathbf{c}_i \cdot \hat{x}$ is the component of the lattice velocity vector $\mathbf{c}_i$ in the $x$-direction. This expression is a linear function of the distributions $f_i(\mathbf{x})$ and we collect the linear coefficients (of the form $c_{ix}\Delta x^3 / \Delta t$) into a single vector $v$ of the same dimension as $f$. Consider a unitary transformation that prepares a state that is proportional to this vector,
\begin{align}
    U_{v}\ket{0\ldots 0} \equiv \ket{v},
\end{align}
where the state $\ket{v}$ is defined such that $v_i(\mathbf{x})=v_i(x,y,z)=||v||_2\langle x,y,z,i\ket{v}$.
Then, we have that
\begin{align}
    \mathcal{F} = v\cdot f = ||v||_2||f||_2\braket{v}{f}.
\end{align}
This quantity can be estimated using quantum amplitude estimation as follows.
Forming the controlled-unitary $V=(c$-$U_{v}^{\dagger})(c$-$U_{\phi})$, the quantity of interest is encoded in a square amplitude as
\begin{align}
\frac{1}{2}\left(1+\frac{\mathcal{F}}{||v||_2||f||_2}\right)
= |(\bra{+}\otimes I)V \ket{+0\ldots 0}|^2 = a.
\label{eq:qamp}
\end{align}
This shows that the quantity on the left-hand-side of \eqref{eq:qamp} is equal to the squared amplitude of the state $V(H\otimes I)\ket{00\ldots 0}$ on the subspace marked by $-H\otimes I$.
Therefore, we can estimate the left-hand side (and thus ${v}$) using quantum amplitude estimation.  The circuit used in quantum amplitude estimation \cite{brassard2002quantum}, illustrated in Figure~\ref{fig:GroverQAE}, is based on the Grover iterate $W=-(V(H\otimes I)R_0(H\otimes I)V^{\dagger})(-H\otimes I)$, where $R_0$ is the reflection about the all-zero state.

Given an estimate $\hat{a}$ of the squared amplitude $a$ where $|\hat{a}-a|\leq \tilde{\epsilon}$, we can then construct an estimate $\hat{\mathcal{F}}$ of $\mathcal{F}$ to within error $\epsilon = 2||{v}||_2||f||_2\tilde{\epsilon}$ as
\begin{align}
\hat{\mathcal{F}}=||{v}||_2||f||_2(2\hat{a}-1).
\end{align}
This estimate also requires knowing $||{v}||_2$ and $||f||_2$. 
As explained in section \ref{subsec:methodolgy}, $||{v}||_2$ can be determined from the structure of the problem (and we establish an upper bound on this quantity for the purposes of resource estimation).
The quantity $||f||_2$, however, 
is a function of the solution to the nonlinear differential equation, which in general is not expected to be known ahead of time.
Accordingly, we must use the quantum computer to generate an estimate of $||f||_2$.
While alternative quantum linear ODE solvers including \cite{Berry2017} might yield a more straight-forward approach to estimating this quantity using quantum amplitude estimation, it is favorable to match the linear ODE solver used to estimate $a$ with the one used to estimate $||f||_2$.
We discuss the need for future work on this topic in section \ref{sec:future}.

\begin{figure}
\centering
\begin{subfigure}{0.8\textwidth}
  \centering
      \yquantset{operator/separation=4mm}
\begin{tikzpicture}
\begin{yquant}
    qubit {$\ket{+}$} ancilla;
    qubit {$\ket{0^n}$} system; 

    box {$W^m$} system | ancilla;
    h ancilla;
    measure ancilla;output {$0$} ancilla;

\end{yquant}
\end{tikzpicture}
\caption{Circuit used for amplitude estimation of $(1+\textup{Re}(\braket{v}{f})/2$, where the probability of $0$ is given by $\textup{Pr}(0|m)=(1+\cos(m\arccos{\textup{Re}(\braket{v}{f}}))/2$.}
\end{subfigure}

\bigskip

\begin{subfigure}{0.8\textwidth}
 \centering
   \yquantset{operator/separation=4mm}
   \begin{tikzpicture}
   \begin{yquantgroup}
       \registers{
        qubit {} q[4];
       }
       \circuit{
        slash q[1];
        slash q[2];
        slash q[3];
        box {$W$} (q);
        slash q[1];
        slash q[2];
        slash q[3];
       } \equals

       \circuit{
        slash q[1];
        slash q[2];
        slash q[3];
        box {$H$} (q[0]);
        box {$U_{v}$} (q[1], q[2]) | q[0];
        box {$U_{\phi}^{\dagger}$} (q[2], q[3]) | q[0];
        box {$H$} (q[0]);
        box {$R_0$} (q);
        box {$H$} (q[0]);
        box {$U_{\phi}$} (q[2], q[3]) | q[0];
        box {$U_{v}^{\dagger}$} (q[1], q[2]) | q[0];
        slash q[1];
        slash q[2];
        slash q[3];
        }
   \end{yquantgroup}
       
\end{tikzpicture}
   \caption{Grover iterate for QAE.}
\end{subfigure}
\caption{Circuits used in QAE subroutine.}\label{fig:GroverQAE}
\end{figure}

To review, the main components of the drag estimation algorithm are:
\begin{enumerate}
    \item Construct unitary $U_{v}$ that encodes drag force coefficients and determine subnormalization $||v||_2$ (which can be done analytically).
    \item construct unitary $U_{\phi}$ that encodes solution state and determine subnormalization $||f||_2$ (which can be done using a separate call to amplitude estimation). 
    \item Use a quantum amplitude estimation algorithm to estimate $a$, the probability of a $(+)$ outcome on the ancilla qubit in $c$-$U_{v}^{\dagger}c$-$U_{\phi}\ket{+0\ldots 0}$.
    \item From $\hat{a}$, the estimate of $a$, return $\hat{\mathcal{F}}=||v||_2||f||_2(2\hat{a}-1)$, the estimate of $\mathcal{F}$.
\end{enumerate}

We conclude by discussing what costs are required to ensure an $\epsilon$-accurate estimate of $\mathcal{F}$ with failure probability less than $\delta$. These will be used in making the resource estimates in the following section.
Various quantum algorithms for amplitude estimation ensure an $\tilde{\epsilon}$-accurate estimate using $\tilde{O}(1/\tilde{\epsilon})$ calls to the Grover iterate per circuit, where $\tilde{O}(x)$ indicates $O(x\cdot\textup{polylog}(x))$.
This means that an $\epsilon$-accurate estimate of $\mathcal{F}$ can be ensured using $\tilde{O}(||{v}||_2||f||_2/\epsilon)$ Grover iterates per circuit. As an example, the iterative quantum amplitude estimation method \cite{grinko2021iterative} ensures an $\tilde{\epsilon}$-accurate estimate of $a$ with confidence $1-\delta$ using no more than
\begin{align}
\frac{32}{(1 - 2 \sin(\frac{\pi}{14}))^2} \log \left( \frac{2}{\delta} \log_2 \left( \frac{\pi}{4\tilde{\epsilon}} \right) \right)
\end{align}
circuit repetitions, with each circuit using a number of Grover iterates per circuit of no more than
\begin{align}
    \frac{\pi}{8 \tilde{\epsilon}}.
\end{align}
Given the logarithmic dependence of the number of samples on its parameters, it depends weakly on $\tilde{\epsilon}$ and $\delta$. For example, from $\tilde{\epsilon}=0.1$ and $\delta=0.1$ to $\tilde{\epsilon}=0.00001$ and $\delta=0.00001$, the number of circuit repetitions only increases from 425 to 1558.

In terms of a relative error $\overline{\epsilon} = \epsilon/\mathcal{F}$, the estimation error required of the amplitude estimation algorithm is
\begin{align}
    \tilde{\epsilon} = \frac{\mathcal{F}\overline{\epsilon}}{2||{v}||_2||f||_2}
\end{align}
and the number of Grover iterates per circuit is no more than
\begin{align}
\label{eq:amp_est_cost}
\frac{\pi ||{v}||_2||f||_2}{4\mathcal{F}\overline{\epsilon}}.
\end{align}

\section{Quantum resource estimates}\label{sec:quantum-resource-estimates}
    
A sufficient QPU for the approach described in section \ref{sec:quantum_workflow} does not exist yet.  One of the main goals of this work is to estimate the quantum resources that would be required to run the calculation.  The \emph{resource estimates} provide a surrogate performance comparison to classical state of the art techniques.  The resource estimates are based on our choice of quantum algorithm workflow, implementation details, and assumptions.  

The purpose of this section is to investigate the question: \textit{what is the estimated cost of solving utility-scale computational fluid dynamics on a quantum computer?} 
We are ultimately interested in determining such costs for the high-utility problem of determining the drag force on a ship hull. However, due to the limitations of existing methods in quantum compilation for addressing such problems, we will instead investigate the costs for the flow-past-a-sphere problem. We conclude this section by exploring what can be inferred about the ship hull problem using the results for the flow-past-a-sphere problem. We begin by describing our approach to generating resource estimates for the drag estimation task in the flow-past-a-sphere problem. 
We then present the LBM parameters used as input for the resource estimates.

\subsection{Methodology, assumptions, and approximations}
\label{subsec:methodolgy}

The resource costs presented are the number of logical qubits and non-Clifford gate counts (specifically $T$-gates). We count $T$-gates because they are believed to be the rate-limiting operations for modern implementations of quantum error correction (specifically lattice surgery on the surface code \cite{horsman2012surface}).  $T$-gate count serves as a proxy for runtime.

Our estimates err on the side of under counting qubits and $T$-gates.
This is because (1) a lower bound on resource estimates is sufficient to establish the impracticality of using the existing methods to solve the drag estimation problem considered, (2) we anticipate that the components throughout the quantum algorithm will be improved and so a detailed accounting at this point may be premature, and (3) the parts of the quantum circuit that are neglected are not expected to introduce substantial costs.

\begin{figure}[p!]
    \centering    

\begin{tikzpicture}[
  node distance=0.75cm and 1cm,
  every node/.style={align=center},
  box/.style={draw, rectangle, minimum width=3.5cm, minimum height=1cm, text centered, rounded corners},
  block/.style={draw, rectangle, minimum width=4cm, minimum height=1cm, text centered, rounded corners},
  arrow/.style={draw, -{Latex[length=3mm]}}
]

\node[box] (lbm) {\it LBM drag force estimation };
\node[box, below=of lbm] (iter1) {\it Iterative quantum amplitude estimation for inner product};
\node[box, below=of iter1] (iter2) {\it Run Iterative quantum amplitude circuit};
\node[box, below=of iter2] (solver) {\it Taylor quantum ODE solver };
\node[box, below=of solver] (amplitude) {\it Oblivious amplitude amplification of solution};
\node[box, below=of amplitude] (prep) {\it Prepare ODE final time state};
\node[box, below=of prep] (qlsa) {\it Quantum Linear solver };

\node[block, below=of qlsa, xshift=-4cm] (odehistory) {\it Block encode ODE history system};
\node[block, below=of qlsa, xshift=4cm] (bvector) {\it ODE history b Vector};

\node[block, below=of odehistory, xshift=0cm] (carleman) {\it  Block encode Carleman Linearized system};
\node[block, below=of carleman, xshift=-4cm] (cubic) {\it Block encode cubic term};
\node[block, below=of carleman, xshift=4cm] (quadratic) {\it Block Encode Quadratic Term};
\node[block, below=of cubic, xshift=4cm] (linear) {\it Block Encode Linear Term};
\node[block, below=of linear, xshift=0cm] (tgates) {\it $T$ gates};

\draw[arrow] (lbm.south) -- (iter1.north);
\draw[arrow] (iter1.south) -- (iter2.north);
\draw[arrow] (iter2.south) -- (solver.north);
\draw[arrow] (solver.south) -- (amplitude.north);
\draw[arrow] (amplitude.south) -- (prep.north);
\draw[arrow] (prep.south) -- (qlsa.north);
\draw[arrow] (qlsa.south) -- (odehistory.north);
\draw[arrow] (qlsa.south) -- (bvector.north);
\draw[arrow] (odehistory.south) -- (carleman.north);
\draw[arrow] (carleman.south) -- (cubic.north);
\draw[arrow] (carleman.south) -- (quadratic.north);
\draw[arrow] (carleman.south) -- (linear.north);
\draw[arrow] (linear.south) -- (tgates.north);
\draw[arrow] (cubic.south) -- (tgates.west);
\draw[arrow] (quadratic.south) -- (tgates.east);

\end{tikzpicture}
 
    \caption{
        Detailed call graph of subroutines used to estimate number of $T$ gates. The higher subroutine calls the lower subroutine a number of times.
    }
    
\label{fig:drag_est_call_graph}
\end{figure}

\paragraph{$T$-gate counting} As shown in figure \ref{fig:drag_est_call_graph}, the quantum algorithm is decomposed into subroutines with subtasks.  We count the number of times a subroutine calls a subtask. The decomposition continues until we reach a subroutine that calls to $T$-gates as its subtask.  (Implemented in \cite{qsub_package}.)

\paragraph{Qubit counting} The qubit counting works in a similar recursive manner.
The difference is that each subroutine is responsible for determining how many qubits it requires as a function of the qubit counts of the various qubit registers used by its subtasks.  (Implemented in \cite{qsub_package}.)  Some calls or qubit counts are approximated.   We review the assumptions for each of the components of the call graph:
\begin{itemize}
    \item \textbf{Quantum amplitude estimation (QAE)} 
    \begin{itemize}
        \item Number of calls: We use \eqref{eq:amp_est_cost} to determine the number of calls to the Grover iterate that uses $U_{\phi}$ and $U_{v}$. In the discussion that follows, we establish upper bounds for the normalization factors.
        \item Qubit count: We properly account for qubits at this layer of the hierarchy.
        \item Note: as explained in section \ref{sec:quantum_approach_to_estimating_drag_force}, while QAE is used to estimate the squared amplitude $a$ (see \eqref{eq:qamp}), in a complete version of the estimation algorithm, QAE may also be used to estimate $||f||_2$. We  ignore the costs of this additional subroutine.  We expect that its cost would not significantly exceed the cost of estimating the squared amplitude $a$ itself due to the similarity of the two tasks. 
    \end{itemize}
    
    \item \textbf{Solution state preparation $U_{\phi}$:} 
    \begin{itemize}
        \item Number of calls to linear ODE matrix $A$ block encoding: we use the analytical upper bounds provided in \cite{jennings2023b} to determine the number of calls to the linear ODE $A$ matrix block encoding.
        \item Number of calls to linear ODE initial state $\phi(0)$: we ignore the costs of this operation as we expect that it will require far fewer $T$-gates than the $A$ matrix block encoding.        
        \item Qubit count: we use the qubit accounting of \cite{jennings2023b} to account for the qubits in this layer.
    \end{itemize}
    
    \item \textbf{Drag vector coefficient state preparation $U_{v}$:} 
    \begin{itemize}
        \item Number of calls to quantum arithmetic operations: we expect this operation will use a number of quantum arithmetic operations which then use $T$-gates. We ignore the $T$-gate costs of this operation because it is called far fewer times than the linear ODE matrix block encoding and has fewer operations than this block encoding.
        \item Qubit count: we will assume that the qubits in this layer are sufficiently accounted for by using the number of qubits required by $U_{\phi}$. This is because $U_{v}$ is expected to use fewer workspace qubits than $U_{\phi}$.
    \end{itemize}
    \item \textbf{ODE history block encoding:} 
    \begin{itemize}
        \item Number of calls to $T$-gates: we only account for the $T$-gates in this operation that are due to the Carleman matrix block encoding addressed below. There will also be $T$-gates used in ancillary operations, though these $T$-gates are expected to be small relative to that of the Carleman matrix block encoding.
        \item Qubit count: we account for most qubits used in these operations, though we exclude the costs of some workspace qubits used for quantum arithmetic (see appendix \ref{sec:quantum_compilation}).
    \end{itemize}
    \item \textbf{Linear ODE (Carleman matrix) block encoding:} 
    \begin{itemize}
        \item Number of calls to nonlinear matrix block encodings: this block encoding requires only a single call to each of the $F_1 +S $, $F_2$, and $F_3$.
        \item Qubit count: we will assume that the qubits in this layer are sufficiently accounted for by using the number of qubits required by the block encodings themselves and their linear combination. There are some additional workspace qubits used in this operation, but their costs are not expected to be significant compared to the qubit costs of the block encodings themselves.
    \end{itemize}
    \item \textbf{$S, F_1$, $F_2$, $F_3$ block encodings:} 
    \begin{itemize}
        \item Number of calls to $T$-gates: we account for most $T$-gates in these operations, though the $T$-gate counts of some lower-level operations are not accounted for (as explained in appendix \ref{sec:quantum_compilation}), though these $T$-gates are expected to be small relative to the other operations.
        \item Qubit count: we account for most qubits used in these operations, though we exclude the costs of some workspace qubits used for quantum arithmetic (see appendix \ref{sec:quantum_compilation}).
    \end{itemize}    
\end{itemize}

In order to determine the runtime costs of quantum amplitude estimation, we must determine bounds for the subnormalizations $||{v}||_2$ and $||\phi||_2$ used to rescale the quantum amplitude estimate.
Both of these quantities must be determined to sufficient accuracy.

\paragraph{Bounding $||f||_2$:}\label{sec:bounding_f}  Recall that the definition of density at some node $\mathbf{x}$ is $\rho(\mathbf{x}) = \sum_i f_i(\mathbf{x})$. We also assumed that the fluid is incompressible and $|1-\rho(\mathbf{x})|< \varepsilon_\rho$ for $0 < \varepsilon_\rho \ll 1$ (see \eqref{eq:lbm-approx-inv-rho}) for all fluid nodes $\mathbf{x}$.  We have also defined the density of solid nodes $\mathbf{y}$ as $\rho(\mathbf{y}) = 0$.  Define the number of fluid nodes as $n_f$ where $0 < n_f \le n$.  The contribution to $||f||^2_2$ from some fluid node $\mathbf{x}$ is $\sum_i f_i(\mathbf{x})^2$.  This is maximized when populations are concentrated in one particular velocity vector bin $j$ and $f_j(\mathbf{x}) = (1 + \varepsilon_\rho)$ and $f_i(\mathbf{x}) = 0, \forall i \in \mathbb{Z}_Q, i \neq j$, thus $\sum_i f_i(\mathbf{x})^2 = (1+\varepsilon_\rho)^2$.  The same contribution to the norm is minimized when populations are uniformly distributed across all velocity vector bins as $f_i(\mathbf{x}) = (1-\varepsilon_\rho/Q)/Q, \forall i \in \mathbb{Z}_Q$, thus $\sum_i f_i(\mathbf{x})^2 = (1-\varepsilon_\rho/Q)^2/Q$.  Thus the bounds are 
\begin{align}
    n_f \left(\frac{(1-\varepsilon_\rho/Q)^2}{Q}\right)
    \le
    ||f||^2_2
    \le
    n_f(1+\varepsilon_\rho)^2
    \le
    n(1+\varepsilon_\rho)^2
    \label{eq:bounds_on_f}
\end{align}

\paragraph{Bounding $||\phi||_2$:}\label{sec:bounding_phi}  $\phi = f\bigoplus f^{\otimes 2} \bigoplus f^{\otimes 3}$ and we will extrapolate our argument for \eqref{eq:bounds_on_f} to bound $||\phi||_2$.
\begin{align}
        ||\phi||^2_2 
        &=
        \sum_{i} \sum_{\mathbf{x}} f_i(\mathbf{x})^2
        +
        \sum_{i,j} \sum_{\mathbf{x},\mathbf{y}} f_i(\mathbf{x})^2f_j(\mathbf{y})^2
        +
        \sum_{i,j,k} \sum_{\mathbf{x},\mathbf{y},\mathbf{z}} f_i(\mathbf{x})^2f_j(\mathbf{y})^2f_k(\mathbf{z})^2.
        \label{eq:phi_norm_squared_definition}
\end{align}
There is only a positive contribution to the summation from fluid grid nodes $\mathbf{x},\mathbf{y},\mathbf{z}$.  We can use the bounds from \eqref{eq:bounds_on_f} directly on the first term.  The second term is minimized when the populations are equally distributed in all velocity vector bins $i \in \mathbb{Z}_Q$, so
\begin{subequations}
    \begin{align}
        \sum_{i,j} \sum_{\mathbf{x},\mathbf{y}} \left(\frac{1-\varepsilon_\rho/Q}{Q}\right)^2\left(\frac{1-\varepsilon_\rho/Q}{Q}\right)^2
        &\le
        \sum_{i,j} \sum_{\mathbf{x},\mathbf{y}} f_i(\mathbf{x})^2f_j(\mathbf{y})^2
        \\
        Q^2 n_f^2 \left(\frac{1-\varepsilon_\rho/Q}{Q}\right)^4
        &\le
        \sum_{i,j} \sum_{\mathbf{x},\mathbf{y}} f_i(\mathbf{x})^2f_j(\mathbf{y})^2
        \\
        n_f^2 \left(\frac{(1-\varepsilon_\rho/Q)^4}{Q^2}\right)
        &\le
        \sum_{i,j} \sum_{\mathbf{x},\mathbf{y}} f_i(\mathbf{x})^2f_j(\mathbf{y})^2.
\end{align}
\end{subequations}
Similarly, the upper bound on the second term occurs when the populations are concentrated in only one velocity vector bin.
\begin{subequations}
    \begin{align}
        \sum_{i,j} \sum_{\mathbf{x},\mathbf{y}} f_i(\mathbf{x})^2f_j(\mathbf{y})^2
        &\le 
        \sum_{\mathbf{x},\mathbf{y}}(1+\varepsilon_\rho)^2(1+\varepsilon_\rho)^2
        \\
        &\le 
        n^2_f(1+\varepsilon_\rho)^4.
\end{align}
\end{subequations}
The same argument is made for the third-order term and the bounds are:
\begin{align}
        n_f^3\left(\frac{(1-\varepsilon_\rho/Q)^6}{Q^3}\right)
        \le 
        \sum_{i,j,k} \sum_{\mathbf{x},\mathbf{y},\mathbf{z}} f_i(\mathbf{x})^2f_j(\mathbf{y})^2f_k(\mathbf{z})^2 
        \le
        n_f^3(1+\varepsilon_\rho)^6.
\end{align}
Combining the bounds on the individual terms we have
\begin{align}
        \underbrace{
        \sum_{k=1}^{3} n_f^k\left(\frac{(1-\varepsilon_\rho/Q)^{2k}}{Q^k}\right)
        }_{\phi^2_{\text{min}}}
        \le 
        ||\phi||^2_2
        \le
        \underbrace{
        \sum_{k=1}^3 n_f^k(1+\varepsilon_\rho)^{2k}
        }_{\phi^2_{\text{max}}}
        \le
        \sum_{k=1}^3 n^k(1+\varepsilon_\rho)^{2k}
        <
        3n^3(1+\varepsilon_\rho)^6
        \label{eq:phi_bounds}
\end{align}
where the summation runs from $k=1,2,3$ because we have truncated our Carleman Linearization at third order.  This analysis was based on the incompressible assumption alone, so these bounds hold for any time $t \in [0,T]$.

\paragraph{Bounding $||{v}||_2$:} This quantity is the square norm of the vector that encodes the linear coefficients of the drag force in \eqref{eq:drag_force}. The elementary volume unit is related to $n$ as $\Delta x^3=V/n$, giving 
\begin{subequations}
\begin{align}
||{v}||_2^2 &=  \left|\frac{\Delta x^3}{\Delta t}\right|^2 \sum_{\mathbf{x} \in \mathbf{X}} \sum_{i \in \mathbb{Z}_Q} \big| \mathbbm{1}_{\mathcal{B}_{\text{links}}}(\mathbf{x},i) (\mathbf{c}_i \cdot \hat{x}) + \mathbbm{1}_{\mathcal{B}_{\text{links}}}(\mathbf{x},\bar{i}) (\mathbf{c}_{\bar{i}} \cdot \hat{x}) \big|^2
\\
& \leq  2\frac{V^2}{n^2\Delta t^2}
|\mathcal{B}_{\text{links}}|
\\
& \leq  2\frac{V^2}{n^2\Delta t^2}
|\mathcal{B}|Q.
\end{align}
\end{subequations}
where
\begin{align}
\mathbbm{1}_{\mathcal{B}_{\text{links}}}(\mathbf{x},i) 
& = 
\begin{cases}
1
& \text{ if } (\mathbf{x},i) \in \mathcal{B}_{\text{links}} \\
0
& \text{ else}\\
\end{cases}
\label{eq:indicator}
\end{align}
is the indicator function selecting $(\mathbf{x},i)$ pairs in $\mathcal{B}_{\text{links}}$ \eqref{eq:set_of_boundary_links}, and $\bar{i}=j$ where $\mathbf{c}_{j} = -\mathbf{c}_i$. To express $|\mathcal{B}|Q$ in terms of $n$, let $R=r/\Delta x$ be the sphere radius in number of grid spacings, so that the number of cells on the boundary of the sphere (i.e. $|\mathcal{B}|$) is less than $4\pi R^2 + \pi\delta^3/3$ (where a cell is considered on the boundary if its center is within $\delta$ of the sphere). We can get an estimate of the polynomial runtime scaling with respect to system size $n$. For the sphere, the number of boundary grid points is
\begin{align}
    |\mathcal{B}|\leq \frac{4\pi r^2 }{V^{2/3}}n^{2/3}+\frac{\pi \delta^3}{3}= O(n^{2/3}).
\end{align}
Using $Q=27$, we have
\begin{align}    ||{v}||_2^2 
&\leq \frac{216\pi V^{4/3}r^2}{ \Delta t^2n^{4/3}}. \label{eq:bounds_on_Fx}
\end{align}
Combining bounds \eqref{eq:bounds_on_f} and \eqref{eq:bounds_on_Fx}, we have that the required number of Grover iterates per circuit \eqref{eq:amp_est_cost} is no more than
\begin{align}
\label{eq:t_gate_scaling}
\frac{\pi||{v}||_2||f||_2}{4\mathcal{F}\overline{\epsilon}}\leq C\frac{ n^{1/6}}{ \mathcal{F}\overline{\epsilon}},
\end{align}
where
\begin{align}
    C=\frac{(3\pi)^{3/2}}{\sqrt{2}} \frac{V^{2/3} r}{\tilde{t}}
\end{align}
and we define $\tilde{t} = n^{1/3}\Delta t$.

Finally we discuss some assumptions that we make regarding the computational problem being solved.
We assume that the physical system obeys periodic boundary conditions in all three spatial directions.
This simplifies the form of the streaming operator (See \ref{sec:quantum_compilation}).
In principle, the width of the simulation can be chosen large enough such that the effects of the boundary do not impact the drag estimate more than the target precision.
Furthermore, due to the encoding scheme we have used, the qubit count only grows logarithmically in the number of grid nodes $n$. (However the number of grid nodes $n$ is driven up when a finer spatial resolution is necessary for simulating flows with large Reynolds numbers.)
The $T$-gate count increases sub-linearly according to \eqref{eq:t_gate_scaling}.
Another approximation we have made is that accommodating the inlet and outlet conditions in the streaming matrix does not substantially increase the $T$-gate counts of the corresponding block encodings.
In a real implementation of the algorithm, the flow of certain nodes must be forced in order to continue driving the fluid. A recent analysis \cite{Jennings2026a} found that incorporating walls, inlets, outlets, and external forcing preserved the asymptotic scaling of the original algorithm \cite{jennings2025a}.

\subsection{LBM parameters for problem instances}\label{sec:lbm-params-for-instances}

For all problem instances, we need to convert the physical parameters into parameters for the lattice topology.  After the conversion, spatial steps and time steps are in terms of integer lattice units.  The LBM parameter choices below for spatial $\Delta x$ and temporal $\Delta t$ discretization fidelity and final evolution time $T$ drive problem size and quantum resource estimates.  Tables \ref{tab:resource-estimate-parameter-definitions} and \ref{tab:resource_estimate_parameters_sphere} provide parameter definitions and values used for the flow-past-a-sphere problem instances, and table \ref{tab:resource_estimate_parameters_hulls} provides the LBM parameter values for the ship hull design problem instances.

\begin{table}[h]
\centering
\begin{tabular}{|l|p{10cm}|}
\hline
Parameter & Description \\
\hline
$A, \phi, b$ 
&
The Carleman-linearized dynamics of the LBM as $\frac{\partial \phi}{\partial t} = A \phi + b$ per equation \eqref{eq:3rd_order_carleman_linearized_system} 
\\
\hline
$||A||_2 \in \mathbb{R}_+$ 
&
The spectral norm of $A$.  We establish an upper bound in appendix \ref{sec:bounding_spectral_norm_of_A}.
\\
\hline
$h \in \mathbb{R}_+$
&
Evolution step size used in \cite{Berry2017}. Note that $h \le 1/||A||_2$.\\
\hline
$\alpha(A) \in \mathbb{R}$
&
Spectral abscissa of $A$, maximum real part of eigenvalues of $A$.  We use a value of $\alpha(A)=0$ per numerical studies \cite{source_for_norm_A_anlysis}.
\\
\hline
$||b||_2 \in \mathbb{R}_+ \cup \{ 0 \}$.
&
We have not included boundary inlet/outlet conditions in our model thus far, so $||b||_2 = 0$.
\\
\hline
$T \in \mathbb{R}_+$
&
Evolution time (in lattice steps). See tables \ref{tab:resource_estimate_parameters_sphere} and \ref{tab:resource_estimate_parameters_hulls}.
\\
\hline
$\phi_{\text{min}}, \phi_{\text{max}} \in \mathbb{R}_+ \cup \{ 0\}$
&
$\phi_{\text{min}} \le ||\phi(t)||_2 \le \phi_{\text{max}}, t \in [0,T]$, lower/upper bounds on the norm of the vector of Carleman-linearized $\phi$-variables. See equation \eqref{eq:phi_bounds}.
\\
\hline
$C_{\text{max}}(T) \in \mathbb{R}_+$
&
$||e^{At}||_2 \le C_{\text{max}}(T), t \in [0,T]$.  For our resource estimates, we use  $C_{\text{max}}(T)=1$ as a placeholder until a better bound can be established and we note that the overall cost will scale proportionally to this value \cite{jennings2023b}.
\\
\hline
\end{tabular}
\caption{Definitions of parameters used for resource estimates.}
\label{tab:resource-estimate-parameter-definitions}
\end{table}

We follow the methodology for parameter selection from \cite{Krueger2017TheLB}.  \cite{Krueger2017TheLB} recommends that the lattice Mach number be less than 0.3 and the maximum lattice velocity be approximately 0.1-0.2.   \cite{Krueger2017TheLB} also establishes the consistency relationship between $\Delta x$, $\Delta t$, and $\tau$ as:
\begin{equation}
    \nu = c_s^2 \left( \tau - \frac{1}{2} \right) \frac{{\Delta x}^2}{\Delta t}.
    \label{eq:lbm_param_consistency_7.14}
\end{equation}
This implies we can choose two out of three parameters from the set $\{ \Delta x, \Delta t, \tau \}$ and compute the third.  In each case, we set $\Delta x$ to be roughly equal to the Kolmogorov scale and set $\tau$ between 0.5 and 1, and calculated $\Delta t$.

Formally, the Kolmogorov scale is equal to 
$(\nu^3/\varepsilon_\mathrm{k})^{1/4}\approx L/\mathrm{Re}^{3/4}$, where $L$ is the characteristic length and $\varepsilon_\mathrm{k}\approx u^3/L$ is the rate of dissipation of turbulence kinetic energy per unit mass \cite{George2013a}. However, in order to maintain a consistent compressibility error ($\propto \mathrm{Ma}^2$), the lattice Mach number and maximum lattice velocity must be uniform across all flow-past-a-sphere instances. Due to the diffusive scaling $\Delta t \propto {\Delta x}^2$ in \eqref{eq:lbm_param_consistency_7.14} for the LBM, this requires that $\Delta x \propto 1/\mathrm{Re}$, and so we set the grid spacing as
\begin{equation} \label{eq:Dx}
    \Delta x = \frac{L}{\mathrm{Re}}, 
\end{equation}
and subsequently calculate the time step $\Delta t$ from \eqref{eq:lbm_param_consistency_7.14} with $\tau=0.6$.

The LBM is a time-marching procedure.  The system is started with initial conditions and step-wise evolved until a steady state is reached.   In practice, simulation results should be tested for convergence.  In our case, we provide an estimate on the evolution time as $T^\star = 2 t^\star_{\text{adv}}$ (in seconds) where $t^\star_{\text{adv}}$ is the advective time scale for each instance. From \cite{Succi_LBE}, the advective time scale is defined as
\begin{equation}
    t^\star_{\text{adv}} = L/u_{\text{max}},
    \label{eq:advective_time_scale}
\end{equation}
where $u_{\text{max}}$ is the expected maximum velocity.  We round $t^\star_{\text{adv}}$ to one significant figure. Recall that physical evolution time $T^{\star}$ and lattice step time $\Delta t$ are related by $T^\star = (\Delta t) T$.

Section \ref{sec:time_domain_convergence_main_body_text} foreshadowed the concerns about error created by truncating the Carleman linearized system of equations.  Theorem \ref{thm:max_convergence_time} establishes lower and upper bounds on the evolution time $T_c$ where the error in the solution $\phi(t)$ for $t \in [0, T_c]$ to the truncated Carleman linearized system can be driven arbitrarily low by increasing the truncation order.  $T_c$ is in lattice time steps, and due to the choices of parameters in the flow-past-a-sphere problem instances, the bounds on $T_c$ are the same for all cases in table \ref{tab:resource_estimate_parameters_sphere} as 
\begin{equation*}
0.0001106 \le T_c \le 0.0001237.
\end{equation*}
The results suggest that there is no guarantee of accuracy past an exceedingly small evolution time $T_c$.  $T_c$ is effectively less than one lattice time step.  Theorem \ref{thm:convergence_thm} also suggests that increasing the truncation order will not improve $T_c$, but rather lower the solution error within the range $t \in [0, T_c]$.  Theorems \ref{thm:convergence_thm} and \ref{thm:max_convergence_time} do not imply the solution will diverge past time $T_c$, but simply say there is no guarantee of accuracy past $T_c$.  Table \ref{tab:time_domain_T_convergence_bounds_for_sphere} translates the bounds on $T_c$ in lattice units to $T^*_c$ in physical units. Furthermore, a recent numerical analysis in 1D and 2D found that the Carleman procedure fails to converge above some threshold Reynolds number, despite an improved quantum LBM algorithm in which the nonlinearity was reduced from third-order to second-order \cite{jennings2025a}. This suggests that the $\mathrm{Ma}$-determined LBM nonlinearity is not entirely independent of $\mathrm{Re}$.  Nonetheless, we continue our journey up the river and deeper into the heart of darkness in pursuit of \textit{quantum resource estimates}.

\clearpage
\begin{landscape}
\begin{table}[H]
  \centering
    \makebox[\textwidth]{\begin{tabular}{|l|l|l|l|l|l|l|l|l|}
    \hline
    $\mathrm{Re}$ & $1\times 10^1$ & $1\times 10^2$ & $1\times 10^3$ & $1\times 10^4$ & $1\times 10^5$ & $1\times 10^6$ & $1\times 10^7$ & $1\times 10^8$ \\
    \hline
    $\nu$ (\si{m^2/s}) & $1.003\times10^{-6}$ & $1.003\times10^{-6}$ & $1.003\times10^{-6}$ & $1.003\times10^{-6}$ & $1.003\times10^{-6}$ & $1.003\times10^{-6}$ & $1.003\times10^{-6}$ & $1.003\times10^{-6}$ \\
    \hline
    $u$ (\si{m/s}) (see \eqref{eq:velocity_as_function_of_Re}) & $1.003\times 10^{-5}$ & $1.003\times 10^{-4}$ & $1.003\times 10^{-3}$ & $1.003\times 10^{-2}$ & $1.003\times 10^{-1}$ & $1.003$ & $1.003\times 10^{1}$ & $1.003\times 10^{2}$ \\
    \hline
    $u_{\text{max}}= 105\% \times u$ & $1.053\times 10^{-5}$ & $1.053\times 10^{-4}$ & $1.053\times 10^{-3}$ & $1.053\times 10^{-2}$ & $1.053\times 10^{-1}$ & $1.053$ & $1.053\times 10^{1}$ & $1.053\times 10^{2}$ \\
    \hline
    $L$ - char. length in $m$ & $1    $ & $1    $ & $1    $ & $1    $ & $1    $ & $1    $ & $1    $ & $1$ \\
    \hline
    $t^\star_{\text{adv}}$ - advection time scale in $s$ & $9.495\times 10^{4}$ & $9.495\times 10^{3}$ & $9.495\times 10^{2}$ & $9.495\times 10^{1}$ & $9.495$ & $9.495\times 10^{-1}$ & $9.495\times 10^{-2}$ & $9.495\times 10^{-3}$ \\
    \hline
    $T^\star = 2t^\star_{\text{adv}}$ - evolution time in $s$ & $1.899\times 10^{5}$ & $1.899\times 10^{4}$ & $1.899\times 10^{3}$ & $1.899\times 10^{2}$ & $1.899\times 10^{1}$ & $1.899$ & $1.899\times 10^{-1}$ & $1.899\times 10^{-2}$ \\
    \hline
    D3Q27 lattice speed of sound $c_s$ & $1/\sqrt{3}$& $1/\sqrt{3}$& $1/\sqrt{3}$ & $1/\sqrt{3}$& $1/\sqrt{3}$& $1/\sqrt{3}$ & $1/\sqrt{3}$& $1/\sqrt{3}$ \\
    \hline
    $\rho_0$ - initial lattice density & 1 & 1 & 1 & 1 & 1 & 1 & 1 & 1 \\
    \hline
    $\tau$ & $0.6  $ & $0.6  $ & $0.6  $ & $0.6  $ & $0.6  $ & $0.6  $ & $0.6  $ & $0.6$ \\
    \hline
    $\Delta x = L/\text{Re}$ in $m$ & $1\times 10^{-1}$ & $1\times 10^{-2}$ & $1\times 10^{-3}$ & $1\times 10^{-4}$ & $1\times 10^{-5}$ & $1\times 10^{-6}$ & $1\times 10^{-7}$ & $1\times 10^{-8}$ \\
    \hline      
    $\Delta t$ in $s$ & $3.323\times 10^{2}$ & $3.323$ & $3.323\times 10^{-2}$ & $3.323\times 10^{-4}$ & $3.323\times 10^{-6}$ & $3.323\times 10^{-8}$ & $3.323\times 10^{-10}$ & $3.323\times 10^{-12}$ \\
    \hline      
    $T = T^{\star}/\Delta t$ - lattice evolution time & $5.714\times 10^{2}$ & $5.714\times 10^{3}$ & $5.714\times 10^{4}$ & $5.714\times 10^{5}$ & $5.714\times 10^{6}$ & $5.714\times 10^{7}$ & $5.714\times 10^{8}$ & $5.714\times 10^{9}$ \\
    \hline
    Approx. lattice max velocity & $0.035$ & $0.035$ & $0.035$ & $0.035$ & $0.035$ & $0.035$ & $0.035$ & $0.035$ \\
    \hline      
    Approx. lattice mach number & $0.061$ & $0.061$ & $0.061$ & $0.061$ & $0.061$ & $0.061$ & $0.061$ & $0.061$ \\
    \hline
    $n_x$ & $1.000\times 10^{2}$ & $1.000\times 10^{3}$ & $1.000\times 10^{4}$ & $1.000\times 10^{5}$ & $1.000\times 10^{6}$ & $1.000\times 10^{7}$ & $1.000\times 10^{8}$ & $1.000\times 10^{9}$ \\
    \hline
    $n_y$ & $8.000\times 10^{1}$ & $8.000\times 10^{2}$ & $8.000\times 10^{3}$ & $8.000\times 10^{4}$ & $8.000\times 10^{5}$ & $8.000\times 10^{6}$ & $8.000\times 10^{7}$ & $8.000\times 10^{8}$ \\
    \hline
    $n_z$ & $8.000\times 10^{1}$ & $8.000\times 10^{2}$ & $8.000\times 10^{3}$ & $8.000\times 10^{4}$ & $8.000\times 10^{5}$ & $8.000\times 10^{6}$ & $8.000\times 10^{7}$ & $8.000\times 10^{8}$ \\
    \hline
    $n = n_x n_y n_z$ & $6.400\times 10^{5}$ & $6.400\times 10^{8}$ & $6.400\times 10^{11}$ & $6.400\times 10^{14}$ & $6.400\times 10^{17}$ & $6.400\times 10^{20}$ & $6.400\times 10^{23}$ & $6.400\times 10^{26}$ \\
    \hline
    $Q$ & 27 & 27 & 27 & 27 & 27 & 27 & 27 & 27 \\
    \hline
    $nQ$ & $1.728\times 10^{7}$ & $1.728\times 10^{10}$ & $1.728\times 10^{13}$ & $1.728\times 10^{16}$ & $1.728\times 10^{19}$ & $1.728\times 10^{22}$ & $1.728\times 10^{25}$ & $1.728\times 10^{28}$ \\
    \hline
    $nQ + (nQ)^2 + (nQ)^3$ & $5.160\times 10^{21}$ & $5.160\times 10^{30}$ & $5.160\times 10^{39}$ & $5.160\times 10^{48}$ & $5.160\times 10^{57}$ & $5.160\times 10^{66}$ & $5.160\times 10^{75}$ & $5.160\times 10^{84}$ \\
    \hline
    $n_f$ & $6.395\times 10^5$ & $6.395\times 10^8$ & $6.395\times 10^{11}$ & $6.395\times 10^{14}$ & $6.395\times 10^{17}$ & $6.395\times 10^{20}$ & $6.395\times 10^{23}$ & $6.395\times 10^{26}$ \\
    \hline
    $||A||_2$ (appendix \ref{sec:bounding_spectral_norm_of_A}) & $901$ & $901$ & $901$ & $901$ & $901$ & $901$ & $901$ & $901$ \\
    \hline
    $h$    & $0.0011$ & $0.0011$ & $0.0011$ & $0.0011$ & $0.0011$ & $0.0011$ & $0.0011$ & $0.0011$ \\
    \hline
    $||b||_2$ & $0    $ & $0    $ & $0    $ & $0    $ & $0    $ & $0    $ & $0    $ & $0$ \\
    \hline
    $\varepsilon_\rho$ \eqref{eq:lbm-approx-inv-rho}, \eqref{eq:phi_bounds}  & $0.001$ & $0.001$ & $0.001$ & $0.001$ & $0.001$ & $0.001$ & $0.001$ & $0.001$ \\
    \hline
    $\phi_\mathrm{min}$ \eqref{eq:phi_bounds} & $3.64\times 10^6$ & $1.15\times 10^{11}$ & $3.64\times 10^{15}$ & $1.15\times 10^{20}$ & $3.64\times 10^{24}$ & $1.15\times 10^{29}$ & $3.64\times 10^{33}$ & $1.15\times 10^{38}$ \\
    \hline
    $\phi_\mathrm{max}$ \eqref{eq:phi_bounds} & $5.13\times 10^8$ & $1.62\times 10^{13}$ & $5.13\times 10^{17}$ & $1.62\times 10^{22}$ & $5.13\times 10^{26}$ & $1.62\times 10^{31}$ & $5.13\times 10^{35}$ & $1.62\times 10^{40}$ \\
    \hline
    $C_\mathrm{max}(T)$ & $1$ & $1$ & $1$ & $1$ & $1$ & $1$ & $1$ & $1$ \\
    \hline
    \end{tabular}}
  \caption{Parameters used for resource estimates for flow-past-a-sphere problem instances. Note we have chosen to set $C_{\textup{max}}(T)$ to be 1 throughout. Future work may provide a refined estimate.}
  \label{tab:resource_estimate_parameters_sphere}
\end{table}
\end{landscape}

\begin{table}[p]
\centering
\begin{tabular}{|l|l|l|l|}
\hline
 Problem instance & JBC & KCS & MV Regal \\
\hline
$\mathrm{Re}$ & $8.23\times 10^6$ & $6.64\times 10^6$ & $9.91\times 10^8$ \\
\hline
        Reynolds Number (Re) & \( 8.23\times 10^6 \) & \( 6.64\times 10^6 \) & \( 9.91\times 10^8 \)  \\
        \hline
        $\nu$ Kinematic Viscosity in $m^2/s$ & 
        $1.003\times10^{-6}$ & 
        $1.003\times10^{-6}$ & 
        $1.003\times10^{-6}$\\
        \hline
        $u$ - velocity in $m/s$ & 1.179 & 0.915 & 7.202 \\
        \hline
        $u_{\text{max}}= 105\% \times u$ in $m/s$ & 1.2380 & 0.96075 & 7.562 \\
        \hline
        $L$ - char. length in $m$ & 7 & 7.27 & 138 \\
        \hline
        $t^\star_{\text{adv}}$ - advection time scale in $s$ & 6 & 7 & 18 \\ 
        \hline
        $T^\star = 2t^\star_{\text{adv}}$ - evolution time in $s$ & 12 & 14 & 36 \\ 
        \hline
        D3Q27 lattice speed of sound $c_s$ & $1/\sqrt{3}$& $1/\sqrt{3}$& $1/\sqrt{3}$\\
        \hline
        $\rho_0$ & 1 & 1 & 1 \\
        \hline
        $\tau$ & 0.6 & 0.6 & 0.6 \\
        \hline
        $\Delta x$ in $m$ & $5\times10^{-6}$ & $5\times10^{-6}$ & $5\times10^{-7}$ \\
        \hline      
        $\Delta t$ in $s$ & $8.3084\times10^{-7}$ & $8.3084\times10^{-7}$ & $8.3084\times10^{-9}$ \\
        \hline      
        $T = T^{\star}/\Delta t$ - lattice evolution time & $1.44 \times 10^{7}$ & $1.68 \times 10^{7}$ & $4.33 \times 10^{9}$ \\
        \hline      
        Approx. lattice max velocity & 0.21 & 0.16 & 0.13\\
        \hline      
        Approx. lattice mach number & 0.36 & 0.28 & 0.22 \\
        \hline
        $n_x$ & $14/\Delta x = 2.8\times10^{6}$ & $14/\Delta x = 2.8\times10^{6}$ & $300/\Delta x = 6\times10^{8}$ \\
        \hline
        $n_y$ & $4/\Delta x = 8\times10^{5}$ & $4/\Delta x = 8\times10^{5}$ & $70/\Delta x = 1.4\times10^{8}$ \\
        \hline
        $n_z$ & $1.4125/\Delta x = $    & $1.34178/\Delta x =$ & $15.25/\Delta x=$ \\
              & $2.825\times10^{5}$  & $2.6836\times10^{5}$ & $3.05\times10^{7}$ \\
        \hline
        $n = n_x n_y n_z$ & $6.328\times10^{17}$ & $6.0112\times10^{17}$ & $2.5620\times10^{24}$ \\
        \hline
        $Q$ & 27 & 27 & 27 \\
        \hline
        $nQ$ & $1.7086\times10^{19}$ & $1.6230\times10^{19}$ & $6.9174\times10^{25}$\\
        \hline
        $nQ + (nQ)^2 + (nQ)^3$ 
        & $4.9876\times10^{57}$ 
        & $4.2754\times10^{57}$ 
        & $3.31\times10^{77}$ \\
\hline
$n_f$ & \(6.068 \times 10^{17}\) & \(5.806 \times 10^{17}\) & \(2.429 \times 10^{24}\) 
\\
\hline
$||A||_2$ (appendix \ref{sec:bounding_spectral_norm_of_A}) & 901 & 901 & 901 
\\
\hline
$h$ & 0.0011 & 0.0011 & 0.0011 
\\
\hline
$||b||_2$ & 0 & 0 & 0 
\\
\hline
$\varepsilon_\rho$ \eqref{eq:lbm-approx-inv-rho}, \eqref{eq:phi_bounds} & 0.001 & 0.001 & 0.001 
\\
\hline
$\phi_{\text{min}}$ \eqref{eq:phi_bounds} & $3.37 \times 10^{24}$ & $3.15 \times 10^{24}$ & $2.70 \times 10^{34}$ 
\\
\hline
$\phi_{\text{max}}$ \eqref{eq:phi_bounds} & $4.74 \times 10^{26}$ & $4.44 \times 10^{26}$ & $3.80 \times 10^{36}$ 
\\
\hline
$C_{\text{max}}(T)$ & 1 & 1 & 1 
\\
\hline
\end{tabular}
\caption{Parameters used for resource estimates for ship hull problem instance. The number of fluid nodes $n_f$ is approximated for ship hulls with bounding rectangle around ship hull. Note we have chosen to set $C_{\textup{max}}(T)$ to be 1 throughout. We leave it to future work to establish a more accurate bound for this quantity and, according to \cite{jennings2023b}, our $T$-gate count resources should be multiplied by the more-accurate estimate of $C_{\textup{max}}(T)$.}
\label{tab:resource_estimate_parameters_hulls}
\end{table}

\clearpage

\begin{table}[H]
    \centering
\begin{tabular}{c c c c}
    \hline
    \textbf{Re} & \(\Delta t\) & \(T_c^{\star} \) & \(T_c^{\star}\)\\
                &              & lower bound & upper bound\\
    \hline
    $10^1$ & $3.323 \times 10^{2}$ & $3.676 \times 10^{-2}$ & $4.112 \times 10^{-2}$ \\
    $10^2$ & $3.323 \times 10^{0}$ & $3.676 \times 10^{-4}$ & $4.112 \times 10^{-4}$ \\
    $10^3$ & $3.323 \times 10^{-2}$ & $3.676 \times 10^{-6}$ & $4.112 \times 10^{-6}$ \\
    $10^4$ & $3.323 \times 10^{-4}$ & $3.676 \times 10^{-8}$ & $4.112 \times 10^{-8}$ \\
    $10^5$ & $3.323 \times 10^{-6}$ & $3.676 \times 10^{-10}$ & $4.112 \times 10^{-10}$ \\
    $10^6$ & $3.323 \times 10^{-8}$ & $3.676 \times 10^{-12}$ & $4.112 \times 10^{-12}$ \\
    $10^7$ & $3.323 \times 10^{-10}$ & $3.676 \times 10^{-14}$ & $4.112 \times 10^{-14}$ \\
    $10^8$ & $3.323 \times 10^{-12}$ & $3.676 \times 10^{-16}$ & $4.112 \times 10^{-16}$ \\
    \hline
    \multicolumn{4}{l}{When $\tau=0.6$ and $\|\phi_0 \|_\infty=0.2958$.}\\
    \hline
\end{tabular}
    \caption{Upper and lower bounds for the evolution time $T_c^{\star}$ in physical units (seconds) that the solution error from truncating the Carleman linearized system can be driven arbitrarily low by increasing the truncation order.  The estimate for $\|\phi_0 \|_\infty=0.2958$ is derived from using \eqref{eq:lbm-feq-coeffs} to construct population distributions $f$ such that $\sum_i f_i \mathbf{c}_i=(u_x, 0, 0)$ where $u_x$ is stated in table \ref{tab:resource_estimate_parameters_sphere}, but converted to lattice units \cite{source_for_norm_A_anlysis}.
    }
    \label{tab:time_domain_T_convergence_bounds_for_sphere}
\end{table}

\subsection{Quantum resource estimates}
\label{subsec:qres}

In this section we generate numerical estimates of the quantum resources required to estimate the drag force for the flow-past-a-sphere problem instances using a quantum computer.
The resource estimates are reported in Figures \ref{fig:resource_estimates_unstructured_block_encodings} and \ref{fig:resource_estimates_bespoke_block_encodings}.
The resource estimates only give a rough assessment of the feasibility of using quantum computers to unlock utility for this problem. We generate resource breakdowns (Figure \ref{fig:flame_graph}) that give numbers of calls to various subroutines and serve as an algorithmic profile of the quantum computation.

Figures \ref{fig:resource_estimates_unstructured_block_encodings} and \ref{fig:resource_estimates_bespoke_block_encodings} display the total estimated $T$-gate count and the total logical qubit count for each of the flow-past-a-sphere instances (parameterized by Reynolds number) in table \ref{tab:resource_estimate_parameters_sphere}. The $T$-gate count serves as a proxy for the runtime of the quantum computation as the actual quantum computation runtime is expected to scale roughly proportionally to this number.
The logical qubit count serves as a proxy for the spatial resources needed to run the quantum computation.
The methodology underlying these resource counts is described in section \ref{subsec:methodolgy} and the software used to generate the estimates can be found at \cite{qsub_package}.  It has been shown for quantum chemistry applications that orders of magnitude can be removed from the quantum resources with efficient block encodings. In this case of our LBM formulation, efficient block encodings for the the streaming and the collision matrices also provide the most significant reduction in resources. In quantum chemistry applications, improvements can be made by taking advantage of symmetries in the Hamiltonian.  In our case, we take advantage of efficient classical descriptions of the matrix elements whose values can be encoded by careful use of tensor algebra and arithmetic (even though subsections of the matrices are dense).

Figures \ref{fig:resource_estimates_unstructured_block_encodings}  and \ref{fig:resource_estimates_bespoke_block_encodings} shows that the estimated $T$-gate counts are in the range of $\lowTgateCountEst$--$\highTgateCountEst$ and 
the estimated number of logical qubits ranges from thousands to hundreds of thousands.  The resource estimates for the toy CFD problem are in the same range as resource estimates for utility-scale quantum chemistry applications \cite{goings2022reliably}.  

Figure \ref{fig:flame_graph} shows a breakdown of the numbers of calls to subroutines used in the drag estimation quantum algorithm and is used to highlight sources of high costs. 
There are two main findings to discuss from the breakdown plots.
First, the breakdowns show that there are numerous factors that \textit{multiplicatively} contribute to the overall $T$-gate counts.
Accordingly, any reductions to the number of calls in one layer translate into a reduction in the total $T$-gate counts.
For example, if the calls to each of five subroutines are reduced by a factor of four, then the total $T$-gate count is reduced by roughly a factor of one thousand. 
Second, the breakdowns show that, among the subroutines, the largest possible reductions could come from reducing (1) the number of Carleman block encodings per ODE solution preparation and (2) the number of $T$-gates per Carleman block encoding. 
In the case of the former, recent quantum algorithms for differential equations have been developed \cite{Fang2023, An2023, pocrnic2025constantfactorimprovementsquantumalgorithms,dalzell2024shortcutoptimalquantumlinear}  that differ from the methods used in our resource estimates. In the case of the latter, we suspect that block encodings of the $F$-matrices cannot be improved more than the bespoke block encodings described in \ref{appendix_d: bespoke_block_encoding}. However, even with bespoke block encodings, the resource estimates are prohibitive.  

\afterpage{
\begin{figure}[H]
    \centering
        \includegraphics[width=1\linewidth]{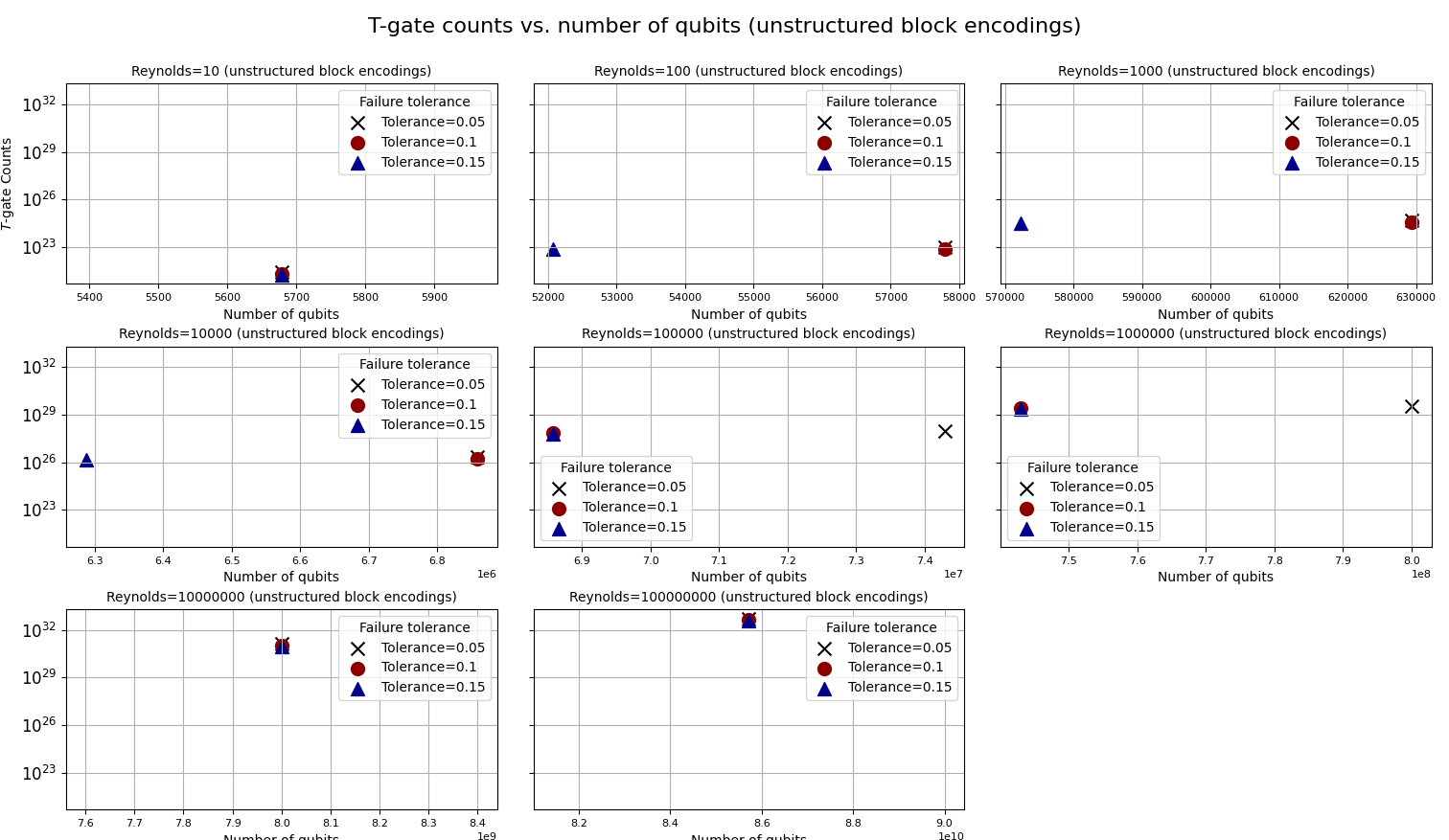}
        \caption{Number of  $T$-gates required for different Reynolds numbers using generic unstructured block encodings.}
        \label{fig:resource_estimates_unstructured_block_encodings}
\end{figure}

\begin{figure}[H]
    \centering
        \includegraphics[width=1\linewidth]{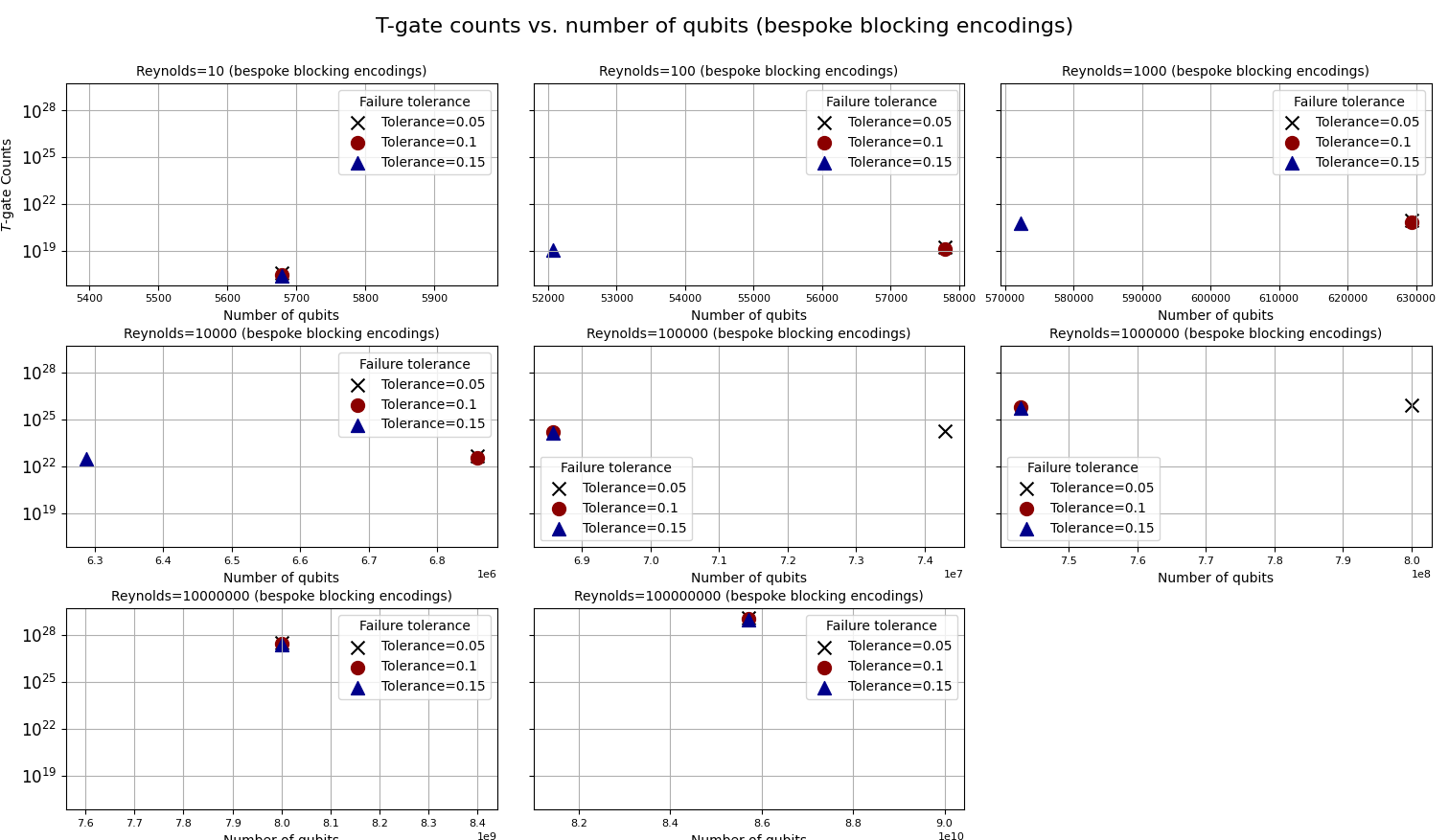}
        \caption{Number of  $T$-gates required for different Reynolds numbers using the bespoke block encodings.}
        \label{fig:resource_estimates_bespoke_block_encodings}
\end{figure}

\newpage
}

\begin{figure}[H]
    \centering
    \subfloat[The breakdown of number of calls as shown in figure \ref{fig:drag_est_call_graph}. We are plotting the number of calls to the different subroutines. The height of a particular box is how many calls one call of the subroutine below has to make.]{
       \includegraphics[width=0.6\linewidth]{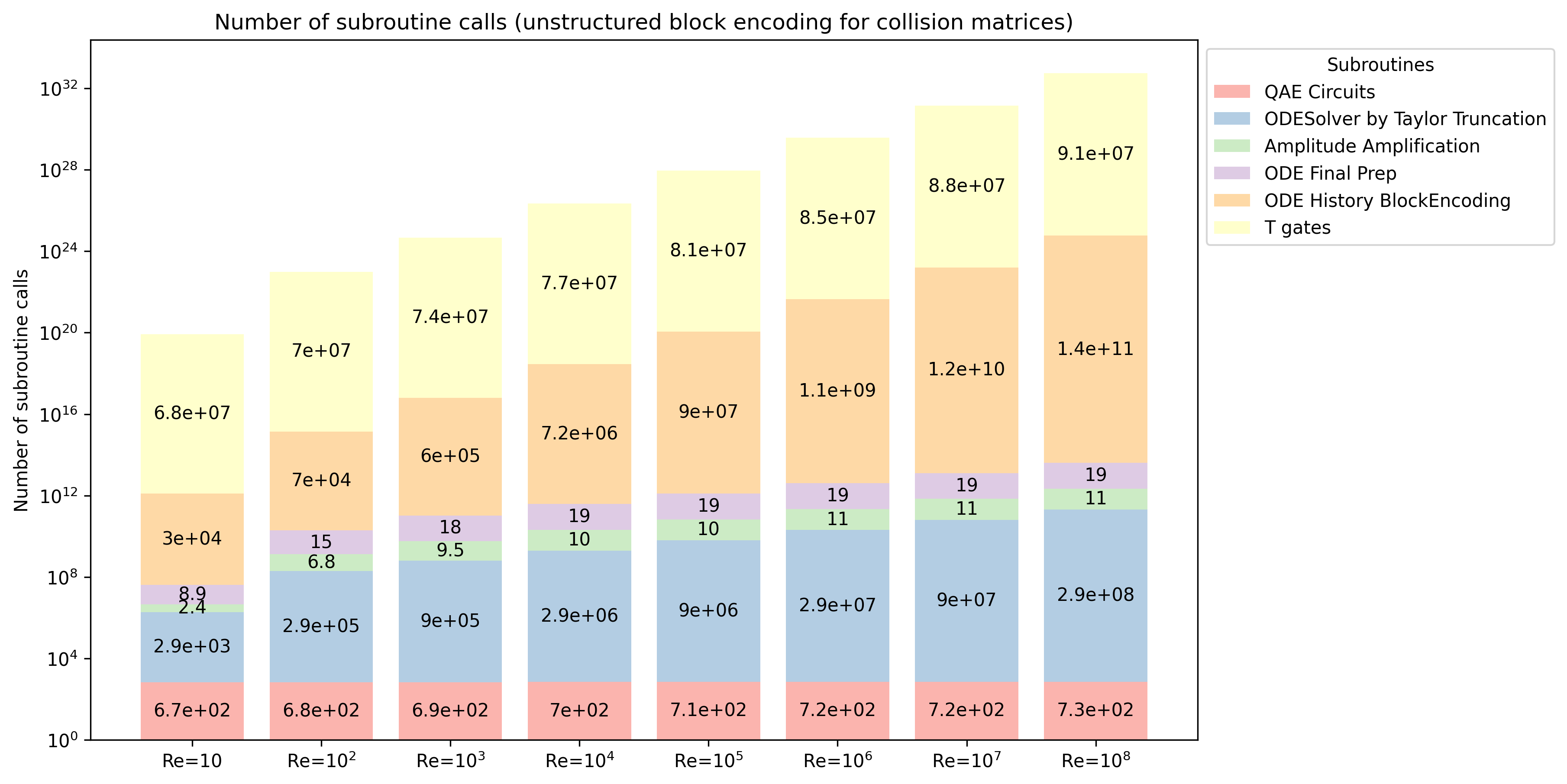}
        \label{fig:flame_graph_unstructured_block_encodings}
    }
        \vspace{0.5cm}
        
    \subfloat[The breakdown of number of calls as shown in figure \ref{fig:drag_est_call_graph}. For this plot we are using the bespoke block encodings for the collision $F$-matrices in \ref{appendix_d: bespoke_block_encoding}. For these block encodings, there are 3-4 orders of magnitude savings.]{
      \includegraphics[width=0.6\linewidth]{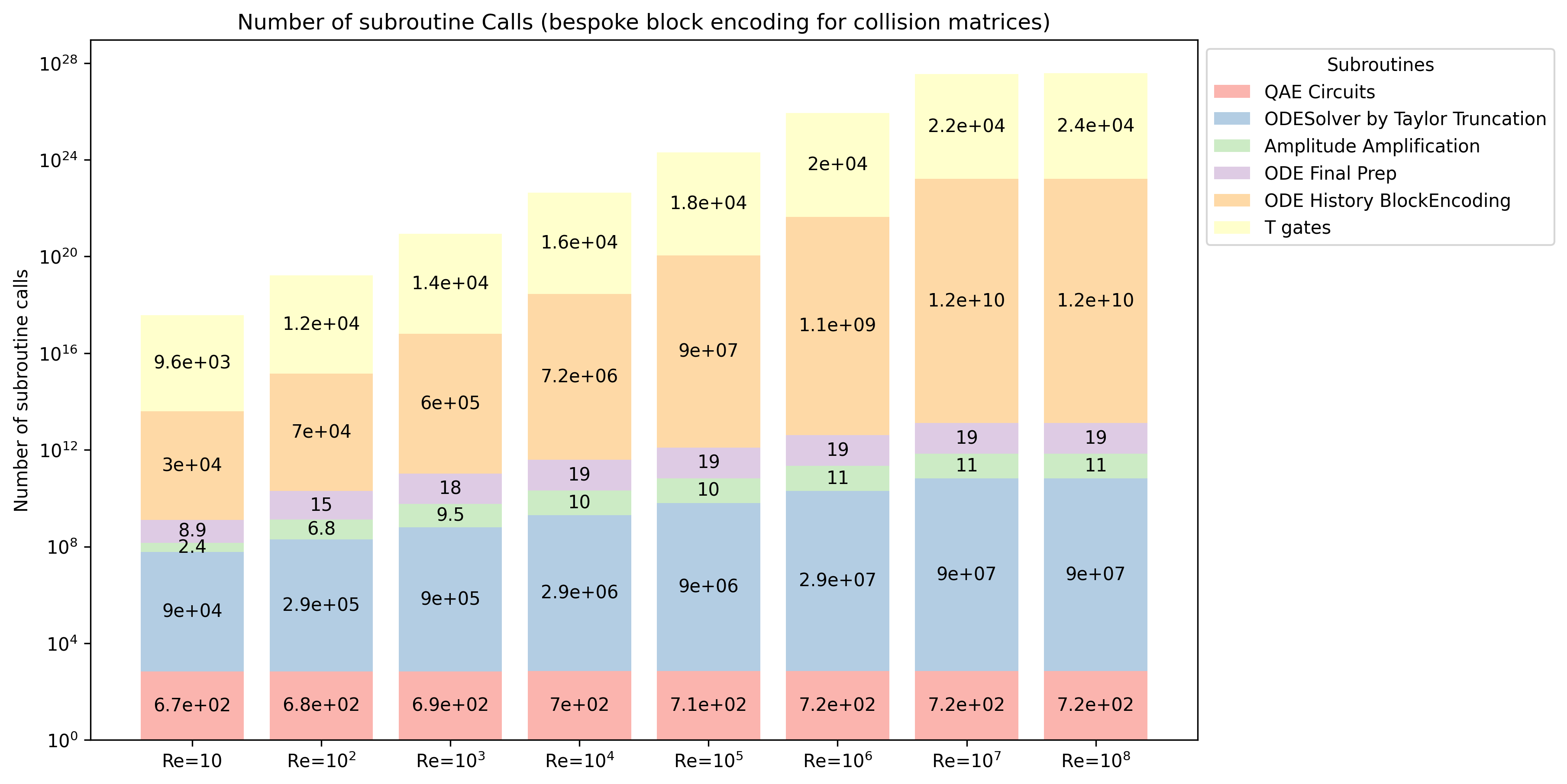}
      \label{fig:flame_graph_bespoke_block_encodings}
    }
    \caption{Resource breakdowns for the flow-past-a-sphere problem instances for the subroutines shown in figure \ref{fig:drag_est_call_graph}.}
    \label{fig:flame_graph}
\end{figure}

To reiterate, the shortcomings of our analysis are
\begin{itemize}    
    \item Resource estimates for state preparation subroutines are not included. This includes the preparation of the state $\ket{\vec{\mathcal{F}_x}}$, initial states used in the unstructured block encoding and initial state preparation for the ODE solver.
    \item We cannot guarantee solution error past a small evolution time $T_c$ (see table \ref{tab:time_domain_T_convergence_bounds_for_sphere}). Furthermore, it was recently shown in \cite{jennings2025a} that even with an improved quantum algorithm based on direct linearization of the discrete-time LBM and the nonlinearity reduced to second order, Carleman truncation fails to converge above some threshold Reynolds number.
    \item We set the parameter $C_{\text{max}}(T) = 1$ (defined in table \ref{tab:resource-estimate-parameter-definitions}) to 1. This plays an important role in the upper bound on the condition number, which in turn impacts the number of calls to the quantum linear solver. Setting the value to 1 is in line with estimates given by \cite{Li2025}.
    \item Our continuous-time formulation of the LBM in \eqref{eq:partial-derivative-LBM-streaming} omits a second-order memory term needed to preserve second-order accuracy with respect to $\Delta t$ \cite{Ubertini2010a, jennings2025a}. Our procedure for setting $\Delta t$ according to \eqref{eq:lbm_param_consistency_7.14} does not satisfy the first-order accuracy requirement $\Delta t \ll \tau$ at low $\mathrm{Re}$ (see Table \ref{tab:time_domain_T_convergence_bounds_for_sphere}).
    \item Our approach to extracting the drag force neglects the fact that higher-order ($k>1$) blocks $f^{\otimes k}$ of the Carleman state vector are no longer equivalent to tensor powers of $f$ after the system has been time-evolved. As noted in \cite{jennings2025a}, amplitude will be exponentially concentrated in the highest-order term (in our case $f^{\otimes 3}$), which has only evolved under the action of the linear part of the LBM. The extracted drag force would therefore correspond to a CFD simulation in which nonlinearities are neglected entirely.
\end{itemize}

\subsection{Relationship between quantum resource estimates and Reynolds number}\label{sec:log_log_relationship_resource_estimates_vs_reynolds_number}

As stated earlier, the computational cost of classical DNS grows rapidly with increasing Reynolds number as $\sim O(\mathrm{Re}^3)$, becoming intractable in the high-turbulence regime $\mathrm{Re}>10^4$. To provide utility, the cost of DNS-accurate solutions calculated on a quantum computer must have better scaling.  Figure \ref{fig:logQRE_vs_logRe} shows a linear relationship between $\log(\text{Re})$ and $\log( (\text{logical qubits)$\times$($T$-gate estimate}))$ with approximate slope 2.69.  This is indicative of a power-law relationship $\sim O(\mathrm{Re}^{2.69})$.  By contrast, a polylogarithmic scaling (i.e., exponential speedup) would have manifested as a sublinear curve.  While the exponent in the quantum algorithm is better than that in the classical DNS scaling (2.69 vs. 3, respectively), this weak polynomial speedup will not lead to practical utility when error-correction overheads are considered.  For example, see the analysis of quadratic speedups by Babbush \textit{et al.} \cite{Babbush2021a}.

\begin{figure}[H]
    \centering
    \includegraphics[width=0.5\textwidth]{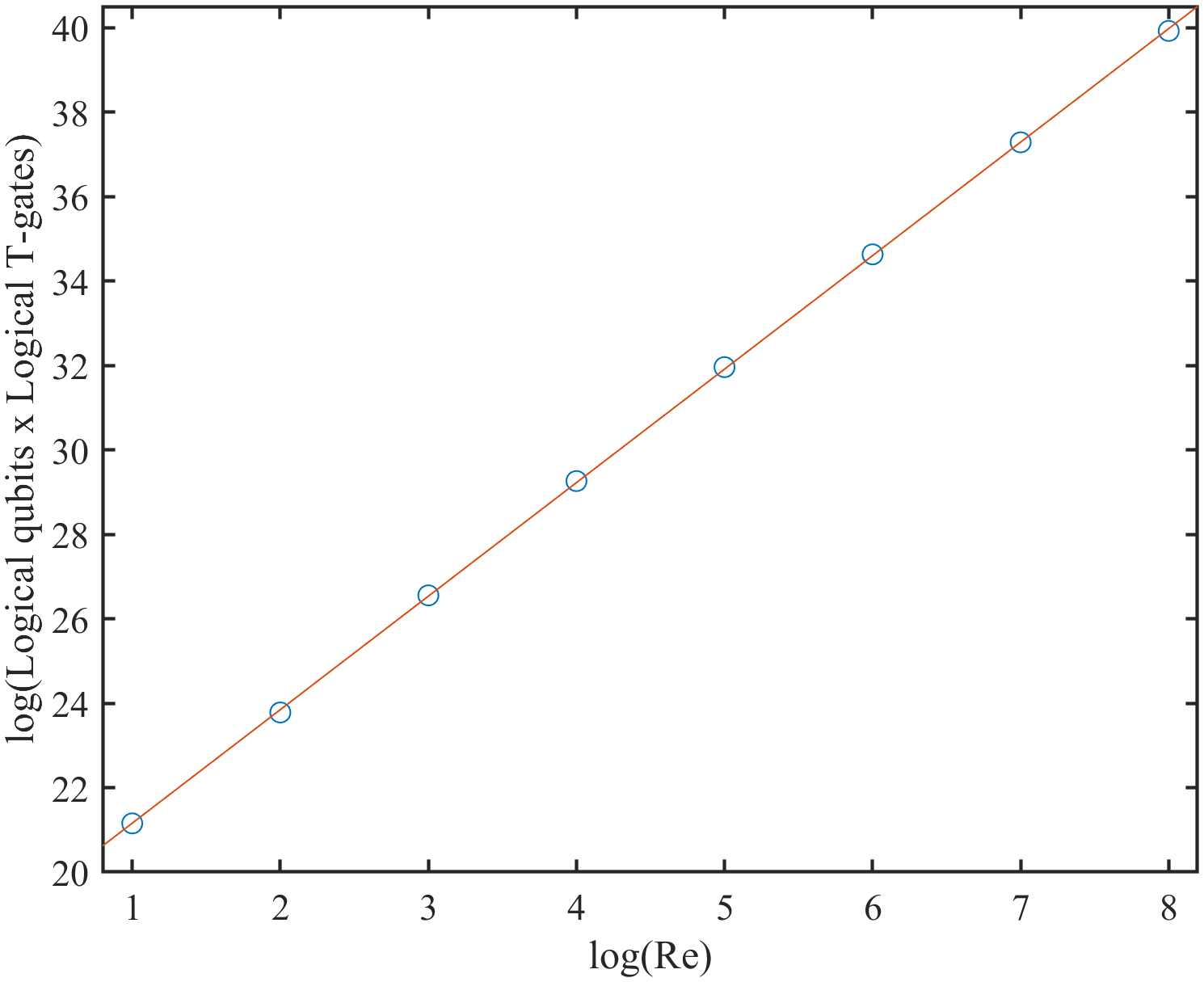}
    \caption{log(Logical QREs) vs. log(Reynolds number) shows a linear relationship with approximate slope 2.69, indicating a power-law relationship.}
    \label{fig:logQRE_vs_logRe}
\end{figure}

The lack of an exponential speedup observed in our QREs can be understood by carefully revisiting the work of Li \textit{et al.} \cite{Li2025} and Krovi \cite{Krovi2023}. They found the query complexity (number of two-qubit quantum gates and oracle queries) to be
linear in the number of lattice time steps ($T$, in our notation) and polylogarithmic in the number of spatial lattice nodes $n$. This suggests an exponential speedup with respect to $n$, since classical DNS complexity is approximately linear in $n$. However, Li \textit{et al.} identify a crucial caveat: If $T$ has implicit polynomial (or worse) dependence on $n$, this will negate the exponential speedup.

We find that $T$ indeed has implicit polynomial dependence on $n$ for the problem instances we consider. This may be understood as a consequence of the diffusive scaling $\Delta t \propto \Delta x^2$ exhibited in \eqref{eq:lbm_param_consistency_7.14}, which is necessary for ensuring the stability and accuracy of the LBM in the incompressible regime \cite{Krueger2017TheLB}. Hence, as $\Delta x$ decreases with increasing Reynolds number, $\Delta t$ also decreases, thereby increasing $T$ because $T=T^\star/\Delta t$. Since $n \propto \frac{1}{\Delta x^{3}}$ and $T$ depends on $\Delta x$, an implicit dependence of $T$ on $n$ is mediated by the Reynolds number. In our case, combining \eqref{eq:velocity_as_function_of_Re}, \eqref{eq:lbm_param_consistency_7.14}, and \eqref{eq:Dx} as
\begin{subequations}
    \begin{align}
        T &= T^{\star}/\Delta t & (T^{\star} \text{ is physical time})\\
         &= \frac{2L}{u} \left(\frac{1}{\Delta t}\right) & (t^{\star}_{\text{adv}}=L/u\text{ from \eqref{eq:advective_time_scale} and } T^{\star} \text{ set to } 2t^{\star}_{\text{adv}})\\
         &= \frac{2L^2}{\nu \text{Re}} \left(\frac{1}{\Delta t}\right) & (u=\nu \text{Re}/L \text{ from \eqref{eq:velocity_as_function_of_Re}})\\
         &= \frac{2L \Delta x}{\nu} \left(\frac{1}{\Delta t}\right) & (\text{Re}=L/\Delta x \text{ from \eqref{eq:Dx} })\\
         &= \left( \frac{2L}{c_s^{2}(\tau - 0.5)} \right) \left(\frac{1}{\Delta x}\right)& (\text{after incorporating \eqref{eq:lbm_param_consistency_7.14} for } \nu)\\
         &= \left( \frac{2L}{c_s^{2}(\tau - 0.5)} \right) \left(\frac{1}{V^{1/3}}\right) n^{1/3}& \text{(where $V$ is the total volume and $n = V/\Delta x^3$)}\\
         &\propto n^{1/3},
    \end{align}
    \label{eq:T_propto_n1/3_for_our_params}
\end{subequations}
shows the dependence $T \propto n^{1/3}$, which negates any exponential speedup.

This lack of exponential speedup was recently confirmed in a similar study wherein the asymptotic query complexity of an improved quantum LBM algorithm was found to exhibit only a weak polynomial advantage over classical DNS \cite{jennings2025a}. The linear system matrix's condition number was found to scale at least linearly in $n$, and hence polynomially in Reynolds number. This is in fact equivalent to our analysis, as the linear dependence on $T$ (and its implicit dependence on $n$) in the query complexity actually arises from a linear dependence on condition number (see \cite{Krovi2023}).

\subsection{Prospects for quantum utility in computational fluid dynamics} \label{sec:prospects}

Even if we relaxed the requirement of consistent compressibility error and diffusive scaling and instead scaled the grid spacing as $\Delta x \propto 1/\text{Re}^\alpha$, where $\alpha$ is left unspecified, the exponential speedup would still be negated by a dependence of evolution time $T$ on the number of grid nodes $n$.  The implicit dependence of $T$ on $n$ occurs in any explicit time-marching method to evolve a system of differential equations due to the Courant–Friedrichs–Lewy (CFL) condition.  The CFL condition requires time discretization $\Delta t \le \Delta x / u_\text{max}$ so that any effects traveling at maximum speed $u_\text{max}$ do not propagate beyond the neighboring lattice node in a single time step.  Hence, $\Delta t$ (on which the number of evolution time steps $T$ depends) depends on $\Delta x$ (on which the number of grid nodes $n$ depends).  We will show that $T$ grows as a power law in the number of grid nodes $n$.  First, note that the number of grid nodes is
\begin{subequations}
\begin{align}
    n &= \gamma \frac{L^3}{\Delta x^3}\\
    n  &=\gamma \text{Re}^{3\alpha}  \iff \text{Re}=(n/\gamma)^{\frac{1}{3\alpha}},
    \label{eq:n_grid_dependence_on_re}
\end{align}
\end{subequations}
where $\gamma$ is some real constant.  Then the lattice evolution time $T$ is
\begin{subequations}
    \begin{align}
        T &= T^{\star}/\Delta t & \text{( $T^{\star}$ is physical evolution time)} \\
         &= \left( \frac{T^{\star} }{C_\text{max}} \right ) \frac{u}{\Delta x} & \text{(incorporating the CFL condition as equality $\Delta t = C_\text{max} \Delta x / u$)} \\
         &= \left( \frac{T^{\star} \nu }{C_\text{max} L } \right ) \frac{\text{Re}}{\Delta x} & \text{(because $u = \text{Re}{\frac{\nu}{L}}$)} \\
         &= \left( \frac{T^{\star} \nu }{C_\text{max} L^2 } \right ) \text{Re}^{1+\alpha} & \text{(as we set $\text{Re}^{\alpha} = L/\Delta x$)} \\
         &= \left( \frac{T^{\star} \nu }{C_\text{max} L^2} \right ) \left( \frac{n}{\gamma} \right)^{\frac{1+\alpha}{3\alpha}} & \text{(from \eqref{eq:n_grid_dependence_on_re})} \\
         &\propto n^{\frac{1+\alpha}{3\alpha}}.
    \end{align}
    \label{eq:generic_cfl_argument_2}
\end{subequations}
To clarify, in \eqref{eq:T_propto_n1/3_for_our_params}, we arrived at $T\propto n^{1/3}$ as we incorporated specific choices of $T^{\star} = \frac{2L}{u}$ and $\alpha=1$ related to the problem instances we defined.  However, in \eqref{eq:generic_cfl_argument_2}, $T^{\star}$ is left as an unspecified constant.  If we specify the same values for $T^{\star}$ and $\alpha$ in \eqref{eq:generic_cfl_argument_2}, then the same $T\propto n^{1/3}$ relationship holds.

In the previous section we approached the argument using the parameter consistency requirements of the lattice Boltzmann method in \eqref{eq:lbm_param_consistency_7.14}. In this section we approached the argument using the more general CFL condition.  In either case, the number of evolution time steps $T$ grows as a power law in the number of grid points $n$. Thus, simulations requiring excessively fine spatial discretization will also require excessively fine temporal discretization. In the case of fluid dynamics simulations that cannot be fast-forwarded, this leads to an a power law scaling in query complexity for a quantum algorithm.

\subsection{Prospects for quantum utility beyond simple geometries}
\label{subsec:prospects_for_utility}

Now that we have considered the flow-past-a-sphere problem instances, we examine how much of the approach can be ported over to utility-scale instances. For the lattice Boltzmann method, the collision matrices are completely independent of the geometry. The number of blocks in the $F$-matrices will increase polynomially.  However, the streaming matrix will change drastically.  What allowed us to progress so far is the simple mechanism to identify where solid or fluid nodes are.  Because of the geometry, \eqref{eq:classical_function} was very easily turned into a quantum circuit since there is a simple arithmetic expression that needs to be evaluated. On the other hand, finding a quantum circuit that can identify the location of solid and fluid nodes in arbitrary geometries will be challenging. Methods to efficiently identify fluid/solid nodes will be critical. Although it was recently shown in \cite{Jennings2026a} that incorporating non-trivial geometries did not impact asymptotic scaling, this assumed geometries with a ``sufficiently compact functional description'' (numerical studies considered 2D flow past a cylinder) so as to enable an efficient block-encoding of the streaming matrix.

As can be noted between the differences in table \ref{tab:resource_estimate_parameters_sphere} (flow-past-a-sphere problem instances) and table \ref{tab:resource_estimate_parameters_hulls} (ship hull problem instances), important parameters like the number of grid points and the evolution times will drastically increase.  The inverse of the $\tau$ parameter enters the matrix elements and subtly increases subnormalization of block encodings.  While respecting the consistency relationship \eqref{eq:lbm_param_consistency_7.14}, future work may explore the trade offs between $\tau$, spatial and temporal discretization, and quantum resource estimates.
\section{Conclusion}\label{sec:conclusion}

In this study we have explored the potential for future fault-tolerant quantum computers to provide utility in applications of incompressible computational fluid dynamics. Choosing simulation-driven ship design as a representative application area, we formulated a set of specific problem instances to serve as benchmarks for evaluating the performance and utility of current and future computational methods. Though formally agnostic to computational approach (e.g., classical vs quantum), our set of problem instances were formulated with two key limitations of quantum differential equation algorithms in mind: The initial state and boundary conditions must have some degree of uniformity or geometrical simplicity to avoid the burdensome cost of data loading, and the outputs of interest should be low-degree polynomial functions (linear, in our case) of the solution vector to make the data extraction more efficient.

To unlock the estimated utility associated with the problem instances, a quantum CFD solver must significantly exceed the performance (runtime and accuracy) of current state-of-the-art classical CFD solvers.  Continual advances in classical CFD make this a moving target, though high-turbulence DNS remains highly intractable. We suggest that a transformational \textit{utility threshold} is the ability to solve CFD problems with DNS-level accuracy at runtimes comparable to current classical RANS solvers (in units of wall clock runtime).  Exceeding the transformational utility threshold would unlock the full utility. For the ship-hull problem instances specifically, we estimated the utility of replacing existing CFD in the shipbuilding industry to lie in the \$10M--\$100M range as a tight lower bound on utility. More broadly, the utility across all applications of incompressible CFD may approach $\sim$\$100B as a weak upper bound. We noted that even exceeding an incremental utility threshold tied to classical DNS runtimes would unlock some incremental utility though improved fundamental understanding of turbulence, though we estimate this utility to be $10\times$ to $100\times$ smaller.

As a surrogate for the runtime of future fault-tolerant quantum computers, we estimated the logical-level resources (qubits and $T$-gate counts) required for our quantum workflow. We did this for the geometrically simple problem instances involving flow past a sphere, as a waypoint toward utility-scale problem instances involving model- and ship-scale hulls. Our quantum approach consisted of the Carleman-linearized lattice Boltzmann method \cite{Li2025} paired with the linear-system-solver-based algorithms of \cite{Berry2017} and \cite{jennings2023b} to evolve a quantum state vector encoding the fluid phase space density to a final evolution time. Lastly, we presented a method for extracting the drag force measurable against the solution vector using amplitude estimation. Our resource estimates incorporated the costs of block-encoding the Carleman-linearized ODE, preparing the final solution vector, and estimating the drag force with quantum amplitude estimation.

We identified the following parameters that drive computational cost and quantum resource estimates.  A high Reynolds number (Re) implies a turbulent flow.  As turbulence increases, an increasingly fine mesh is required to resolve flow features.  This leads to a larger system of difference equations to solve over the mesh.  Since we only seek the steady-state solution, the evolution time is only applicable to the quantum LBM approach (the classical RANS-based approach is intrinsically steady-state) and clearly a longer evolution time will increase quantum resource estimates. Finally, the complexity of the geometry itself drives up the quantum resource estimates as all nodes in the LBM must be identified as either solid or fluid nodes for block encoding the bounce-back conditions of the streaming matrix $S$.  

Our logical-level quantum resource estimates for the problem instances involving flow past a sphere, in terms of $T$-gate counts range from $\lowTgateCountEst$ to $\highTgateCountEst$. 
These counts are relatively high compared to corresponding estimates for problems in quantum chemistry  that are around $10^{9}$ \cite{caesura2025a}, or $10^{10}$ \cite{goings2022reliably} $T$-gates.
Our estimated $T$-gate counts are also high in an absolute sense in that each $T$-gate is expected to take in the range of $\mu$s to ms \cite{beverland2022assessing}, with parallelization of these gates being at most the number of logical qubits (hundreds of thousands).
Considering the largest of the flow-past-a-sphere instances and taking the maximum possible parallelization, this leads to runtime estimates on the order of hundreds of thousands of years.  Through the resource estimation process, we have identified paths for potential resource optimization and algorithmic improvements, as discussed further in section \ref{sec:future}.

Of even greater concern is the lack of exponential speedup  in the scaling of our quantum resource estimates with respect to Reynolds number (a result also found by \cite{jennings2025a}). We found that the predicted polylogarithmic dependence on spatial grid size is negated by an intrinsic dependence between grid size and the number of time steps needed to reach steady state. This finding extends beyond the LBM-based quantum approach, to the broader case of explicit time-evolution in nonlinear differential equations that are subject to the CFL condition or similar convergence conditions linking time step size to spatial grid resolution. After a number of compilation optimizations there may yet be regimes where quantum computers could offer some small gains. If so, constant factor analysis of quantum differential equation solvers like \cite{pocrnic2025constantfactorimprovementsquantumalgorithms} might yet prove to be useful.

We also emphasize that the analysis throughout this work relies on several optimistic assumptions, and thus the reported resource estimates should be interpreted as lower bounds on the true costs.  For example, we assumed that the highest-order term in the Carleman encoding  $f^{\otimes K}$, which evolves only via the linear dynamics, is equivalent to $K$ copies of the first-order term $f$ that evolves under the full nonlinear dynamics, thereby neglecting the observable-estimation challenges discussed in \cite{costa2025further,jennings2025a}.

\subsection{Future work}\label{sec:future}

The recent study by \cite{jennings2025a} identified four major obstacles to unlocking utility via the quantum LBM approach of \cite{Li2025}:
\begin{enumerate}
    \item Carleman truncation fails to converge at high Reynolds numbers typical of turbulent flows
    \item No exponential quantum advantage with respect to Reynolds number
    \item A continuous-time formulation of the LBM requires prohibitively small time steps
    \item Inefficient data extraction
\end{enumerate} 
Although the improved quantum LBM algorithm presented in \cite{jennings2025a} overcomes obstacles 3 \& 4, obstacles 1 \& 2 remain. As we have shown obstacle 2 to arise from the CFL condition (Sec. \ref{sec:prospects}), it applies not only to CFD but to explicit time-evolution of nonlinear differential equations in general.

Due to the bleak outlook for explicitly time-evolved nonlinear systems, future work should focus on identifying applications and formulations where a steady-state solution can be calculated directly, or via implicit schemes that may allow sufficient relaxation of the CFL condition. Looking to applications of differential equations beyond CFD, recent investigations into large systems of coupled oscillators \cite{Babbush2023a,danz2025b} and applications of the finite element method \cite{Clader2013a,Montanaro2016a,danz2025a} suggest these areas warrant further study for potential quantum utility. Nonlinear Schr\"odinger equations (e.g., in describing Bose-Einstein condensate dynamics or laser propagation through nonlinear media) are likewise intriguing, though recent results show that quantum utility is unlikely in the case of unitary (i.e., non-dissipative) dynamics \cite{Brustle2025a}.

Future CFD work may consider the following scalar measurables as estimated from a system at steady-state, though the path to steady-state formulation is unclear at this time.
\begin{itemize}
    \item \textbf{Vorticity} is relevant to fundamental studies of turbulence (e.g., recent verification of the link between turbulence and entropy \cite{Yao2024a}) and can offer additional utility by informing better turbulence models such as RANS.  Vorticity at an inspection point $\mathbf{x}$ is calculated as $\nabla \times \mathbf{u}(\mathbf{x})$, where $\mathbf{u}(\mathbf{x})$ is the 3D velocity field constructed from the LBM population distributions $f(\mathbf{x})$.
    \item \textbf{Velocity Correlation Coefficients} can be calculated for turbulent flows inspection points $\mathbf{x}$ and $\mathbf{y}$ to determine if the velocities at said points are correlated.  This can inform models for the spatial scale of turbulent velocity fluctuations.   
    \item \textbf{Flow separation points} (where laminar flow transitions to turbulent) on the surface of the solid can be identified by temporal fluctuations in velocity or lack thereof.  
\end{itemize}

Future work should also quantify the additional resource costs incurred by nontrivial geometries such as ship hulls. The costs associated with both the block encoding of the LBM streaming matrix and the drag force estimation are tied to the notoriously burdensome task of loading classical data. In the case of a sphere (or cylinder, as studied in \cite{Jennings2026a}), the determination of whether a lattice node is inside or outside the sphere boundary can be represented by a simple arithmetic equation that can be implemented efficiently in a quantum circuit.  In the case of a ship hull a large number of polytopes and corresponding equations must be incorporated, potentially adding significant cost.

\section*{Acknowledgements}

This research was sponsored by the DARPA Quantum Benchmarking program, and by the State of Maryland through the QSANDRA program at the University of Maryland’s Applied Research Laboratory for Intelligence and Security (ARLIS). The authors wish to acknowledge the DARPA QB team and the Test \& Evaluation teams at NASA QuAIL and MIT-LL for helpful guidance and feedback. We thank B. Das, D. Fang,  D. Jennings, H. Krovi, A. Paler, and T. Watts for helpful discussions regarding quantum differential equation algorithms and resource estimation. We thank A. Nayak, D. Ponkratov, W. Qiu, D. Radosavlijevic, S. Verma, and M. Wheeler for helpful discussions regarding classical CFD and ship hull design. Authors PK and BR acknowledge support from DARPA under Air Force Contract No. FA8702-15-D-001.

\newpage
\appendix
\appendixpage
\addcontentsline{toc}{section}{Appendices}

\section{Utility estimates and bounds}\label{apx:utility}

\subsection{Software/compute market utility}

\begin{itemize}
\item A 2021 study \cite{Nely_2021} identified key entities in the CFD software market as
    \begin{itemize}
        \item Ansys, Inc., approximately 40\% of market share
        \item Open source software (including OpenFOAM), approximately 23\% of market share
        \item COMSOL, Inc., approximately 20\% of market share.
    \end{itemize}
\item Ansys, Inc., (a publicly traded company) reported an annual revenue of \$2.27B in 2023, a 9.9\% increase from 2022.  This includes all Ansys software sales and is therefore a \emph{weak upper bound} on the sales due to Ansys CFD software.  
\item Compute services and hardware vendors also claim some market share, but we cannot identify hardware sales where CFD was the dedicated end use.
\item According to \cite{engys-cfd-trends}, naval and marine applications accounted for roughly 6\% of CFD users in 2023. 
\end{itemize}

Compiling the above information, we arrive at a rough order-of-magnitude estimate for the incompressible CFD software/compute market utility of $\sim$\$1B, and for the application of simulation-driven ship design we estimate software/compute market utility in the range \$10M--\$100M.

\subsection{Engineering services market utility}

We assess this user-group utility strictly in the context of the shipbuilding market. Shipbuilders understand that system optimization (hull, propeller, energy saving devices) should be done at full ship scale \cite{wheeler2021a}. However, the computational time is still a challenge as they may only have a week or two to produce the final design. Whereas model-scale simulations require about an hour of compute time per design variant, ship-scale simulations require about a week per variant, and so model-scale simulations are used. It is anticipated that ship design will improve significantly if calculations can be sped up sufficiently \cite{mizzi2015a}.  In this work, we consider two model-scale design simulations (sections \ref{sec:jbc} and \ref{sec:kcs}) and one full-scale design simulation (section \ref{sec:mv-regal}). The cost associated with CFD in modern-day ship design workflows is estimated from the following data:

\begin{itemize}
    \item Shipbuilding market size \$142B in 2020, expected to reach \$195B by 2030 \cite{shipbuilding-market}
    \begin{itemize}
        \item This is whole market, including many non-CFD influences.
        \item Example large firm: HD Hyundai Heavy Industries Co., Ltd. Reported shipbuilding sales and profit of \$1.5B and \$57M, respectively in Q1 of 2024 \cite{hyundai-2024Q1-earnings}.
    \end{itemize}
    \item Many major shipbuilders operated at a loss in 2021-2022, driving the need for streamlining production (including engineering design) activities \cite{usni-maritime-review-2023}.
    \item A 2009 study \cite{deschamps2009a} found that for a multi-ship construction program where each of eight ships built was quoted at roughly \$30M, the distributed non-recurring costs associated with design, engineering, and planning amounted to roughly \$4.5M per ship (15\%).
    \item Representative annual compute cost lower bound for a small engineering firm: \$63K.  Based on the following assumptions:
    \begin{itemize}
        \item 4x (small cluster) on-demand Amazon Web Services (AWS) EC2 hpc7a.48xlarge instance at \$7.20 per hour (AMD 96 vCPU, 768GiB instance)
        \item 2,200 hours (roughly 25\% utilization),
        \item Not including S3 storage costs.
        \item AWS pricing varies based on reserved instances vs. on-demand.
        \item hpc7a.48xlarge selected based on \cite{aws-cfd-workshop} and updated from the legacy hpc6a.48xlarge instance.
    \end{itemize}
    \item Representative annual compute costs for large engineering firm were calculated in \cite{wheeler2021a}.  In the study, a bare hull resistance calculation at ship scale typically costs about \$1,000 per design variant, based on Amazon Cloud computing services (AWS) in 2021. For a design optimization project involving 100 design variants, this amounts to roughly \$100k \emph{per hull design}.
    \item At model-scale the cost is far lower, estimated at roughly \$10 per design variant, based on a cost of \$0.1 per core-hour \cite{gatin2019a}. In this case, the labor cost associated with preparing and running the simulations becomes the dominant cost, estimated at 40 minutes per design variant (roughly \$100 per variant or \$10k for 100 variants).
    \item For the example design cost spread out over 8 ships at \$4.5M per ship \cite{deschamps2009a}, CFD costs of \$10k--\$100k would then account for roughly 0.03\%--0.3\% of total design costs.
    
\end{itemize}

Compiling the above cost data, with CFD costs accounting for roughly 0.03\%--0.3\% of design costs, and design costs accounting for roughly 15\% of total ship design and build costs, the order-of-magnitude estimate for current and near-future CFD costs associated with the overall shipbuilding market falls in the same \$10M--\$100M range that was estimated for the ship-design segment of the software/compute market. We refer to this estimate of CFD costs in modern-day ship design practices as the ``classical state-of-the-art replacement utility'' for our considered application instance. This represents a \emph{tight lower bound} on the utility to the shipbuilding industry that would be unlocked by surpassing the transformational utility threshold.

The above utility estimate may also be regarded as a \emph{weak lower bound} on the utility to \emph{all} segments of the engineering services market (i.e., beyond shipbuilding) that leverage CFD platforms. The broader utility of incompressible CFD to the general engineering services market may be obtained by multiplying the \$1B estimate for the incompressible CFD software/compute market by return-on-investment (ROI) estimates. Reference \cite{hyperion-ansys-blog-roi} estimates that each \$1 invested in HPC resources for simulation returns approximately \$151.90 of revenue and \$34.9 of profit, through a combination of (1) making better products (40\% of projects) and (2) creating cost savings (20\% of projects). This implies a \emph{weak upper bound} of $\sim$\$100B on the utility of incompressible CFD to the engineering services market.

Beyond reducing current CFD costs, exceeding the transformational utility threshold may unlock additional utility in shipbuilding by reducing or even eliminating the need for experimental testing in the ship design workflow:

\begin{itemize}    
    
    \item The cost of testing a model-scale ship in a towing-tank facility is estimated to start at \$20k which includes fabricating the model and running the tests \cite{dms_towing_tank}. This process does not scale for ship hull optimization, where a variety of hull geometry modifications are tested.
    \item To estimate towing-tank costs, we surveyed towing tank operators affiliated with the International Towing Tank Conference (ITTC) \cite{ittc_towing_tanks}. The 12 responses we received quoted costs ranging from $\sim$\$10k for a basic bare hull resistance test in a small-scale tank to $\sim$\$1M for a comprehensive set of tests of large-scale models across a range of design parameters. Omitting outliers and averaging the ranges' mid-points, we obtained a typical cost estimate of \$77k. 
    \item Testing a model-scale ship at a towing-tank facility also comes with significant scheduling lead time. The test facility must be booked months or years in advance.
 
\end{itemize}

Overall, the cost of model-scale experimental testing in the ship design workflow exceeds model-scale CFD costs by at least $\sim10\times$. We therefore estimate the utility associated with rendering towing tank tests obsolete to lie in the range \$100M--\$1B. This represents a \emph{tight upper bound} on the utility to the shipbuilding industry that would be unlocked by surpassing the transformational utility threshold.

\subsection{Derived product end-user utility}
Surpassing the transformative utility threshold would enable simulation-driven workflows that result in better-performing ship designs \cite{siemens_sim_design}.  Various forms of resulting utility can be quantified:
\begin{itemize}
    \item Reduced capital expense (CAPEX) of building the vehicle.
    \begin{itemize}
        \item In 2022, new-build construction escalated about 25\% and new orders declined 22\% due to the rising cost of steel and labor \cite{usni-maritime-review-2023}, underscoring a new need to optimize designs early in the engineering and build cycle.
    \end{itemize}
    \item Reduced operational expense (OPEX) throughout the lifetime of the vessel.
    \begin{itemize}
        \item Ship hulls fall into many size and purpose categories.  Consider the average ``Handymax'' bulk carrier that is approximately 150-200m long and can carry approximately 40,000-59,999 dead weight tons of carrying capacity.  
        \item The average Handymax carrier consumes approximately 25-30 metric tons of diesel fuel per day \cite{maritime-fuel}.
        \item As of 2022, there were approximately 4000 Handymax-class vessels in operation \cite{dry-cargo-international}
        \item Assuming 25\% of the fleet is underway at any time, then this singular class of ships would consume approximately 25,000 metric tons of diesel per day, or 9M metric tons of diesel per year.
        \item Assuming a cost of \$3 per gallon, the fuel costs for the fleet would be approximately \$7B.
        \item This represents a \emph{weak lower bound} on operating fuel costs for the worldwide fleet since we have focused on only one type of vessel: Handymax
    \end{itemize}
    \item In \cite{iea-org-net-zero}, the International Energy Agency posits that the worldwide shipping fleet is \emph{not} on track to meet net zero emissions standards by 2050, which would require an almost 15\% reduction in emissions from 2022-2030.  Ship owners will begin to demand CFD-optimized ship hulls from ship builders to ensure compliance with regulations.      
    \item Design requirements sensitivity analysis:  during the design of the hull, identify end user requirements that are overly burdensome and give the end user an opportunity to relax or alter requirements in favor of a design with reduced CAPEX or OPEX.  

\end{itemize}
These estimates focused on end users purchasing and operating resistance-optimized shipping vessels only, and thus represent a \emph{weak lower bound} on all end users that benefit from improved products from CFD platforms.
\section{Notation and variable indexing}\label{sec:notation}

The goal of this section is to define the ordering/indexing scheme for the $f_i(\mathbf{x})$ variables used in this paper.  Alternative schemes may be used. 

The computational domain of the LBM is the set of three dimensional grid points $\mathbf{X} = \{ (x,y,z) \in \mathbb{Z}_{n_x} \times \mathbb{Z}_{n_y} \times \mathbb{Z}_{n_z} \}$.  $n_x$, $n_y$, and $n_z$ are the grid dimensions and the total number of grid nodes $n=n_x n_y n_z$.  First we define a bijective mapping $\mathcal{L}_{\mathbf{x}(1)}$ of the $n$ grid points in the set $\mathbf{X}$ to a subset of the nonnegative integers $\mathbb{Z}_n$.  There are many ways to implement $\mathcal{L}_{\mathbf{x}(1)}$. We use the following. First order first by $x$ coordinate, then $y$, then $z$, so that our expression for the position index $\alpha\in \mathbb{Z}_n$ for some node $\mathbf{x}_\alpha = (x_\alpha, y_\alpha, z_\alpha)$ is
\begin{subequations}
    \begin{align}
        \mathcal{L}_{\mathbf{x}(1)} : \{ \mathbf{x}_\alpha \in \mathbf{X} \} &\rightarrow \{ \alpha \in \mathbb{Z}_n \}\\
        \mathcal{L}_{\mathbf{x}(1)}(\mathbf{x}_\alpha) & = x_\alpha + n_x y_\alpha + n_x n_y z_\alpha \\
        &= \alpha.
    \end{align}
    \label{eq:L_x-ordering-implementation}
\end{subequations}
 The inverse is
\begin{subequations}
    \begin{align}
     \mathcal{L}^{-1}_{\mathbf{x}(1)} : \mathbb{Z}_n &\rightarrow \mathbf{X} \\
     \mathcal{L}^{-1}_{\mathbf{x}(1)}(\alpha)
     &= (  \alpha \bmod n_x , \quad (\alpha\bdiv n_x) \bmod n_y , \quad \alpha\bdiv (n_x n_y) )\\
     &= (  x_\alpha , y_\alpha , z_\alpha ) \\
     &= \mathbf{x}_\alpha,
    \end{align}
    \label{eq:L_x-ordering-implementation-inverse}
\end{subequations}
where $a\bdiv b \equiv \lfloor a/b \rfloor$.

Next, we establish the mapping a bijective mapping of the $f$-variables to a subset of the nonnegative integers $\mathbb{Z}_{nQ} = \{0,1,2,...,nQ-1\}$.  For some variable  $f_{i}(\mathbf{x}_\alpha)$ with corresponding pair $( \mathbf{x}_\alpha , i ) \in \mathbf{X} \times \mathbb{Z}_Q$ we have:
\begin{subequations}
\begin{align}
\mathcal{L}_1 : \{ ( \mathbf{x}_\alpha , i ) \in \mathbf{X} \times \mathbb{Z}_Q \} &\rightarrow \{ \mu \in \mathbb{Z}_{nQ} \}\\
\mathcal{L}_1( \mathbf{x}_\alpha, i ) 
& = Q\mathcal{L}_{\mathbf{x}(1)}(\mathbf{x}_\alpha) + i\\
& = Q(x_\alpha + n_x y_\alpha + n_x n_y z_\alpha) + i\\
& = \mu.
\end{align}
\label{eq:L_1-ordering-implementation}
\end{subequations}

The inverse is
\begin{subequations}
    \begin{align}
        \mathcal{L}^{-1}_1 : \{ \mu \in \mathbb{Z}_{nQ} \} & \rightarrow \{ ( \mathbf{x}_\alpha , i ) \in \mathbf{X} \times \mathbb{Z}_Q \}\\
        \mathcal{L}^{-1}_1(\mu) &= ( \mathcal{L}^{-1}_{\mathbf{x}(1)}(\mu \bdiv Q) , \quad \mu \bmod Q )\\
        &= ( \mathcal{L}^{-1}_{\mathbf{x}(1)}(\alpha) , i )\\
        &= ( \mathbf{x}_\alpha , i ).
    \end{align}
    \label{eq:L_1-ordering-inverse}
\end{subequations}
E.g., we can now refer to the $f$-variable at index $\mu$ as $f_\mu = f_i(\mathbf{x}_\alpha)$, where $\mu$, $i$, $\mathbf{x}_\alpha$ satisfy \eqref{eq:L_x-ordering-implementation}-\eqref{eq:L_1-ordering-inverse}. 

The above ordering can be implemented naturally with a tensor product of registers $|z_\alpha\rangle |y_\alpha\rangle |x_\alpha\rangle |\mathbf{c}_i\rangle$.

Given we have some $\mathcal{L}_1$, we can now make statements of the form ``$\mathcal{L}_1\left( f_i(\mathbf{x}_\alpha )\right) < \mathcal{L}_1\left( f_j(\mathbf{x}_\beta) \right)$'', and thus we can establish a sort-order of $f_{i}(\mathbf{x})$-variables.  We will also use this to establish the index/order of the $f, f^{\otimes 2}$, and $f^{\otimes 3}$ vectors.  This ordering will also impact the row/column position that coefficients appear in the matrices $S, F_1, F_2$ and $F_3$.

In the Carleman Linearized differential equation in \ref{sec:cl-lbm}, the variables we need to evolve are $f_i(\mathbf{x})$.  During the linearization process, we introduced additional ``variables'' of the form $f_j(\mathbf{x}_\beta)f_k(\mathbf{x}_\gamma)$ (a second-order monomial) and $f_j(\mathbf{x}_\beta)f_k(\mathbf{x}_\gamma)f_\ell(\mathbf{x}_\delta)$ (a third-order monomial). Since we are limiting our Carleman Linearization to third order, we limit the monomial to three factors.    

Using multilinear algebra notation per section \ref{sec:cl-lbm}, 
$f$ is a vector of length $nQ$ that holds all monomials of distribution functions $f_i(\mathbf{x})$ of order 1 (for $i \in \mathbb{Z}_Q$ and $\mathbf{x}\in \mathbf{X}$).  Similarly $f \otimes f = f^{\otimes 2}$ is a vector of length $(nQ)^2$ that holds all monomials of distribution functions of order 2.  For example, $f_{j} (\mathbf{x}_\beta) f_{k} (\mathbf{x}_{\gamma} )$ is an element of $f^{\otimes 2}$.  But $f_{k} (\mathbf{x}_\gamma) f_{j} (\mathbf{x}_\beta)$ is also an element of $f^{\otimes 2}$.  Scalar multiplication is commutative and $f_{j} (\mathbf{x}_\beta) f_{k} (\mathbf{x}_{\gamma} ) = f_{k} (\mathbf{x}_\gamma) f_{j} (\mathbf{x}_\beta)$.  We call $f_{j} (\mathbf{x}_\beta) f_{k} (\mathbf{x}_{\gamma} )$ and $f_{k} (\mathbf{x}_\gamma) f_{j} (\mathbf{x}_\beta)$ degenerate monomial terms. See table \ref{tab:full-list-of-second-order-monomials} for more examples.  Finally $f \otimes f \otimes f = f^{\otimes 3}$ is a vector of length $(nQ)^3$ that holds all monomials of order 3.  $f^{\otimes 3}$ contains even more degenerate terms.  See table \ref{tab:full-list-of-third-order-monomials} for more examples.

\begin{table}[H]
\centering
\begin{tabular}{|c|l|l|l|l|l|}
\hline
      & Second-order Monomial   & Distinct  &  Degenerate  & Total  \\
      &  & Factors & Variants & Variants \\
      \hline
1    &  $f_j(\mathbf{x}_\beta)f_j(\mathbf{x}_\beta) = f_j(\mathbf{x}_\beta)^{2}$ & 1 & None & 1\\
\hline
2    &  $f_j(\mathbf{x}_\beta)f_k(\mathbf{x}_\gamma)$ & 2  & $f_k(\mathbf{x}_\gamma)f_j(\mathbf{x}_\beta)$ & 2\\
\hline
\end{tabular}
\caption{Full list of second-order monomials composed of factors $f_j(\mathbf{x}_\beta)$ and $f_k(\mathbf{x}_\gamma)$.}
\label{tab:full-list-of-second-order-monomials}
\end{table}

\begin{table}[H]
\centering
\begin{tabular}{|c|l|l|l|l|l|}
\hline
      & Third-order Monomial    & Distinct & Degenerate  & Total  \\
      &    & Factors & Variants & Variants\\
      \hline
1    &  $f_j(\mathbf{x}_\beta)f_j(\mathbf{x}_\beta)f_j(\mathbf{x}_\beta)=f_j(\mathbf{x}_\beta)^{3}$ & 1 & None & 1\\
\hline
2    &  $f_j(\mathbf{x}_\beta)f_j(\mathbf{x}_\beta)f_k(\mathbf{x}_\gamma) = f_j(\mathbf{x}_\beta)^{2}f_k(\mathbf{x}_\gamma) $ & 2 & $f_k(\mathbf{x}_\gamma)f_j(\mathbf{x}_\beta)^{2}$ & 3\\
 & &  &$f_j(\mathbf{x}_\beta)f_k(\mathbf{x}_\gamma)f_j(\mathbf{x}_\beta)$ & \\
\hline
3    &  $f_j(\mathbf{x}_\beta)f_k(\mathbf{x}_\gamma)f_\ell(\mathbf{x}_\delta) $ & 3 &$f_k(\mathbf{x}_\gamma)f_j(\mathbf{x}_\beta)f_\ell(\mathbf{x}_\delta)$ & 6\\
    & &      & $f_k(\mathbf{x}_\gamma)f_\ell(\mathbf{x}_\delta)f_j(\mathbf{x}_\beta)$ & \\
    & &      & $f_j(\mathbf{x}_\beta)f_\ell(\mathbf{x}_\delta)f_k(\mathbf{x}_\gamma)$ & \\
    & &      & $f_\ell(\mathbf{x}_\delta)f_j(\mathbf{x}_\beta)f_k(\mathbf{x}_\gamma)$ & \\
    & &      & $f_\ell(\mathbf{x}_\delta)f_k(\mathbf{x}_\gamma)f_j(\mathbf{x}_\beta)$ & \\    
\hline
\end{tabular}
\caption{Full list of third-order monomials composed of factors $f_j(\mathbf{x}_\beta), f_k(\mathbf{x}_\gamma)$, and $f_\ell(\mathbf{x}_\delta)$.}
\label{tab:full-list-of-third-order-monomials}
\end{table}

We will extend the sort-order defined in \eqref{eq:L_1-ordering-implementation} to distinguish and access (by index) degenerate second- and third-order monomials as follows.
\begin{subequations}
\begin{align}    
\mathcal{L}_{\mathbf{x}(2)} : \{ (\mathbf{x}_\beta, \mathbf{x}_\gamma) \in \mathbf{X}^2 \} &\rightarrow \mathbb{Z}_{n^2}\\
\mathcal{L}_{\mathbf{x}(2)}(\mathbf{x}_\beta, \mathbf{x}_\gamma) 
&= n\mathcal{L}_{\mathbf{x}(1)}(\mathbf{x}_\beta) + \mathcal{L}_{\mathbf{x}(1)}(\mathbf{x}_\gamma),
\end{align}
\label{eq:L_x2-ordering-definition}
\end{subequations}

\begin{subequations}
\begin{align}    
\mathcal{L}_2 : \{ (f_j(\mathbf{x}_\beta),f_k(\mathbf{x}_\gamma)) : j,k \in \mathbb{Z}_Q, (\mathbf{x}_\beta, \mathbf{x}_\gamma) \in \mathbf{X}^2 \} &\rightarrow \mathbb{Z}_{(nQ)^2}\\
\mathcal{L}_2(f_j(\mathbf{x}_\beta),f_k(\mathbf{x}_\gamma)) 
&= 
Q^2\left( \mathcal{L}_{\mathbf{x}(2)}( \mathbf{x}_\beta, \mathbf{x}_\gamma )  \right) + Qj + k\\
\mathcal{L}_2(f_j(\mathbf{x}_\beta),f_k(\mathbf{x}_\gamma)) 
&= 
Q^2\left( n\mathcal{L}_{\mathbf{x}(1)}( \mathbf{x}_\beta) + \mathcal{L}_{\mathbf{x}(1)}( \mathbf{x}_\gamma )  \right) + Qj + k\\
\mathcal{L}_2(f_j(\mathbf{x}_\beta),f_k(\mathbf{x}_\gamma)) 
&= 
Q^2n\mathcal{L}_{\mathbf{x}(1)}( \mathbf{x}_\beta)  + Q\mathcal{L}_{\mathbf{x}(1)}( \mathbf{x}_\gamma )  + Qj + k.
\end{align}
\label{eq:L_2-ordering-definition}
\end{subequations}
Equation \eqref{eq:L_2-ordering-definition} provides multiple equivalent definitions for $\mathcal{L}_2$.

\begin{subequations}
\begin{align}    
\mathcal{L}_{\mathbf{x}(3)}: \{ (\mathbf{x}_\beta, \mathbf{x}_\gamma, \mathbf{x}_\delta ) \in \mathbf{X}^3 \} &\rightarrow \mathbb{Z}_{n^3}\\
\mathcal{L}_{\mathbf{x}(3)}(\mathbf{x}_\beta, \mathbf{x}_\gamma, \mathbf{x}_\delta) 
&= 
n\mathcal{L}_{\mathbf{x}(2)}(\mathbf{x}_\beta, \mathbf{x}_\gamma) + \mathcal{L}_{\mathbf{x}(1)}(\mathbf{x}_\delta)\\
\mathcal{L}_{\mathbf{x}(3)}(\mathbf{x}_\beta, \mathbf{x}_\gamma, \mathbf{x}_\delta) 
&= 
n^2\mathcal{L}_{\mathbf{x}(1)}(\mathbf{x}_\beta) + 
n\mathcal{L}_{\mathbf{x}(1)}(\mathbf{x}_\gamma) +
\mathcal{L}_{\mathbf{x}(1)}(\mathbf{x}_\delta)
\end{align}
\label{eq:L_x3-ordering-definition}
\end{subequations}
Equation \eqref{eq:L_x3-ordering-definition} provides multiple equivalent definitions for $\mathcal{L}_{\mathbf{x}(3)}$.

\begin{subequations}
\begin{align}    
\mathcal{L}_3 
: \{ (f_j(\mathbf{x}_\beta),f_k(\mathbf{x}_\gamma),f_\ell(\mathbf{x}_\delta)  ) &: j,k,\ell \in \mathbb{Z}_Q, (\mathbf{x}_\beta, \mathbf{x}_\gamma, \mathbf{x}_\delta) \in \mathbf{X}^3 \} \rightarrow \mathbb{Z}_{(nQ)^3}\\
\mathcal{L}_3(f_j(\mathbf{x}_\beta),f_k(\mathbf{x}_\gamma),f_\ell(\mathbf{x}_\delta)  ) 
&= 
Q^3 \left( \mathcal{L}_{\mathbf{x}(3)}( \mathbf{x}_\beta, \mathbf{x}_\gamma,\mathbf{x}_\delta )  \right) + Q^2j + Qk + \ell\\
\mathcal{L}_3(f_j(\mathbf{x}_\beta),f_k(\mathbf{x}_\gamma),f_\ell(\mathbf{x}_\delta)  ) 
&= 
Q^3 \left( n^2 \mathcal{L}_{\mathbf{x}(1)}( \mathbf{x}_\beta ) + n\mathcal{L}_{\mathbf{x}(1)}( \mathbf{x}_\gamma ) + \mathcal{L}_{\mathbf{x}(1)}( \mathbf{x}_\delta ) \right)
 + Q^2j + Qk + \ell\\
\mathcal{L}_3(f_j(\mathbf{x}_\beta),f_k(\mathbf{x}_\gamma),f_\ell(\mathbf{x}_\delta)  ) 
&= 
Q^3n^2 \mathcal{L}_{\mathbf{x}(1)}( \mathbf{x}_\beta ) + Q^3n\mathcal{L}_{\mathbf{x}(1)}( \mathbf{x}_\gamma ) + Q^3\mathcal{L}_{\mathbf{x}(1)}( \mathbf{x}_\delta ) + Q^2j + Qk + \ell
\end{align}
\label{eq:L_3-ordering-definition}
\end{subequations}
Equation \eqref{eq:L_3-ordering-definition} provides some--but not all--equivalent definitions for $\mathcal{L}_{\mathbf{x}(3)}$.

The inverses $\mathcal{L}_{\mathbf{x}(2)}^{-1}$, $\mathcal{L}_{\mathbf{x}(3)}^{-1}$, $\mathcal{L}_2^{-1}$ and $\mathcal{L}_3^{-1}$ exist.

\section{Functions to compute $S$- and  $F$-matrix elements}\label{sec:derivation-of-matrix-coeffs}

The goal of this section is to derive functions $\tilde{S}$, $\tilde{F}_1$, $\tilde{F}_2$, $\tilde{F}_3$ that produce the elements of the $S$- and $F$-matrices and can be called for block encoding.  Table \ref{tab:tildeF_lookup} is provided as a quick reference.  The input to functions $\tilde{S}$, $\tilde{F}_1$, $\tilde{F}_2$, $\tilde{F}_3$ is a tuple of monomials which correspond to a row and column of the matrix and the functions return the element of the matrix at the row/column. The row/column variable indexing schemes are described in appendix \ref{sec:notation}.  The functions are derived for the D3Q27 lattice.  See table \ref{tab:D3Q27-constants} for a complete list of constants for the D3Q27 lattice.  As a reminder, $\tau$ is not a D3Q27 constant, but is related to the viscosity of the medium and thus specific to the problem instance.  

\begin{table}[h]
\centering
\begin{tabular}{|c|l|l|l|}
\hline
Function  & Example input  & Row/Column & Dependencies \\ 
\hline
$\tilde{S}(m_r, m_c)$  \eqref{eq:tildeS} & $m_r = f_i(\mathbf{x}_\alpha)$ & row = $\mathcal{L}_1(m_r)$ &    \\ 
                                         & $m_c = f_j(\mathbf{x}_\beta)$ & col = $\mathcal{L}_1(m_c)$ &    \\
\hline
$\tilde{F}_1(m_r, m_c)$  \eqref{eq:tildeF_1} & $m_r = f_i(\mathbf{x}_\alpha)$ & row = $\mathcal{L}_1(m_r)$ &  $\tilde{T}_{0[1]}$ \eqref{eq:tildeT_{0[1]}}  \\ 
                                         &  $m_c = f_j(\mathbf{x}_\beta)$                     &       col = $\mathcal{L}_1(m_c)$                    & $\tilde{T}_{1[1]}$ \eqref{eq:tildeT_{1[1]}}   \\
                                         &  &  &  $\tilde{T}_{2[1]}$ \eqref{eq:tildeT_{2[1]}}  \\
\hline
$\tilde{F}_{2\text{-dense}}(m_r, m_c)$  \eqref{eq:tildeF_2_dense} & $m_r = f_i(\mathbf{x}_\alpha)$ & row = $\mathcal{L}_1(m_r)$ &  $\tilde{T}_{3[2]\text{-dense}}$ \eqref{eq:tildeT_3[2]_dense}  \\ 
                                         & $m_c = f_j(\mathbf{x}_\beta)f_k(\mathbf{x}_\gamma)$                      &     col = $\mathcal{L}_2(m_c)$                      & $\tilde{T}_{4[2]\text{-dense}}$ \eqref{eq:tildeT_4[2]_dense}   \\
\hline
$\tilde{F}_{3\text{-dense}}(m_r, m_c)$  \eqref{eq:tildeF_3_dense} & $m_r = f_i(\mathbf{x}_\alpha)$ & row = $\mathcal{L}_1(m_r)$ &  $\tilde{T}_{3[3]\text{-dense}}$ \eqref{eq:tildeT_3[3]_dense}  \\ 
                                         & $m_c = f_j(\mathbf{x}_\beta)f_k(\mathbf{x}_\gamma)f_\ell(\mathbf{x}_\delta)$                      &   col = $\mathcal{L}_3(m_c)$                        & $\tilde{T}_{4[3]\text{-dense}}$ \eqref{eq:tildeT_4[3]_dense}   \\
\hline
$\tilde{F}_{2\text{-sparse}}(m_r, m_c)$  \eqref{eq:tildeF_2_sparse} & $m_r = f_i(\mathbf{x}_\alpha)$ & row = $\mathcal{L}_1(m_r)$ &  $\tilde{T}_{3[2]\text{-sparse}}$ \eqref{eq:tildeT_3[2]_sparse}  \\ 
                                         & $m_c = f_j(\mathbf{x}_\beta)f_k(\mathbf{x}_\gamma)$                      &   col = $\mathcal{L}_2(m_c)$                        & $\tilde{T}_{4[2]\text{-sparse}}$ \eqref{eq:tildeT_4[2]_sparse}   \\
\hline
$\tilde{F}_{3\text{-sparse}}(m_r, m_c)$  \eqref{eq:tildeF_3_sparse} & $m_r = f_i(\mathbf{x}_\alpha)$ & row = $\mathcal{L}_1(m_r)$ &  $\tilde{T}_{3[3]\text{-sparse}}$ \eqref{eq:tildeT_3[3]_sparse}  \\ 
                                         &  $m_c = f_j(\mathbf{x}_\beta)f_k(\mathbf{x}_\gamma)f_\ell(\mathbf{x}_\delta)$                    &     col = $\mathcal{L}_3(m_c)$                      & $\tilde{T}_{4[3]\text{-sparse}}$ \eqref{eq:tildeT_4[3]_sparse}   \\
\hline
\end{tabular}
\caption{\label{tab:tildeF_lookup} List of the $\tilde{S}$ and $\tilde{F}$ functions that return elements of the $S$- and $F$-matrices.}
\end{table}

As discussed in section \ref{sec:cl-lbm}, the approximation for the LBM dynamics is:
\begin{equation}
    \frac{\partial }{\partial t} f \approx Sf + F_1 f + F_2 f^{\otimes2} + F_3 f^{\otimes3}. \tag{repeating \ref{eq:lbm-multilinear-algebra}}
\end{equation}

The structure of the $F_1$, $F_2$ and $F_3$ matrices are shown in figures \ref{fig:f_1_structure}, \ref{fig:f_2_structure} and \ref{fig:f_3_structure}, respectively.

\begin{figure}[H]
\centering
\begin{tikzpicture}[scale=1] 
\draw (0,0) rectangle (4,4);

\fill[black] (0,3) rectangle ++(1,1);
\fill[black] (1,2) rectangle ++(1,1);
\node at (2.5,1.5) {$\ddots$};
\fill[black] (3,0) rectangle ++(1,1);

\draw [decorate, decoration={brace, amplitude=5pt, raise=5pt}] (1,2) -- (1,3);
\node[left] at (0,2.35) {$f_i(\mathbf{x}), \forall i \in \mathbb{Z}_Q$};

\draw [decorate, decoration={brace, amplitude=5pt, raise=5pt}] (1,3) -- (2,3);
\node[above] at (1.5,4.5) {$f_j(\mathbf{x})$};
\node[above] at (1.5,4.0) {$\forall j \in \mathbb{Z}_Q$};

\end{tikzpicture}
\caption{Structure of the $F_1$ matrix.  This structure is produced if the row/column indexing of $f$-variables is implemented per appendix \protect{\ref{sec:notation}}.  The matrix is $nQ \times nQ$.  Nonzero terms may appear in the black $Q \times Q$ block of rows/columns that correspond to $f$-variables that share the same grid point $\mathbf{x}$.  The black blocks are all the same.  We later refer to one instance of a black block as $F_1^{\mathbf{x}} \in \mathbb{R}^{Q \times Q}$.}
\label{fig:f_1_structure}
\end{figure}

\begin{figure}[H]
\centering
\begin{tikzpicture}[scale=0.75] 
\draw (0,0) rectangle (16,4);

\fill[pattern=north west lines, pattern color=gray] (0,3) rectangle ++(4,1);
\fill[black] (0,3) rectangle ++(1,1);

\fill[pattern=north west lines, pattern color=gray] (4,2) rectangle ++(4,1);
\fill[black] (5,2) rectangle ++(1,1);

\node at (10,1.5) {$\ddots$};

\fill[pattern=north west lines, pattern color=gray] (12,0) rectangle ++(4,1);
\fill[black] (15,0) rectangle ++(1,1);

\draw [decorate, decoration={brace, amplitude=5pt, raise=5pt}] (4,2) -- (4,3);
\node[left] at (0,2.35) {$f_i(\mathbf{x}), \forall i \in \mathbb{Z}_Q$};

\draw [decorate, decoration={brace, amplitude=5pt, raise=5pt}] (5,3) -- (6,3);
\node[above] at (5.5,4.75) {$f_j(\mathbf{x})f_k(\mathbf{x})$};
\node[above] at (5.5,4.25) {$\forall j, k \in \mathbb{Z}_Q$};

\draw [decorate, decoration={brace, amplitude=5pt, raise=5pt}] (8,2) -- (4,2);
\node[below] at (5.5,0) {$f_j(\mathbf{x})f_k(\mathbf{y})$};
\node[below] at (5.5,-0.6) {$\forall \mathbf{y} \in \mathbf{X}$};
\node[below] at (5.5,-1.2) {$\forall j, k \in \mathbb{Z}_Q$};

\end{tikzpicture}
\caption{Structure of the $F_2$ matrix.  This structure is produced if the row/column indexing of $f$-variables is implemented per appendix \protect{\ref{sec:notation}}.  The matrix is $nQ \times (nQ)^2$.  Nonzero terms may appear in the black $Q\times Q^2$ blocks. The black blocks have rows corresponding to $f$-variables and columns corresponding to $f^2$-variables that share the same grid point $\mathbf{x}$. 
The patterned blocks are size $Q \times nQ^2$ and will be zero, implying there will be some all-zero columns.  The black blocks are all the same. We later refer to one instance of a black block as $F_2^{\mathbf{x}} \in \mathbb{R}^{Q \times Q^2}$.}
\label{fig:f_2_structure}
\end{figure}

\begin{figure}[H]
\centering
\begin{tikzpicture}[scale=0.75] 
\draw (0,0) rectangle (16,4);

\fill[pattern=north west lines, pattern color=gray] (0,3) rectangle ++(4,1);
\fill[black] (0,3) rectangle ++(1,1);

\fill[pattern=north west lines, pattern color=gray] (4,2) rectangle ++(4,1);
\fill[black] (5,2) rectangle ++(1,1);

\node at (10,1.5) {$\ddots$};

\fill[pattern=north west lines, pattern color=gray] (12,0) rectangle ++(4,1);
\fill[black] (15,0) rectangle ++(1,1);

\draw [decorate, decoration={brace, amplitude=5pt, raise=5pt}] (4,2) -- (4,3);
\node[left] at (0,2.35) {$f_i(\mathbf{x}), \forall i \in \mathbb{Z}_Q$};

\draw [decorate, decoration={brace, amplitude=5pt, raise=5pt}] (5,3) -- (6,3);

\node[above] at (5.5,4.75) {$f_j(\mathbf{x})f_k(\mathbf{x})f_\ell(\mathbf{x})$};
\node[above] at (5.5,4.25) {$\forall j, k, \ell \in \mathbb{Z}_Q$};

\draw [decorate, decoration={brace, amplitude=5pt, raise=5pt}] (8,2) -- (4,2);
\node[below] at (5.5,0) {$f_j(\mathbf{x})f_k(\mathbf{y})f_k(\mathbf{z})$};
\node[below] at (5.5,-0.6) {$\forall \mathbf{y},\mathbf{z} \in \mathbf{X}$};
\node[below] at (5.5,-1.2) {$\forall j, k, \ell \in \mathbb{Z}_Q$};

\end{tikzpicture}

\caption{Structure of the $F_3$ matrix.  This structure is produced if the row/column indexing of $f$-variables is implemented per appendix \protect{\ref{sec:notation}}.  The matrix is $nQ \times (nQ)^3$. Nonzero terms may appear in the black $Q\times Q^3$ blocks.  The black blocks have rows corresponding to $f$-variables and columns corresponding to $f^3$-variables that share the same grid point $\mathbf{x}$. The patterned blocks are size $Q \times n^2Q^3$ and will be zero, implying there will be some all-zero columns.  The black blocks are all the same. We later refer to one instance of a black block as $F_3^{\mathbf{x}} \in \mathbb{R}^{Q \times Q^3}$.}
\label{fig:f_3_structure}
\end{figure}

\afterpage{
We begin our analysis by expanding the differential equation for $\frac{\partial f_i(\mathbf{x})}{\partial t}$ into separate terms.
}

\begin{subequations}
\begin{align}
\frac{\partial f_i(\mathbf{x}) }{\partial t} 
&=
\underbrace{
T_{s[1]}
}_{\text{streaming term }}
+ 
\underbrace{
\frac{-1}{\tau} \left( f_i(\mathbf{x}) - f^{\text{eq}}_i(\mathbf{x}) \right)
}_{\text{collision term}}
\\
&\approx
T_{s[1]}
+ 
\underbrace{
\frac{-1}{\tau} \left( f_i(\mathbf{x}) - \tilde{f}^{\text{eq}}_i(\mathbf{x}) \right)
}_{\text{collision term with approximation } \tilde{f}^{\text{eq}}_i(\mathbf{x}) \text{ from \eqref{eq:lbm-feq-approx}}}
\\
&\approx
T_{s[1]} + 
\underbrace{
\frac{-1}{\tau} \left( 
f_i(\mathbf{x}) - w_i \left( \rho + 3 \rho \mathbf{u}\cdot \mathbf{c}_i + \frac{9}{2}
\underbrace{(2-\rho)}_{\text{See \eqref{eq:lbm-approx-inv-rho}}}(\rho \mathbf{u} \cdot \mathbf{c}_i)^2 + \frac{-3}{2} \underbrace{(2-\rho)}_{\text{See \eqref{eq:lbm-approx-inv-rho}}} ( \rho \mathbf{u} \cdot \rho \mathbf{u}) \right)
\right) 
}_{\text{approximate collision term}}
\\
&\approx T_{s[1]}
+ 
\underbrace{
\left( \frac{-1}{\tau} \right) f_i(\mathbf{x})
}_{T_{0[1]}}
+
\underbrace{
\left( \frac{w_i}{\tau} \right) \rho 
}_{T_{1[1]}}
+ 
\underbrace{
\left ( \frac{3w_i}{\tau} \right) 
\left( \rho \mathbf{u}\cdot \mathbf{c}_i \right)
}_{T_{2[1]}} 
\\
&\hspace{3cm}+
\underbrace{
\left( \frac{9w_i}{2\tau} \right) (2-\rho)
(\rho \mathbf{u} \cdot \mathbf{c}_i)^2
}_{T_3 = T_{3[2]} + T_{3[3]}}
+
\underbrace{
\left ( \frac{-3w_i}{2\tau} \right) (2-\rho) ( \rho \mathbf{u} \cdot \rho \mathbf{u}) 
}_{T_4 = T_{4[2]}+ T_{4[3]}}
\notag
\end{align}
\label{eq:lbe-diff-eq-expanded}
\end{subequations}
Equation \eqref{eq:lbe-diff-eq-expanded} corresponds to one ``row'' of equation \eqref{eq:lbm-multilinear-algebra} for the row associated with variable $f_i(\mathbf{x})$. The streaming term $T_{s[1]}$ will contain coefficients for first-order monomials that will be elements of the $S$ matrix.
Terms $T_{0[1]}, T_{1[1]}$, and $T_{2[1]}$  will contain the coefficients for first-order monomials that will be elements in the $F_1$ matrix (hence the subscript $[1]$).  Terms $T_3$ and $T_4$ will contain coefficients for second- and third-order monomials that will be elements of the $F_2$ and $F_3$ matrices.

Consider the streaming term $T_{s[1]}:$
\begin{align}
T_{s[1]} 
&=
\begin{cases}
f_i (\mathbf{x} - \mathbf{c}_i ) - f_i (\mathbf{x}) 
&\text{ if } \mathbf{x} \text{ is a fluid node}\\
&\text{ and } \mathbf{x} - \mathbf{c}_i \text{ is a fluid node},
\\
f_j (\mathbf{x}) - f_i (\mathbf{x}) 
&\text{ if } \mathbf{x} \text{ is a fluid node}\\
&\text{ and } \mathbf{x} - \mathbf{c}_i \text{ is a solid node}\\
&\text{ and } \mathbf{c}_j = - \mathbf{c}_i,
\\
0
&\text{ else. }
\end{cases}
\label{eq:T_{s[1]}}
\end{align}
Our goal is to develop a function $\tilde{S}$ whose input is a tuple of first-order variables (i.e., first-order monomials) $(f_i(\mathbf{x_\alpha}), f_j(\mathbf{x}_\beta) )$, and returns the element in  of the matrix $S$ (from equation \eqref{eq:lbm-multilinear-algebra}), in row $r$ and column $c$, where $r = \mathcal{L}_1(f_i(\mathbf{x}_\alpha) )$ and $c = \mathcal{L}_1(f_j(\mathbf{x}_\beta ))$.
\begin{align}
\tilde{S}( f_i(\mathbf{x}_\alpha), f_j(\mathbf{x}_\beta) )
&=
\begin{cases}
-1 
&\text{ if } \mathbf{x}_\beta = \mathbf{x}_\alpha \text{ (local to node $\mathbf{x}_\alpha$)}\\
&\text{ and } \mathbf{x}_\alpha \text{ is a fluid node}\\
&\text{ and } \mathbf{c}_j = \mathbf{c}_i \neq \mathbf{0} \text{ (same velocity vector)},\\
1 
&\text{ if } \mathbf{x}_\beta = \mathbf{x}_\alpha - \mathbf{c}_i \text{ (not local to node $\mathbf{x}_\alpha$)}\\
&\text{ and } \mathbf{x}_\alpha \text{ is a fluid node}\\
&\text{ and } \mathbf{x}_\beta \text{ is a fluid node}\\
&\text{ and } \mathbf{c}_j = \mathbf{c}_i \neq \mathbf{0} \text{ (same velocity vector)},\\
1
&\text{ if } \mathbf{x}_\beta = \mathbf{x}_\alpha \text{ (local to node $\mathbf{x}_\alpha$)}\\
&\text{ and } \mathbf{x}_\alpha \text{ is a fluid node}\\
&\text{ and } \mathbf{x}_\alpha - \mathbf{c}_i \text{ is a solid node}\\
&\text{ and } \mathbf{c}_j = -\mathbf{c}_i \neq \mathbf{0} \text{ (reverse velocity vector)},\\
0
&\text{ else.}
\end{cases}
\label{eq:tildeS}
\end{align}

Next, we consider collision terms $T_{0[1]}$ through $T_4$.  Term $T_{0[1]}$ is trivial.  First consider $T_{1[1]}$:

\begin{align}
T_{1[1]} 
&=
\left( \frac{w_i}{\tau} \right)
\sum_j f_j(\mathbf{x}),
\label{eq:T_{1[1]}}
\end{align}
when \eqref{eq:lbm-rho} is substituted for $\rho$.

\begin{align}
T_{2[1]} 
&=
\left ( \frac{3w_i}{\tau} \right)  \sum_j ( \mathbf{c}_j \cdot \mathbf{c}_i) f_j(\mathbf{x})
\label{eq:T_{2[1]}}
\end{align}
when \eqref{eq:lbm-u-2} is substituted for $\rho \mathbf{u}$.

Now we define functions $\tilde{T}_{0[1]}$, $\tilde{T}_{1[1]}$, and $\tilde{T}_{2[1]}$, which accept a pair of first-order monomials $(m_r, m_c)$ and return only the coefficients of $m_c$ from terms $T_{0[1]}, T_{1[1]}$, and $T_{2[1]}$, respectively.

\begin{align}
\tilde{T}_{0[1]} ( f_i(\mathbf{x}_\alpha ) , f_j(\mathbf{x}_\beta ) )
&= 
\begin{cases}
\frac{-1}{\tau} &\text{ if } \mathbf{x}_\alpha = \mathbf{x}_\beta \text{ and } i = j, \\
0 &\text{ else.}
\end{cases}
\label{eq:tildeT_{0[1]}}
\end{align}

\begin{align}
\tilde{T}_{1[1]} ( f_i(\mathbf{x}_\alpha) , f_j(\mathbf{x}_\beta ) )
&= 
\begin{cases}
\frac{w_i}{\tau} &\text{ if } \mathbf{x}_\alpha = \mathbf{x}_\beta, \\
0 &\text{ else.}
\end{cases}
\label{eq:tildeT_{1[1]}}
\end{align}

\begin{align}
\tilde{T}_{2[1]} ( f_i(\mathbf{x}_\alpha) , f_j(\mathbf{x}_\beta ) )
&= 
\begin{cases}
\frac{3w_i}{\tau}(\mathbf{c}_j \cdot \mathbf{c}_i) &\text{ if } \mathbf{x}_\alpha = \mathbf{x}_\beta, \\
0 &\text{ else.}
\end{cases}
\label{eq:tildeT_{2[1]}}
\end{align}

Armed with $\tilde{T}_{0[1]}$, $\tilde{T}_{1[1]}$, and $\tilde{T}_{2[1]}$, we can define $\tilde{F}_1( m_r , m_c)$ for first-order monomials $m_r, m_c$.  $\tilde{F}_1( m_r , m_c)$ will return the element of the $F_1$ matrix in row $r = \mathcal{L}_1(m_r)$ and column $c = \mathcal{L}_1(m_c)$.
\begin{equation}
\tilde{F}_1(m_r, m_c) = 
\underbrace{\tilde{T}_{0[1]}( m_r , m_c )}_{\text{\eqref{eq:tildeT_{0[1]}}}}
+ \underbrace{\tilde{T}_{1[1]}( m_r , m_c )}_{\text{\eqref{eq:tildeT_{1[1]}}}}
+ \underbrace{\tilde{T}_{2[1]}( m_r , m_c )}_{\text{\eqref{eq:tildeT_{2[1]}}}}
\label{eq:tildeF_1}
\end{equation}

We continue our analysis on terms $T_3$ and $T_4$ with a similar goal.  First we consider $T_3.$
\begin{align}
T_3 
&= 
\underbrace{
\left( \frac{9w_i}{\tau} \right)(\rho \mathbf{u} \cdot \mathbf{c}_i)^2
}_{T_{3[2]}}
+ 
\underbrace{
\left( \frac{-9w_i}{2\tau} \right) \rho 
(\rho \mathbf{u} \cdot \mathbf{c}_i)^2,
}_{T_{3[3]}}
\label{eq:T_3}
\end{align}
$T_{3[2]}$ is a second-order polynomial with second-order monomials and $T_{3[3]}$ is a third-order polynomial with third-order monomials.  Note the subscripts $[2]$ and $[3]$.

\begin{subequations}
\begin{align}
T_{3[2]} 
&= 
\left( \frac{9w_i}{\tau} \right)(\rho \mathbf{u} \cdot \mathbf{c}_i)^2 \\
&= 
\left( 
    \frac{9w_i}{\tau} 
\right)
\left( 
    \sum_j (\mathbf{c}_i \cdot \mathbf{c}_j) f_j(\mathbf{x}) 
\right)^2\\
&= 
\left( 
    \frac{9w_i}{\tau} 
\right)
\left( 
    \sum_{j,k} (\mathbf{c}_i \cdot \mathbf{c}_j) (\mathbf{c}_i \cdot \mathbf{c}_k) f_j(\mathbf{x})f_k(\mathbf{x}) 
\right)
\end{align}
\label{eq:T_3[2]}
\end{subequations}

The expansion in \eqref{eq:T_3[2]} includes potentially nonzero coefficient for degenerate monomial terms.  If we accept \eqref{eq:T_3[2]} as is, we can then construct $\tilde{T}_{3[2]\text{-dense}}$, which will produce a potentially nonzero element of the $F_2$ matrix for each degenerate second-order monomial (where the degenerate second-order monomials correspond to distinct columns of the matrix).  The input to $\tilde{T}_{3[2]\text{-dense}}$ is a pair of monomials $(m_r, m_c)$, where $m_r$ is a first-order monomial and $m_c$ is a second-order monomial.

\begin{subequations}
\begin{align}
& \tilde{T}_{3[2]-\text{dense}} ( f_i(\mathbf{x}_\alpha ) , f_j(\mathbf{x}_\beta )f_k(\mathbf{x}_\gamma ) ) = \\
&\begin{cases}
\left(
    \frac{9w_i}{\tau} 
\right) 
(\mathbf{c}_i \cdot \mathbf{c}_j) (\mathbf{c}_i \cdot \mathbf{c}_k) 
  &\text{ if } \mathbf{x}_\alpha = \mathbf{x}_\beta = \mathbf{x}_\gamma,\\
0 & \text{ else.} 
\end{cases}
\end{align}
\label{eq:tildeT_3[2]_dense}
\end{subequations}

If we want to produce the sparsest $F_2$ matrix possible, then one method is to construct $\tilde{T}_{3[2]\text{-sparse}}$ as follows.

\begin{subequations}
\begin{align}
& \tilde{T}_{3[2]\text{-sparse}} ( f_i(\mathbf{x}_\alpha ) , f_j(\mathbf{x}_\beta )f_k(\mathbf{x}_\gamma ) ) = \\
&\begin{cases}
\left(
    \frac{9w_i}{\tau} 
\right) 
(\mathbf{c}_i \cdot \mathbf{c}_j)^2
& \text{ if } \mathbf{x}_\alpha = \mathbf{x}_\beta = \mathbf{x}_\gamma\\
 & \text{ and }
 \mathcal{L}_1(f_j(\mathbf{x}_{\beta})) =
 \mathcal{L}_1(f_k(\mathbf{x}_{\gamma})), \\
\left(
    \frac{9w_i}{\tau} 
\right) 
(2) (\mathbf{c}_i \cdot \mathbf{c}_j) (\mathbf{c}_i \cdot \mathbf{c}_k) 
& \text{ if } \mathbf{x}_\alpha = \mathbf{x}_\beta = \mathbf{x}_\gamma\\
 & \text{ and }
 \mathcal{L}_1(f_j(\mathbf{x}_{\beta})) <
 \mathcal{L}_1(f_k(\mathbf{x}_{\gamma})), \\
0 & \text{ else.} 
\end{cases}
\end{align}
\label{eq:tildeT_3[2]_sparse}
\end{subequations}
The inclusion of the clause $\mathcal{L}_1(f_j(\mathbf{x}_{\beta})) =
\mathcal{L}_1(f_k(\mathbf{x}_{\gamma}))$ implies that $f_j(\mathbf{x}_\beta ) = f_k(\mathbf{x}_\gamma )$ and the second-order monomial in question is the unique $f_j(\mathbf{x}_\beta )^2$.  The inclusion of the clause $\mathcal{L}_1(f_j(\mathbf{x}_{\beta})) <
\mathcal{L}_1(f_k(\mathbf{x}_{\gamma}))$ implies that we only assign a potentially nonzero coefficient to one of the two degenerate second-order monomials and the other degenerate monomial will be assigned a coefficient of zero.

Next we consider term $T_{3[3]}$ of \eqref{eq:T_3} which contains third-order monomials.
\begin{subequations}
\begin{align}
T_{3[3]} &= 
\left( \frac{-9w_i}{2\tau} \right) \rho 
(\rho \mathbf{u} \cdot \mathbf{c}_i)^2 \\
T_{3[3]} &= 
\left( 
    \frac{-9w_i}{2\tau}
\right) 
\left ( 
    \sum_j f_j(\mathbf{x} ) 
\right) 
\left( 
    \sum_k (\mathbf{c}_i \cdot \mathbf{c}_k)f_k(x)
\right)^2\\
T_{3[3]} &= 
\left( 
    \frac{-9w_i}{2\tau}
\right) 
\left ( 
    \sum_j f_j(\mathbf{x} ) 
\right) 
\left( 
    \sum_{k,\ell} (\mathbf{c}_i \cdot \mathbf{c}_k)(\mathbf{c}_i \cdot \mathbf{c}_\ell)f_k(x)f_\ell(x)
\right)
\label{eq:T_{3[3]}_expanded_puc2}
\\
T_{3[3]} &= 
\left( \frac{-9w_i}{2\tau} \right) 
\sum_{k,\ell} (\mathbf{c}_i \cdot \mathbf{c}_k)(\mathbf{c}_i \cdot \mathbf{c}_\ell)\underbrace{f_0(\mathbf{x})}f_k(\mathbf{x})f_\ell(\mathbf{x})
\label{eq:T_{3[3]}_distributed_start}
\\
&\quad \quad \vdots
\notag
\\
&\quad \quad +
\left( \frac{-9w_i}{2\tau} \right)
\sum_{k,\ell} (\mathbf{c}_i \cdot \mathbf{c}_k)(\mathbf{c}_i \cdot \mathbf{c}_\ell)\underbrace{f_j(\mathbf{x})}f_k(\mathbf{x})f_\ell(\mathbf{x})
\notag
\\
&\quad \quad \vdots
\notag
\\
&\quad \quad +
\left( \frac{-9w_i}{2\tau} \right)
\sum_{k,\ell} (\mathbf{c}_i \cdot \mathbf{c}_k)(\mathbf{c}_i \cdot \mathbf{c}_\ell)\underbrace{f_{Q-1}}(\mathbf{x})f_k(\mathbf{x})f_\ell(\mathbf{x}) 
\label{eq:T_{3[3]}_distributed_end}
\end{align}
\label{eq:T_3[3]}
\end{subequations}
And we can clearly see the coefficient we seek for the third-order monomial $f_j(\mathbf{x})f_k(\mathbf{x})f_\ell(\mathbf{x})$.  The expansion in \eqref{eq:T_{3[3]}_expanded_puc2} included all second-order (including degenerate) monomials.  The distribution of the $(\rho \mathbf{u}\cdot \mathbf{c}_i)^2$ term into the $\rho$ term then created all third-order (including degenerate) monomials.  Define $\tilde{T}_{3[3]\text{-dense}}$ for input pair of monomials $(m_r, m_c)$ where $m_r$ is a first-order monomial and $m_c$ is a third-order monomial as:
\begin{align}
\label{eq:tildeT_3[3]_dense}
& \tilde{T}_{3[3]\text{-dense}} ( f_i(\mathbf{x}_\alpha) , f_j(\mathbf{x}_\beta )f_k(\mathbf{x}_\gamma )f_\ell(\mathbf{x}_\delta) ) = \\
&\begin{cases}
\left( 
\frac{ -9 w_i }{2 \tau}
\right)
(\mathbf{c}_i \cdot \mathbf{c}_k) (\mathbf{c}_i \cdot \mathbf{c}_\ell) 
&\text{ if } \mathbf{x}_\alpha = \mathbf{x}_\beta = \mathbf{x}_\gamma = \mathbf{x}_\delta, \\
0 & \text{ else.} \\
\end{cases}
\notag
\end{align}

To construct the sparse version of $\tilde{T}_{3[3]}$ we need to consider the factors of the monomial $m_c = f_j(\mathbf{x}_\beta )f_k(\mathbf{x}_\gamma )f_\ell(\mathbf{x}_\delta)$.  If the third-order monomial $m_c$ has one unique factor, e.g., $m_c = f_j(\mathbf{x}_\beta)^3$, then $m_c$ only appears once in lines \eqref{eq:T_{3[3]}_distributed_start}-\eqref{eq:T_{3[3]}_distributed_end} and so the set of coefficients for degenerate variants of $m_c$ only has one item.  We also consider third-order monomials $m_c$ of the form $m_c= f_j(\mathbf{x}_\beta)^2f_k(\mathbf{x}_\gamma)$ with two distinct factors:  $f_j(\mathbf{x}_\beta)$ and $f_k(\mathbf{x}_\gamma)$ where $f_j(\mathbf{x}_\beta) \neq f_k(\mathbf{x}_\gamma)$, and of the form $m_c = f_j(\mathbf{x}_\beta)f_k(\mathbf{x}_\gamma)f_\ell(\mathbf{x}_\delta)$ with three distinct factors $f_j(\mathbf{x}_\beta)$, $f_k(\mathbf{x}_\gamma)$ and $f_\ell(\mathbf{x}_\delta)$, where $f_j(\mathbf{x}_\beta) \neq f_k(\mathbf{x}_\gamma) \neq f_\ell(\mathbf{x}_\delta)$.

\begin{align}
\label{eq:tildeT_3[3]_sparse}
& \tilde{T}_{3[3]\text{-sparse}} ( f_i(\mathbf{x}_\alpha) , f_j(\mathbf{x}_\beta )f_k(\mathbf{x}_\gamma )f_\ell(\mathbf{x}_\delta) ) = \\
&\begin{cases}
\left( 
    \frac{ -9 w_i }{2 \tau}
\right)
(\mathbf{c}_i \cdot \mathbf{c}_j)^2 
&\text{ if } \mathbf{x}_\alpha = \mathbf{x}_\beta = \mathbf{x}_\gamma = \mathbf{x}_\delta \\
& \text{ and } 
\mathcal{L}_1( f_j(\mathbf{x}_\beta )) = 
\mathcal{L}_1( f_k(\mathbf{x}_\gamma )) \\
& \text{ and } 
\mathcal{L}_1( f_k(\mathbf{x}_\gamma )) = 
\mathcal{L}_1( f_\ell(\mathbf{x}_\delta )), \\
\left( 
    \frac{ -9 w_i }{2 \tau}
\right)
\left( 
    (\mathbf{c}_i \cdot \mathbf{c}_j)^2 + 2(\mathbf{c}_i \cdot \mathbf{c}_j)(\mathbf{c}_i \cdot \mathbf{c}_\ell) 
\right)
&\text{ if } \mathbf{x}_\alpha = \mathbf{x}_\beta = \mathbf{x}_\gamma = \mathbf{x}_\delta \\
& \text{ and } 
\mathcal{L}_1( f_j(\mathbf{x}_\beta )) = 
\mathcal{L}_1( f_k(\mathbf{x}_\gamma )) \\
& \text{ and } 
\mathcal{L}_1( f_k(\mathbf{x}_\gamma )) < 
\mathcal{L}_1( f_\ell(\mathbf{x}_\delta )), \\
\left( 
    \frac{ -9 w_i }{2 \tau}
\right)
\left( 
    2(\mathbf{c}_i \cdot \mathbf{c}_j)(\mathbf{c}_i \cdot \mathbf{c}_k)
    +
    2(\mathbf{c}_i \cdot \mathbf{c}_j)(\mathbf{c}_i \cdot \mathbf{c}_\ell)
    +
    2(\mathbf{c}_i \cdot \mathbf{c}_k)(\mathbf{c}_i \cdot \mathbf{c}_\ell)
\right)
&\text{ if } \mathbf{x}_\alpha = \mathbf{x}_\beta = \mathbf{x}_\gamma = \mathbf{x}_\delta \\
& \text{ and } 
\mathcal{L}_1( f_j(\mathbf{x}_\beta )) < 
\mathcal{L}_1( f_k(\mathbf{x}_\gamma )) \\
& \text{ and } 
\mathcal{L}_1( f_k(\mathbf{x}_\gamma )) < 
\mathcal{L}_1( f_\ell(\mathbf{x}_\delta )), \\
0 & \text{ else.}
\end{cases}
\notag
\end{align}

Next, consider term $T_4$ from \eqref{eq:lbe-diff-eq-expanded}:

\begin{align}
T_4 = 
\underbrace{
\left ( \frac{-3w_i}{\tau} \right) (\rho\mathbf{u} \cdot \rho\mathbf{u})
}_{T_{4[2]}}
+
\underbrace{
\left ( \frac{3w_i}{2\tau} \right) \rho (\rho\mathbf{u} \cdot \rho\mathbf{u}),
}_{T_{4[3]}}
\label{eq:T_4}
\end{align}
where $T_{4[2]}$ is a second-order polynomial and $T_{4[3]}$ is a third-order polynomial.

\begin{subequations}
\begin{align}
T_{4[2]} &= 
\left ( \frac{-3w_i}{\tau} \right) (\rho\mathbf{u} \cdot \rho\mathbf{u}) \\
T_{4[2]} &= 
\left ( \frac{-3w_i}{\tau} \right) 
\left( \sum_j f_j(\mathbf{x}) \mathbf{c}_j \right) \cdot
\left( \sum_j f_j(\mathbf{x}) \mathbf{c}_j \right)\\
T_{4[2]} &= 
\left ( \frac{-3w_i}{\tau} \right) 
\left( 
\left( \sum_j f_j(\mathbf{x}) c_{ix} \right)^2 
+ \left( \sum_j f_j(\mathbf{x}) c_{iy} \right)^2 
+\left( \sum_j f_j(\mathbf{x}) c_{iz} \right)^2 
\right)\\
T_{4[2]} &= 
\left ( \frac{-3w_i}{\tau} \right) 
\sum_{j,k} 
c_{jx} c_{kx} f_j(\mathbf{x})f_k(\mathbf{x})
\label{eq:T_4[2]_expanded_x}
\\
&\quad \quad+
\left ( \frac{-3w_i}{\tau} \right) 
\sum_{j,k} 
c_{jy} c_{ky} f_j(\mathbf{x})f_k(\mathbf{x})
\label{eq:T_4[2]_expanded_y}\\
&\quad \quad+
\left ( \frac{-3w_i}{\tau} \right) 
\sum_{j,k} 
c_{jz} c_{kz} f_j(\mathbf{x})f_k(\mathbf{x})
\label{eq:T_4[2]_expanded_z}
\\
T_{4[2]} 
&= 
\left ( \frac{-3w_i}{\tau} \right) 
\sum_{j,k} (\mathbf{c}_j \cdot \mathbf{c}_k) f_j(\mathbf{x})f_k(\mathbf{x})
\end{align}
\label{eq:T_4[2]}
\end{subequations}
Recall that $\mathbf{c}_i = (c_{ix} , c_{iy} , c_{iz} )$.  See table \ref{tab:D3Q27-constants} for a complete list of $\mathbf{c}_i$ lattice velocity vectors.  Similar to $T_{3[2]}$ in \eqref{eq:T_3[2]}, our expansion of term $T_{4[2]}$ in \eqref{eq:T_4[2]} contains all degenerate second-order monomials.  We can then choose to construct a dense $F_2$ by relying on $\tilde{T}_{4[2]\text{-dense}}$ in \eqref{eq:tildeT_4[2]_dense} or a sparse $F_2$ matrix by relying on $\tilde{T}_{4[2]\text{-sparse}}$ in \eqref{eq:tildeT_4[2]_sparse}.

\begin{align}
\label{eq:tildeT_4[2]_dense}
& \tilde{T}_{4[2]\text{-dense}} ( f_i(\mathbf{x}_{\alpha}) , f_j(\mathbf{x}_{\beta})f_k(\mathbf{x}_{\gamma}) ) ) = & \\
& \begin{cases}
\left( \frac{-3w_i}{\tau} \right) (\mathbf{c}_j \cdot \mathbf{c}_k) 
&\text{ if } \mathbf{x}_{\alpha} = \mathbf{x}_{\beta} = \mathbf{x}_{\gamma},
\\
0 &\text{ else.} 
\end{cases}
\notag
\end{align}

\begin{align}
    \label{eq:tildeT_4[2]_sparse}
    &
    \tilde{T}_{4[2]\text{-sparse}} ( f_i(\mathbf{x}_{\alpha}) , f_j(\mathbf{x}_{\beta})f_k(\mathbf{x}_{\gamma}) ) ) 
    =&
    \\
    & 
    \begin{cases}
    \left( \frac{-3w_i}{\tau} \right)
    (\mathbf{c}_j \cdot \mathbf{c}_k) 
    &\text{ if } \mathbf{x}_{\alpha} = \mathbf{x}_{\beta} = \mathbf{x}_{\gamma} \\
    &\text{ and } 
    \mathcal{L}_1(f_j(\mathbf{x}_{\beta})) =
    \mathcal{L}_1(f_k(\mathbf{x}_{\gamma})),
    \\
    \left( \frac{-3w_i}{\tau} \right)
    (2)(\mathbf{c}_j \cdot \mathbf{c}_k) 
    &\text{ if } \mathbf{x}_{\alpha} = \mathbf{x}_{\beta} = \mathbf{x}_{\gamma}\\
    &\text{ and } 
    \mathcal{L}_1(f_j(\mathbf{x}_{\beta})) <
    \mathcal{L}_1(f_k(\mathbf{x}_{\gamma})),
    \\
    0 &\text{ else.} 
    \end{cases}
    \notag
\end{align}

Consider term $T_{4[3]}$ of \eqref{eq:T_4}:

\begin{subequations}
\begin{align}
T_{4[3]} &= 
\left ( \frac{3w_i}{2\tau} \right) \rho (\rho\mathbf{u} \cdot \rho\mathbf{u}) \\
T_{4[3]} &= 
\left ( \frac{3w_i}{2\tau} \right) 
\left(
    \rho
\right)
\left( 
    \sum_{k,\ell} (\mathbf{c}_k \cdot \mathbf{c}_\ell) f_k(\mathbf{x})f_\ell(\mathbf{x})
\right)\label{eq:T_4[3]_algebra_from_T_4[2]}\\
T_{4[3]} &= 
\left ( \frac{3w_i}{2\tau} \right) 
\left(
    \sum_j f_j(\mathbf{x})
\right)
\left( 
    \sum_{k,\ell} (\mathbf{c}_k \cdot \mathbf{c}_\ell) f_k(\mathbf{x})f_\ell(\mathbf{x})
\right)\\
T_{4[3]} &= 
\left ( \frac{3w_i}{2\tau} \right) 
\sum_{k,\ell} (\mathbf{c}_k \cdot \mathbf{c}_\ell) \underbrace{f_0(\mathbf{x})}f_k(\mathbf{x})f_\ell(\mathbf{x})\\
&\quad \quad \vdots
\notag
\\
&\quad \quad+ 
\left ( \frac{3w_i}{2\tau} \right) 
\sum_{k,\ell} (\mathbf{c}_k \cdot \mathbf{c}_\ell) \underbrace{f_j(\mathbf{x})}f_k(\mathbf{x})f_\ell(\mathbf{x})
\notag
\\
&\quad \quad \vdots
\notag
\\
&\quad \quad + 
\left ( \frac{3w_i}{2\tau} \right) 
\sum_{k,\ell} (\mathbf{c}_k \cdot \mathbf{c}_\ell) \underbrace{f_{Q-1}(\mathbf{x})}f_k(\mathbf{x})f_\ell(\mathbf{x})
\notag
\end{align}
\label{eq:T_4[3]}
\end{subequations}

Line \eqref{eq:T_4[3]_algebra_from_T_4[2]} is derived by substituting similar steps of \eqref{eq:T_4[2]}.  Similar to our analysis of $T_{3[3]}$, the expansion of $T_{4[3]}$ in \eqref{eq:T_4[3]} includes all third-order monomials, including degenerate variants.  We then proceed to construct $\tilde{T}_{4[3]\text{-dense}}$, which assigns a potentially nonzero coefficient to each degenerate third-order monomial (corresponding to multiple columns in the $F_3$ matrix), and we construct $\tilde{T}_{4[3]\text{-sparse}}$, which assigns a potentially nonzero coefficient to one specific representation of a third-order monomial (corresponding to only one column in the $F_3$ matrix).

\begin{align}
\label{eq:tildeT_4[3]_dense}
& \tilde{T}_{4[3]\text{-dense}} ( f_i(\mathbf{x}_\alpha) , f_j(\mathbf{x}_\beta )f_k(\mathbf{x}_\gamma )f_\ell(\mathbf{x}_\delta) ) = \\
&\begin{cases}
\left ( \frac{3w_i}{2\tau} \right) 
\left( 
\mathbf{c}_k \cdot \mathbf{c}_\ell
\right)
&\text{ if } \mathbf{x}_\alpha = \mathbf{x}_\beta = \mathbf{x}_\gamma = \mathbf{x}_\delta,\\
0 & \text{ else.} 
\end{cases}
\notag
\end{align}

To construct $\tilde{T}_{4[3]\text{-sparse}}$, we again need to consider third-order monomials of the form $m_c = f_j(\mathbf{x}_\beta)^3$, $m_c = f_j(\mathbf{x}_\beta )^2f_k(\mathbf{x}_\gamma )$, or $m_c = f_j(\mathbf{x}_\beta )f_k(\mathbf{x}_\gamma )f_\ell(\mathbf{x}_\delta)$, where $f_j(\mathbf{x}_\beta )\neq  f_k(\mathbf{x}_\gamma ) \neq f_\ell(\mathbf{x}_\delta)$ are unique factors.   

\begin{align}
\label{eq:tildeT_4[3]_sparse}
& \tilde{T}_{4[3]\text{-sparse}} ( f_i(\mathbf{x}_\alpha) , f_j(\mathbf{x}_\beta )f_k(\mathbf{x}_\gamma )f_\ell(\mathbf{x}_\delta) ) = \\
&\begin{cases}
\left ( \frac{3w_i}{2\tau} \right)  ( \mathbf{c}_j \cdot \mathbf{c}_j)
&\text{ if } \mathbf{x}_\alpha = \mathbf{x}_\beta = \mathbf{x}_\gamma = \mathbf{x}_\delta \\
& \text{ and } 
\mathcal{L}_1( f_j(\mathbf{x}_\beta )) = 
\mathcal{L}_1( f_k(\mathbf{x}_\gamma )) \\
& \text{ and } 
\mathcal{L}_1( f_k(\mathbf{x}_\gamma )) = 
\mathcal{L}_1( f_\ell(\mathbf{x}_\delta )),
\\
\left ( \frac{3w_i}{2\tau} \right)  
\left( 
    \mathbf{c}_j \cdot \mathbf{c}_j + 2(\mathbf{c}_j \cdot \mathbf{c}_k)
\right)
&\text{ if } \mathbf{x}_\alpha = \mathbf{x}_\beta = \mathbf{x}_\gamma = \mathbf{x}_\delta \\
& \text{ and } 
\mathcal{L}_1( f_j(\mathbf{x}_\beta )) = 
\mathcal{L}_1( f_k(\mathbf{x}_\gamma )) \\
& \text{ and } 
\mathcal{L}_1( f_k(\mathbf{x}_\gamma )) < 
\mathcal{L}_1( f_\ell(\mathbf{x}_\delta )),
\\
\left ( \frac{3w_i}{2\tau} \right) 
\left(
 2(\mathbf{c}_j \cdot \mathbf{c}_k) 
 + 2(\mathbf{c}_j \cdot \mathbf{c}_\ell)
 + 2(\mathbf{c}_k \cdot \mathbf{c}_\ell)
\right)
&\text{ if } \mathbf{x}_\alpha = \mathbf{x}_\beta = \mathbf{x}_\gamma = \mathbf{x}_\delta \\& \text{ and } 
\mathcal{L}_1( f_j(\mathbf{x}_\beta )) < 
\mathcal{L}_1( f_k(\mathbf{x}_\gamma )) \\
& \text{ and } 
\mathcal{L}_1( f_k(\mathbf{x}_\gamma )) < 
\mathcal{L}_1( f_\ell(\mathbf{x}_\delta )),
\\
0 & \text{ else.}
\notag
\end{cases}
\end{align}

Now that we have derived $\tilde{T}_{3[2]}, \tilde{T}_{3[3]}, \tilde{T}_{4[2]}$, and $\tilde{T}_{4[3]}$ we can present the sparse and dense variants of functions $\tilde{F}_2$ and $\tilde{F}_3$.  $\tilde{F}_2$ is defined for a pair of monomials $m_r, m_c$, where $m_r$ is a first-order monomial and $m_c$ is a second-order monomial.  $\tilde{F}_2$ will return the element of the $F_2$ matrix at row $\mathcal{L}_1(m_r)$ and column $\mathcal{L}_2(m_c)$ as

\begin{align}
    \tilde{F}_{2\text{-dense}}( m_r , m_c)
    = 
    \underbrace{\tilde{T}_{3[2]\text{-dense}} ( m_r , m_c)}_{\text{\eqref{eq:tildeT_3[2]_dense}}}
    + 
    \underbrace{\tilde{T}_{4[2]\text{-dense}} (m_r , m_c)}_{\text{\eqref{eq:tildeT_4[2]_dense}}},
\label{eq:tildeF_2_dense}
\end{align}
or
\begin{align}
    \tilde{F}_{2\text{-sparse}}( m_r , m_c) = 
    \underbrace{\tilde{T}_{3[2]\text{-sparse}} ( m_r , m_c)}_{\text{\eqref{eq:tildeT_3[2]_sparse}}}
    + 
    \underbrace{\tilde{T}_{4[2]\text{-sparse}} (m_r , m_c)}_{\text{\eqref{eq:tildeT_4[2]_sparse}}}.
\label{eq:tildeF_2_sparse}
\end{align}

$\tilde{F}_3$ is defined for a pair of monomials $m_r, m_c$, where $m_r$ is a first-order monomial and $m_c$ is a third-order monomial.  $\tilde{F}_3$ will return the element of the $F_3$ matrix at row $\mathcal{L}_1(m_r)$ and column $\mathcal{L}_3(m_c)$ as

\begin{align}
    \tilde{F}_{3\text{-dense}}( m_r , m_c) 
    &= 
    \underbrace{\tilde{T}_{3[3]\text{-dense}} ( m_r , m_c)}_{\text{\eqref{eq:tildeT_3[3]_dense}}}
    + 
    \underbrace{\tilde{T}_{4[3]\text{-dense}} (m_r , m_c)}_{\text{\eqref{eq:tildeT_4[3]_dense}}}
\label{eq:tildeF_3_dense}
\end{align}
or
\begin{align}
    \tilde{F}_{3\text{-sparse}}( m_r , m_c) = 
    \underbrace{\tilde{T}_{3[3]\text{-sparse}} ( m_r , m_c)}_{\text{\eqref{eq:tildeT_3[3]_sparse}}}    
    + 
    \underbrace{\tilde{T}_{4[3]\text{-sparse}} (m_r , m_c)}_{\text{\eqref{eq:tildeT_4[3]_sparse}}}  .  
\label{eq:tildeF_3_sparse}
\end{align}

It is possible to construct other combinations of the sparse and dense variations of the functions.  We use the dense version of the functions for our analysis.

\section{Indexing streaming matrix elements}\label{sec:streaming}

The goal of this section is to develop relationships and expressions for row/column variable index pairs where nonzero entries appear in the streaming matrix $S$.  This section will use the variable ordering index implementation described in \eqref{eq:L_x-ordering-implementation}-\eqref{eq:L_1-ordering-inverse} from appendix \ref{sec:notation}.  In our discussion, we may refer to the $f$-variable at index $\mu \in \mathbb{Z}_{nQ}$ as $f_\mu = f_i(\mathbf{x}_\alpha) \equiv f_{i(\mu)}(\mathbf{x}_{\alpha(\mu)})$, where $i, \mathbf{x}_\alpha$, and $\mu$ satisfy our index mapping \eqref{eq:L_x-ordering-implementation}-\eqref{eq:L_1-ordering-inverse}.  When the situation is clear enough, we may drop the $i(\mu)$, $\alpha(\mu)$ notation and simply refer to positions and velocity vectors as $\mathbf{x}_\mu \equiv \mathbf{x}_{\alpha(\mu)}$ and $\mathbf{c}_{\mu} \equiv \mathbf{c}_{i(\mu)}$.

The streaming matrix $S$, has dimensions $nQ\times nQ$, and has the action of redistributing population within the position-velocity population distribution vector from $f_\eta\equiv f_{j(\eta)}(\mathbf{x}_{\beta(\eta)})$ to some $f_\mu\equiv f_{i(\mu)}(\mathbf{x}_{\alpha(\mu)})$.  Each row of $S$ has a $-1$ on the diagonal corresponding to outward streaming, and a single $+1$ corresponding to inward streaming, and all zeros elsewhere.  

The streaming matrix may be broken down into the outgoing and incoming streaming terms: \\
\begin{align}
    \label{eq:streaming_matrix}
    S=S^{\mathrm{out}}+S^{\mathrm{in}}.
\end{align}

The outgoing streaming terms are trivial, serving to remove population from each $f_\eta$:
\begin{equation}
    S^{\mathrm{out}}=-I.
\end{equation}

For rows/columns corresponding to zero velocity vectors (i.e., $i(\mu)=j(\eta)=0$), there should be no streaming (outgoing or incoming).  In this case, $f_\mu$ is decreased by $-1$ from $S_{\mu \mu}^{\mathrm{out}}$ and increased by $+1$ from $S_{\mu \mu}^{\mathrm{in}}$, the net result being a zero change.  This leads to zero rows/columns in the $S$ matrix, thus $S$ is not full rank.

We define a binary vector to flag solid (vs. fluid) nodes
\begin{equation}\label{eq:issolid}
   \mathcal{N}_\alpha(\mu) = 
\begin{cases}
     0 \text{ if } \mathbf{x}_{\alpha(\mu)} \text{ is a fluid node} \\
     1 \text{ if } \mathbf{x}_{\alpha(\mu)} \text{ is a solid node},
\end{cases}
\end{equation}
and then consider inbound streaming terms from node $\mathbf{x}_{\beta(\eta)}$ into $\mathbf{x}_{\alpha(\mu)}$:
\begin{align}
     S^{\mathrm{in}}_{\mu\eta} 
     & = 
     \underbrace{
     (1 - \mathcal{N}_\mu)
     (1 - \mathcal{N}_\eta)
     (\delta_{\mathbf{x}_\mu, \mathbf{x}_\eta+\mathbf{c}_\eta})
     (\delta_{\mathbf{c}_\mu,\mathbf{c}_\eta})
     }_{\text{$\mu$ and $\eta$ both correspond to fluid nodes}}
     +
     \underbrace{
     (\mathcal{N}_{\zeta})
     (1 - \mathcal{N}_\eta)
     (\delta_{\mathbf{x}_\mu,\mathbf{x}_\eta})
     (\delta_{\mathbf{c}_\mu,-\mathbf{c}_\eta})
     (\delta_{\mathbf{x}_\zeta,\mathbf{x}_\eta+\mathbf{c}_\eta})
     }_{\text{$\zeta$ corresponds to an adjacent solid node}}, 
     \label{eq:Sin}
\end{align}
with simplified notation $\mathbf{x}_\mu\equiv\mathbf{x}_{\alpha(\mu)}$ and $\mathbf{c}_\mu\equiv\mathbf{c}_{i(\mu)}$.  The fluid streaming scenario in the first term is depicted in figure \ref{fig:streaming-figure}.  The bounce-back scenario in the second term is depicted in figure \ref{fig:bb-figure} due to an adjacent solid node $\mathbf{x}_\zeta = \mathbf{x}_\eta + \mathbf{c}_\eta$.

Beginning with the first term, the condition $\mathbf{x}_\mu = \mathbf{x}_\eta+\mathbf{c}_\eta$ is satisfied for the above-specified position ordering when
\begin{align}
    \mu
    =
    \eta 
    + (Q)c_{\eta x}
    + (Qn_x)c_{\eta y}
    + (Qn_xn_y)c_{\eta z}, 
    \label{eq:mu-fluid}
\end{align}
where in the special case of periodic boundary conditions the addition of the second, third, and fourth terms must be carried out modulo $Qn_x$, modulo $Qn_x n_y$, and modulo $Qn_xn_yn_z$, respectively. The condition $\mathbf{c}_\mu=\mathbf{c}_\eta$ is likewise satisfied, since the velocity vector index $i$ \eqref{eq:L_1-ordering-inverse} does not change when variable index $\mu$ changes by an integer multiple of $Q$.

\begin{figure}
    \centering
    \begin{tikzpicture}
        
        \draw (6,0) 
        node[
            circle, 
            draw, 
            fill=white, 
            inner sep=3pt, 
            minimum size=5mm, 
            label=below:{$\mathbf{x}_{\beta(\eta)}$},
            label=above:{Fluid}
        ] 
        (F1) 
        {};
        
        \draw (12,0)
        node[
            circle, 
            draw, 
            fill=white, 
            inner sep=3pt, 
            minimum size=5mm, 
            label=below:{$\mathbf{x}_{\alpha(\mu)}= \mathbf{x}_{\beta(\eta)} + \mathbf{c}_{j(\eta)}$},
            label=above:{Fluid}
        ]
        (F2)
        {};

        \draw[->] (F1) -- node[above] {$\mathbf{c}_{j(\eta)} $} (9,0);
        \draw[->] (F2) -- node[above] {$\mathbf{c}_{i(\mu)} = \mathbf{c}_{j(\eta)} $} (15,0);
    \end{tikzpicture}
    \caption{Fluid streaming: the population $f_\mu$ (right) is replaced by the population $f_\eta$ (left). $f_\mu$ corresponds to the population at node $\mathbf{x}_{\alpha(\mu)}$ with velocity vector $\mathbf{c}_{i(\mu)}$.  $f_\eta$ corresponds to the population at node $\mathbf{x}_{\beta(\eta)}$ with velocity vector $\mathbf{c}_{j(\eta)}$.  In this case $\mathbf{c}_{i(\mu)} = \mathbf{c}_{j(\eta)}$.}
    \label{fig:streaming-figure}
\end{figure}

\begin{figure}
\centering
\begin{tikzpicture}

\draw (6,0) 
node[
    circle, 
    draw, 
    fill=white, 
    inner sep=3pt, 
    minimum size=5mm, 
    label=below:{$\mathbf{x}_{\alpha(\mu)} = \mathbf{x}_{\beta(\eta)}$},
    label=above:{Fluid}
] 
(F) 
{};

\draw (12,0)
node[
    circle, 
    draw, 
    fill=black, 
    inner sep=3pt, 
    minimum size=5mm, 
    label=below:{$\mathbf{x}_{\gamma(\zeta)} = \mathbf{x}_{\beta(\eta)} + \mathbf{c}_{j(\eta)}$},
    label=above:{Solid}
]
(S)
{};

\draw[->] (F) -- node[above] {$\mathbf{c}_{i(\eta)}$} (9,0);
\draw[->] (F) -- node[above] {$\mathbf{c}_{i(\mu)} = -\mathbf{c}_{j(\eta)} $} (3,0);

\end{tikzpicture}
\caption{Bounce back: the population $f_\mu$ (left-pointing arrow) is replaced by the population $f_\eta$ (right-pointing arrow). $f_\mu$ corresponds to the population at fluid node $\mathbf{x}_{\alpha(\mu)}=\mathbf{x}_{\beta(\eta)}$ with velocity vector $\mathbf{c}_{i(\mu)}$ traveling away from the solid.  $f_\eta$ corresponds to the population at fluid node $\mathbf{x}_{\beta(\eta)}$ with velocity vector $\mathbf{c}_{j(\eta)}$ traveling towards the solid.  In this case $\mathbf{c}_{i(\mu)} = -\mathbf{c}_{j(\eta)}$.}
\label{fig:bb-figure}
\end{figure}

In the second term of \eqref{eq:Sin}, we instead have the condition $\mathbf{x}_\mu = \mathbf{x}_\eta$ since bounce-back is a local process. However, it must first be determined whether the adjacent node at $\mathbf{x}_{\zeta} = \mathbf{x}_\eta+\mathbf{c}_\eta$ is solid. The expression for index $\zeta$ is the same as \eqref{eq:mu-fluid},
\begin{equation}
    \zeta = \eta + (Q)c_{\eta x} + (Qn_x) c_{\eta y} + (Qn_xn_y) c_{\eta z}, 
    \label{eq:zeta}
\end{equation}
with similar modulo operations for periodic boundaries.  

For the velocity vector ordering specified in appendix \ref{sec:D3Q27_lattice_constants}, the bounce-back condition $\mathbf{c}_\mu=-\mathbf{c}_\eta$ is satisfied  when
\begin{align}
    \mu & = \eta + c_{\eta x} + 3 c_{\eta y} + 9 c_{\eta z}
\label{eq:mu-solid}
\end{align}
which makes use of the fact that the lattice vector components $c_{i x}$, $c_{i y}$, and $c_{i z}$ cycle through values $\{1,-1,0\}$ with periodicities of 3, 9, and 27, respectively. The condition $\mathbf{x}_\mu = \mathbf{x}_\eta$ is likewise satisfied since it can be shown that
\begin{equation}
    0 \le i + c_{i x} + 3 c_{i y} + 9 c_{i z} \le 26, \forall i\in [0,26],
\end{equation}
and therefore $\mu$ is in the same $Q$-length position block as $\eta$.

The lattice vector components may be expressed as arithmetic functions of index $i(\mu)=\mu\bmod Q$ \eqref{eq:L_1-ordering-inverse}, consistent with the ordering specified in appendix \ref{sec:D3Q27_lattice_constants}:
\begin{subequations}
\begin{align}
c_{\mu x} &= \left[ (\mu \bmod Q) + 2\right] \bmod{3}  - 1 \\
c_{\mu y} &= \left[ (\mu \bmod Q)\bdiv 3 + 2\right] \bmod{3}  - 1 \\
c_{\mu z} &= \left[ (\mu \bmod Q)\bdiv 9 + 2\right] \bmod{3} - 1
\end{align}
\label{eq:c_mu}
\end{subequations}

The $+1$ entries of $S^{\mathrm{in}}_{\mu\eta}$ correspond to all row index $\mu$, column index $\eta$ pairings that satisfy either:
\begin{itemize}
    \item \eqref{eq:mu-fluid} with $\mathcal{N}_\mu=\mathcal{N}_\eta=0$, or
    \item \eqref{eq:mu-solid} with $\mathcal{N}_\eta=0$ and $\mathcal{N}_{\zeta}=1$ per \eqref{eq:zeta}
\end{itemize}

Note that in \eqref{eq:Sin} for $S^{\mathrm{in}}_{\mu\eta}$, both $(1-\mathcal{N}_{\mu})$ in the first term and $\mathcal{N}_{\zeta}$ in the second term require evaluating $\mathcal{N}$ for the same expression (\eqref{eq:mu-fluid} and \eqref{eq:zeta}). In fact, in the first term one could replace $(1- \mathcal{N}_\mu)$ with $(1- \mathcal{N}_{\zeta})$ since in this case $\mu=\zeta$. There is a one-to-one correspondence between $\eta$ and $\zeta$. We can think of $\mathcal{N}_{\zeta}$ as a flag that tells us whether the population streaming out from index $\eta$ will stream into index $\mu=\zeta$ (if $\mathbf{x}_{\zeta}$ is a fluid node) or instead stream into the index $\mu$ given by \eqref{eq:mu-solid} if $\mathbf{x}_{\zeta}$ is a solid node.

\subsection{Flow past a sphere}

In the simple example of flow past a sphere, we can provide an explicit expression for $\mathcal{N}(\mathbf{x})$. We consider a sphere of radius $r$, centered at position $\mathbf{x}_\mathrm{c}=(x_\mathrm{c},y_\mathrm{c},z_\mathrm{c})$. All nodes within the sphere are deemed solid, so we have
\begin{align}
   \label{eq:classical_function}
    \mathcal{N(\mathbf{x}_\alpha}) = H\big( r^2 -(x_\alpha - x_\mathrm{c})^2 - (y_\alpha - y_\mathrm{c})^2 - (z_\alpha - z_\mathrm{c})^2 \big),
\end{align}
where $H$ is the Heaviside step function.

\subsection{Rectangular prisms}

An explicit expression for $\mathcal{N}(\mathbf{x})$ can also be provided for the simple case of a rectangular prism spanning $(x_0,y_0,z_0)$ to $(x_0+L_x,y_0+L_y,z_0+L_z)$:
\begin{align}
   \label{eq:rec_prism}
    \mathcal{N(\mathbf{x}_\alpha}) =  &\Big[H\big( x_\alpha - x_0 \big) - H\big( x_\alpha - (x_0 + L_x) \big)\Big] \\
    & \times \Big[H\big( y_\alpha - y_0 \big) - H\big( y_\alpha - (y_0 + L_y) \big)\Big] \notag \\
    & \times \Big[H\big( z_\alpha - z_0 \big) - H\big( z_\alpha - (z_0 + L_z) \big)\Big]. \notag
\end{align}

\section{Quantum circuit details}
\label{sec:quantum_compilation}

Quantum mechanics only affords us unitary operations.  We will use \emph{block encoding} to embed non-unitary matrices into unitary matrices and apply matrix arithmetic on the embedded operations.  Block encoding was introduced in \cite{Gilyen2019long}.

\begin{definition}[Block Encoding]\label{def:block_encoding}
    Suppose A is an a-qubit operator, $\alpha$, $\epsilon \in \mathbb{R}_+$,  $b \in \mathbb{Z}_+$, then an $a + b$ qubit
unitary $U_A$ is said to be an $(\alpha, b, \epsilon)$ block encoding of A if 
 $ \| A - \alpha\left( \bra{0}^b \otimes I\right) U_A\left(\ket{0}^b \otimes I \right) \|_2 \leq \epsilon$.
\end{definition}

We can also block encode rectangular matrices by padding the left-most block with the necessary amount of zeros. We have the requirement that $\alpha \le \|A\| $ to ensure the embedding is possible. But $\alpha$ also ends up playing a crucial role in the cost of the block encoding. This can be seen by noting that
\begin{equation}
    U_A \ket{0}^b \ket{\psi} = \frac{1}{\alpha} \ket{0}^b A \ket{\psi} + \ket{\text{garbage}}.
\end{equation}
Thus the probability of getting the application of $A$ we want, i.e., measuring the ancillae in  $\ket{0}^b$,  is $\|A \ket{\psi}\|_2\alpha^{-2}$. Thus, we prefer a sub-normalization $\alpha$ close to 1.

Sometimes in block encodings there will be ancilla qubits that are introduce within the course of the block encoding but are not measured in order to get the matrix in the right sub-block. The following definition is a very slight extension to the previous definition.

\begin{definition}\label{def:block_encoding_definition_2}
Let $a_c,a_p,n\in\mathbb{N}$ and $m=a_c+a_p+n$. We say that an $m$-qubit unitary operator $U$ is an $(\alpha ,(a_c,a_p),\epsilon )$ \emph{block encoding} of an n-qubit operator $A$ (not necessarily unitary) if $\lVert A-\tilde{A}\rVert _2<\epsilon$ where
\begin{equation}
    \tilde{A}=\frac{\left(\langle 0|^{\otimes a}\otimes I_n\right) U\left( |0\rangle ^{\otimes a}\otimes I_n\right)}{\left\lVert\left(\langle 0|^{\otimes a}\otimes I_n\right) U\left( |0\rangle ^{\otimes a}\otimes I_n\right)\right\rVert}
\end{equation}
and $\alpha =\left\lVert\left(\langle 0|^{\otimes a}\otimes I_n\right) U\left( |0\rangle ^{\otimes a}\otimes I_n\right)\right\rVert$. Here, we take $\lVert A\rVert =\max |A_{ij}|$.
\end{definition}

To avoid arbitrarily selecting an error scaling for each block encoding, we will ``normalize" each matrix $A$ such that $\lVert A\rVert =1$. The $a_c$ ancillary qubits are referred to as \emph{clean ancillae} in that their states are reset to $|0\rangle$ after the block-encoding circuit is complete; the other $a_p$ ancillary qubits are called \emph{persistent ancillae} and their states must be left undisturbed for the remainder of the computation. In the block-encoding circuits illustrated below, only the data qubits and persistent ancillae are shown. A multi-controlled gate with controls on $n$ qubits will have $n-1$ clean ancillae using a standard decomposition.

\subsection{Block encoding strategy overview}
\label{subsec:carleman_block_encoding}

In this section we overview the block encoding strategy that will be used to block encode the $S$- and $F$-matrices ($F_1, F_2$, and $F_3$), as well as how these are used to form the Carleman matrix block encoding.

Throughout this section, we will ignore the costing of some workspace ancilla qubits since our primary focus is counting $T$-gates. Qubit re-use should be possible in all computations and thus the number of extra qubits needed is expected to only add a small overhead.
Instead of an approximation, we view our results as lower bounds on the qubit counts.

We first note that $F$-matrices are only functions of the of the local populations. Whatever structure we have at a node is the same for every node.  We use the $f$-variable indexing scheme in appendix \ref{sec:notation}.  This gives us the structure of block matrices with nonzero blocks acting on $f$-variables that are all related to the same grid node $\mathbf{x}$.  The same nonzero block structure repeats for all grid nodes within the $F$-matrices.  See Figures \ref{fig:f_1_structure}, \ref{fig:f_2_structure}, and \ref{fig:f_3_structure}.  We refer to each nonzero block in the $F$-matrices as $F_1^{\mathbf{x}} \in \mathbb{R}^{Q \times Q}, F_2^{\mathbf{x}} \in \mathbb{R}^{Q \times Q^2}, \text{ and } F_3^{\mathbf{x}} \in \mathbb{R}^{Q \times Q^3}$, where $Q=27$ and the subscript indicates which full-size $F$-matrix they are related to.  We will use a generic \textit{sparse structured} input model to capture the repetition of the $F^{\mathbf{x}}$ blocks. For the $F^{\mathbf{x}}$ blocks themselves we compare a generic \textit{dense unstructured} input model with a bespoke \textit{dense structured} input model, with the latter reducing quantum resource estimates by several orders of magnitude.

$F_{1}^{\mathbf{x}}$ is a square matrix, while $F_{2}^{\mathbf{x}}$ and $F_{3}^{\mathbf{x}}$ are rectangular. $F_{2}^{\mathbf{x}}$ is padded with $Q^2-Q$ zero rows to create the square matrix $\bar{F}_{2}^{\mathbf{x}}$ and $F_{3}^{\mathbf{x}}$ is padded with $Q^3-Q$ zero rows to create the square matrix $\bar{F}_{3}^{\mathbf{x}}$.

To block encode the streaming $S$-matrix, we will construct and cost out reversible operations for the shifting and indexing operations detailed in appendix \ref{sec:streaming}, conditioned on an oracle function to flag fluid vs. solid grid nodes.

We then describe how the $S$- and $F$-matrices can be used to form the Carleman matrix block encoding.
We develop a method for embedding the block structure of this matrix into qubit registers and a method for using just one call to each block encoding to construct the Carleman block encoding, thus improving the scaling of the number of $T$-gates with respect to truncation order.
We ultimately arrive at a theorem that gives the parameters of the Carleman matrix block encoding as a function of the parameters of the $S$- and $F$-matrix block encodings.
This informs the resource estimates presented in section \ref{sec:quantum-resource-estimates}.

\subsection{Generic input models}
\label{appendix d: Generic input models}
We consider two generic input models for block encoding the matrices: \textit{sparse structured} and \textit{dense unstructured}. For a \textit{sparse structured} matrix, there exist arithmetic expressions to calculate column/row indices given a list of nonzero matrix elements. 

\paragraph*{Sparse structured matrix}
Following \cite{sunderhauf2024}, there are three equivalent ways of representing the nonzero matrix elements:
\begin{enumerate}
    \item a data item  and multiplicity index $(d,m)$ for $d \in \mathbb{Z}_D, m \in \mathbb{Z}_M$, or
    \item a column index  and the sparsity of that column $(j,s_c) $ for $j \in \mathbb{Z}_N, s_c \in \mathbb{Z}_{S_c}$, or
    \item a row index and the sparsity of that row $(i, s_r)$ for  $i \in \mathbb{Z}_N, s_c \in \mathbb{Z}_{S_r}$.
\end{enumerate}
We must obey
\begin{equation}
    \label{eq:block_encoding_restriction}
    MD = N\check{S},
\end{equation}
where $\check{S} = \text{max} \{ S_c, S_r \}$, where $S_c$ and $S_r$ are the column and row sparsities, and we are allowed to add padding to $M$ or $\check{S}$ to make the equation hold.

We reproduce the circuit from \cite{sunderhauf2024} in Figure \ref{fig:block_encoding_structured_matrix} for block encoding some structured matrix $B$.
\begin{figure}[h]
    \centering
      \yquantset{operator/separation=3mm}
\begin{tikzpicture}
  \begin{yquant}   
    qubit {$\ket{0}_{\text{data}}$} data;
    qubit {$\ket{0}_{\text{del}}$} del;
    qubit {$\ket{j}$} j; 
    qubit {$\ket{s_c}$} s;
    slash j;
    slash s;
    hspace {0.9cm} s;
    box {$H_{\check{S}}$} s;
    box {$O_c^\dagger$} (s, j);
    box {$O_{\text{range}}$} (s, j, del);
    box {$R_y\left(\frac{A_d}{\max_d \{A_d\}}\right)$} data | j;
    box {$O_r$} (s, j);
    box {$H_{\check{S}}^\dagger$} s;
  \end{yquant}
\end{tikzpicture}
\caption{Structured input block encoding.}
\label{fig:block_encoding_structured_matrix}
\end{figure}
Oracle $O_c$ and oracle $O_r$ are
\begin{align}
    O_c \ket{d} \ket{m} &= \ket{j}\ket{s_c}, \\
    O_r \ket{d} \ket{m} &= \ket{i}\ket{s_r},
\end{align}
where $i, j$ are the row and column indices respectively and $s_c, s_r$ are the sparsities for the columns and rows. We have a data oracle $O_{\text{data}}$ that has the following action
\begin{equation}
    O_{\text{data}}\ket{0} \ket{d}
    =
    \left( 
        \left( 
            \frac{A_d}{\max_d \{ A_d \}}
        \right)
        \ket{0}
        +
        \left(
            \sqrt{1- \left(\frac{A_d}{\max_d \{ A_d \}}\right)^2} 
        \right)
        \ket{1} 
    \right)
    \otimes
    \ket{d}, 
\end{equation}
where $d$ is the label of the nonzero matrix element $A_d$,
\begin{equation}
    O_{\text{range}} \ket{d} \ket{m} \ket{0}_{\text{del}} = \begin{cases}
        \ket{d} \ket{m} \ket{0}_{\text{del}} & \text{if } A_{i(d, m), j(d,m)} = A_d,\\
        \ket{d} \ket{m} \ket{1}_{\text{del}} & \text{if } A_{i(d, m), j(d,m)} = 0,
    \end{cases}
\end{equation}
 $i(d, m), j(d,m)$ are arithmetic expression for column and row indices, and
\begin{equation}
    H_{\check{S}}\ket{0} = \frac{1}{\sqrt{\check{S}}} \sum_{s=0}^{\check{S}-1} \ket{s}.
\end{equation}
This gives us a $(\sqrt{S_c S_r}\left( \max_d \{ A_d \}\right), 2 + \log(\check{S}),0 )$  block encoding of $A$ in general.

We also use an alternative model from \cite{Camps2023} where multiplicities are unknown.  This is depicted in Figure~\ref{fig:Oprime_data_block_encoding}.
\begin{figure}[ht ]
      \yquantset{operator/separation=4mm}
      \centering
\begin{tikzpicture}
\begin{yquant}
    qubit {$\ket{0}$} q0;
    qubit {$\ket{0}^{\otimes m}$} nonzero; 
    qubit {$\ket{j}$} qj;

    box {$D_{s_c}$} nonzero;
    box {$O'_{\text{data}}$} (q0, nonzero, qj);
    box {$O'_c$} (nonzero, qj);
    box {$D_{s_c}^{\dagger}$} (nonzero);
    hspace {1.54cm} q0; 
    measure nonzero;
    hspace {0.6cm} q0; 
    measure  q0;

\end{yquant}
\end{tikzpicture}
\caption{The second version of the structured input model.  This allows us to cost out the block encoding of a matrix when the multiplicities of the nonzero elements are unknown.}
\label{fig:Oprime_data_block_encoding}
\end{figure}
In this circuit, the top register encodes the matrix element, the middle register encodes label $\ell$ for the list of $s_c$ nonzero elements in the $j^{\text{th}}$ column and the bottom register encodes the column index.  The operators are defined as
\begin{align}
    O'_c\ket{\ell,j} = \ket{\ell, c(\ell,j)}, \\
    D_{s_c} \ket{0}^{\otimes m} = \frac{1}{\sqrt{s_c}} \sum_{\ell=0}^{s_c-1} \ket{\ell}, 
\end{align}
where $s_c$ is the number of nonzero matrix elements in the $j^{\text{th}}$ column and $c(\ell,j)$ is an arithmetic expression that calculates the row index. This gives us an $(s_c, m+1, 0)$ block encoding.

For both \textit{structured} block encoding schemes, we use $ O_{\text{data}}$ and $O'_{\text{data}}$ to represent any unitary matrix that actually contains the nonzero matrix elements. The actual implementation will depend on the structure of the surrounding circuit elements.

\paragraph*{Dense unstructured matrix}
For a \textit{dense unstructured} square matrix $A \in \mathbb{R}^{N \times N}$ where $N = 2^q$, the circuit will have quantum registers for the column and row labels and will consist mostly of a data oracle $O''_{\text{data}}$ that stores the data items appropriately. We use the circuit shown in figure \ref{fig:unstructured_input_circuit} from \cite{Camps2022}.

\begin{figure}[h!]
\centering
\yquantset{operator/separation=6.1mm}
\begin{tikzpicture}
  \begin{yquant}
    qubit {$\ket{0}$} ancilla;
    qubit {$\ket{0}^{\otimes q}$} register;
    qubit {$\ket{\psi}$} psi;
    nobit out;

    box {$H_{S}$} register;
    box {$O''_{\text{data}}$} (ancilla, register, psi);
    swap (psi , register);
    box {$H_S^{\dagger}$} register;

    hspace {2.42cm} ancilla;

    measure ancilla; output {$0$} ancilla;
    measure register; output {$0$} register;
    output {$\frac{A \ket{\psi}}{\|A\ket{\psi}\|_2}$} psi;
  
  \end{yquant}
\end{tikzpicture}
\caption{Unstructured input.}
\label{fig:unstructured_input_circuit}
\end{figure}
where
\begin{equation}
O''_{\text{data}} \ket{0} \ket{i} \ket{j}=  \left(A_{ij}\ket{0} + \sqrt{1-A_{ij}}\ket{1} \right)\ket{i}\ket{j}.
\end{equation}
Since we have real elements, we can use \textit{uniformly controlled} $R_y$ rotations per figure \ref{fig:uniformly_controlled_Ry_rotations}. 

\begin{figure}[h!]
 \centering
   \yquantset{operator/separation=4mm}
   \begin{tikzpicture}
       \begin{yquantgroup}
           \registers{
            qubit {$\ket{0}$} data;
            qubit {$\ket{i}$} row;
            qubit {$\ket{j}$} column;
           }
           \circuit{
            slash row;
            slash column;
            box {$O''_{\text{data}}$} (data, row, column);
           } \equals

           \circuit{
            slash row;
            slash column;  
            box {$R_y(\vec{\theta})$} data | row, column ;
}
    \end{yquantgroup}       
   \end{tikzpicture}
   \caption{Data matrix that is uniformly controlled on the column and row registers.}
   \label{fig:uniformly_controlled_Ry_rotations}
\end{figure}

The above naive implementation of $O''_{\text{data}}$ for generic square matrix $A$ requires $O(N^4)$ operations where we have $O(N^2)$  $\mathrm{CTRL}^{2q}$-$(R_y)$ gates (each using $2q$ control qubits).  Each requires $O(N^2)$ CNOT operations and $R_y$ operations. This is worse than what we can achieve classically. In order to solve this problem, \cite{Camps2022} uses a circuit from \cite{Mottonen2004} that reduces the number of CNOT operations to $O(N^2)$  and $R_y$ operations where the control qubit for the $\ell^{\text{th}}$ CNOT operation is determined by where the $\ell^{\text{th}}$ and $(\ell+1)^{\text{th}}$ gray code differ. The trade off is that the new angles used in the rotations come from solving

 \begin{equation}
   \left( H^{\otimes 2q} P_{G} \right) \vec{\theta}' = \vec{\theta}, 
\end{equation}
where $P_{G}$ changes from gray code to base-2 integer representation and can be exponentially smaller than before, so the Clifford gate synthesis for the single qubit $R_y(\theta_{new})$ has an exponential increase in accuracy required which ultimately translates to more $T$-gates.

For the unstructured model we get a $(\frac{1}{2^q}, q+1, 0)$ block encoding for $A\in \mathbb{R}^{2^q \times 2^q}$. This assumes that $|A_{ij}| \leq 1$,  $ \forall ~ i, j$.

\subsection{Block encoding the $F$-matrices via generic input models}
\label{generic}

In this section, we shall use the generic models of block encoding to cost out the classical upload for the algorithm. In the later sections we shall present bespoke block encodings for the collision matrices.
 We will produce the sub-normalization constants for the block encodings and the number of $T$-gates required. 

We can construct the full-size $F$-matrices with the following tensor products.
\begin{enumerate}
    \item The $F_1$ matrix is a square matrix can be written in the following form $\mathcal{D}^{(1)}\otimes F_1^{\mathbf{x}}$ where $\mathcal{D}^{(1)}= I_{n \times n}$ and $n$ is the number of lattice grid points.
    \item The ${F}_2$ matrix has the structure $\mathcal{D}^{(2)} \otimes {F}_2^{\mathbf{x}}$, where $\mathcal{D}^{(2)}_{ij} = \delta_{(n+1)i,j}$ and $\mathcal{D}^{(2)}$ is an $n \times n^2$ matrix.
    \item The ${F}_3$ matrix has the structure $\mathcal{D}^{(3)} \otimes {F}_3^{\mathbf{x}}$, where $\mathcal{D}^{(3)}_{ij} = \delta_{(n^2+n+1)i,j}$ and $\mathcal{D}^{(3)}$ is an $n \times n^3$ matrix.
\end{enumerate}
In this approach, we have a tensor product of a \textit{sparse structured} matrix $\mathcal{D}^{(\cdot)}$ and the \textit{dense unstructured} $F^{\mathbf{x}}$-matrices.  To cost out the block encoding of the $F^{\mathbf{x}}$-matrices, we focus on the data oracle $O''_{\text{data}}$ whose costs will be dominated by the $T$-gate counts that come from implementing the $R_y$ operations. Using the Repeat Until Success (RUS) circuits from \cite{Bocharov2015} and synthesizing the $R_y$ operations to to accuracy $\epsilon$, we require $1.15 \log_2(\frac{1}{\epsilon})$ $T$-gates. 

\paragraph*{Cost of block encoding $F_1$}
We restrict ourselves to considering what the matrix looks like at a specific node $\mathbf{x}$, i.e., the block $F^{\mathbf{x}}_1$. The formula for the elements of $F_1$ is \eqref{eq:tildeF_1} and we can use the same formula to generate the elements of the $F^{\mathbf{x}}_1$ block:
\begin{subequations}\label{eq:f1_block}
\begin{align}
    \left[ F_1^{\mathbf{x}} \right]_{ii}
    &=
    \left( -1 + w_i + 3 w_i (\mathbf{c}_i \cdot \mathbf{c}_i) \right) \frac{1}{\tau}
    \\
    \left[ F_1^{\mathbf{x}} \right]_{ij} 
    &=
    \left( 1 + 3 (\mathbf{c}_i \cdot \mathbf{c}_j) \right) \frac{w_i}{\tau}. 
\end{align}    
\end{subequations}

Straightforward calculations using table \ref{tab:D3Q27-constants} show that we have 729 nonzero elements and only 15 unique nonzero elements (for the D3Q27 lattice).

The block encoding $U_{F_1}$ for $F_1$ is $\underbrace{I_{2 \times 2} \otimes I_{2 \times 2} \otimes \dots \otimes I_{2 \times 2}}_{\lceil \log_2 n \rceil } \otimes   
U_{F_1^{\mathbf{x}}}$. 
The unstructured block encoding input model gives a $(\frac{1}{1.58950617Q}, \lceil\log_2 Q  \rceil +1, 0)$  block encoding of $F_1^{\mathbf{x}}$ and the cost of $F_1$ is simply the cost of one of the blocks. The prefactor ensures all the matrix elements are less than or equal to 1.  This gives us a $(\frac{1}{1.58950617Q}, 0, 0)$ block encoding of $F_1$. 

If we can tolerate some error $\epsilon_{F_1}$ of $U_{F_1}$, then we need a total of $(1.15)(729) \log_2(\frac{729}{\epsilon_{F_1}})$ $T$-gates. 
 
\paragraph*{Cost of block encoding padded matrices $\bar{F}_2$ and $\bar{F}_3$}

For matrices generated with the functions $\tilde{F}_{2\text{-dense}}$ from \eqref{eq:tildeF_2_dense} and $\tilde{F}_{3\text{-dense}}$ from \eqref{eq:tildeF_3_dense}, we have the following $F^{\mathbf{s}}$ blocks 
\begin{align}
    \label{eq:f2_block}
    \left[ {F}_2^{\mathbf{x}} \right]_{a_i, b_{j,k}} 
    &= \Big( 3 (\mathbf{c}_i \cdot \mathbf{c}_j)(\mathbf{c}_i \cdot \mathbf{c}_k) - (\mathbf{c}_j \cdot \mathbf{c}_k ) \Big)\left(\frac{3 w_i}{\tau}\right), 
    \\
    \label{eq:f3_block}
     \left[ F_3^{\mathbf{x}} \right]_{a_i, b_{j,k,\ell}} 
     &=  \Big( 3(\mathbf{c}_i \cdot \mathbf{c}_k)(\mathbf{c}_i \cdot \mathbf{c}_\ell) - (\mathbf{c}_k \cdot \mathbf{c}_\ell) \Big) \left(\frac{3 w_i}{\tau}\right)\left(\frac{-1}{2}\right).
\end{align}
Note that $F^{\mathbf{x}}_3$ is a scaled tiling of $F^{\mathbf{x}}_2$.  The $j$ index does not appear in \eqref{eq:f3_block}, yet the structure is the same as \eqref{eq:f2_block}, thus, the information in $F^{\mathbf{x}}_2$ is repeated for each $j=0,1,...,Q-1$ in the construction of $F^{\mathbf{x}}_3$ by \eqref{eq:f3_block}.  

In total we have $42$ nonzero unique elements for the $F_2^{\mathbf{x}}$ block with a total of $15,180$ nonzero elements.  The $F_3^{\mathbf{x}}$ block has $42$ unique nonzero elements with $409,860$ nonzero elements \cite{source_for_norm_A_anlysis}.

As before, the block encoding of  $\bar{F}_2^{\mathbf{x}} \in \mathbb{R}^{Q^2 \times Q^2}$ and $\bar{F}_3^{\mathbf{x}} \in \mathbb{R}^{Q^3 \times Q^3}$ is accomplished with the unstructured input model. We get $(\frac{1}{4.\overline{4}Q}, 
 \lceil \log_2 Q \rceil +1, 0)$ block encoding for $\bar{F}_2^{\mathbf{x}}$ and the $(\frac{1}{2.\overline{2}Q}, 
 \lceil \log_2 Q \rceil +1, 0)$  for  $\bar{F}_3^{\mathbf{x}}$.  Although we are dealing with $Q^2 \times Q^2$ and $Q^3 \times Q^3$, the extra rows are just padded zeros, so we only need to prepare an equal superposition of labels for $Q$ rows. 
 
Next we consider padding $\mathcal{D}^{(2)}$ and $\mathcal{D}^{(3)}$ to $\bar{\mathcal{D}}^{(2)}$ and $\bar{\mathcal{D}}^{(3)}$ by adding rows of zeros. $\bar{\mathcal{D}}^{(2)}$ and $\bar{\mathcal{D}}^{(3)}$ will be $n^{2}\times n^{2}$ and $n^{3}\times n^{3}$ matrices, respectively.

Using the structured input model we have the 
that in the equation $MD = NS$ both $D$ and $\check{S}$ are equal to 1 since $\check{S} = \max\{ S_c, S_r \} = 1$  and the length of the data item list is only 1.   $N = n^2$ (or $n^3$) and $M$ has been padded from $n$ to get to $n^2$ (or $n^3$). With padded multiplicity register in mind, we have the following arithmetic expression for row and column indices (base-0 indexing)
\begin{align}
    i(d, m) &= m, \\
    j(d, m) &= (n+1)m ,
\end{align}
with $m \in \mathbb{Z}_{n^2}$ for $\bar{\mathcal{D}}^{(2)}$ and 

\begin{align}
    i(d, m) &= m, \\
    j(d, m) &= (n^2 + n +1)m, 
\end{align}
for $\bar{\mathcal{D}}^{(3)}$
with $m \in \mathbb{Z}_{n^3}$.
 
The out-of-range oracle $O_{\text{range}}$ will flag a qubit when $m > n$.  This is accomplished with an oracle $O_{\text{comp\_const}(\nu)}$ that computes a comparison with a constant $\nu$
  \begin{equation}
        O_{\text{comp\_const}(\nu)}\ket{m}\ket{0}  = \begin{cases}
      \ket{m}\ket{1} & \text{if } m > \nu \\
      \ket{m} \ket{0} &  \text{else}.
  \end{cases}
   \end{equation}
We have from \cite{Kan2023}, the resources for $O_{\text{comp\_const}(\nu=n)}$ is $(8 \lceil \log_2 n^2 \rceil - 16)$ $T$-gates and an additional $(3\lceil \log_2 n^2 \rceil  -2)$ ancilla qubits for $\bar{\mathcal{D}}^{(2)}$.  For $\bar{\mathcal{D}}^{(3)}$, the cost for $O_{\text{comp\_const}(\nu=n)}$ is $(8 \lceil \log_2 n^3 \rceil - 16)$ $T$-gates and an additional $(3\lceil \log_2 n^3 \rceil  -2)$ ancilla qubits. 

For the $O_{j(d,m)}$ oracle, we need an oracle that multiplies an integer in a register by a constant. Using \cite{Kan2023},  we need $\left( 2 \lceil \log_2n^2  \rceil \left(  \lceil \log_2 n^2 \rceil -1)\right) \right)$ $T$-gates for $\bar{\mathcal{D}}^{(2)}$ and $\left( 2 \lceil \log_2n^3  \rceil \left(  \lceil \log_2 n^3 \rceil -1)\right) \right)$ $T$-gates for $\bar{\mathcal{D}}^{(3)}$.
\begin{figure}[ht!]
  \centering
\yquantset{operator/separation=5mm}
\begin{tikzpicture}
  \begin{yquant}

    qubit {$\ket{j}$} j;
    qubit {$\ket{s_c}$} s;
    qubit {$\ket{0}_{\text{del}}$} del;
    qubit {$\ket{0}_{\text{range ancilla}}$} range;
    qubit {$\ket{0}_{\text{data}}$} data;

    box {$O_{j(d,m)}^\dagger$} (s, j);
    box {$O_{\text{comp\_const}(n)}$} (s, del, range);

    hspace {0.7cm} del;
    hspace {5.6cm} data;
    measure data; output {$0$} data;
    hspace {0.4cm} range; 
    measure del; output {$0$} del;
    hspace {0.4cm} data; 

    hspace {0.4cm} s; 
    hspace {0.4cm} j; 
    hspace {0.4cm} del; 
  \end{yquant}
\end{tikzpicture}
    \caption{Block encoding for $\bar{\mathcal{D}}^{(2)}$ or $\bar{\mathcal{D}}^{(3)}$.}
 \label{fig:d_nu_block_encoding}
\end{figure}
The circuit shown in figure \ref{fig:d_nu_block_encoding} gives us a $(1, 2, 0)$ block encoding of $\bar{\mathcal{D}}^{(2)}$ and $\bar{\mathcal{D}}^{(3)}$.

Lemma 1 from \cite{Camps2020} shows that when block encoding tensor products, the subnormalizations should be multiplied and the number of ancillae should be added.  Using this and our construction $\bar{F}_2^{'} = \bar{\mathcal{D}}^{(2)} \otimes \bar{F}^{\mathbf{x}}_2$ and $\bar{F}_3^{'} = \bar{\mathcal{D}}^{(3)} \otimes \bar{F}^{\mathbf{x}}_3$, we get 
a $(\frac{1}{4.\overline{4}Q}, \lceil \log_2 Q \rceil + 3  , 0)$ block encoding for both $\bar{F}_2$ and $\bar{F}_3$. 
The prime in $\bar{F}_2^{'}$ and $\bar{F}_3^{'}$ signifies that, at this point, the matrices have rows of zero between each block $\bar{F}^{\mathbf{x}}_2$ in $\bar{F}_2^{'}$ and $\bar{F}^{\mathbf{x}}_3$ in $\bar{F}_3^{'}$. To get $\bar{F}_2$ and $\bar{F}_3$ we can conjugate the prime version by some permutation matrix containing $\text{SWAP}$ and $NOT$ gates that acts in the part of the Hilbert space that gets tensored to the actual space we want $F_2$ and $F_3$ to act on. 
The total $T$-gate counts are \newline
$\left[ 8 \lceil \log_2 n^{2} \rceil + 17457\log_2\frac{15180}{\epsilon_{F_2}} -16 + 2\lceil \log_2n^2\rceil \left( \lceil \log_2 n^2\rceil -1  \right)\right]$
and \newline
$\left[ 8 \lceil \log_2 n^{3} \rceil + (1.15)(409860)\log_2\frac{409860} {\epsilon_{\bar{F}_3}} -16 + 2 \lceil \log_2n^3  \rceil \left(  \lceil \log_2 n^3 \rceil -1)\right)
\right]$,
for $\bar{F}_2$ and $\bar{F}_3$, respectively. The results are summarized in table \ref{table:summary_of_F_block_encoding_results}.

\begin{table}[ht]
\centering
\begin{tabular}{|c|c|c|}
\hline
Matrix & Block Encoding  &  $T$-gate count \\
\hline
$F_1$ & $\left( \frac{1}{1.58950617Q}, \lceil \log_2 Q \rceil +1, 0 \right)$ & $838.35 \log_2(\frac{729}{\epsilon_{F_1}})$   
\\
\hline
$\bar{F}_2$ & $\left( \frac{1}{4.\overline{4}Q}, \lceil \log_2 Q \rceil + 3 , 0 \right)$  & $8 \lceil \log_2 n^{2} \rceil + 17457\log_2\frac{15180}{\epsilon_{F_2}} -16 + 2 \lceil \log_2n^2\rceil \left( \lceil \log_2 n^2\rceil -1  \right) $  
\\
\hline
$\bar{F}_3$ &  $\left( \frac{1}{2.\overline{2}Q}, \lceil \log_2 Q \rceil + 3 , 0 \right)$   & $8 \lceil \log_2 n^{3} \rceil$ + $471339\log_2\frac{409860} {\epsilon_{F_3}} -16 + 2 \lceil \log_2n^3  \rceil \left(  \lceil \log_2 n^3 \rceil -1)\right)$ 
\\
\hline
\end{tabular}
\caption{Summary of results for block encoding $F_1, \bar{F}_2, \bar{F}_3$ matrices}
\label{table:summary_of_F_block_encoding_results}
\end{table}

\subsection{Bespoke block encodings of the $F$-matrices}
\label{appendix_d: bespoke_block_encoding}

To reduce the quantum resource estimates associated with block encoding the $F$-matrices, we developed a bespoke methodology based on decomposing the expressions in \eqref{eq:f1_block}, \eqref{eq:f2_block}, and \eqref{eq:f3_block} into matrix operations which can be efficiently simulated by well-known block-encoding circuits \cite{pyLIQTR}.

We first document a collection of tools we will use to block encode the $F$ matrices below. We estimate the resources required by these block encodings in terms of $T$-gates; the three primative sources of $T$-gates in these circuits are controlled Hadamard gates costing 2 $T$-gates each, Toffoli gates costing 4 $T$-gates each, and rotation gates which up to 2-norm error $\epsilon$ cost $1.15\log _2(1/\epsilon )+9.2$ $T$-gates each \cite{bocharov15}. From these pieces, we can cost out more complicated structures such as an $n$-qubit controlled $X$ gate costing $4(n-1)$ $T$-gates using a Toffoli construction from Gidney \cite{Gidney2018} and a controlled rotation costing $2.3\log _2(1/\epsilon )+20.7$ $T$-gates (constructed with two uncontrolled rotations each with accuracy $\epsilon /2$). Sequences of multi-qubit controlled gates can lead to significant savings if the Toffoli expansions are selected carefully.

Critical block-encoding circuits include those which combine blocks via matrix operations. The linear combination of unitaries (LCU) \cite{Gui_Lu_2006} allows one to add block encodings together in a weighted sum. If $A$ and $B$ are two block encodings with ancilla registers, we can form a \emph{product} block encoding by placing the data components in sequence and the ancilla components on separate registers. Meanwhile we can form a \emph{tensor product} block encoding by stacking the data registers as well as the ancilla registers as in figure \ref{fig:prod-tensor}. See Camps \cite{PhysRevA.102.052411} for more details.

\begin{figure}[ht]
\centering
\begin{tikzpicture}
    \begin{yquantgroup}
    \registers{
    qubit {} q[6];
    }
    \circuit{
      discard q[4-5];
      
      [name=A] box {\Ifnum\idx<1 $A_a$\Else $A_d$\Fi} (q[1]), (q[2],q[3]);
      \draw (A-0) -- (A-1);      
      hspace {2mm} -;

      [name=B] box {\Ifnum\idx<1 $B_a$\Else $B_d$\Fi} (q[0]), (q[2],q[3]);
      \draw (B-0) -- (B-1);  
    }
    \circuit{
    discard q[0-5];
    }
    \circuit{
      [name=A] box {\Ifnum\idx<1 $A_a$\Else $A_d$\Fi} (q[1]), (q[4],q[5]);
      \draw (A-0) -- (A-1);      
      hspace {2mm} -;

      [name=B] box {\Ifnum\idx<1 $B_a$\Else $B_d$\Fi} (q[0]), (q[2],q[3]);
      \draw (B-0) -- (B-1);        
    }
    \end{yquantgroup}
\end{tikzpicture}
\caption{(\emph{Left}) Product block encoding $BA$ (\emph{Right}) Tensor product block encoding $B\otimes A$. Here, $A_a$ is the ancilla register of $A$ and $A_d$ is the data register of $A$, likewise for $B$.}
\label{fig:prod-tensor}
\end{figure}

One particular matrix we will frequently use in the following block encodings is the projection operator. For example, if we wanted to project off of an individual basis state of the data register, we could use a multi-control $X$ gate (control sequence matching the binary representation of the basis state) and target on an ancilla register. Importantly, this operation allows us to `zero-out' entries of the data block by essentially moving them to an ancilla register. Placing a projection operator after a block encoding eliminates rows of the block while placing the projection before the block encoding eliminates columns. Figure \ref{fig:row-col} illustrates this concept.

\begin{figure}
\centering
\begin{tikzpicture}
    \begin{yquantgroup}
    \registers{
    qubit {} q[4];
    }
    \circuit{
      [name=A] box {\Ifnum\idx<1 $A_a$\Else $A_d$\Fi} (q[1]), (q[2],q[3]);
      \draw (A-0) -- (A-1);      
      hspace {2mm} -;

      cnot q[0] | q[2],q[3];
    }
    \circuit{
    discard q[0-3];
    }
    \circuit{
      cnot q[0] | q[2],q[3];
      hspace {2mm} -;
      
      [name=A] box {\Ifnum\idx<1 $A_a$\Else $A_d$\Fi} (q[1]), (q[2],q[3]);
      \draw (A-0) -- (A-1);         
    }
    \end{yquantgroup}
\end{tikzpicture}
\caption{Projection operator off of the $|11\rangle$ state applied to a block encoding of $A$ (\emph{Left}) Eliminates the last row of $A$ (\emph{Right}) Eliminates the last column of $A$.}
\label{fig:row-col}
\end{figure}

\paragraph{Converting \eqref{eq:f1_block}, \eqref{eq:f2_block}, and \eqref{eq:f3_block} to matrix arithmetic}

Let $c=[c_{ix},c_{iy},c_{iz}]\in\mathbb{R}^{Q\times 3}$ be a matrix containing the set of $Q=27$ 3D lattice vectors (see Table \ref{tab:D3Q27-constants}). Noticing that $c$ can be written as a concatenation of tensor products, we let $x=[1,-1,0]^\top$ and let ${\bf 1}_{3\times 1}=[1,1,1]^\top$, so that we have
\begin{equation}
  c_{ix}={\bf 1}\otimes{\bf 1}\otimes x,\hspace{1cm}c_{iy}={\bf 1}\otimes x\otimes{\bf 1},\hspace{1cm}c_{iz}=x\otimes{\bf 1}\otimes{\bf 1}.  
\end{equation}
Next, we define the matrix of all possible dot products $C=cc^\top\in\mathbb{R}^{Q\times Q}$, which describes the dot product in (\ref{eq:f1_block}). Then the matrices containing each of the first two dot products in \eqref{eq:f2_block} and \eqref{eq:f3_block} can be expressed as
\begin{equation}
    \left( c_i\cdot c_k\right) _{i,k,l}=C\otimes{\bf 1}_{1\times Q},\hspace{1cm}\left( c_i\cdot c_l\right) _{i,k,l}={\bf 1}_{1\times Q}\otimes C .
\end{equation}
The matrix describing the product of these two expressions is given by the Hadamard product, or entry-wise product of two matrices
\begin{equation}
    \left( (c_i\cdot c_k)(c_i\cdot c_l)\right) _{i,k,l}=\left( C\otimes{\bf 1}_{1\times Q}\right)\odot\left( {\bf 1}_{1\times Q}\otimes C\right) \label{eq:dotprodprod}
\end{equation}
Rather than block encoding the Hadamard product outright, we can simplify the system by noting that $(A\otimes B)\odot (C\otimes D) =(A\odot C)\otimes (B\odot D)$ when $\text{dim}(A)=\text{dim}(C)$ and $\text{dim}(B)=\text{dim}(D)$. Thus, we instead consider the matrix
\begin{subequations}
\begin{align}
\mathcal{C} & = \left( C\otimes{\bf 1}_{Q\times Q}\right)\odot\left( {\bf 1}_{Q\times Q}\otimes C\right) \\
 & = C\otimes C,
\end{align}
\end{subequations}
where the $(Q+1)i^\text{th}$ row is the $i^\text{th}$ row of \eqref{eq:dotprodprod}. Note that for the Hadamard product, the identity matrix ${\bf 1}_{Q\times Q}$ is a $Q\times Q$ matrix of all ones.

The third dot product of \eqref{eq:f2_block} and \eqref{eq:f3_block} does not contain $i$, so there is no dependence on the row index. Because of this, we can express this term as
\begin{equation}
\left( c_k\cdot c_l\right) _{i,k,l}={\bf 1}_{Q\times 1}\otimes\left( c_{ix}^\top\otimes c_{ix}^\top+c_{iy}^\top\otimes c_{iy}^\top+c_{iz}^\top\otimes c_{iz}^\top\right).
\label{eq:thirddotprod} 
\end{equation}

The factor of $w_i$ in \eqref{eq:f2_block} and \eqref{eq:f3_block} can be applied by multiplying the matrix formed by summing \eqref{eq:dotprodprod} and \eqref{eq:thirddotprod} by a matrix $W=\text{diag}(w_i)$. We can further simplify this by noting that 
\begin{equation}
    w_i = \frac{8}{27} \left( \frac{1}{4} \right)^{|c_{ix}|} \left( \frac{1}{4} \right)^{|c_{iy}|} \left( \frac{1}{4} \right)^{|c_{iz}|} ,
\end{equation}
and expressing this dependence on the components of $\mathbf{c}_i$ via the tensor product
\begin{subequations}
    \begin{align}
        W & = w\otimes w\otimes w, \label{eq:W} \\
        \mathrm{where}\quad w & = \text{diag}\left(\left[\frac{1}{4},\frac{1}{4},1\right]\right).
    \end{align}
\end{subequations}

\paragraph{Block encoding $F_1$.}

We use several of the resources developed above to create a block encoding for the smaller $F_1$ matrix. We can rewrite $F_1$ in terms of the following matrix:
$$F_1=P^{\otimes 3}(W({\bf 1}_{Q\times Q}+3C)-I_6)P^{\otimes 3}$$
where $P$ is the projection off of $|11\rangle$.  Let $G_n$ represent an $n$-qubit Grover matrix. Since $G_n+ I_n=\frac{1}{2^{n-1}}{\bf 1}_{2^n\times 2^n}$ and we can write $C=64\tilde{C}$  we can rewrite an expression for $F_1$ in a form better adapted for the LCU as
\begin{equation}
    F_1=P^{\otimes 3}(32W(G_6+I_6+6\tilde{C})-I_6)P^{\otimes 3},
\end{equation}

We use \cite{kuklinski2025efficientblockencodingsrequirestructure} and cite theorems for what it takes to block encode $F_1$ wherein the reader will find the details.
\begin{theorem}
    \textup{(Theorem 5 of \cite{kuklinski2025efficientblockencodingsrequirestructure})} There exists a circuit that is a $(0.0225,(7, 20), \epsilon)$ block-encoding of $F_1$ costing
$11.5 \log_2
(1/\epsilon)$ + 404.5 T-gates.
\end{theorem}

\paragraph{Block encoding $F_2$}

Now we are prepared to convert the equations we've developed into block-encoding circuits. To precisely account for resources (e.g., reusable ancilla), we use definition \ref{def:block_encoding_definition_2} (cf. definition \ref{def:block_encoding}).

One other caveat we must mention in these circuits is that the block encoding of $F_2\in\mathbb{R}^{Q\times Q^2}$ will be embedded in a larger $64\times 64^2$ matrix. Working with powers of 2 will facilitate the circuit construction.

First, we will block encode a matrix which contains both $x$ and ${\bf 1}_{3\times 1}$ in its columns. In particular, we will write
\begin{equation}
A=\begin{pmatrix}
1 & 1 & \cdot & \cdot \\
1 & -1 & \cdot & \cdot \\
1 & 0 & \cdot & \cdot \\
\cdot & \cdot & \cdot & \cdot
\end{pmatrix}=
\frac{1}{2}\begin{pmatrix}
1 & 1 & 1 & 1 \\
1 & -1 & 1 & -1 \\
1 & 1 & -1 & -1 \\
1 & -1 & -1 & 1
\end{pmatrix}
+\frac{1}{2}\begin{pmatrix}
1 & 1 & 1 & 1 \\
1 & -1 & 1 & -1 \\
1 & -1 & -1 & 1 \\
1 & 1 & -1 & -1
\end{pmatrix}
\end{equation}
The first matix in this sum is just $H\otimes H$, and the second swaps columns 2 and 4. We write the block-encoding circuit below.
\begin{lemma}
The circuit in figure \ref{fig:Amat} is a $(\frac{1}{2},(1,1),0)$ block-encoding circuit of $A$ costing 4 $T$-gates. The controlled circuit costs 16 $T$-gates. 
\end{lemma}
\begin{figure}[H]
\centering
\begin{tikzpicture}
    \begin{yquantgroup}
    \registers{
      qubit {} q[3];
      }
    \circuit{
      [name=A1] box {\Ifnum\idx<1 $A_a$\Else $A_d$\Fi} q[0], (q[1],q[2]);      
      \draw (A1-0) -- (A1-1);
      }
    \equals
    \circuit{
      h q[0];
      cnot q[1] | q[0],q[2];
      h q[0];
      h q[1];
      h q[2];
      }
    \end{yquantgroup}
\end{tikzpicture}
\caption{Block-encoding circuit for $A$, where $A_a$ indicates the part acting on a persistant ancilla and $A_d$ indicates the part acting on the two data qubits.}
\label{fig:Amat}
\end{figure}
We include the $T$-gate cost of the controlled circuit since we will be combining figure \ref{fig:Amat} in a linear combination of unitaries (LCU) later.

Next, notice that $A\otimes A\otimes A$ contains $c_{ix}$ in its $2^\text{nd}$ column, $c_{iy}$ in its $5^\text{th}$ column, and $c_{iz}$ in its $17^\text{th}$ column. Thus, we arrive at the following block-encoding circuit for $c$.
\begin{lemma}
The circuit in figure \ref{fig:c} is a $(\frac{1}{8},(5,4),0)$ block encoding of $c$ costing 44 $T$-gates. The controlled circuit costs 84 $T$-gates.
\end{lemma}
\begin{figure}[H]
\centering
\begin{tikzpicture}
   \begin{yquant}
      qubit {} q[10];
      
      cnot q[0] | ;
      cnot q[0] | q[9] ~ q[4], q[5], q[6], q[7], q[8];
      cnot q[0] | q[7] ~ q[4], q[5], q[6], q[9], q[8];
      cnot q[0] | q[5] ~ q[4], q[9], q[6], q[7], q[8];

      [name=A1] box {\Ifnum\idx<1 $A_a$\Else $A_d$\Fi} q[3], (q[8],q[9]);      
      \draw (A1-0) -- (A1-1);
      hspace {2mm} -;
     
      [name=A2] box {\Ifnum\idx<1 $A_a$\Else $A_d$\Fi} (q[2]), (q[6],q[7]);
      \draw (A2-0) -- (A2-1);      
      hspace {2mm} -;

      [name=A3] box {\Ifnum\idx<1 $A_a$\Else $A_d$\Fi} (q[1]), (q[4],q[5]);
      \draw (A3-0) -- (A3-1);  

   \end{yquant}
\end{tikzpicture}
\caption{Block-encoding circuit for $c$. The top four qubits are persistent ancillae.}
\label{fig:c}
\end{figure}
Notice that in figure \ref{fig:c}, we do not expend any resources shifting the aforementioned columns to the first three columns. This is because the matrix $C=cc^\top$ is invariant under column permutations in $c$ in the larger $64\times 64$ matrix.

To block encode $C$, we simply take the product $C=cc^\top$ by reversing the circuit in figure \ref{fig:c} and putting it in series with figure \ref{fig:c} with a separate ancillary register. The gate count is slightly less than doubled because we can advantageously commute the multi-controls to reduce $T$-gate count.
\begin{lemma}
The circuit in figure \ref{fig:bigC} is a $(\frac{3}{64},(5,8),0)$ block encoding of $C$ costing 56 $T$-gates. The controlled circuit costs 132 $T$-gates.
\begin{figure}[H]
\centering
\begin{tikzpicture}
   \begin{yquantgroup}
    \registers{
    qubit {} q[14];
    }
    \circuit{
      
      [name=Ad1] box {\Ifnum\idx<1 $A_a^\dagger$\Else $A_d^\dagger$\Fi} q[3], (q[12],q[13]);      
      \draw (Ad1-0) -- (Ad1-1);
      hspace {2mm} -;
     
      [name=Ad2] box {\Ifnum\idx<1 $A_a^\dagger$\Else $A_d^\dagger$\Fi} (q[2]), (q[10],q[11]);
      \draw (Ad2-0) -- (Ad2-1);      
      hspace {2mm} -;

      [name=Ad3] box {\Ifnum\idx<1 $A_a^\dagger$\Else $A_d^\dagger$\Fi} (q[1]), (q[8],q[9]);
      \draw (Ad3-0) -- (Ad3-1);  
      hspace {2mm} -;

      cnot q[0], q[4] | ;
      cnot q[0], q[4] | q[12] ~ q[8], q[9], q[10], q[11], q[13];
      cnot q[0], q[4] | q[10] ~ q[8], q[9], q[12], q[11], q[13];
      cnot q[0], q[4] | q[8] ~ q[12], q[9], q[10], q[11], q[13];

      [name=A1] box {\Ifnum\idx<1 $A_a$\Else $A_d$\Fi} q[7], (q[12],q[13]);      
      \draw (A1-0) -- (A1-1);
      hspace {2mm} -;
     
      [name=A2] box {\Ifnum\idx<1 $A_a$\Else $A_d$\Fi} (q[6]), (q[10],q[11]);
      \draw (A2-0) -- (A2-1);      
      hspace {2mm} -;

      [name=A3] box {\Ifnum\idx<1 $A_a$\Else $A_d$\Fi} (q[5]), (q[8],q[9]);
      \draw (A3-0) -- (A3-1);
    }
    \equals
    \circuit{

      discard q[1-3];
      discard q[5-7];

      slash q[0];
      slash q[4];

      [name=A3d2] box {\Ifnum\idx<1 $(A_a^\dagger)^{\otimes 3}$\Else $(A_d^\dagger)^{\otimes 3}$\Fi} q[0], (q[8-]);      
      \draw (A3d2-0) -- (A3d2-1);
      hspace {2mm} -;

      [name=CP2] box {$C-PREP$} (q[0-4]), (q[8-]);      
      \draw (CP2-0) -- (CP2-1);
      hspace {2mm} -;

      [name=A32] box {\Ifnum\idx<1 $(A_a)^{\otimes 3}$\Else $(A_d)^{\otimes 3}$\Fi} q[4], (q[8-]);      
      \draw (A32-0) -- (A32-1);
    }
   \end{yquantgroup}
\end{tikzpicture}
\caption{Block-encoding circuit for $C$. Slashes in the wires on the RHS denote 4-qubit ancillary registers, where it is understood that $C-PREP$ acts on the first qubit of each register, whereas $(A_a)^{\otimes 3}$ and $(A_a^\dagger)^{\otimes 3}$ act on the other three qubits in each register.}
\label{fig:bigC}
\end{figure}
\end{lemma}

To block encode $(c_i\cdot c_k)(c_i\cdot c_l)$, we stack the circuits from figure \ref{fig:bigC} on top of each other and use we use a series of 6 CNOT gates to move row $65i$ to row $i$.  We will worry about zeroing out the top 6 qubits later.
\begin{lemma}
The circuit in figure \ref{fig:calC} is a $(\frac{9}{4096},(5,8),0)$ block encoding of $(c_i\cdot c_k)(c_i\cdot c_l)$ costing 112 $T$-gates. The controlled circuit costs 264 $T$-gates.
\end{lemma}
\begin{figure}[H]
\centering
\begin{tikzpicture}
   \begin{yquant}
      qubit {} q[16];
      slash q[0-3];
      
      [name=Ad1] box {\Ifnum\idx<1 $(A_a^\dagger)^{\otimes 3}$\Else $(A_d^\dagger)^{\otimes 3}$\Fi} q[3], (q[10-]);      
      \draw (Ad1-0) -- (Ad1-1);
      hspace {2mm} -;

      [name=CP1] box {$C-PREP$} (q[2,3]), (q[10-]);      
      \draw (CP1-0) -- (CP1-1);
      hspace {2mm} -;

      [name=A1] box {\Ifnum\idx<1 $(A_a)^{\otimes 3}$\Else $(A_d)^{\otimes 3}$\Fi} q[2], (q[10-]);      
      \draw (A1-0) -- (A1-1);
      hspace {2mm} -;
      
      [name=Ad2] box {\Ifnum\idx<1 $(A_a^\dagger)^{\otimes 3}$\Else $(A_d^\dagger)^{\otimes 3}$\Fi} q[1], (q[4-9]);      
      \draw (Ad2-0) -- (Ad2-1);
      hspace {2mm} -;

      [name=CP2] box {$C-PREP$} (q[0,1]), (q[4-9]);      
      \draw (CP2-0) -- (CP2-1);
      hspace {2mm} -;

      [name=A2] box {\Ifnum\idx<1 $(A_a)^{\otimes 3}$\Else $(A_d)^{\otimes 3}$\Fi} q[0], (q[4-9]);      
      \draw (A2-0) -- (A2-1);
      hspace {2mm} -;

      cnot q[9] | q[15];
      cnot q[8] | q[14];
      cnot q[7] | q[13];
      cnot q[6] | q[12];
      cnot q[5] | q[11];
      cnot q[4] | q[10];
   \end{yquant}
\end{tikzpicture}
\caption{Block-encoding circuit for $(c_i\cdot c_k)(c_i\cdot c_l)$ via block encoding of $C$ followed by row permutation. Slashes indicate 4-qubit ancillary registers. The upper six data qubits (CNOT targets) will be zeroed out later.}
\label{fig:calC}
\end{figure}

To block encode $(c_k\cdot c_l)$, consider the matrix $(A^\top)^{\otimes 6}$ (note $A$ is real so $A^\top=A^\dagger$) which contains $c_{ix}^\top\otimes c_{ix}^\top$ in its $66^\text{th}$ row, $c_{iy}^\top\otimes c_{iy}^\top$ in its $261^\text{st}$ row, and $c_{iz}^\top\otimes c_{iz}^\top$ in its $1041^\text{st}$ row. We use a preparation matrix $vPREP_1$ to move these rows to the top.
\begin{lemma} \label{lem:Adag6perm}
The circuit in figure \ref{fig:Adag6perm} is a $(\frac{1}{64},(1,6),0)$ block encoding of
$\begin{pmatrix}
c_{ix}^\top\otimes c_{ix}^\top \\
c_{iy}^\top\otimes c_{iy}^\top \\
c_{iz}^\top\otimes c_{iz}^\top
\end{pmatrix}$
costing 24 $T$-gates. The controlled circuit costs 116 $T$-gates.
\end{lemma}
\begin{figure}[H]
\centering
\begin{tikzpicture}
   \begin{yquantgroup}
   \registers{
      qubit {} q[13];
    }
    \circuit{
      slash q[0];
      
      [name=Ad] box {\Ifnum\idx<1 $(A_a^\dagger)^{\otimes 6}$\Else $(A_d^\dagger)^{\otimes 6}$\Fi} q[0], (q[1-]);      
      \draw (Ad-0) -- (Ad-1);
      hspace {2mm} -;

      cnot q[6] | q[12];
      cnot q[4] | q[10];
      cnot q[2] | q[8];
      hspace {0mm} -;
      cnot q[10], q[12] | ;
      cnot q[8], q[11], q[12] | q[10];
      cnot q[10] | q[11];
    }
    \equals
    \circuit{
      slash q[0];
      
      [name=Ad] box {\Ifnum\idx<1 $(A_a^\dagger)^{\otimes 6}$\Else $(A_d^\dagger)^{\otimes 6}$\Fi} q[0], (q[1-]);      
      \draw (Ad-0) -- (Ad-1);
      hspace {2mm} -;

      box {$vPREP_1$} (q[2-]);
    }
   \end{yquantgroup}
\end{tikzpicture}
\caption{Block-encoding circuit for the matrix in Lemma \ref{lem:Adag6perm}. The slash in the top line denotes a 6-qubit ancillary register.}
\label{fig:Adag6perm}
\end{figure}
To sum the top three vector rows, we block encode $[1,1,1,0]$ in the top row of a matrix and apply those data qubit gates to the bottom two qubits. We then multiply this by a $64\times 64$ matrix of 1s in the first column and zeros elsewhere to copy the top row to the other 63 rows below. The procedure is summarized below.
\begin{lemma}
The circuit in figure \ref{fig:ckdotcl} is a $(\frac{3}{1024},(5,8),0)$ block encoding of
$(c_k\cdot c_l)$ costing 52 $T$-gates. The controlled circuit costs 172 $T$-gates.
\end{lemma}
\begin{figure}[H]
\centering
\begin{tikzpicture}
   \begin{yquantgroup}
   \registers{
      qubit {} q[15];   
   }
   \circuit{
      slash q[2];

      [name=Ad] box {\Ifnum\idx<1 $(A_a^\dagger)^{\otimes 6}$\Else $(A_d^\dagger)^{\otimes 6}$\Fi} q[2], (q[3-]);   
      \draw (Ad-0) -- (Ad-1);
      hspace {2mm} -;

      box {$vPREP_1$} (q[4-14]);
      hspace {2mm} -;

      h q[1];
      h q[13];
      h q[14];
      zz (q[1,13,14]);
      h q[13], q[14] | q[1] ;
      hspace {0mm} -;
      cnot q[0];
      h q[1];
      cnot q[0] | ~ q[9-];
      h q[9-];
    }
    \equals
    \circuit{
      slash q[0,3];
      discard q[1,2],q[4-];
      box {$(A^\dagger)^{\otimes 6}$} (q[0-3]);   
      box {$vPREP$} (q[0-3]);  
    }
   \end{yquantgroup}
\end{tikzpicture}
\caption{Block-encoding circuit for $(c_k\cdot c_l)$. The slash on the LHS denotes a 6-qubit ancillary register upon which $(A_a^\dagger)^{\otimes 6}$ acts; there are an additional two ancillae above for the block encoding of a matrix with top row [1,1,1,0]. The slashes on the RHS denote an 8-qubit ancillary register (upper) and a 12-qubit data register (lower) upon which $(A_d^\dagger)^{\otimes 6}$ acts.}
\label{fig:ckdotcl}
\end{figure}

To access the sum in \eqref{eq:f2_block} and \eqref{eq:f3_block}, we use a simple LCU circuit to combine the circuits in figures \ref{fig:calC} and \ref{fig:ckdotcl}.
\begin{lemma}
The circuit in figure \ref{fig:sum} is a $(\frac{3}{1664},(6,9),\epsilon )$ block encoding of
$-3(c_i\cdot c_k)(c_i\cdot c_l)+(c_k\cdot c_l)$ costing $268+2.3\log (1/\epsilon )$ $T$-gates.
\begin{figure}[H]
\centering
\begin{tikzpicture}
  \begin{yquantgroup}
    \registers{
        qubit {} q[3];
    }
    \circuit{
        slash q[1,2];
        hspace {0mm} -;
        box {$\left(A^\dagger\right)^{\otimes 6}$} (q[1-2]);
        box {$R_y(\theta_1)$} q[0];
        box {$vPREP$} (q[1-2]) | ~ q[0];
        box {$CPREP$} (q[1-2]) | q[0];
        box {$A^{\otimes 6}$} (q[1-2]) | q[0];
        hspace {0mm} -;
        box {$R_y(\theta_1)$} q[0];    
    }
    \equals
    \circuit{
        slash q[1,2];
        box {$LCU$} (q);
    }
   \end{yquantgroup}

\end{tikzpicture}
\caption{Block-encoding circuit for $-3(c_i\cdot c_k)(c_i\cdot c_l)+(c_k\cdot c_l)$. The slash on the middle line denotes an 8-qubit ancillary register, while the slash on the bottom line denotes the 12-qubit data register.}
\label{fig:sum}
\end{figure}
\end{lemma}
{\bf Proof:} Notice that figure \ref{fig:calC} block encodes $\frac{1}{2^{12}}(c_i\cdot c_k)(c_i\cdot c_l)$ and figure \ref{fig:ckdotcl} block encodes $\frac{1}{2^{10}}(c_k\cdot c_l)$. Therefore, we choose $\theta _1$ such that 
\begin{equation}
    R_y(\theta _1)=\frac{1}{\sqrt{13}}\begin{pmatrix}
    1 & -\sqrt{12} \\
    \sqrt{12} & 1
\end{pmatrix}.
\end{equation}
This implies the matrix being block encoded by the LCU sum is $\frac{1}{13\cdot 2^{10}}\left( -3(c_i\cdot c_k)(c_i\cdot c_l)+\frac{1}{2^{10}}(c_k\cdot c_l)\right)$. By placing figures \ref{fig:calC} and \ref{fig:ckdotcl} in the LCU, notice $(A^\top)^{\otimes 6}$ can be implemented uncontrolled, thus reducing gate count. Not including the rotations, there are 252 $T$-gates in this circuit. If each rotation in approximated to an accuracy $\epsilon$, then the total circuit has accuracy $3328\epsilon /3$. By normalizing the error of the total circuit, we arrive at the result. 

To arrive at the final block encoding of $F_2$, we place figure \ref{fig:W} in series with figure \ref{fig:sum} and use a series of multi-controls to zero-pad the rest of the matrix.

\begin{figure}[H]
\centering
\begin{tikzpicture}
  \begin{yquant}
    qubit {} q[9];
    
    box {$R_y(\theta_2)$} q[2] | ~ q[7];
    box {$R_y(\theta_2)$} q[1] | ~ q[5];
    box {$R_y(\theta_2)$} q[0] | ~ q[3];
    
  \end{yquant}
\end{tikzpicture}
\caption{Block-encoding circuit for $W = w\otimes w\otimes w$, with $w=\text{diag}\left(\frac{1}{4},\frac{1}{4},1,1\right)$. Controlled rotations are applied to the upper three ancillae.}
\label{fig:W}
\end{figure}

\begin{theorem} \label{prop:F2}
The circuit in figure \ref{fig:F2circuit} is a $(\frac{3}{13312},(6,22),\epsilon )$ block encoding of $F_2$ costing $328+5.75\log (1/\epsilon )$ $T$-gates.
\begin{figure}[H]
\centering
\begin{tikzpicture}
  \begin{yquant}
    qubit {} q[25];

    cnot q[12] | q[23],q[24];
    cnot q[11] | q[21],q[22];
    cnot q[10] | q[19],q[20];
    cnot q[9] | q[17],q[18];
    cnot q[8] | q[15],q[16];
    cnot q[7] | q[13],q[14];

    box {$LCU$} (q[13-]);

    box {$R_y(\theta_2)$} q[6] | ~ q[23];
    box {$R_y(\theta_2)$} q[5] | ~ q[21];
    box {$R_y(\theta_2)$} q[4] | ~ q[19];

    hspace {0mm} -;

    cnot q[3];
    cnot q[3] | ~ q[13],q[14],q[15],q[16],q[17],q[18];
    cnot q[2] | q[23],q[24];
    cnot q[1] | q[21],q[22];
    cnot q[0] | q[19],q[20];
  \end{yquant}
\end{tikzpicture}
\caption{Block-encoding circuit for $F_2$. The bottom 12 lines are the data qubit register; all lines above are ancillae. The $LCU$ subcircuit includes 9 persistent ancillae not shown here (see figure \ref{fig:sum}).}
\label{fig:F2circuit}
\end{figure}
\end{theorem}
{\bf Proof:} Since the error of a product of two block encodings with respective errors $\epsilon _1,\epsilon _2$ has error $\epsilon _1+\epsilon _2$, we can optimize the $T$-gate count by choosing the $W$ operation to have error $2\epsilon /5$ and the LCU operator to have error $3\epsilon /5$; this leads to an extra 4 $T$-gates coming out of the logarithms. There are 56 additional $T$-gates required to zero-pad the circuit. $\hfill\Box$

\paragraph{Block encoding $F_3$.}

To extend theorem \ref{prop:F2} to a block encoding of $F_3$, we need to tensor $F_2$ with a line of $[1,1,1,0]$ tensored with itself three times.
\begin{theorem}
The circuit in figure \ref{fig:F3circuit} is a $(\frac{3}{106496},(6,25),\epsilon )$ block encoding of $F_3$ costing $340+5.75\log (1/\epsilon )$ $T$-gates.
\begin{figure}[H]
\centering
\begin{tikzpicture}
  \begin{yquant}
    qubit {} q[22];

    cnot q[3] | q[8],q[9];
    cnot q[2] | q[6],q[7];
    cnot q[1] | q[4],q[5];

    hspace {0mm} -;
    
    cnot q[0];

    box {$H$} q[4];
    box {$H$} q[5];
    box {$H$} q[6];
    box {$H$} q[7];
    box {$H$} q[8];
    box {$H$} q[9];

    cnot q[0] | ~ q[4-9];

    box {$F_2$} (q[10-]);
  \end{yquant}
\end{tikzpicture}
\caption{Block-encoding circuit for $F_3$.}
\label{fig:F3circuit}
\end{figure}

\end{theorem}
{\bf Proof:} Compared to the results in proposition \ref{prop:F2}, we only need an extra 12 $T$-gates for the zero-padding, an additional 3 qubits, and multiplying the scaling factor by $1/8$ due to the Hadamard gates.

\subsection{Block encoding the $S$-matrix}

In this section, we block encode the streaming matrix $S \in \{-1,0,1\}^{nQ \times nQ}$ defined in (\ref{eq:tildeS}). We first re-write the definition so that it's evident how it acts on vector spaces of position and velocity.

We define the following vector space structure
\begin{equation}
   \label{eq:streaming_vector_space}
   \mathcal{V}= V_{\text{position}} \otimes V_{\text{velocity}} \otimes V_{\text{auxiliary}}.
\end{equation}
The third tensor factor $V_{\text{auxiliary}}$ will be used to encode the constraints that correspond to different grid points being fluid or solid nodes.  

The $S^{\text{in}}$ term from \eqref{eq:streaming_matrix} can be composed of operations that employ arithmetic expressions 
and a comparator, and so it can readily be created by composing a series of reversible computations.  This implies that we have a sum of unitaries for the streaming operators. 
Block encoding each term in \eqref{eq:streaming_matrix}, we have the following equation
\begin{equation}
    \label{eq:lcu_streaming}
    U_{S}= -I + U_{S^{\text{in}}}.
\end{equation}

For the implementation of $U_{S^{\text{in}}}$, we break $S^{\text{in}}$ into the terms
\begin{equation}
    S^{\text{in}} = S^{1} + S^{2}.
\end{equation}
$S^{1},S^{2}$ will be applied conditioned on the state of some flag qubits. $S_1$ will only act if grid point $\mathbf{x}$ is a fluid node and $S_2$ will only act if $\mathbf{x} - \mathbf{c}_i$ is a solid node. We define the following $S_{\text{shift}}$ operator:
\begin{align}
    \label{eq:shifting_operator}
    S_{\text{shift}} 
    \ket{\mathbf{x}} 
    \ket{\mathbf{c}_i}
    \ket{0}_{\text{auxiliary}}
    &= 
    \ket{\mathbf{x}}
    \ket{\mathbf{c}_i}
    \ket{\mathbf{x} - \mathbf{c}_i}_{\text{auxiliary}}.
\end{align}

Thus the procedure for applying $U_{S^{\text{in}}}$ is as follows. First use oracle $O_{\mathcal{N}}$ to compute whether a point $\mathbf{x}$ is a fluid or solid node, as in \eqref{eq:issolid}.  Then condition the application of $U_{S^1}$ on the result.  To apply $U_{S^2}$, first apply $S_{\text{shift}}$ to calculate the $\ket{\mathbf{x} - \mathbf{c}_i}$ position.  Then apply oracle $O_{\mathcal{N}}$ again on the shifted position and condition the application of of $U_{S^2}$ on the result.  See Figure \ref{fig:circuit_U_S_in}.

\begin{figure}[ht!]
     \centering
      \yquantset{operator/separation=3mm}
\begin{tikzpicture}
  \begin{yquant}
    qubit {$\ket{0}_{S^2}$} flags2;
    qubit {$\ket{0}^{\otimes k}_{\text{aux}}$} auxiliary;
    qubit {$\ket{0}_{S^1}$} flags1;
    qubit {$\ket{\mathbf{x}}$} position;
    qubit {$\ket{\mathbf{c}}$} velocity;
    qubit {$\ket{0}_{\text{ancillae}}^{\otimes b}$} ancillae;

    box {$O_{\mathcal{N}}$} (auxiliary, flags1, position);
    box {$U_{S^{1}}$} (position, velocity, ancillae) ~ flags1;
    
    [name=SS] box {$S_{\text{shift}}$} (auxiliary), (position, velocity);      
    \draw (SS-0) -- (SS-1);
    
    box {$O_{\mathcal{N}}$} (auxiliary, flags2);
    box {$U_{S^{2}}$} (position, velocity, ancillae) | flags2;  

  \end{yquant}
\end{tikzpicture}
\caption{Circuit for implementing $U_{S^{\text{in}}}$. The sufficiently large auxiliary space is the $\ket{0}^{\otimes k}_{\text{aux}}$ register, which is used both to store $\ket{\mathbf{x} - \mathbf{c}_i}$ and as a source of ancillae for constructing $O_{\mathcal{N}}$. $\ket{0}_{S^1}$ and $\ket{0}_{S^2}$ are the flag qubits for $S^1$ and $S^2$, respectively. The register $\ket{0}_{\text{ancillae}}^{\otimes b}$ is used for block encoding and $b$ is sufficiently large.}
\label{fig:circuit_U_S_in}
\end{figure}

Similar to the previous section, quantum costs will be $T$-gate counts. The count will be calculated by costing out oracle $O_{\mathcal{N}}, U_{S^{1}}, U_{S^{2}} \text{ and } S_{\text{shift}}$. The $T$-gate counts that follow depend on $n_p$, defined as 
\begin{equation}
    n_p = \max \{ \lceil \log_2 n_x \rceil , \lceil \log_2 n_y \rceil , \lceil \log_2 n_z \rceil \},
\end{equation}
which is the size of a qubit register required to encode the integer lattice position in the $x$, $y$, or $z$ direction. 

\paragraph{$S_{\text{shift}}$}
To cost out (\ref{eq:shifting_operator}), we first state how we shall use a \textit{two's complement} scheme to encode the velocity vectors. The velocity magnitudes are $(c_x, c_y, c_z) \in \{1, 0, -1 \}^{3}$.  We encode the magnitudes as bit strings as $ 00 \longmapsto 0, 01 \longmapsto 1, 11 \longmapsto -1$. The bit string $10$ is unphysical and not used. During the initial state preparation, the amplitude of states that include the bit string $10$ are fixed to zero.  We need to pad representation to the left with the appropriate number of ones in the case of $-1$ or zeros  in the case of $0$ and $1$ in order to facilitate the subtraction. The $S_{\text{shift}}$ operator will need to do the following subtractions
\begin{align}
    \ket{x} \ket{c_x} \ket{0}_{\text{aux}} 
    &\longmapsto
    \ket{x} \ket{c_x} \ket{x-c_x}_{\text{aux}}
    \\
    \ket{y} \ket{c_y} \ket{0}_{\text{aux}} 
    &\longmapsto 
    \ket{y} \ket{c_y}\ket{y-c_y}_{\text{aux}}
    \\
    \ket{z} \ket{c_z} \ket{0}_{\text{aux}} 
    &\longmapsto 
    \ket{z} \ket{c_z}\ket{z-c_z}_{\text{aux}}.
\end{align}
Using \cite{Gidney2018}, we need $12(n_p-1)$ $T$-gates.

\paragraph{Oracle $O_\mathcal{N}$ for the Sphere}
Again using \cite{Gidney2018, Kan2023}, to implement $O_{\mathcal{N}}$ (or ultimately \eqref{eq:classical_function}), we need to 
\begin{enumerate}
    \item subtract a constant for at most three possible different directions, using $\left(3(4n_p-8)\right)$ $T$-gates and at most $3(2n_p-2)$ ancilla qubits, and 
    \item square numbers for at most three possible directions using $\left(12n_p^2\right)$ $T$-gates, and
    \item add numbers for at most three possible directions using $3(4n_p-4)$ $T$-gates, and
    \item perform a comparison with sphere radius squared $r^2$ using $\left(8n_p-16\right)$ $T$-gates.
\end{enumerate}

\paragraph{Oracle $O_\mathcal{N}$ for the Rectangular prism} For the rectangular prism, we use the following subroutines.
\begin{enumerate}
    \item  6 comparators to implement the Heaviside function, which is $6(8n_p -16)$ $T$-gates.
    \item 6 adding of variables to constants which is $6(4n_p-8)$ $T$-gates.
    \item Subtractions of the outcomes from the Heaviside functions. 
  However, since we are working binary addition, this can be done with simple CNOT gates. The outcome of this leaves 3 ancilla qubits in either the $\ket{0}$ or \ket{1} state.
    \item 3 multiplications from subtracting 3 pairs of Heaviside functions can be accomplished by a 3 bit-Toffoli gate which would cost 6 $T$-gates using \cite{Gidney2021}.
\end{enumerate}
This gives a total cost of $72n_p -138$ $T$-gates for a rectangular prism.  We have $N_{prism}$ prisms.

\paragraph{Block encoding $S^1$ and $S^2$}

We shall use the second model for the \textit{structured} input represented by the circuit in Figure \ref{fig:Oprime_data_block_encoding} to create $U_{S^1}$ and $U_{S^2}$. $S^1, S^2 \in \{0,1\}^{nQ \times nQ}$, so it's clear that $O'_{\text{data}}$ is the identity since the only nonzero matrix elements are $1$. The next task is to cost out the oracle $O_c'$. We can use the combination of (\ref{eq:mu-fluid}) and (\ref{eq:c_mu}) to obtain an algebraic equation for the row index given a column index. For the streaming matrix $S_{\mu\eta}$
\begin{equation}
  \begin{split}
    \mu &= \eta + \left( Q \left[ \left(\eta \bmod Q \right)\bmod 3 + 2 \bmod 3 \right] -Q \right) + \\ 
   & \left( Q n_x \left[\left(\frac{\eta \bmod Q}{3}\right) \bmod 3 + 2 \bmod 3 \right] - Q n_x \right) + \\
   & \left( Q n_x n_y \left[\left(\frac{\eta \bmod Q}{9}\right) \bmod 3 + 2 \bmod 3 \right] - Q n_x n_y \right)
   \end{split}
\end{equation}
Noting that $ \left\lfloor \frac{x \bmod Q }{q}\right\rfloor \bmod 3  \equiv \left\lfloor \frac{x}{q} \right\rfloor \bmod 3$ if $3|Q$ (true for most commonly used values of $Q$, including our $Q=27$ case) means that our statements are, in fact,
\begin{equation}
\label{eq:streaming_matrix_row_column_eq}
  \begin{split}
    \mu &= \eta + \left( Q \left[ \eta \bmod 3 + 2 \bmod 3 \right] - Q \right) + \\ 
   & \left( Q n_x \left[ \left\lfloor \frac{\eta }{3} \right\rfloor \bmod 3  + 2 \bmod 3 \right] -Q n_x \right) + \\
   & \left( Q n_x n_y \left[ \left\lfloor \frac{\eta }{9} \right\rfloor \bmod 3  + 2 \bmod 3 \right] - Q n_x n_y \right)
   \end{split}
\end{equation}

In costing out the block encodings, we shall use some small bespoke circuits for the $ \bmod 3$ computations. Our two-qubit state representation will be $0 \longmapsto \ket{00},1 \longmapsto \ket{01}, 2 \longmapsto \ket{11}$ with $\ket{10}$ not being a valid state and therefore will be a fixed point. The first gadget is the $INC_3(x)$ which just does the operation $x  \leftarrow (x+1)\bmod 3$ that is represented in Figure~\ref{fig:$INC_3$}. 

\begin{figure}[h!]
\centering
\yquantset{operator/separation=6.1mm}
\begin{tikzpicture}
  \begin{yquant}
    qubit {$\ket{b}$} b;
    qubit {$\ket{a}$} a;

     cnot b | a;
     cnot a | ~b;

  \end{yquant}
\end{tikzpicture}
\caption{$INC_3(x)$ for $\ket{ba}$}
\label{fig:$INC_3$}
\end{figure}

The algorithm (\ref{alg:carry-stream-remainders}) for producing quantum registers that contain the values $\eta \bmod 3, \left\lfloor \frac{\eta}{3} \right\rfloor \bmod 3,\left\lfloor \frac{\eta}{9} \right\rfloor \bmod 3$ is presented below. For proof of correctness, see Lemmas \ref{lemma:Concatenating_lemma} -- \ref{lemma:carrybit_lemma} further below.

\begin{algorithm}[H]
\caption{ Modular division from successive carry streams}
\label{alg:carry-stream-remainders}

\KwIn{Binary digits $t=(q_{L-1}\cdots q_0)_2$.}
\KwOut{
$r_0=t\bmod3$,
$r_1=\lfloor t/3\rfloor\bmod3$, and
$r_2=\lfloor t/9\rfloor\bmod3$.
}

$r_0\gets0$; $r_1\gets0$; $r_2\gets0$\;

\For{$j=L-1,L-2,\ldots,0$}{
    $x\gets q_j$\;

    \tcp{Process the next bit of $t$.}
    $r_0'\gets(2r_0+x)\bmod3$\;
    $c_0\gets x\oplus[r_0'=1]$\;

    \tcp{Feed the next bit of $\lfloor t/3\rfloor$ into $R_1$.}
    $r_1'\gets(2r_1+c_0)\bmod3$\;
    $c_1\gets c_0\oplus[r_1'=1]$\;

    \tcp{Feed the next bit of $\lfloor t/9\rfloor$ into $R_2$.}
    $r_2'\gets(2r_2+c_1)\bmod3$\;

    $(r_0,r_1,r_2)\gets(r_0',r_1',r_2')$\;
}

\Return{$(r_0,r_1,r_2)$}\;
\end{algorithm}

Algorithm (\ref{alg:carry-stream-remainders}) is then reproduced using the quantum circuit presented in Figure~\ref{fig:mod3-iteration}

\begin{figure}[h!]
\centering
\resizebox{\textwidth}{!}{
\begin{tikzpicture}
\begin{yquant}

    qubit {$\ket{q_j}$} x;

    qubit {$\ket{b_0}$} b0;
    qubit {$\ket{a_0}$} a0;

    qubit {$\ket{b_1}$} b1;
    qubit {$\ket{a_1}$} a1;

    qubit {$\ket{b_2}$} b2;
    qubit {$\ket{a_2}$} a2;

    qubit {$C_0:\ket{0}$} c0;
    qubit {$C_1:\ket{0}$} c1;

    cnot b0 | a0;
    box {$\mathrm{INC}_3$} (b0, a0) | x;

    cnot c0 | x;
    cnot c0 | a0 ~ b0;

    cnot b1 | a1;
    box {$\mathrm{INC}_3$} (b1, a1) | c0;

    cnot c1 | c0;
    cnot c1 | a1 ~ b1;

    cnot b2 | a2;
    box {$\mathrm{INC}_3$} (b2, a2) | c1;

    cnot c1 | a1 ~ b1;
    cnot c1 | c0;

    cnot c0 | a0 ~ b0;
    cnot c0 | x;

\end{yquant}
\end{tikzpicture}
}

\caption{
One iteration of the mod-3 remainder scheme, processing
the next input bit $x=q_j$ in MSB-to-LSB order.
Each register $R_i=\ket{b_i a_i}$ encodes a trit using
$0\mapsto00$, $1\mapsto01$, and $2\mapsto11$.
If $P$ is the previously processed binary prefix, the
incoming logical values are
$r_0=P\bmod3$,
$r_1=\lfloor P/3\rfloor\bmod3$, and
$r_2=\lfloor P/9\rfloor\bmod3$.
After this iteration they are
$r_0'=(2P+q_j)\bmod3$,
$r_1'=\lfloor(2P+q_j)/3\rfloor\bmod3$, and
$r_2'=\lfloor(2P+q_j)/9\rfloor\bmod3$.
The first carry supplies the next binary digit of the
quotient by 3; the second supplies the next binary digit
of the quotient by 9.
Both carry qubits return to $\ket{0}$, and $x$ is unchanged.
Starting with all trits zero and repeating for every input
bit therefore produces
$t\bmod3$, $\lfloor t/3\rfloor\bmod3$, and
$\lfloor t/9\rfloor\bmod3$ in $R_0,R_1,R_2$, respectively.
}
\label{fig:mod3-iteration}
\end{figure}

Finally, to calculate the streaming matrix row index (\ref{eq:streaming_matrix_row_column_eq}) we have the circuit in Figure~\ref{fig:trit-controlled-additions} with the total number of $T$ gates counts reproduced in Table \ref{tab:t-count-combined-subtraction}.

\begin{figure}[H]
\centering
\resizebox{\textwidth}{!}{
\begin{tikzpicture}
\begin{yquant}

    qubit {$\ket{b_0}$} b0;
    qubit {$\ket{a_0}$} a0;

    qubit {$\ket{b_1}$} b1;
    qubit {$\ket{a_1}$} a1;

    qubit {$\ket{b_2}$} b2;
    qubit {$\ket{a_2}$} a2;

    qubits {$\ket{\eta}$} m;

    box {$\mathrm{INC}_3$} (b0, a0);
    box {$\mathrm{INC}_3$} (b0, a0);

    box {$-B_0$} m;
    box {$+B_0$} m | a0;
    box {$+B_0$} m | b0;


    box {$\mathrm{INC}_3$} (b1, a1);
    box {$\mathrm{INC}_3$} (b1, a1);

    box {$-B_1$} m;
    box {$+B_1$} m | a1;
    box {$+B_1$} m | b1;


    box {$\mathrm{INC}_3$} (b2, a2);
    box {$\mathrm{INC}_3$} (b2, a2);

    box {$-B_2$} m;
    box {$+B_2$} m | a2;
    box {$+B_2$} m | b2;

\end{yquant}
\end{tikzpicture}
}
\caption{
Constant additions controlled by the extracted trits.
Initially,
$r_0=\eta\bmod3$,
$r_1=\lfloor\eta/3\rfloor\bmod3$, and
$r_2=\lfloor\eta/9\rfloor\bmod3$.
Each register $R_i=\ket{b_i a_i}$ uses the encoding
$0\mapsto00$, $1\mapsto01$, $2\mapsto11$.
Two applications of $\mathrm{INC}_3$ produce
$s_i=(r_i+2)\bmod3$.
In this encoding $s_i=a_i+b_i$, so the two singly
controlled adders contribute $B_i a_i+B_i b_i=B_i s_i$:
neither fires for $00$, one fires for $01$, and both
fire for $11$.
Including the unconditional subtraction contributes
$B_i(s_i-1)$.
With $B_0=Q$, $B_1=Qn_x$, and $B_2=Qn_xn_y$,
the target $M$, initially containing $\eta$, finishes
with $\mu=\eta+\sum_{i=0}^{2}B_i(s_i-1)$.
The trit registers retain their shifted values $s_i$.
The arithmetic boxes denote reversible constant
addition or subtraction; on an $m$-qubit target,
these operations are interpreted modulo $2^m$.
}
\label{fig:trit-controlled-additions}
\end{figure}

\begin{table}[H]
\centering
\resizebox{\textwidth}{!}{
\begin{tabular}{|l|c|c|c|}
\toprule
Operation
& Instances
& $T$ gates per instance
& Total $T$ gates \\
\midrule
\hline

Doubling a trit modulo $3$ (CNOT)
& $3\lceil\log_2(nQ)\rceil$
& $0$
& $0$ \\[4pt]
\hline

Controlled-$\mathrm{INC}_3$ (two Toffolis)
& $3\lceil\log_2(nQ)\rceil$
& $14$
& $42\lceil\log_2(nQ)\rceil$ \\[4pt]
\hline

Carry computation
& $2\lceil\log_2(nQ)\rceil$
& $7$
& $14\lceil\log_2(nQ)\rceil$ \\[4pt]
\hline

Carry uncomputation
& $2\lceil\log_2(nQ)\rceil$
& $7$
& $14\lceil\log_2(nQ)\rceil$ \\[4pt]
\hline

Uncontrolled-$\mathrm{INC}_3$ (twice per trit)
& $6$
& $0$
& $0$ \\[4pt]
\hline
Singly controlled addition of $B_i$
& $6$
& $11\lceil\log_2(nQ)\rceil-15$
& $66\lceil\log_2(nQ)\rceil-90$ \\[4pt]
\hline

Combined subtraction of $B_0+B_1+B_2$
& $1$
& $4\lceil\log_2(nQ)\rceil-5$
& $4\lceil\log_2(nQ)\rceil-5$ \\

\midrule
\textbf{Total}
& & &
$\boldsymbol{140\lceil\log_2(nQ)\rceil-95}$ \\
\bottomrule
\end{tabular}
}
\caption{
$T$-gate accounting for three-trit extraction followed by
controlled constant additions, with
$B_0=Q$, $B_1=Qn_x$, and $B_2=Qn_xn_y$.
The three unconditional subtractions are combined into
one subtraction of $B_0+B_1+B_2$.
Each Toffoli costs seven $T$ gates.
Constant-adder costs follow Fedoriaka,
\emph{New Circuit for Quantum Adder by Constant},
arXiv:2501.07060.
The count assumes $\lceil\log_2(nQ)\rceil\ge4$ and
uses full-width adder costs without constant-specific
simplifications.
Carry cleanup is included; the shifted trit registers
are retained.
}
\label{tab:t-count-combined-subtraction}
\end{table}

Table \ref{table:T_gates_for_streaming} summarizes the results.
 \begin{table}[H]
\centering
\begin{tabular}{|c|c|c|}
\hline
Operator & $T$-gate counts (sphere) & $T$-gate counts (rectangular prism)\\
\hline
$O_{\mathcal{N}}$ & 12 $n_p^2 +32 n_p - 52$ & $N_{prism}(72n_p -138)$ \\
\hline
$S_{\text{shift}}$ & $12(n_p -1) $ & $12(n_p -1)$ \\
\hline
$S^1, S^2$ & \multicolumn{2}{c|}{$ 140\lceil\log_2(nQ)\rceil-95 $} \\
\hline
\end{tabular}
\caption{$T$-gate counts for $O_{\mathcal{N}}, S_{\text{shift}}$, $S^1$, $S^2$.}
\label{table:T_gates_for_streaming}
\end{table}

We are finally in a position to write the block-encoding circuit for streaming matrix as defined in 
(\ref{eq:lcu_streaming}).  The circuit is shown in Figure \ref{fig:block_encoding_streaming_matrix_S}.

\begin{figure}[H]
     \centering
      \yquantset{operator/separation=3mm}
\begin{tikzpicture}
  \begin{yquant}
    qubit {$\ket{0}$} ancilla;
    qubit {$\ket{\psi}$} psi;
 
    box {$H$} ancilla;
    box {$U_{S^{in}}$} psi ~ ancilla;
    box {$H$} ancilla;
    box {$X$} ancilla;
    measure ancilla; output {$0$} ancilla;
  \end{yquant}
\end{tikzpicture}
\caption{The above gives us an (1,1,0) Block Encoding for streaming matrix, $S$ in (\ref{eq:lcu_streaming}). }
\label{fig:block_encoding_streaming_matrix_S}
\end{figure}
We finish this section by giving a series of lemmas that prove the correctiness of Algorithm~(\ref{alg:carry-stream-remainders})

\begin{lemma}[Concatenating Lemma]
  \label{lemma:Concatenating_lemma}
  Let $P_l$ be the bitstring $P_l = (q_1q_2 \dots q_l)$ then $P_{l+1}$ with bits $(q_1q_2 \dots q_{l+1})$ has the form $2P_{_l} + x_{l+1}$ where $x_{l+1}$ is the bit appended to the end of $P_l$
\end{lemma}
\begin{proof}
    Let $0 < l \le n-1 $ and $P_0=0$ and $P_l = \sum_{h=1}^{l}x_h2^{l-h}$ i.e $P_l = (x_1x_2 \dots x_l)$, then $P_{l+1}=(x_1x_2x_3 \dots x_l x_{l+1})= x_12^l +x_22^{l-1}+ x_32^{l-2}+ \dots x_l2^0= 2(x_12^{l-1} +x_22^{l-2}+ x_32^{l-3}+ \dots x_l2^0 + x_{l+1} = 2 P_l + x_{l+1} $.  
\end{proof}
The concatenating lemma tells us that building a bit-string and can be reinterpreted as a combination of adding a bit and multiplying the previous iteration of concatenation by 2. The goal of constructing a register with values $P \bmod 3$, $\left\lfloor \frac{P}{3}\right\rfloor \bmod 3$ and $\left\lfloor \frac{P}{9}\right\rfloor \bmod 3$ will be to find out what we are multiplying by 2 and what bit we are adding. 

\begin{lemma}[Remainder Lemma]
\label{lemma:remainder_lemma}
    Let $r_l = P_l \bmod 3$ then $r_{l+1} = 2{r_1} + x_{l+1} \bmod 3$
\end{lemma}
\begin{proof}
    By the concatenating lemma we have $P_{l+1}=2P_l + x_{l+1} = 2(3k_l + r_l) + x_{l+1}$ with $r_l \in \{ 0,1,2\}$. Thus if $r_{l+1} = P_{l+1} \bmod 3 \implies r_{l+1} = 2 r_l + x_{l+1} \bmod 3$
\end{proof} 
Note that $2 r_l + x_{x+1}$
takes values in $[0,5]$. So we can write it as $3c_{l+1} + r'_{l+1}$ where $c_{l+1}$ is some carry bit and $r'_{l+1} \in \{0,1,2 \}$. This immediately tells us that $x_{l+1} = (c_{l+1} + r'_{l+1}) \bmod 2$ or $x_{l+1} + r'_{l+1} = c_{l+1} \bmod 2$. Thus when we append a bit $x_{l+1}$ we can calculate what the carry bit is supposed to be. 
One can now see what the skeleton of the algorithm is supposed to be: we imagine we are looping through the bits of $P$ and at each iteration we use the lemma~\ref{lemma:remainder_lemma} which tells us to multiply the remainder accumulator by 2 and add the bit we want to append and moreover it tells us what the carry bit should be. The reason we care about the carry bit is because of another lemma we prove below.

\begin{lemma}
\label{lemma:carrybit_lemma}
    Let $k_{l-1} = \left\lfloor \frac{P_{l-1}}{3}\right\rfloor$ then $k_l = 2 k_{l-1} +c_{l}$ where
\end{lemma}
\begin{proof}
    Let $P_{l-1} = 3k_{l-1} + r_{l-1}$ and note that $k_{l-1} = \left\lfloor \frac{P_l}{3} \right\rfloor$. Again using concatenating lemma we have
    \begin{align*}
        P_{l} &= 2 P_l + x_l \\
         &= 2(3k_{l-1} + r_{l-1}) + x_l \\
         &= 3k_l + r'_{l}
    \end{align*}
    where $k_l = 2k_{l-1} + c_1$.
\end{proof}
With this result we can see that if we keep an accumulator with $r_l$ then another register which keeps an accumulator implementing $2 k_{l-1} + c_l$ where $k_0=0$ and $c_l$ is the carry bit for $P_l \bmod 3$ we can calculate $\left\lfloor \frac{P_l}{3} \right\rfloor \bmod 3$.
To calculate $\left\lfloor \frac{P_l}{9} \right\rfloor \bmod 3$ we can see that we can borrow the procedure for calculating $\left\lfloor \frac{P_l}{3} \right\rfloor \bmod 3$ by adding another register for keep track of the carry bit $c^{(1)}_l$ for $\left\lfloor \frac{P_l}{3} \right\rfloor \bmod 3$ and in another register calculating $a_l = 2a_l + c^{(1)}_l$ where $a_l = \left\lfloor \frac{P_1}{9} \right\rfloor \bmod 3 $ with $a_0 = 0$. 
The concatenating lemma and remainder lemma guarantee that this procedure will give us $\left\lfloor \frac{P}{3} \right\rfloor \bmod 3$ and $\left\lfloor \frac{P}{9} \right\rfloor \bmod 3$.
\subsection{Block Encoding the Linear Term ($S + F_1$)}
Now that we have a block encoding for $F_1$ and $S$, the last step is to present the block encoding for $S+F$. We will use a combination of the definition of a state preparation pair Definition 51 in  \cite{Gilyen2019long}, Lemma 52 also in \cite{Gilyen2019long} and Theorem 1 in \cite{Camps2020}. We reproduce them here for the convenience of the reader

\begin{definition}[State preparation pair]
\label{def:state_prep_pair}
Let $y \in \mathbb{C}^n$ and $\|y\|_1 \leq \beta$, the pair of unitaries $(P_L, P_R)$ 
is called a $(\beta,b,\varepsilon)$-state-preparation-pair if 
\[
P_L|0\rangle^{\otimes b} = \sum_{j=0}^{2^b-1} c_j |j\rangle 
\quad \text{and} \quad
P_R|0\rangle^{\otimes b} = \sum_{j=0}^{2^b-1} d_j |j\rangle
\]
such that 
\[
\sum_{j=0}^{m-1} \big| \beta(c_j^* d_j) - y_j \big| \leq \varepsilon
\]
and for all $j \in m, \dots, 2^b-1$ we have $c_j^* d_j = 0$.
\end{definition}

\begin{lemma}[Linear combination of block-encoded matrices]
\label{lemma:lcu_of_block_encoded_matrices}
Let $A = \sum_{j=1}^m y_j A_j$ be an $s$-qubit operator and $\varepsilon \in \mathbb{R}_+$. 
Suppose that $(P_L, P_R)$ is a $(\beta,b,\varepsilon_1)$-state-preparation-pair for $y$, 
\[
W = \sum_{j=0}^{m-1} |j\rangle \langle j| \otimes U_j 
  + \Big( I - \sum_{j=0}^{m-1} |j\rangle \langle j| \Big) \otimes I_a \otimes I_s
\]
is an $s+a+b$ qubit unitary such that for all $j \in 0,\dots,m$ we have that $U_j$ is an 
$(\alpha,a,\varepsilon_2)$-block-encoding of $A_j$. 

Then we can implement a 
$(\alpha\beta, a+b, \alpha\varepsilon_1 + \alpha\beta\varepsilon_2)$-block-encoding of $A$, 
with a single use of $W$, $P_R$ and $P_L^{\dagger}$.
\end{lemma}

\begin{theorem}
Let $A$ be as defined in Lemma~\ref{lemma:lcu_of_block_encoded_matrices} with 
$(\alpha^{(j)}, a^{(j)}, \epsilon^{(j)})$-block-encodings of 
$A_j$, 
for $j \in [m]$, 

Assume that all block-encodings are extended to 
$a = \max_j a^{(j)}$ ancilla qubits, 
$\alpha = \max_j \alpha^{(j)}$, 
and $\epsilon_1 = \max_j \epsilon^{(j)}$. 

Then, by Lemma~\ref{lemma:lcu_of_block_encoded_matrices}, we can construct a unitary $U_{b+n}$ that is an 
$(\alpha\beta, a+b, \alpha\epsilon_2 + \beta\epsilon_1)$-block-encoding of $B_s$.
\end{theorem}

For construction a block encoding o $U_{S+F_1}$ given block encoding of $U_S$ and $U_{F_1}$, we have $(H,H)$ state preparation pair which is a $(2,1,0)$-state preparation pair for vector $y =(1,1)$. Using the block encoding of $S$, $F_1$, the following parameters 
\begin{equation}
    a_{S+F_1} = \max\{1,9\}, \hspace{1mm} \alpha_{S+F_1} = \max\{\frac{1}{257},1\}, \hspace{1mm} \epsilon_{S+F_1} = \max\{\epsilon_{F_1}, \epsilon_S \},
\end{equation}
 we have a $(2,10,2\epsilon_{F_1})$ block encoding of $S+F_1$.

\subsection{Block encoding the Carleman $A$-matrix}
\label{subsec:carleman_matrix_block_encoding}

Lastly, we need to block encode the Carleman matrix $A$, as defined in \eqref{eq:3rd_order_carleman_linearized_system_renamed}, using the block encodings of $F_1$, $F_2$, $F_3$, and $S$ that were developed in the previous sections.
A challenge presented by the definition of $A$ is that each of the blocks in the matrix are different sizes, making the choice of how to map the domain onto the space of qubits nontrivial.
The strategy we take is to define $\tilde{A}$ that acts on an extended space, while preserving the action of $A$ on a particular subspace.  The original Carleman matrix $A$ is $\left[ nQ + (nQ)^2 + (nQ)^3\right] \times \left[ nQ + (nQ)^2 + (nQ)^3 \right]$.  The extended space matrix $\tilde{A}$ will be $3(nQ)^3 \times 3(nQ)^3$.  For each block in $A$ there will be a corresponding block in $\tilde{A}$.  In our construction of $\tilde{A}$, we will ensure that all blocks in $\tilde{A}$ are conveniently the same size: $(nQ)^3 \times (nQ)^3$.
The advantage of this embedding is that we can factorize the domain and range of $\tilde{A}$ as $ \mathbb{R}^3 \otimes \mathbb{R}^{(nQ)} \otimes \mathbb{R}^{(nQ)}  \otimes \mathbb{R}^{(nQ)}$.
The basis vectors in the first factor $\mathbb{R}^3$ label the row and column locations of the blocks.
This factorization facilitates the quantum compilation that incorporates the block encodings of $S$, $F_1$, $F_2$, and $F_3$ into a single block encoding of $\tilde{A}$.

We now define the matrices comprising the blocks of $\tilde{A}$ in terms of how they relate to the six nonzero blocks of $A$ (see \eqref{eq:3rd_order_carleman_linearized_system}), which we refer to as $A_{\alpha,\beta}$, where $\alpha$ and $\beta$ label the block row and block column locations. 
One step we will make is to map the nonsquare $F_2$ and $F_3$ into square matrices $\bar{F}_2$ and $\bar{F}_3$ by padding with zero rows.

The padded versions of these matrices can be defined by appending zero rows via a tensor product as $\bar{F}_2 = \hat{e}_0^{T}\otimes F_2$ and $\bar{F}_3 = \hat{e}_0^{T}\otimes\hat{e}_0^{T}\otimes F_3$, where we choose $\hat{e}_0$ to be the first basis vector in $\mathbb{R}^{(nQ)}$. Defining $F'_1=S+F_1$, the embeddings for the blocks are:
\begin{align}
    A_{1,1}=F'_1 &\longmapsto \tilde{A}_{1,1}= F'_1 \otimes I \otimes I\\
    A_{1,2}=F_2 =  & \longmapsto\tilde{A}_{1,2}=  \hat{e}_0^{T} \otimes F_2 \otimes I \\
    A_{1,3}=F_3  &\longmapsto\tilde{A}_{1,3}= \hat{e}_0^{T} \otimes \hat{e}_0^{T} \otimes F_3 \\
    A_{2,2}=F'_1 \otimes I + I\otimes F'_1 &\longmapsto\tilde{A}_{2,2}=  \left(F'_1 \otimes I + I\otimes F'_1 \right)\otimes I \\
    A_{2,3}= F_2 \otimes I + I \otimes F_2  &\longmapsto\tilde{A}_{2,3}= \hat{e}_0^{T} \otimes (F_2 \otimes I + \otimes I  \otimes F_2)\\
    A_{3,3}= F'_1 \otimes I \otimes I + I\otimes F'_1 \otimes I + I \otimes I \otimes F'_1 &\longmapsto\tilde{A}_{3,3}=  F'_1 \otimes I \otimes I + I\otimes F'_1 \otimes I + I \otimes I \otimes F'_1 
\end{align}
where the identity matrices act on $\mathbb{R}^{(nQ)}$, and the embedding for vectors in the domain of $A$ is defined as
\begin{align}
    \phi=\phi_1\oplus \phi_2\oplus \phi_3&\longmapsto\tilde{\phi}=(\hat{e}_0\otimes\hat{e}_0\otimes \phi_1) \oplus (\hat{e}_0\otimes \phi_2)\oplus \phi_3,
\end{align}
where the direct product structure on the left corresponds to $\mathbb{R}^{(nQ)}
\oplus
\left( \mathbb{R}^{(nQ)} \right)^{\otimes 2}
\oplus
\left( \mathbb{R}^{(nQ)} \right)^{\otimes{3}}$.
Though tedious, one can confirm that applying $\tilde{A}$ to any such $\tilde{\phi}$ will result in a vector of the form $\tilde{\phi}=(\hat{e}_0\otimes\hat{e}_0\otimes \phi_1) \oplus (\hat{e}_0\otimes \phi_2)\oplus \phi_3$.
Accordingly, we say that the above embedding preserves the action of $A$.
Critically, this ensures that solutions to $\frac{\partial}{\partial t}\tilde{\phi}(t)=\tilde{A}\tilde{\phi}(t)$ are simply solutions to $\frac{\partial}{\partial t}\phi(t)=A\phi(t)$ whenever the initial state $\tilde{\phi}(t=0)$ is an embedded vector.

Since we have shown that it is sufficient to generate solutions to the extended linear ODE $\frac{\partial}{\partial t}\tilde{\phi}(t)=\tilde{A}\tilde{\phi}(t)$, we will develop an approach to block encoding $\tilde{A}$.
First, we decompose $\tilde{A}$ into three terms $\tilde{A}=\tilde{A}_1+\tilde{A}_2+\tilde{A}_3$, where
\begin{align}
    \tilde{A}_1&=\ket{1}\bra{1}\otimes \tilde{A}_{1,1}+\ket{2}\bra{2}\otimes \tilde{A}_{2,2}+\ket{3}\bra{3}\otimes \tilde{A}_{3,3}\\
    \tilde{A}_2&=\ket{2}\bra{1}\otimes \tilde{A}_{1,2}+\ket{1}\bra{2}\otimes \tilde{A}_{2,1}\\
\tilde{A}_3&=\ket{3}\bra{1}\otimes \tilde{A}_{1,3}.
\end{align}
Because each of these terms only involves one of $F'_1$, $F_2$, and $F_3$ and depends linearly on them, we can block encode each $\tilde{A}_r$ individually using a single call to each respective block encoding, as illustrated in Figure~\ref{fig:transferblock}.

Although in our specific case we are interested in $F_r$ only up to $r=3$, we will establish a result that applies for arbitrary number of $F_r$ and $\tilde{A}_r$. We will use this to give a general theorem for constructing Carleman matrix block encodings. For the resource estimates of section \ref{sec:quantum-resource-estimates}, we will use this theorem to obtain the qubit counts, subnormalizations, and failure rates of the Carleman matrix block encoding.

Our first result is a lemma that constructs each block encoding of $\tilde{A}_r$ from a single call to the block encoding of $\bar{F}_r$.
\begin{lemma} 
\label{lem:F_r_matrix_BE}
Given $U_{\bar{F}_r}$, an $(\alpha_r, m_r, \epsilon_r)$ block encoding of the square matrix $\bar{F}_r=\ket{0}^{\otimes r-1}\otimes F_r$, the following circuit acting on $nT+m_r+3\lceil\log_2(T-r+1)\rceil$ qubits gives an
\[
\left(\alpha_r\frac{(T-r+1)(T-r + 2)}{2}, m_r+2\log_2(T-r+1),(T-r+1)\epsilon_r \right)
\]
block encoding of the $r^{\text{th}}$ degree term in the Carleman matrix
\begin{align}
    \tilde{A}_r = \sum_{k=1}^{T-r+1} |k\rangle\langle k+r-1|\otimes \ket{0}^{\otimes r-1}\sum_{j=1}^k I^{\otimes (j-1)}\otimes F_r\otimes I^{\otimes (T-r-j+1)},
\end{align}
where each $I$ is a $2^n\times 2^n$ matrix and $F_r$ is a $2^n\times (2^n)^r$ matrix:

\begin{figure}
      \yquantset{operator/separation=3mm}
\begin{tikzpicture}
  \begin{yquant*}[every nobit output/.style={}, register/separation=3mm]
    qubit {$\ket{ 0^{\lceil \log_2(T-r+1) \rceil }}$} preparepermutation;
    qubit {$\ket{ 0^{\lceil \log_2(T-r+1) \rceil}}$} trunc;
    qubit {$\ket{ 0^{\lceil \log_2(T-r+1) \rceil}}$} prepareblock;
    qubit {$\ket{\psi}_{j=1}$} x;
    qubit {$\vdots$} x[+1]; discard x[1];
    qubit {$\ket{\psi}_{j=T-r+1}$} x[+1];
    qubit {$\ket{\psi}$} y;
    qubit {$\ket{\psi}_1$} z;
    qubit {$\vdots$} z[+1]; discard z[1];
    qubit {$\ket{\psi}_{r-1}$} z[+1];
    qubit {$\ket{0^{m_r}}$} block_encoding_ancilla;
    
    box {$C_r$} trunc;
    box {$U_K$} prepareblock;
    box {$ U_S$}  preparepermutation| trunc, prepareblock;
    box {$P$} (x,y,z) | preparepermutation ;
    box {$U_{\overline{F}_r}$} (y,z, block_encoding_ancilla);
    box {$P^{'}$} (x,y) | preparepermutation;
    box {$ U_S^{\dagger}$}  preparepermutation| trunc, prepareblock; 
    box {$U_K^{\dagger}$} prepareblock;
  
    output {$\ket{0^{\lceil \log_2(T-r+1) \rceil}}$} prepareblock;
    output {$\ket{0^{N}}$} z[0];
    output {$\vdots$} z[1];
    output {$\ket{0^{N}}$}z[2];
    output {$\ket{0^{\lceil \log_2(T-r+1) \rceil}}$} preparepermutation;
    output {$\ket{0^{m_r}}$} block_encoding_ancilla;
  \end{yquant*}
\end{tikzpicture}
\caption{Block Encoding for the transfer matrices $\tilde{A}_r$. }\label{fig:transferblock}
\end{figure}

\noindent where $C_r=\sum_{k=1}^{T-r+1}\ket{k}\bra{k+r-1}$,
where $c$-$P$ and $c$-$P'$ implement the controlled permutation operations
\begin{align}
c\text{-}P &= \sum_{j=1}^{T-r+1} |j\rangle \langle j| \otimes \text{SWAP}_{1j}\ldots \text{SWAP}_{r(r+j-1)}\\
c\text{-}P'&= \sum_{j=1}^{T-r+1} |j\rangle \langle j| \otimes \text{SWAP}_{1j}
\end{align}
with $\text{SWAP}_{\ell\ell} = I$, where $\text{c-c-}U_S$ implements the (doubly) controlled state preparation that prepares the coefficients for the permutations in $P$ (ensuring that the sum of permutations is truncated to a level depending on the matrix block register value)
\[ 
\text{c-c-U}_{S} = \sum_{k=1}^{T-r+1}\sum_{j=1}^k\frac{1}{\sqrt{k}}\ket{j}\bra{0}\otimes \ket{k}\bra{k}\otimes \ket{k}\bra{k} + R_{\perp},
\]
with $R_{\perp}|0\rangle=0$, and where $U_K$ re-weights the blocks such that they are implemented with equal weight
\[ 
U_{K} = \sum_{k=1}^{T-r+1}\sqrt{\frac{2k}{({T-r+1})({T-r+2})}} \ket{k}\bra{0} + R_{\perp}.
\]
\end{lemma}
\begin{proof}
First, we compute the result of the $c$-$P' U_{\bar{F}_r} c$-$P$ sub-circuit:
\begin{align}
\langle 0^m| c\text{-}P' U_{\bar{F}_r} c\text{-}P |0^{m}\rangle &= \langle 0^m |(\sum_{j,j'=1}^{T} |j\rangle \langle j| \otimes \text{SWAP}_{1j})U_{\bar{F}_r}( |j'\rangle \langle j'| \otimes \text{SWAP}_{1j'}\ldots \text{SWAP}_{r(r+j'-1)}|0^m\rangle\nonumber\\
&= \frac{1}{\alpha_r}\sum_{j=1}^{T} |j\rangle \langle j|\otimes\ket{0}^{\otimes r-1} \otimes I^{\otimes j-1}\otimes F_r \otimes I^{\otimes T-j}.
\end{align}
Next, we compute the result of the $(\text{c-c-}U_s)U_KC_r$ sub-circuit:
\begin{align}
(\text{c-c-}U_{S}) U_{K}C_r\ket{0}\ket{0}
&= (\sum_{k=1}^{T-r+1}\sum_{j=1}^k\frac{1}{\sqrt{k}}\ket{j}\bra{0}\otimes \ket{k}\bra{k+r-1}\otimes \ket{k}\bra{k} + R_{\perp})U_K\ket{0}\ket{0}\nonumber\\
&= (\sum_{k=1}^{T-r+1}\sum_{j=1}^k\frac{1}{\sqrt{k}}\ket{j}\otimes \ket{k}\bra{k+r-1}\otimes \ket{k}\bra{k})\sum_{k'=1}^{T-r+1}\sqrt{\frac{2k'}{(T-r+1)(T-r+2)}}\ket{k'}\nonumber\\
&= (\sum_{k=1}^{T-r+1}\sum_{j=1}^k\sqrt{\frac{2}{(T-r+1)(T-r+2)}}\ket{j}\otimes \ket{k}\bra{k+r-1}\otimes \ket{k} )\nonumber
\end{align}
Finally, defining $\omega_r=\frac{(T-r+1)(T-r+2)}{2}$, we compose these two subcircuits to arrive at the final result
\begin{align}
\bra{0}\bra{0}\bra{0}U_K^{\dagger}(\text{c-c-}U_s)^{\dagger}c\text{-}P' U_{\bar{F}_r} c\text{-}P (\text{c-c-U}_{S}) U_{K}C_r\ket{0}\ket{0}\ket{0} 
&= (\sum_{k=1}^{T-r+1}\sum_{j=1}^k\frac{1}{\sqrt{\omega_r}}\bra{j}\otimes \ket{k}\bra{k}\otimes \bra{k}) \\
&(\frac{1}{\alpha_r}\sum_{j=1}^{T} |j\rangle \langle j|\otimes\ket{0}^{\otimes r-1} \otimes F_r^{(j)})\\
&(\sum_{k=1}^{T-r+1}\sum_{j=1}^k\frac{1}{\sqrt{\omega_r}}\ket{j}\otimes \ket{k}\bra{k+r-1}\otimes \ket{k})\\
&= \frac{1}{\alpha_r\omega_r}\sum_{k=1}^{T-r+1}\sum_{j=1}^k \ket{k}\bra{k+r-1}\otimes F_r^{(j)}
\end{align}
where $F_r^{(j)}\equiv I^{\otimes j-1}\otimes F_r \otimes I^{\otimes T-j}$ and we see that the subnormalization has been scaled by a factor of $\omega_r=(T-r+1)(T-r+2)/2$ and two additional registers of size $\lceil \log_2 (T-r+1)\rceil$ have been used for the overall block encoding.
The total number of qubits used is therefore $Tn+rn+3\lceil m+\log_2 (T-r+1)\rceil$.
The error from the $F_r$ block encoding propagates to the error from the $\tilde{A}_r$ block encoding as follows
\begin{align}
    |\tilde{A}_r-\alpha_r\omega_r\bra{0}U_{\tilde{A}_r}\ket{0}|&=\left|\sum_{k=1}^{T-r+1} \ket{k}\bra{k}\otimes (\sum_{j=1}^kF_r-\alpha\bra{0}U_{\bar{F}_r}\ket{0})^{(j)}\right|\nonumber\\
    &\leq \max_k\left|\sum_{j=1}^{k}( F_r-\alpha\bra{0}U_{\bar{F}_r}\ket{0})^{(j)}\right|\nonumber\\
    &\leq\left|\sum_{j=1}^{T-r+1}( F_r-\alpha\bra{0}U_{\bar{F}_r}\ket{0})^{(j)}\right|\nonumber\\
    &\leq (T-r+1)\epsilon    
\end{align}
\end{proof}
With the $\tilde{A}_r$ block encodings established, we are prepared to create the full Carleman matrix block encoding by creating a sum of the $\tilde{A}_r$.
To do so, we establish a corollary of an existing result from the literature that will serve as the basis for constructing the sum of the block encodings.
\begin{corollary}[Following Lemma 52 \cite{Gilyen2019long}] Given a set $U_{B_r}$ of $(\alpha_r, m_r, \epsilon_r)$ block encodings of $B_r$,

we can construct a $(D\alpha,m+\lceil\log_2D\rceil, D\alpha\epsilon)$-block encoding
of the matrix $B=\sum_{r=1}^D B_r$, where $\alpha=\max_r \alpha_r$, $m=\max_r m_r$, and $\epsilon=\max_r \epsilon_r$.
\end{corollary}
\begin{proof}
    We use Lemma 52. From the set of block encodings, we first construct a SELECT operator
    \begin{align}
        W = \sum_{r=1}^{D} \ket{r-1}\bra{r-1} \otimes U_{B_r} + \sum_{r=D+1}^{2^{\lceil\log_2D\rceil} - 1} \ket{r-1}\bra{r-1} \otimes I_n. 
    \end{align}
    Next, we define a PREPARE operator such that $(V,V^{\dagger})$ is a $(D,\lceil\log_2(D)\rceil,0)$-state preparation pair for the vector $y=(1,\overset{D}{\ldots},1,0,\overset{2^b-D}{\ldots},0)$:
    \begin{align}
        V = \frac{1}{\sqrt{D}}\sum_{r=1}^{D} \ket{r-1}\bra{0} + R_{\perp 0}.
    \end{align}
Then, invoking Lemma 52 of \cite{Gilyen2019long}, $V^{\dagger}W V$ is a $(D\alpha,m+\lceil\log_2D\rceil, D\alpha\epsilon)$-block encoding
of $B=\sum_{r=1}^D B_r$, where $\alpha=\max_r \alpha_r$, $m=\max_rm_r$, and $\epsilon=\max_r \epsilon_r$.
\end{proof}
With this Corollary established, we can give the general theorem establishing the Carleman matrix block encoding along with its qubit costs.
\begin{theorem}
Given a set $U_{\bar{F}_r}$ of $(\beta_r,a_r,\epsilon_r)$-block encodings of $\bar{F}_r$, where $r=1,\ldots,D$ and where $F_1$ is encoded in $n$ qubits, we can construct a circuit acting on 
\begin{align}
nT+a+3\lceil\log_2(T)\rceil+\lceil\log_2 D\rceil
\textup{ qubits}
\end{align}
that gives a
\[\left(D\frac{T(T+1)\beta}{2}, a+2\log_2(T),D\frac{T(T+1)\beta\epsilon}{2}\right)\textup{-block encoding} 
\]
of the
$T$-level truncation Carleman matrix, where $\beta=\max_r\beta_r$, $a=\max_ra_r$, and $\epsilon=\max_r\epsilon_r$.
\end{theorem}
\begin{proof}
We simply apply the uniform linear combination result of Lemma 52 (from \cite{Gilyen2019long})
to the $F_r$ block encoding result of Lemma \ref{lem:F_r_matrix_BE}. 
\end{proof}
Applying this theorem to our case where $T=3$ and $D=3$, we have that the Carleman matrix block encoding acts on $3n+a+16$ qubits, with subnormalization bounded by $54\beta$.

\section{D3Q27 lattice constants}\label{sec:D3Q27_lattice_constants}

For the D3Q27 lattice, the dimension $D=3$ and the number of lattice velocity vectors is $Q=27$. Table \ref{tab:D3Q27-constants} explicitly shows all constants for the D3Q27 lattice that are used in this work. The lattice vector components $c_{i x}$, $c_{i y}$, and $c_{i z}$ cycle through values $\{1,-1,0\}$ with periodicities of 3, 9, and 27, respectively.

\begin{table}[H]
\label{table:lattice_constants}
\centering
\begin{tabular}{|c|c|rrrrrr|}
\hline
index & weight & velocity vector  \\
$i$ & $w_i$ & $\mathbf{c}_i$ =& (&$c_{ix}$,& $c_{iy}$,&  $c_{iz}$&) \\
\hline
0 & $1/216$ & $\mathbf{c}_0 = -\mathbf{c}_{13} =$& (&1,& 1,& 1&) \\
\hline
1 & $1/216$ & $\mathbf{c}_1 = -\mathbf{c}_{12} =$& (&-1,& 1,& 1&) \\
\hline
2 & $1/54$ & $\mathbf{c}_2 = -\mathbf{c}_{14} =$& (&0,& 1,& 1&) \\
\hline
3 & $1/216$ & $\mathbf{c}_3 = -\mathbf{c}_{10} =$& (&1,& -1,& 1&) \\
\hline
4 & $1/216$ & $\mathbf{c}_4 = -\mathbf{c}_9 =$& (&-1,& -1,& 1&) \\
\hline
5 & $1/54$ & $\mathbf{c}_5 = -\mathbf{c}_{11} =$& (&0,&-1,&1&) \\
\hline
6 & $1/54$ & $\mathbf{c}_6 = -\mathbf{c}_{16} =$& (&1,&0,&1&) \\
\hline
7 & $1/54$ & $\mathbf{c}_7 = -\mathbf{c}_{15} =$& (&-1,&0,&1&) \\
\hline
8 & $2/27$ & $\mathbf{c}_8 = -\mathbf{c}_{17} =$& (&0,&0,&1&) \\
\hline
9 & $1/216$ & $\mathbf{c}_9 = -\mathbf{c}_4 =$& (&1,&1,&-1&) \\
\hline
10 & $1/216$ & $\mathbf{c}_{10} = -\mathbf{c}_3 =$& (&-1,&1,&-1&) \\
\hline
11 & $1/54$ & $\mathbf{c}_{11} = -\mathbf{c}_5 =$& (&0,&1,&-1&) \\
\hline
12 & $1/216$ & $\mathbf{c}_{12} = -\mathbf{c}_1 =$& (&1,&-1,&-1&) \\
\hline
13 & $1/216$ & $\mathbf{c}_{13} = -\mathbf{c}_0 =$& (&-1,&-1,&-1&) \\
\hline
14 & $1/54$ & $\mathbf{c}_{14} = -\mathbf{c}_2 =$& (&0,&-1,&-1&) \\
\hline
15 & $1/54$ & $\mathbf{c}_{15} = -\mathbf{c}_7 =$& (&1,&0,&-1&) \\
\hline
16 & $1/54$ & $\mathbf{c}_{16} = -\mathbf{c}_6 =$& (&-1,&0,&-1&) \\
\hline
17 & $2/27$ & $\mathbf{c}_{17} = -\mathbf{c}_8 =$& (&0,&0,&-1&) \\
\hline
18 & $1/54$ & $\mathbf{c}_{18} = -\mathbf{c}_{22} =$& (&1,&1,&0&) \\
\hline
19 & $1/54$ & $\mathbf{c}_{19} = -\mathbf{c}_{21} =$& (&-1,&1,&0&) \\
\hline
20 & $2/27$ & $\mathbf{c}_{20} = -\mathbf{c}_{23} =$& (&0,&1,&0&) \\
\hline
21 & $1/54$ & $\mathbf{c}_{21} = -\mathbf{c}_{19} =$& (&1,&-1,&0&) \\
\hline
22 & $1/54$ & $\mathbf{c}_{22} = -\mathbf{c}_{18} =$& (&-1,&-1,&0&) \\
\hline
23 & $2/27$ & $\mathbf{c}_{23} = -\mathbf{c}_{20} =$& (&0,&-1,&0&) \\
\hline
24 & $2/27$ & $\mathbf{c}_{24} = -\mathbf{c}_{25} =$& (&1,&0,&0&) \\
\hline
25 & $2/27$ & $\mathbf{c}_{25} = -\mathbf{c}_{24} =$& (&-1,&0,&0&) \\
\hline
26 & $8/27$ & $\mathbf{c}_{26} = -\mathbf{c}_{26} =$& (&0,&0,&0&) \\
\hline

\end{tabular}
\caption{D3Q27 constants used in this work.}
\label{tab:D3Q27-constants}
\end{table}
\section{Derivation of the blocks of the Carleman-linearized matrix}\label{sec:appendix_derivation_of_submatrices}

We will apply a third-order truncated Carleman Linearization procedure to \eqref{eq:lbm-multilinear-algebra}:
\begin{equation}
    \frac{\partial }{\partial t}f \approx (S+F_1)f + F_2 f^{\otimes2} + F_3 f^{\otimes3}.
\tag{Repeated \ref{eq:lbm-multilinear-algebra}}
\end{equation}
Our objective is to linearize \eqref{eq:lbm-multilinear-algebra} and capture the (truncated) dynamics in a linear system of the form:
\begin{subequations}
\begin{align}
    \frac{\partial}{\partial t} 
    \underbrace{\begin{bmatrix} f \\ f^{\otimes 2} \\ f^{\otimes 3} \end{bmatrix}}_{\phi}
    &\approx
    \underbrace{
    \begin{bmatrix} 
    (S + F_1) & F_2 & F_3 \\ 
    \star & \star & \star \\
    \star & \star & \star 
    \end{bmatrix}
    }_{A}
\underbrace{\begin{bmatrix} f \\ f^{\otimes 2} \\ f^{\otimes 3} \end{bmatrix}}_{\phi}
    + \underbrace{\mathbf{0}}_{b}\\
    \frac{\partial}{\partial t}\phi 
    &\approx
    A\phi + b,
\end{align}
\label{eq:carlman_incomplete}
\end{subequations}
where now we are treating $f^{\otimes 2}=f\otimes f$, and $f^{\otimes 3}=f\otimes f\otimes f$ as independent variables.  We need to derive the $\star$-values of the blocks of $A$ in \eqref{eq:carlman_incomplete}.  The second ``row'' of blocks of $A$ describes the dynamics of $\frac{\partial}{\partial t}f^{\otimes 2}$ and the third ``row'' of blocks of $A$ describes the dynamics of $\frac{\partial}{\partial t}f^{\otimes 3}$.

Considering $\frac{\partial}{\partial t}f^{\otimes 2}$:
\begin{subequations}
    \begin{align}
        \frac{\partial}{\partial t}f^{\otimes 2} 
        &=
        f \otimes \left(\frac{\partial}{\partial t}f \right) + \left(\frac{\partial}{\partial t}f \right) \otimes f \\
        &=
        f \otimes \left( (S+F_1)f + F_2 f^{\otimes2} + F_3 f^{\otimes3}
        \right) 
        + \left( (S+F_1)f + F_2 f^{\otimes2} + F_3 f^{\otimes3}
        \right) \otimes f \\
        &\approx
        f \otimes ((S+F_1)f) + f \otimes (F_2 f^{\otimes2}) + 
        \underbrace{
        \cancel{f \otimes(F_3 f^{\otimes3})}
        }_{\text{higher order terms}} \\
        &
        \quad \quad + (S+F_1)(f \otimes f) + (F_2 f^{\otimes 2})\otimes f 
        + \underbrace{
        \cancel{(F_3 f^{\otimes 3})\otimes f}
        }_{\text{higher order terms}}
        \notag
        \\
        &\approx 
        f \otimes ((S+F_1)f) 
        + f \otimes (F_2 f^{\otimes2})
        + ((S+F_1)f) \otimes f 
        + (F_2 f^{\otimes 2})\otimes f \\
        &\approx 
        f \otimes ((S+F_1)f) 
        + ((S+F_1)f) \otimes f \\
        &
        \quad \quad +
        f \otimes (F_2 f^{\otimes2}) 
        + (F_2 f^{\otimes 2}) \otimes f
        \notag
        \\
        &\approx 
        (\mathbb{I}f) \otimes ((S+F_1)f) 
        + ((S+F_1)f) \otimes (\mathbb{I}f) \\
        &
        \quad \quad +
        (\mathbb{I} f) \otimes (F_2 f^{\otimes2}) 
        + ((F_2) f^{\otimes 2}) \otimes (\mathbb{I}f) 
        \notag
        \\
        &\approx 
        (\mathbb{I} \otimes (S+F_1))(f \otimes f) 
        + ((S+F_1) \otimes \mathbb{I}) (f \otimes f)
        \quad \quad \text{(by the mixed product rule)}
        \\
        &
        \quad \quad + 
        (\mathbb{I} \otimes F_2)(f \otimes f^{\otimes2}) 
        + (F_2 \otimes \mathbb{I})(f^{\otimes2} \otimes f) 
        \notag
        \\
        &\approx 
        \underbrace{
        \left( (\mathbb{I} \otimes (S+F_1)) + ((S+F_1) \otimes \mathbb{I}) \right)
        }_{(S+F_1)^{[2]} \text{ per \eqref{eq:S_plus_F_1_2_submatrix}}}(f \otimes f)\\
        &\quad \quad +
        \underbrace{
        (\mathbb{I} \otimes F_2  + F_2 \otimes \mathbb{I} )
        }_{F_2^{[2]} \text{ per \eqref{eq:F_2_2_submatrix}}}
        (f \otimes f \otimes f)
        \notag
        \\
        &\approx 
        \mathbf{0}f + (S+F_1)^{[2]}(f^{\otimes 2}) + F_2^{[2]}(f^{\otimes 3}).
    \end{align}
    \label{eq:derivation_of_S_plus_F_1_2_and_F_2_2}
\end{subequations}
The identity matrix $\mathbb{I}$ that appears in \eqref{eq:derivation_of_S_plus_F_1_2_and_F_2_2} is $nQ \times nQ$, which is consistent with the dimension of our original $f$-variables.
The same dimensions for the identity matrix are also used in \eqref{eq:S_plus_F_1_2_submatrix}, \eqref{eq:F_2_2_submatrix}, \eqref{eq:S_plus_F_1_3_submatrix}, and the following analysis in \eqref{eq:derivation_of_S_plus_F_1_3}.

We follow a similar analysis for $\frac{\partial}{\partial t}f^{\otimes 3}$:
\begin{subequations}
    \begin{align}
        \frac{\partial}{\partial t}f^{\otimes 3} 
        &=
        f \otimes f \otimes \left(\frac{\partial}{\partial t}f \right) 
        \\
        &\quad \quad + 
        f \otimes \left(\frac{\partial}{\partial t}f \right) \otimes f 
        \notag
        \\ 
        &\quad \quad + 
        \left(\frac{\partial}{\partial t}f \right) \otimes f \otimes f
        \notag
        \\ 
        &=
        f \otimes f \otimes \left((S+F_1)f + F_2 f^{\otimes2} + F_3 f^{\otimes3}\right) 
        \\ 
        &\quad \quad + 
        f \otimes \left((S+F_1)f + F_2 f^{\otimes2} + F_3 f^{\otimes3}\right) \otimes f 
        \notag
        \\
        &\quad \quad + 
        \left((S+F_1)f + F_2 f^{\otimes2} + F_3 f^{\otimes3}\right) \otimes f \otimes f
        \notag
        \\ 
        &\approx
        f \otimes f \otimes ((S+F_1)f)+ 
        \cancel{(\text{higher order terms})}
        \\ 
        &\quad \quad + 
        f \otimes ((S+F_1)f) \otimes f + 
        \cancel{(\text{higher order terms})}
        \notag
        \\ 
        &\quad \quad + 
        ((S+F_1)f) \otimes f \otimes f + 
        \cancel{(\text{higher order terms})}
        \notag
        \\ 
        &\approx 
        (\mathbb{I}f) \otimes (\mathbb{I}f) \otimes ((S+F_1)f)
        \\ 
        &\quad \quad + 
        (\mathbb{I}f) \otimes ((S+F_1)f) \otimes (\mathbb{I}f)
        \notag
        \\ 
        &\quad \quad + 
        ((S+F_1)f) \otimes (\mathbb{I}f) \otimes (\mathbb{I}f)
        \notag
        \\ 
        &\approx 
        (\mathbb{I}  \otimes \mathbb{I} \otimes (S+F_1))(f \otimes f \otimes f)
        \quad \quad \text{(by the mixed product rule)}
        \\
        &\quad \quad + 
        (\mathbb{I} \otimes (S+F_1) \otimes \mathbb{I})(f \otimes f \otimes f)
        \notag
        \\
        &\quad \quad + 
        ((S+F_1) \otimes \mathbb{I} \otimes \mathbb{I})(f \otimes f \otimes f)
        \notag
        \\
        &\approx 
        \underbrace{
        \left( 
        (\mathbb{I}  \otimes \mathbb{I} \otimes (S+F_1))
        + (\mathbb{I} \otimes (S+F_1) \otimes \mathbb{I})
        + ((S+F_1) \otimes \mathbb{I} \otimes \mathbb{I})
        \right)
        }_{(S+F_1)^{[3]} \text{per \eqref{eq:S_plus_F_1_3_submatrix}}}
        f^{\otimes 3}
        \\
        &\approx 
        \mathbf{0}f + \mathbf{0}f^{\otimes 2} + (S+F_1)^{[3]} f^{\otimes 3}.
    \end{align}
    \label{eq:derivation_of_S_plus_F_1_3}
\end{subequations}

\section{Bounding the spectral norm of the Carleman-linearized matrix}\label{sec:bounding_spectral_norm_of_A}

We need to estimate the spectral norm $\| \cdot \|_2$ of the Carleman matrix $A$ from \eqref{eq:3rd_order_carleman_linearized_system}.  The spectral norm of $\|A\|_2$ sets an upper bound on the time step size as $h \le 1/\|A\|_2$ per \cite{Berry2017}, and thus quantum resource estimates increase with a larger $\|A\|_2$. We will establish an upper bound on $\|A\|_2$ using the inequality $\|A\|_2 \le \sqrt{ \|A\|_1 \|A\|_\infty }$, and estimating upper bounds on the 1-norm $\|A\|_1$ and the $\infty$-norm $\|A\|_\infty$.  We define the norms as
\begin{equation}
    \|A\|_1 = \max_j \sum_i | a_{ij} | \quad (\text{max abs. column sum}),
\end{equation}
\begin{equation}
\|A\|_2 = \max_{\mathbf{x} \neq 0} \frac{|A\mathbf{x}|}{|\mathbf{x}|} \quad (\text{spectral norm}),
\end{equation}
and
\begin{equation}
    \|A\|_\infty = \max_i \sum_j | a_{ij} | \quad (\text{max abs. row sum}).
\end{equation}

As we proceed to estimate a \emph{bound} on the 1- and $\infty$-norm of $A$, we will leverage the norms of the blocks of $A$: $S$, $F_1$, $F_2$, $F_3$, $(S+F_1)^{[2]}$, $F_2^{[2]}$, and $(S+F_1)^{[3]}$.  We begin with some preliminary bounds on the sub matrices.

First note that $\|S\|_1 = \|S\|_\infty = 2$ because $S$ has at most 2 nonzero elements of magnitude 1: $+1$ for the inbound streaming population and $-1$ for the outbound streaming population.

Next, we note that the norm of the sub matrices $(S+F_1)^{[2]}$, $F_2^{[2]}$ that appear in the second and third ``row'' are bounded as follows.

\begin{subequations}
    \begin{align}
        \|(S+F_1)^{[2]}\|_{p=1,\infty}
        &\le \|\mathbb{I} \otimes (S+F_1) + (S+F_1) \otimes \mathbb{I}\|_{p=1,\infty} \\
        &\le \|\mathbb{I} \otimes (S+F_1) \|_{p=1,\infty} + \| (S+F_1) \otimes \mathbb{I}\|_{p=1,\infty} \label{eq:kron_norm_I}\\
        &\le \|S+F_1\|_{p=1,\infty} + \|S+F_1\|_{p=1,\infty} \label{eq:kron_norm_simplified}\\
        &\le 2\|S+F_1\|_{p=1,\infty} \\
        &\le 2\|S\|_{p=1,\infty} +2\|F_1\|_{p=1,\infty} \\
        &\le 4 + 2\|F_1\|_{p=1,\infty} 
    \end{align}
\end{subequations}
The step from \eqref{eq:kron_norm_I} to \eqref{eq:kron_norm_simplified} is justified via Theorem 8 from \cite{lancaster1972norms} and the fact that $\|\mathbb{I}\|=1$. A similar argument will show that
\begin{equation}
\|F_2^{[2]}\|_{p=1,\infty} \le 2\|F_2\|_{p=1,\infty},
\end{equation}
and also 
\begin{subequations}
\begin{align}
    \|(S+F_1)^{[3]}\|_{p=1,\infty}
    &\le 3\|S + F_1\|_{p=1,\infty} \\
    &\le 3(2) + 3\|F_1\|_{p=1,\infty} \\
    &\le 6 + 3\|F_1\|_{p=1,\infty}.
\end{align}    
\end{subequations}

We now make some observations about the structure of the $F$-matrices.  The $F$-matrices describe the collision operation.  Collision is a local operation to each grid node $\mathbf{x}$.  For some row of the $F$-matrix associated with variable $f_i(\mathbf{x})$, the same $F$-matrix will only have nonzero elements in \emph{columns} associated with variables $f_j(\mathbf{x}), j \in \mathbb{Z}_Q$ with the \emph{same} grid point $\mathbf{x}$.  This implies that the structure of the $F$-matrices has repeating blocks in rows/columns that are associated with variables at the same spatial grid point $\mathbf{x}$.  These blocks then impose the same collision dynamics at each node $\mathbf{x}$.  Thus, it is sufficient to study one block of the $F$-matrices.  Additionally, the 1- and $\infty$-norms of the $F$-matrices are independent of the number of grid points $n$.

Since the norms of the $F$-matrices are independent of $n$, we constructed a system with $n=1$ and implemented equations \eqref{eq:tildeF_1}, \eqref{eq:tildeF_2_dense}, and \eqref{eq:tildeF_3_dense} in Python SymPy \cite{sympy}.  The parameters for the D3Q27 lattice (see appendix \ref{sec:D3Q27_lattice_constants}) are fixed, and thus the only remaining simulation design parameter is $\tau$.  Source code is available at \cite{source_for_norm_A_anlysis}.  The results are (rounded up the nearest integer):

\begin{subequations}
    \begin{align}
        \|F_1\|_\infty &\le  \left( \frac{1}{\tau} \right) 8.407    \le  \left( \frac{1}{\tau} \right) 9 \\
        \|F_1\|_1 &\le \left( \frac{1}{\tau} \right) 3.481              \le \left( \frac{1}{\tau} \right) 4 \\
        \|F_2\|_\infty &\le \left( \frac{1}{\tau} \right) 565.333       \le \left( \frac{1}{\tau} \right) 566 \\
        \|F_2\|_1 &\le \left( \frac{1}{\tau} \right) 7.333              \le \left( \frac{1}{\tau} \right) 8 \\
        \|F_3\|_\infty &\le \left( \frac{1}{\tau} \right) 7632.0      \le \left( \frac{1}{\tau} \right) 7633 \\
        \|F_3\|_1 &\le \left( \frac{1}{\tau} \right) 3.667              \le \left( \frac{1}{\tau} \right) 4 
        \end{align}
\end{subequations}

The 1-norm (maximum absolute column sum) of $A$ will be maximized the third ``column'' of the Carleman matrix $A$, and so we can establish the following upper bound:
\begin{subequations}
    \begin{align}
        \|A\|_1 &\le \|F_3\|_1 + \|F_2^{[2]}\|_1 + \|(S+F_1)^{[3]}\|_1 \\
        &\le \|F_3\|_1 + 2\|F_2\|_1 + (6 + 3\|F_1\|_1) \\
        &\le 6 + \frac{1}{\tau} (32)
    \end{align}
\end{subequations}
The $\infty$-norm (maximum absolute row sum) of $A$ will be maximized in the first ``row'' of the Carleman matrix $A$, so we can establish the following upper bound:
\begin{subequations}
    \begin{align}
        \|A\|_\infty &\le \|S + F_1 \|_\infty + \|F_2\|_\infty + \|F_3\|_\infty \\
        &\le \|S\|_\infty + \|F_1\|_\infty + \|F_2\|_\infty+ \|F_3\|_\infty \\
        &\le 2 + \frac{1}{\tau}(8208)
    \end{align}
\end{subequations}
And finally we arrive at an upper bound on the spectral norm of $A$:
\begin{subequations}
\begin{align}
\|A\|_2 &\le \sqrt{\|A\|_1 \|A\|_\infty } \\
&\le \sqrt{ 12 + 49,312 \frac{1}{\tau} + 262,656 \frac{1}{\tau^2} }
\end{align}
\end{subequations}
Recall that $0.5 < \tau \le 1$. Thus the bound on the spectral norm of $A$ is between 559 and 1072.  In particular, if $\tau=0.6$, then $\| A\|_2 \le 901$, which is the value used for quantum resource estimates and displayed in tables \ref{tab:resource_estimate_parameters_sphere} and \ref{tab:resource_estimate_parameters_hulls}. Small test cases show that the true spectral norm of $A$ is less than the derived bound.

\section{Time domain convergence of the Carleman linearized system}
\label{sec:appendix_time_domain_convergence_of_CL}

The only norm used in this section is the the $\infty$-norm.  We shorten our notation to $\|A\| = \|A\|_\infty = \max_i \sum_j |a_{ij}|$.  For a vector $x$,  $\|x\| = \|x\|_\infty = \max_i |x_i|$.

\begin{lemma}[Proposition 3.4 from \cite{Forets2018}, rewritten with our notation.]\label{lem:forets_3.4}
    
 Consider the general (i.e., not necessarily LBM) $K^{\text{th}}$ order system with $K \geq 2$:
\begin{equation}
\label{eq:k-order-non-linear-ode}
\frac{\partial}{\partial t} f = F_1 f + F_2 f^{\otimes 2} + \dots + F_K f^{\otimes K},
\end{equation}
where $f \in \mathbb{R}^{nQ}$.

The $K^{\text{th}}$ order system \eqref{eq:k-order-non-linear-ode} reduces to a quadratic system in $\phi$
\begin{equation}
\label{eq:effective-second-order-ode}
\frac{\partial}{\partial t} \phi = \tilde{F}_1 \phi + \tilde{F}_2 \phi^{\otimes 2}, 
\end{equation}
\textit{where} $\phi = f^{\otimes 1} \oplus f^{\otimes 2} \oplus ...\oplus f^{\otimes K-1}$, \textit{and the matrices} $\tilde{F}_1$ \textit{and} $\tilde{F}_2$ are
\begin{align*}
\tilde{F}_1 &=
\begin{pmatrix}
F_1^{[1]} & F_2^{[1]} & F_3^{[1]} & \dots & F_{K-1}^{[1]} \\
0 & F_1^{[2]} & F_2^{[2]} & \dots & F_{K-1}^{[2]} \\
0 & 0 & \ddots & \ddots & \vdots \\
\vdots & \vdots & \ddots & \ddots & \vdots \\
0 & 0 & \dots & 0 & F_1^{[K-1]}
\end{pmatrix},
\end{align*}
and 
\begin{align*}
\tilde{F}_2 &=
\left(
\begin{array}{ccccccc}
0 \cdots 0      & F^{[1]}_K     & 0 \cdots 0 & 0            & 0 \cdots 0 & \cdots &  0 \\
0  \cdots 0     & F^{[2]}_{K-1} & 0 \cdots 0 & F^{[2]}_{K}  & 0 \cdots 0 & \cdots &  0 \\
\vdots          &\vdots         & \vdots     & \vdots       &   \vdots   & \vdots &\vdots  \\
    0\cdots 0 
    & F^{[K-1]}_2 
    & 0 \cdots 0 
    & F^{[K-1]}_{3} 
    & 0 \cdots 0  
    & \cdots &F^{[K-1]}_{K} 
\end{array}
\right), \\
\end{align*}
\text{ with}
\begin{equation}\label{eq:F^[i]_j_definition}
F_{j}^{[i]} = \sum_{\ell=1}^{i} 
\overbrace{\mathbb{I}_{nQ \times nQ} \otimes \cdots \otimes 
\underbrace{F_j}_{\ell^{\text{th}} \text{ position}} 
\otimes \cdots \otimes 
\mathbb{I}_{nQ \times nQ}}^{\text{i factors}},
\end{equation}
where $j=1,...,K.$  Moreover, the $\infty$-norm of the linear and quadratic parts satisfy, respectively,
\begin{equation*}
    \|\tilde{F}_1\| \leq \max_{1 \leq i \leq K-1} (K - i) \sum_{j=1}^{i} \|F_j\|
\quad \text{and} \quad
\|\tilde{F}_2\| \leq (K - 1) \sum_{j=2}^{K} \|F_j\|.
\end{equation*}
\end{lemma}

Below we cite the result relating to time convergence. The result is particularly useful because it does not require knowledge of the norm of the solution. 
\begin{theorem}[Theorem 4.3 from \cite{Forets2018}, 
 rewritten with our notation.]
\label{thm:convergence_thm}
For time $t \in [0,T]$, let $\phi(t)$ be an exact solution to the quadratic system \eqref{eq:effective-second-order-ode} and let $\hat{\phi}(t)$ be a solution to the Carleman linearization of \eqref{eq:effective-second-order-ode}, truncated at order $K$. Define
\begin{equation}
\beta_0 = \frac{\|\phi(t=0)\| \, \|\tilde{F_2}\|}{\|\tilde{F_1}\|}. 
\end{equation}
Then, the error $\varepsilon(t) = \phi(t) - \hat{\phi}(t)$ satisfies the estimate
\[
\|\varepsilon(t)\| \leq \mathcal{E}(t) = 
\frac{
    \|\phi(t=0)\| \exp (\|\tilde{F_1}\| t)
}
{
    (1 + \beta_0) - \beta_0 \exp (\|\tilde{F_1}\| t) 
}
\left(
    \beta_0 (\exp (\|\tilde{F_1}\| t) - 1)
\right)^K
. 
\]

\noindent
\textit{Moreover, for all $0 < t < T_c$, where}
\begin{equation}
T_c = \frac{1}{\|\tilde{F_1}\|} \ln\left(1 + \frac{1}{\beta_0}\right), 
\label{eq:T_c_definition}
\end{equation}
\textit{the solution of the truncated system converges, that is,}
\[
\lim_{K \to \infty} \|\varepsilon(t)\| = 0, \quad \text{for all } t < T_c.
\]
\end{theorem}

Our goal for this section is to approximate $T_c$ for our problem formulation.  In our LBM formulation with third-order truncation ($K=3$), 
\begin{align}
\tilde{F}_1 = \begin{pmatrix}
(S+F_1)^{[1]} & F_2^{[1]}  \\
0 & (S+F_1)^{[2]}  \\
\end{pmatrix}
\end{align}
and
\begin{align}
\tilde{F}_2 = \begin{pmatrix}
0\cdots0  & F_3^{[1]} & 0\cdots0 & 0  \\
0\cdots0 & F_2^{[2]} & 0\cdots0  & F_3^{[2]}  \\
\end{pmatrix}.
\end{align}
Lemma \ref{lem:forets_3.4} implies
\begin{equation}
\label{eq:maxnorm_for_transfer_matrices_3}
\| \tilde{F_1} \| \le \max \big\{ \|S+ F_1\| + \|F_2 \|, 2\|S+F_1\|  \big \},
\end{equation}
and
\begin{equation}
    \| \tilde{F_2} \| \le 2 \big( \|F_2 \| + \| F_3 \|  \big).
    \label{eq:maxnorm_for_transfer_matrices_4}
\end{equation}

Appendix \ref{sec:derivation-of-matrix-coeffs} shows that the structure of the $F$-matrices has uncoupled sub-blocks for each lattice node $\mathbf{x}$. Therefore, the calculations of the norms in \eqref{eq:maxnorm_for_transfer_matrices_3}, \eqref{eq:maxnorm_for_transfer_matrices_4} for the collision matrices can be done at a particular point $\mathbf{x}$ i.e., $\|F_1\| = \| F_1^{\textbf{x}}\|, \|F_2\| = \| F_2^{\textbf{x}}\|$, and $\|F_3\| = \| F_3^{\textbf{x}}\|$. 

First, we provide Lemma \ref{lem:inf_norm_expansion_holds}, which will be used later help us establish bounds on $T_c$.
\begin{lemma}\label{lem:inf_norm_expansion_holds}
For any $m \times n$ matrix $A$ and its related $A^{[i]}$ defined by \eqref{eq:F^[i]_j_definition},
\begin{equation}
\| A^{[i]} \| = i\| A \|.
\end{equation}
\end{lemma}
\begin{proof}
First note that $\|A \otimes \mathbb{I}\|=\|\mathbb{I} \otimes A\|=\|A\|$.  By induction, each term in the summation in \eqref{eq:F^[i]_j_definition} will have
\begin{equation}
\|
\overbrace{
\mathbb{I} \otimes \cdots \otimes 
\underbrace{A}_{\ell^{\text{th}} \text{ position}} 
\otimes \cdots \otimes 
\mathbb{I}}^{\text{i factors}}
\|= \| A \| = \max_r \sum_k | A_{rk} |.
\end{equation}
Recombining all $i$ terms in the summation in \eqref{eq:F^[i]_j_definition} will show that $\| A^{[i]} \|$ is
\begin{equation}
    \| A^{[i]} \| = \max_{r_1, r_2, ... , r_i } \Bigg( \sum_k | A_{r_1 k} |+ \sum_k | A_{r_2 k} |+ ... + \sum_k | A_{r_i k} | \Bigg),
\end{equation}
and each term can maximize independently.  There are $i$ terms, so the relationship reduces to
\begin{subequations}
\begin{align}
    \| A^{[i]} \| &= i \Bigg( \max_{r} \sum_k | A_{r k} | \bigg) \\
    &= i \| A \|.
\end{align}
\end{subequations}
\end{proof}
A similar argument can be made for the 1-norm, but our analysis only relies on the $\infty$-norm.

We now present the following argument to tighten the inequality in \eqref{eq:maxnorm_for_transfer_matrices_4} to
\begin{equation}
    \| \tilde{F_2} \| = 2 \big( \|F_2 \| + \| F_3 \|  \big).
    \label{eq:maxnorm_for_transfer_matrices_5}
\end{equation}
As emphasized in \eqref{eq:f3_block}, the $F_3$ matrix is a scaled tiling of the $F_2$ matrix.  Each $F^{\mathbf{x}}_3$ sub-block of the $F_3$ matrix specific to grid node $\mathbf{x}$ has the form
\begin{align}
    F^{\mathbf{x}}_3 = \left(\frac{-1}{2}\right)
    \Bigg[ 
    \underbrace{F^{\mathbf{x}}_2 \quad F^{\mathbf{x}}_2 \quad  \cdots \quad F^{\mathbf{x}}_2}_{Q \text{ tiles}}
    \Bigg].
\end{align}
Recall that $\| F^\mathbf{x}_2\|=\| F_2\|$ and $\| F^\mathbf{x}_3\|=\| F_3\|$.

Let 
\begin{equation}
i^* = \argmax_i \sum_{j} \bigg| \left[ F^\mathbf{x}_2 \right]_{ij}  \bigg|.
\end{equation}
Due to degeneracies, $i^*$ may not be unique. Then
\begin{align}
\sum_{j} \bigg| \left[ F^\mathbf{x}_2 \right]_{ij} \bigg| &\le \sum_{j} \bigg| \left[ F^\mathbf{x}_2 \right]_{i^*j} \bigg| & \forall i \\
\sum_{j} \bigg| \left(\frac{-1}{2}\right)\left[ 
\underbrace{
F^\mathbf{x}_2 \cdots F^\mathbf{x}_2
}_{Q \text{ tiles}}
\right]_{ij} \bigg| 
&\le
\sum_{j} \bigg| \left(\frac{-1}{2}\right)\left[ 
\underbrace{
F^\mathbf{x}_2 \cdots F^\mathbf{x}_2
}_{Q \text{ tiles}}
\right]_{i^*j} \bigg| 
& \forall i\\
\sum_{j} \bigg| 
\left[ 
F^\mathbf{x}_3 
\right]_{ij}
\bigg| 
&\le
\sum_{j} \bigg| 
\left[ 
F^\mathbf{x}_3
\right]_{i^*j}
\bigg| & \forall i, 
\end{align}
and we can conclude the same set of row indices $i^*$ that maximize $\sum_j \big| \left[F^{\mathbf{x}}_2\right]_{ij} \big|$ also maximize $\sum_j \big| \left[F^{\mathbf{x}}_3\right]_{ij} \big|$.

When considering
\begin{align}
\| \tilde{F}_2  \| &= 
\Bigg\|
\begin{pmatrix}
0\cdots0  & F_3^{[1]} & 0\cdots0 & 0  \\
0\cdots0 & F_2^{[2]} & 0\cdots0  & F_3^{[2]}  \\
\end{pmatrix} 
\Bigg\|,
\end{align}
we want to show that the same row index $i^*$ that maximizes the $\sum_j \Bigg| \left[F^{[2]}_2\right]_{ij} \Bigg|$ also maximizes $\sum_j \Bigg| \left[F^{[2]}_3\right]_{ij} \Bigg|$.  Assume for contradiction that
\begin{align}
    \label{eq:bad_assumption}
    \exists \ell^* &\neq i^* \text { such that } \sum_j \Bigg| \left[F^{[2]}_3\right]_{\ell^*j} \Bigg| > \sum_j \Bigg| \left[F^{[2]}_3\right]_{i^*j} \Bigg| & \text{(assumption for contradiction).}
\end{align}
Then
\begin{subequations}
    \begin{align}
        \bigg\|F^{[2]}_3 \bigg\| &= \bigg\| \mathbb{I} \otimes F_3  + F_3 \otimes \mathbb{I} \bigg\| \\
         &= \bigg\| \mathbb{I} \otimes \left(\frac{-1}{2}\right) \left[ F_2 \cdots F_2 \right]  + \left(\frac{-1}{2}\right) \left[ F_2 \cdots F_2 \right] \otimes \mathbb{I} \bigg\| \\
         &=\left(\frac{Q}{2}\right) \bigg\| \mathbb{I} \otimes F_2   + F_2 \otimes \mathbb{I} \bigg\| \\
         &= \left(\frac{Q}{2}\right) \sum_j \bigg| \big[ \mathbb{I} \otimes F_2   + F_2 \otimes \mathbb{I} \big]_{\ell^* j} \bigg|,
    \end{align}
\end{subequations}
which implies that 
\begin{align}
        \left(\frac{Q}{2}\right) \sum_j \bigg| \big[ \mathbb{I} \otimes F_2   + F_2 \otimes \mathbb{I} \big]_{\ell^* j} \bigg| &> \left(\frac{Q}{2}\right) \sum_j \bigg| \big[ \mathbb{I} \otimes F_2   + F_2 \otimes \mathbb{I} \big]_{i^* j} \bigg| & \text{ (from assumption \eqref{eq:bad_assumption})}
\end{align}
and 
\begin{align}
        \sum_j \bigg| \big[ \mathbb{I} \otimes F_2   + F_2 \otimes \mathbb{I} \big]_{\ell^* j} \bigg| &>  \sum_j \bigg| \big[ \mathbb{I} \otimes F_2   + F_2 \otimes \mathbb{I} \big]_{i^* j} \bigg|
        =
        \sum_j \bigg| \left[ F^{[2]}_2 \right]_{i^* j} \bigg|,
\end{align}
which is a contradiction to our valid assumption that row index $i^*$ maximizes $\sum_j \Bigg| \left[F^{[2]}_2\right]_{ij} \Bigg|$.  Thus, assumption \eqref{eq:bad_assumption} is false and the same row indices $i^*$ that maximize $\sum_j \Bigg| \left[F^{[2]}_2\right]_{ij} \Bigg|$ also maximize $\sum_j \Bigg| \left[F^{[2]}_3\right]_{ij} \Bigg|$.  In other words,
\begin{equation}
    i^* = \argmax_i \sum_j \Bigg| \left[F^{[2]}_2\right]_{ij} \Bigg| = \argmax_i \sum_j \Bigg| \left[F^{[2]}_3\right]_{ij} \Bigg|.
    \label{eq:i_star}
\end{equation}

Then
\begin{subequations}
    \begin{align}
2 ( \| F_2 \| + \| F_3 \|) &= \bigg\| F^{[2]}_2 \bigg\| + \bigg\| F^{[2]}_3 \bigg\|\\
&= \max_i \sum_j \bigg| \left[ F^{[2]}_2 \right]_{ij} \bigg| + \max_i \sum_j \bigg| \left[ F^{[2]}_3  \right]_{ij} \bigg| \\
&= \sum_j \bigg| \left[ F^{[2]}_2 \right]_{i^*j} \bigg| + \sum_j \bigg| \left[ F^{[2]}_3  \right]_{i^*j} \bigg| \label{eq:sum_1}\\
&= \sum_j \bigg| \left[ F^{[2]}_2 \quad  F^{[2]}_3  \right]_{i^*j} \bigg| \label{eq:sum_2}\\
&= \sum_j \bigg| \left[ 0 \cdots 0 \quad F^{[2]}_2 \quad 0 \cdots 0 \quad F^{[2]}_3  \right]_{i^*j} \bigg| \\
&= \bigg\| \tilde{F}_2 \bigg\|.
\end{align}
\end{subequations}
The step from \eqref{eq:sum_1} to \eqref{eq:sum_2} hold due to \eqref{eq:i_star}.  Thus, we confirm that \eqref{eq:maxnorm_for_transfer_matrices_5} holds.

We now present the following theorem that places upper and lower bounds on the  maximum time for the convergence time.
\begin{theorem}
    \label{thm:max_convergence_time}
    Given $T_c$ specified as \eqref{eq:T_c_definition}, and $\| \phi_0 \|$ be the norm of the initial state, then for \eqref{eq:k-order-non-linear-ode} with $K=3$ (third-order Carleman truncation) we have the following following bounds on $T_c$:
    \begin{equation}
    \frac{1}{2\|\phi_0\| \left( \| F_2\| + \|F_3\|\right) + \|S\| + \|F_1\| + \|F_2 \| } \leq T_c \le \frac{1}{2\|\phi_0\| \left( \|F_2\| + \|F_3\| \right)}.
    \end{equation}
\end{theorem}

\begin{proof}
We begin by using the fact that $\frac{\gamma}{1+\gamma} \leq \ln(1+ \gamma)\leq \gamma$ $ \text{for }  \gamma>-1$.  When applied to the logarithm in $T_c$, we have 
\begin{subequations}
    \label{eq:proof_starting_algebra}
    \begin{align}
        \frac{\frac{1}{\beta_0}}{1+\frac{1}{\beta_0}} &\le \log \left(1 + \frac{1}{\beta_0} \right)\le \frac{1}{\beta_0}, \text{ and}\\
        \frac{1}{\| \phi_0 \| \| \tilde{F}_2 \| + \| \tilde{F}_1 \| } &\le \underbrace{\frac{1}{\| \tilde{F}_1 \|}\log \left( 1 + \frac{1}{\beta_0} \right)}_{T_c}\le \frac{1}{\| \phi_0 \| \| \tilde{F}_2 \|}.
    \end{align}
\end{subequations}
Direct numerical calculation in \cite{source_for_norm_A_anlysis} shows that for $\tau=0.6$,
\begin{subequations}
\begin{align}
    \|F_1\|&=14.012, \\
    \|F_2\|&=942.222, \text{ and}\\
    \|F_3\|&=12,720.0.
\end{align}
\end{subequations}
On general considerations, we have that $\|S\| =2$.
Using the \textit{reverse triangle inequality} and the \textit{triangle inequality} we have that
\begin{equation}
    \Big| \|S\| - \|F_1 \| \Big| \le \| S+ F_1 \| \le \|S\| + \| F_1 \|.
\end{equation}
This implies $12.012 \le \| S+ F_1 \| \le 16.012$. This bound also implies that 
\begin{equation}
   954.234 \le \Big| \|S\| - \| F_1 \| \Big| + \| F_2 \| \le \| S+ F_1 \| + \| F_2 \|.
\end{equation}
With the numerical calculations, we can evaluate the max function in \eqref{eq:maxnorm_for_transfer_matrices_3}, and conclude that
\begin{equation}
   \label{eq:first_bound_for_f1}
   \| \tilde{F_1} \| \le \| S+ F_1 \| + \| F_2 \| \le  \| S \| + \| F_1 \| + \| F_2 \|  \le 958.234.
\end{equation}
Substituting the terms from \eqref{eq:first_bound_for_f1} and \eqref{eq:maxnorm_for_transfer_matrices_5} into \eqref{eq:proof_starting_algebra} yields the stated $T_c$ bounds.
\end{proof}

\section{Costing the doubly-controlled quantum adder}
\label{sec:appendix_doubly_controlled_adder}

In costing out the doubly controlled quantum adder we shall use the controlled quantum adder from \cite{Munoz2019}. They explicitly break down its construction into 7 steps, 3 of which involve Toffoli ($\text{CTRL}^2$-X) gates that take the control qubit as an input. We shall construct the doubly controlled quantum adder by replacing these Toffoli gates by the gate $\text{CTRL}^3$-X. Using \cite{He2017}, we know that the relative phase version of this gate has 8 $T$-gates. We have the following.
\begin{enumerate}
    \item Step 2 has 1 Toffoli  which we replace with $\text{CTRL}^3$-X and thus get 8 $T$-gates.
    \item Step 4 has 2 Toffoli which we replace with $\text{CTRL}^3$-X and thus get 16 $T$-gates.
    \item Step 5 has $(n-1)$ Toffoli which we replace with $\text{CTRL}^3$-X and thus get  $8(n-1)$ $T$-gates.
\end{enumerate}
Therefore, we have ($8n+16$) $T$-gates from the introduced $\text{CTRL}^3$-X gates. 

Next we count the number of $T$-gates in the remaining Toffoli gates.
\begin{enumerate}
    \item  Step 3 (as is shown in \cite{Munoz2019}) has $4(n-1)$ $T$-gates respectively.
    \item Step 4 has 2 Toffoli gates and so a combined 14 $T$-gates.
    \item Step 5 has $(n-1)$ Toffoli gates and therefore $4(n-1)$ $T$-gates.
\end{enumerate}
Therefore, we have a total $8n+6$ $T$-gates from the remaining Toffoli gates.
This gives us a total of $16n+22$ $T$-gates.

\printbibliography
 
\end{document}